\documentclass{lmcs}
\pdfoutput=1

\usepackage{lastpage}
\lmcsdoi{19}{1}{20}
\lmcsheading{}{\pageref{LastPage}}{}{}%
{Apr.~29,~2021}{Mar.~15,~2023}{}


\usepackage[utf8]{inputenc}
\usepackage[T1]{fontenc}
\usepackage[french,UKenglish]{babel}

\usepackage{amsfonts,amssymb,bm}

\usepackage{url}
\usepackage{hyperref}

\usepackage{subfigure}
\usepackage{wrapfig}
\usepackage{float}

\usepackage{stmaryrd} %

\usepackage{xspace} %
\usepackage{xparse} %
\usepackage{xcolor,colortbl} %

\usepackage{hhline} %
\usepackage{multirow}
\usepackage{eurosym} %

\usepackage{tikz}
\usetikzlibrary{arrows,automata,shapes,calc,positioning,intersections,backgrounds}

\tikzset{>=latex,initial text=}

\tikzstyle{player1}=[draw,rounded rectangle, minimum size=7mm]
\tikzstyle{player2}=[draw,rectangle,minimum size=7mm]

\usepackage[ruled,boxed,vlined,linesnumbered]{algorithm2e}

\theoremstyle{plain}
\newtheorem*{claim}{Claim}
\newcommand*{\myproofname}{Proof of Claim}
\newenvironment{claimproof}[1][\myproofname]{\begin{proof}[#1]}{\end{proof}}

\theoremstyle{definition}
\newtheorem{hypothesis}{Hypothesis}

\newcommand\newmath[2]{\newcommand#1{\ensuremath{#2}\xspace}}
\newcommand\renewmath[2]{\renewcommand#1{\ensuremath{#2}\xspace}}
\newcommand\newmathope[2]{\newcommand#1{\ensuremath{\operatorname{#2}}\xspace}}

\newmath{\NL}{\mathsf{NL}}
\newmath{\coNL}{\mathsf{coNL}}
\renewmath{\P}{\mathsf{PTIME}}
\newmath{\NP}{\mathsf{NP}}
\newmath{\coNP}{\mathsf{coNP}}
\newmath{\PSPACE}{\mathsf{PSPACE}}
\newmath{\EXP}{\mathsf{EXPTIME}}
\newmath{\coNEXP}{\mathsf{coNEXPTIME}}
\newmath{\EXPSPACE}{\mathsf{EXPSPACE}}
\newmath{\NPSPACE}{\mathsf{NPSPACE}}
\newmath{\coNPSPACE}{\mathsf{coNPSPACE}}
\newmath{\coPSPACE}{\mathsf{coPSPACE}}

\newmath{\N}{\mathbb{N}}
\newmath{\Z}{\mathbb{Z}}
\newmath{\Q}{\mathbb{Q}}
\newmath{\R}{\mathbb{R}}

\newmath{\Zbar}{\Z_\infty}
\newmath{\Qbar}{\Q_\infty}
\newmath{\Rbar}{\R_\infty}

\newmath{\Qpos}{\Q_{{\geq}0}}
\newmath{\Rpos}{\R_{{\geq}0}}

\newmath{\Nspos}{\N_{{>}0}}
\newmath{\Rspos}{\R_{{>}0}}
\newmath{\Qspos}{\Q_{{>}0}}

\renewcommand\vec[1]{\boldsymbol{#1}}
\renewmath\leq{\leqslant}
\renewmath\geq{\geqslant}

\newmath\arginf{\text{arginf}}
\newmath\argsup{\text{argsup}}

\newmath{\OExp[2]}{\mathsf{PolyExp}(#1,#2)}
\newmath{\OPoly[1]}{\mathsf{Poly}(#1)}

\newmath{\Ball}{\mathcal{B}_\infty}

\newmath{\st}{\mid}
\newmath{\fract}{\mathsf{fract}}

\newmath{\struct[1]}{\langle#1\rangle}

\newmath{\czero}{\mathbf{0}}
\newmath{\cone}{\mathbf{1}}
\newmath{\cinfty}{\boldsymbol\infty}

\newmath{\game}{\mathcal G}
\newmath{\gamea}{\game_{\automaton}}

\newmath{\hgame}{\mathcal S(\game)}

\newmath{\Pl}{\mathsf{P}}
\newmath{\ConPl}{\mathsf{Ctrl}}
\newmath{\EnvPl}{\mathsf{Env}}

\newmath{\MinPl}{\mathsf{Min}}
\newmath{\MaxPl}{\mathsf{Max}}

\newmath{\initial}{\mathsf 0}
\newmath{\target}{\mathsf t}
\newmath{\deadlock}{\bot}

\newmath{\state}{s}
\newmath{\States}{S}
\newmath{\StatesPl}{\States_{\Pl}}
\newmath{\StatesCon}{\States_{\ConPl}}
\newmath{\StatesEnv}{\States_{\EnvPl}}
\newmath{\StatesMin}{\States_{\MinPl}}
\newmath{\StatesMax}{\States_{\MaxPl}}
\newmath{\StatesT}{\States_{\target}}
\newmath{\stateI}{\state_{\initial}}
\newmath{\stateT}{\state_{\target}}
\newmath{\StatesD}{\States_{\deadlock}}
\newmath{\StatesK}{\States_{\Kernel}}

\newmath{\labl}{a}
\newmath{\Labels}{\Sigma}
\newmath{\LabelsT}{\Labels_{\target}}

\newmath{\trans}{t}
\newmath{\Trans}{T}
\newmath{\TransT}{\Trans_{\target}}
\newmath{\TransK}{\Trans_{\Kernel}}

\newmath{\play}{\rho}

\newmath{\first}{\mathsf{first}}
\newmath{\last}{\mathsf{last}}

\newmath{\FPlays}{\mathsf{FPlays}}
\newmath{\FPlaysPl}{\FPlays^{\Pl}}
\newmath{\FPlaysCon}{\FPlays^{\ConPl}}
\newmath{\FPlaysEnv}{\FPlays^{\EnvPl}}
\newmath{\FPlaysMin}{\FPlays^{\MinPl}}
\newmath{\FPlaysMax}{\FPlays^{\MaxPl}}

\newmath{\strat}{\sigma}
\newmath{\stratpl}{\sigma_{\Pl}}
\newmath{\stratcon}{\sigma_{\ConPl}}
\newmath{\stratenv}{\sigma_{\EnvPl}}
\newmath{\stratmin}{\sigma_{\MinPl}}
\newmath{\stratmax}{\sigma_{\MaxPl}}
\newmath{\outcome}{\mathsf{play}}

\newmath{\automaton}{\mathcal A}

\newmath{\clock}{\mathsf x}
\newmath{\clocky}{\mathsf y}
\newmath{\clockz}{\mathsf z}
\newmath{\clockd}{\mathsf d}
\newmath{\Clocks}{\mathcal X}
\newmath{\clockbound}{M}
\newmath{\clockgranu}{N}

\newmath{\val}{\nu}
\newmath{\valI}{\val_{\initial}}
\newmath{\valnull}{\mathbf{0}}
\newmath{\ValSpace}{\Rpos^{\Clocks}}
\newmath{\ValSpaceBound}{[0,\clockbound)^{\Clocks}}
\newmath{\loc}{\ell}
\newmath{\Locs}{L}
\newmath{\LocsPl}{\Locs_{\Pl}}
\newmath{\LocsCon}{\Locs_{\ConPl}}
\newmath{\LocsEnv}{\Locs_{\EnvPl}}
\newmath{\LocsMin}{\Locs_{\MinPl}}
\newmath{\LocsMax}{\Locs_{\MaxPl}}
\newmath{\LocsT}{\Locs_{\target}}
\newmath{\locI}{\loc_{\initial}}
\newmath{\locT}{\loc_{\target}}

\newmath{\edge}{e}
\newmath{\Edges}{E}
\newmath{\guard}{g}
\newmath{\Guards}{\mathsf{Guards}}
\newmath{\Guardsnd}{\Guards^{\mathrm{nd}}}
\newmath{\Guardsd}{\Guards^{\mathrm{d}}}
\newmath{\reset}{\mathcal Y}

\newmath{\epath}{\pi}

\newcommand\sem[1]{\ensuremath{\llbracket#1\rrbracket}}
\newmath{\delay}{d}

\newmath{\timed}{\mathsf{Timed}}
\newmath{\untimed}{\mathsf{Untimed}}

\newmath{\rautomaton}{\mathcal{R}(\automaton)}
\newmath{\rgame}{\mathcal{R}(\game)}
\newmath{\rgamea}{\mathcal{R}(\gamea)}

\newmath{\Nrautomaton[1]}{\mathcal{R}_{#1}(\automaton)}
\newmath{\Nrgame[1]}{\mathcal{R}_{#1}(\game)}

\newmath{\reg}{r}
\newmath{\regI}{\reg_{\initial}}
\newmath{\Regions[2]}{\mathsf{Reg}(#1,#2)}
\newmath{\NRegions[3]}{\mathsf{Reg}_{#1}(#2,#3)}
\newmath{\Regs}{\Regions{\Clocks}{\clockbound}}
\newmath{\NRegs}{\NRegions{\clockgranu}{\Clocks}{\clockbound}}

\newmath{\RStates}{S}
\newmath{\RStatesMin}{\RStates_{\MinPl}}
\newmath{\RStatesMax}{\RStates_{\MaxPl}}
\newmath{\RStatesT}{\RStates_\target}
\newmath{\RStatesK}{\RStates_\Kernel}
\newmath{\rtrans}{t} %
\newmath{\RTrans}{T} %
\newmath{\RTransT}{\RTrans_\target}
\newmath{\RTransK}{\RTrans_\Kernel}

\newmath{\rpath}{\mathsf p}

\newmath{\zone}{Z}
\newmath{\zoneI}{\zone_{\initial}}
\newmath{\Zones}{\mathsf{Zones}}

\newmath{\Pre}{\mathsf{Pre}}
\newmath{\Post}{\mathsf{Post}}
\newmath{\ZReset}{\mathsf{Reset}}
\newmath{\Unreset}{\mathsf{Unreset}}
\newmath{\Guard}{\mathsf{Guard}}
\newmath{\Pretime}{\Pre}%
\newmath{\Posttime}{\mathsf{PostTime}}
\newmath{\PretimeZ}{\mathsf{PreTime}_{>0}}
\newmath{\PosttimeZ}{\mathsf{PostTime}_{>0}}

\newmath{\CPre}{\mathsf{CPre}}
\newmath{\Shrink}{\mathsf{Shrink}}
\newmath{\CPreD}{\mathsf{CPre}_{\delta}}
\newmath{\ShrinkD}{\mathsf{Shrink}_{\delta}}
\newmath{\PretimeD}{\mathsf{PreTime}_{>\delta}}
\newmath{\PosttimeD}{\mathsf{PostTime}_{>\delta}}

\newmath{\PretimeT}{\mathsf{PreTime}_{>t}}
\newmath{\PosttimeT}{\mathsf{PostTime}_{>t}}

\newcommand\doubledshrinknoarg[1][]{\ensuremath{\mathsf{shrink}
    \if#1\delta_{[-\delta,\delta]}\fi}}

\newmath{\Norm}{\mathsf{Norm}}
\newmath{\dbm}{\mathsf M}
\newmath{\dbmD}{\mathsf P}
\newmath{\dbmla}{{\prec}}
\newcommand{\dbmlaop}{\prec}
\newmath{\dbmleq}{{\leq}}

\newmath{\dbmlt}{{<}}

\newmath{\Bounds}{\mathsf{Bounds}}
\newmath{\PWBounds}{\mathsf{PWBounds}}

\newmath{\Reach}{\mathsf{Reach}}

\newmathope{\relation}{R}

\newmath{\intval}{\mathsf v}
\newmath{\corner}{\intval}

\newmath{\FOG}{\mathsf{FOG}}

\newmath{\cgame}{\Gamma(\game)}
\newmath{\Ncgame[1]}{\Gamma_{#1}(\game)}

\newmath{\cplay}{\overline\rho}

\newmath{\cgclock}{\mathsf X}
\newmath{\cgclocky}{\mathsf Y}
\newmath{\cgclocku}{\mathsf U}
\newmath{\cgclockv}{\mathsf V}
\newmath{\cgclockZ}{\cgclock_0}
\newmath{\clockgen}{\tau}
\newmath{\cgleft}{{\boldsymbol\ell}}
\newmath{\cgright}{\mathbf r}
\newmath{\cgany}{\mathbf k}

\newmath{\weight}{\mathsf{wt}}

\newmath{\weightT}{\mathsf{f}\weight}
\newmath{\weightC}{\weight_\Sigma}
\newmath{\wmax}{w_{\mathrm{max}}}
\newmath{\wmaxTimed}{\wmax}
\newmath{\globalmax}{\mathrm{MaxW}}

\newmath{\Val}{\mathsf{Val}}
\newmath{\uVal}{\overline{\Val}}
\newmath{\lVal}{\underline{\Val}}

\newmath{\ValIteOpe}{\mathcal{F}}
\newmath{\ValIteVec}{\vec V}

\newcommand\WTG{WTG\xspace}

\newmath{\Tube}{\mathsf{Tube}}

\newmath{\Kernel}{{\mathsf{K}}}
\newmath{\tgame}{\mathcal T(\game)}

\newmath{\splitnum}{\mathsf{Splits}}
\newmath{\lipconst}{\Lambda}
\newmath{\addconstbound}{\Kappa}

\renewcommand\paragraph[1]{\smallskip\noindent\textbf{#1.}}

\newcommand\regpt[3][black!80!white]{\node[fill=#1,circle,inner sep=0pt, minimum size=4pt] () at (#2,#3) {}}

\newcommand\regvl[3][black!80!white]{\path[fill=#1,rounded corners=1pt,inner sep=0pt] (#2cm-1.5pt,#3cm+2.5pt) -- (#2cm+1.5pt,#3cm+2.5pt) -- (#2cm+1.5pt,#3cm+1cm-2.5pt) -- (#2cm-1.5pt,#3cm+1cm-2.5pt) -- (#2cm-1.5pt,#3cm+2.5pt) -- (#2cm+1.5pt,#3cm+2.5pt)}

\newcommand\reghl[3][black!80!white]{\path[fill=#1,rounded corners=1pt,inner sep=0pt] (#2cm+2.5pt,#3cm-1.5pt) -- (#2cm+2.5pt,#3cm+1.5pt) -- (#2cm+1cm-2.5pt,#3cm+1.5pt) -- (#2cm+1cm-2.5pt,#3cm-1.5pt) -- (#2cm+2.5pt,#3cm-1.5pt) -- (#2cm+2.5pt,#3cm+1.5pt)}

\newcommand\regdl[3][black!80!white]{\path[fill=#1,rounded corners=1pt,inner sep=0pt] (#2cm+2pt,#3cm+4pt) -- (#2cm+4pt,#3cm+2pt) -- (#2cm+1cm-2pt,#3cm+1cm-4pt) -- (#2cm+1cm-4pt,#3cm+1cm-2pt) -- (#2cm+2pt,#3cm+4pt) -- (#2cm+4pt,#3cm+2pt)}

\newcommand\reght[3][black!80!white]{\path[fill=#1,rounded corners=1pt,inner sep=0pt] (#2cm+2.5pt,#3cm+6pt) -- (#2cm+2.5pt,#3cm+1cm-2.5pt) -- (#2cm+1cm-6pt,#3cm+1cm-2.5pt) -- (#2cm+2.5pt,#3cm+6pt) -- (#2cm+2.5pt,#3cm+1cm-2.5pt)}

\newcommand\reglt[3][black!80!white]{\path[fill=#1,rounded corners=1pt,inner sep=0pt] (#2cm+6pt,#3cm+2.5pt) -- (#2cm+1cm-2.5pt,#3cm+2.5pt) -- (#2cm+1cm-2.5pt,#3cm+1cm-6pt) -- (#2cm+6pt,#3cm+2.5pt) -- (#2cm+1cm-2.5pt,#3cm+2.5pt)}

\newcommand\resp{respectively}

\def\ie{{\em i.e.}}

\def\eg{{\em e.g.}}

\begin{document}

\title{Optimal controller synthesis for timed systems}

\author[D.~Busatto-Gaston]{Damien Busatto-Gaston\lmcsorcid{0000-0002-7266-0927}}[a]
\author[B.~Monmege]{Benjamin Monmege\lmcsorcid{0000-0002-4717-9955}}[b]
\author[P.-A.~Reynier]{Pierre-Alain Reynier}[b]

\address{Univ Paris Est Creteil, LACL, F-94010 Creteil, France}
\email{damien.busatto-gaston@u-pec.fr}

\address{Aix Marseille Univ, LIS, CNRS, Marseille, France}
\email{benjamin.monmege@univ-amu.fr, pierre-alain.reynier@univ-amu.fr}

\thanks{We thank the reviewers of this article, as well as the two conference versions from which most of the results come from, for their valuable feedback.}

\begin{abstract}

   Weighted timed games are zero-sum games played by two players on a
   timed automaton equipped with weights, where one player wants to
   minimise the cumulative weight while reaching a target.
    Used in a reactive synthesis perspective, this quantitative
    extension of
   timed games allows one to measure the quality of controllers in
    real-time systems.
   Weighted timed games are notoriously difficult and quickly
   undecidable, even when restricted to non-negative weights.
   For non-negative weights, a fragment of weighted timed games that can be analysed
   has been introduced by Bouyer, Jaziri and Markey in 2015.
   Though the value problem is undecidable even in this fragment, the authors show
   how to approximate the value by considering regions with a refined
   granularity.
   In this work, we extend this class to incorporate negative weights,
   allowing one to model energy for instance, and prove that the value
   can still be approximated, with the same complexity (provided that clocks are bounded).  A
   restriction also allows us to obtain a class of decidable weighted
   timed games with negative weights and an arbitrary number of
   clocks. In addition, we show that a symbolic algorithm, relying on
   the paradigm of value iteration, can be used as an
   approximation/computation schema over these classes. We also
   consider the special case of untimed weighted games, where the same
   fragments are solvable in polynomial time: this contrasts with the
   pseudo-polynomial complexity, known so far, for weighted games
   without restrictions.

\end{abstract}
\maketitle

\section{Introduction}

We are interested in the design of programs sensitive to real-time,
where keeping track of how much time elapses between the decisions
taken by the program is required to differentiate the good and bad
behaviours of the system. This is a common requirement for embedded
systems, as they interact with the real world. The design of such
programs is a notoriously difficult problem, because they must take
care of delicate timing issues, and are difficult to debug \textit{a
posteriori}. In order to ease the design of real-time software, it
appears important to automatise the process by using formal methods.
The situation may be modelled into a \emph{timed
  automaton}~\cite{AluDil94}, namely a transition system equipped with
real-valued variables, called clocks, evolving with a uniform
rate. Transitions are equipped with timing constraints expressed over
the clocks, and may only be taken when these constraints are met. Finally, clocks may be reset so as to keep track of how much time has elapsed since a particular transition was taken.

In order to check whether the real-time system fulfils a certain
specification, one determines whether there exists an accepting
execution in the timed automaton: this is the classical
\emph{model-checking problem}. A simple, yet realistic specification
asks that a target state is reached at some point. It has been proven
in early works~\cite{AluDil94} that the reachability problem on timed
automata is \PSPACE-complete: in particular, the \PSPACE upper bound is obtained by partitioning the state space into a finite number of
\emph{regions}.  While optimal from a theoretical complexity point of
view, practical tools tend to favour efficient \emph{symbolic}
algorithms for solving  such model-checking problems like reachability, that use \emph{zones} instead
of regions, as they allow an on-demand partitioning of the state
space. This leads to much better performances, as witnessed by
successful model-checking tools like
\textsc{Uppaal}~\cite{Larsen1997}, \textsc{Kronos}~\cite{BFb0028779},
or \textsc{TChecker}~\cite{TChecker, HerbreteauSW10}.

Instead of verifying a system, we can try to synthesise one
automatically. One setting consists in using \emph{game
  theory}. The set of configurations of the system is then partitioned into two
\emph{players}: a controller whose role is to fulfil the
specification, and an antagonistic environment. The goal becomes to
find automatically a good \emph{strategy} for the controller, which is
called the \emph{controller synthesis problem}. In timed systems,
both players alternatively choose transitions and delays in a timed
automaton: this is called a \emph{timed game}. Strategies of players
are recipes dictating how to play (timing delays and transitions to
follow). In this ambitious setting, we will focus on reachability
objectives, and we are thus looking for a strategy of the controller
so that the target is reached no matter how the environment
plays. Reachability timed games are decidable~\cite{AsaMal99}, and
\EXP-complete~\cite{JurTri07}.

\begin{figure}
  \centering

\begin{tikzpicture}[node distance=3cm]

  \node[player1,initial] (A) {$A$};
  \node[player1,right of=A,node distance=5cm] (C) {$C$};
  \node[player1,below of=A] (B) {$B$};
  \node[player1,below of=C] (D) {$D$};
  \draw[->]
  (A) edge node[above] {road} node[below]{$1h$} (C)
  (D) edge node[above] {road} node[below]{$1h$} (B)
  (A) edge[bend left=15] node[right,yshift=.25cm] {road} node[right,yshift=-.25cm]{$[2h,3h]$} (B)
  (A) edge[bend right=15] node[left,yshift=.25cm] {highway} node[left,yshift=-.25cm]{$1h$} (B)
  (C) edge[bend left=15] node[right,yshift=.25cm] {highway} node[right,yshift=-.25cm]{$0.5h$} (D)
  (C) edge[bend right=15] node[left,yshift=.25cm] {road} node[left,yshift=-.25cm]{$[0.5h,1h]$} (D)
  ;
\end{tikzpicture}

  \caption{A ride-sharing decision diagram.}
  \label{fig:intro-weighted}
\end{figure}

If the controller has a winning strategy in a given reachability timed
game, several such winning strategies could exist.  Weighted
extensions of these games have been considered in order to measure the
quality of the winning strategy for the controller~\cite{BCFL04}.
This means that the game now takes place over a \emph{weighted (or
  priced) timed automaton}~\cite{BehFeh01,AluLa-04}, where edges are
equipped with weights, and locations with rates of weights (the cost
is then proportional to the time spent in this location, with the rate
as proportional coefficient). The optimal reachability problem asks
what is the lowest cumulative weight that the controller can guarantee
for reaching a target from a given initial state, against any
decision made by the antagonistic environment. The controller is
therefore the minimiser player, while the environment wants to
maximise the weight accumulated along the execution. This lowest
weight is called the \emph{value} of the game, and computing it can be
seen as a natural generalisation of the classical shortest path
problem in a weighted graph to the case of two-player timed games.

\begin{figure}
  \centering

\begin{tikzpicture}[node distance=3cm]

  \node[player1,initial] (A) {$A$};
  \node[player2,below right of=A,node distance=2cm] (ABr) {$\text{rd}_{A\rightarrow B}$};
  \node[player2,below left of=A,node distance=2cm] (ABh) {$\text{hw}_{A\rightarrow B}$};
  \node[player2,right of=A] (ACr) {$\text{rd}_{A\rightarrow C}$};
  \node[player1,accepting,below of=A,node distance=4cm] (B) {$B$};
  \node[player1,right of=ACr,node distance=5cm] (C) {$C$};
  \node[player2,below left of=C,node distance=2cm] (CDr) {$\text{rd}_{C\rightarrow D}$};
  \node[player2,below right of=C,node distance=2cm] (CDh) {$\text{hw}_{C\rightarrow D}$};
  \node[player1,below of=C,node distance=4cm] (D) {$D$};
  \node[player2,right of=B,node distance=5cm] (DBr) {$\text{rd}_{D\rightarrow B}$};

  \node()[above of=ACr,node distance=6mm,color=blue]{\EUR{$1$}/h};
  \node()[below of=DBr,node distance=6mm,color=blue]{\EUR{$1$}/h};
  \node()[left of=ABh,node distance=14mm,color=blue]{\EUR{$2$}/h};
  \node()[right of=ABr,node distance=14mm,color=blue]{\EUR{$1$}/h};
  \node()[right of=CDh,node distance=14mm,color=blue]{\EUR{$0$}/h};
  \node()[left of=CDr,node distance=15mm,color=blue]{\EUR{$-1$}/h};
  \draw[->]
  (A) edge node[above] {$\clock=0$}  (ACr)
  (ACr) edge node[above] {$\clock= 1$} node[below]{$\clock:=0$} (C)
  (A) edge[bend right=10] node[left] {$\clock=0$} node[right,yshift=-1mm,color=red] {\EUR{$1$}} (ABh)
  (ABh) edge[bend right=10] node[left] {$\clock= 1$} (B)
  (A) edge[bend left=10] node[right] {$\clock=0$} (ABr)
  (ABr) edge[bend left=10] node[right] {$2\leq \clock\leq 3$} (B)
  (C) edge[bend right=10] node[right,xshift=-1mm,yshift=-1mm] {$\clock=0$} (CDr)
  (CDr) edge[bend right=10] node[left,xshift=-1mm,yshift=0.25cm] {$\frac{1}{2}\leq\clock\leq 1$} node[left,yshift=-0.25cm]{$\clock:=0$} (D)
  (C) edge[bend left=10] node[above right] {$\clock=0$} node[right,xshift=1mm,yshift=-1mm,color=red] {\EUR{$1$}} (CDh)
  (CDh) edge[bend left=10] node[right,xshift=1mm,yshift=0.25cm] {$\clock= \frac{1}{2}$} node[right,yshift=-0.25cm]{$\clock:=0$} (D)
  (D) edge node[above] {$\clock=0$}  (DBr)
  (DBr) edge node[above] {$\clock= 1$} (B)
  ;
\end{tikzpicture}

  \caption{A weighted timed game modelling \figurename~\ref{fig:intro-weighted}. Transitions are labelled by a guard over clock \clock and by the reset $\clock:=0$ when needed.
  The cost of waiting in a state is displayed in blue, the cost of taking a transition
  is in red. States and transitions without costs have a weight of \EUR{$0$}}
  \label{fig:intro-wtg}
\end{figure}
\begin{exa}
  As a motivating example, we present a ride-sharing scenario. As a
  driver, we wish to travel from point $A$ to point $B$, and must
  choose between several options, as displayed in
  \figurename~\ref{fig:intro-weighted}.  We can use a direct road, and
  reach $B$ in two to three hours, or a highway that lets us reach
  our destination in one hour.  Alternatively, we can make a detour:
  another traveller is waiting at point $C$, and wishes to reach point
  $D$.  For this portion too, a faster highway is available.

  While all four possible paths satisfy the objective "reaching $B$",
  we want to select the one that lets us spend as little money as possible
  for the trip.
  The cost of each path depends on several factors. There are fixed entry fees
  (of~\EUR{$1$}) for the highways, and we need to keep track of fuel consumption,
  as the rate at which fuel is used differs in roads and highways.
  Thus, we say that roads cost~\EUR{$1$} per hour,
  while highways cost~\EUR{$2$} per hour.
  Moreover, if we share the portion from~$C$ to~$D$, the other traveller
  will pay for his trip (at a rate of~\EUR{$2$}/h),
  and that can lower our costs.
  A shared road therefore costs us~\EUR{$-1$} per hour
  (negative rate means we are making a profit), while a shared highway
  costs~\EUR{$0$}/h. %

  The situation can be modelled as a weighted timed game,
  displayed in \figurename~\ref{fig:intro-wtg}.
  The controller chooses delays and transitions in circle states,
  while the environment controls the square ones. For example,
  if we choose to use the direct road from $A$ to $B$,
  we go (immediately) to state $\text{rd}_{A\rightarrow B}$,
  and stay there until going to state $B$. This requires letting between
  two and three hours elapse in $\text{rd}_{A\rightarrow B}$, with
  a cost of \EUR{$1$}/h. The delay is chosen by the environment,
  as it depends on external influences like traffic density.

  In this example, the optimal strategy is to share the road from $C$ to $D$.
  This lets us ensure a total weight of at most \EUR{$1.5$}:
  going from $A$ to $C$ costs \EUR{$1$} in the worst case;
  going from $D$ to $B$ similarly costs at most \EUR{$1$};
  and sharing the trip from $C$ to $D$ is guaranteed to bring us
  at least \EUR{$0.5$}.
\end{exa}

  While deciding whether the value of a weighted timed automata is lower than a given threshold has
  been shown to be \PSPACE-complete~\cite{BouBri07} (\ie~the same
  complexity as the non-weighted version, that is the reachability problem), the same problem is
  known to be undecidable in weighted timed games~\cite{BBR05}. This justifies the study of
  restrictions in order to regain decidability, the first and most
  interesting one being the class of strictly non-Zeno cost with only
  non-negative weights (in edges and locations)~\cite{BCFL04}:
  this hypothesis states that every execution of the timed automaton
  that follows a cycle of the region abstraction has a weight far from 0
  (in the interval $[1,+\infty)$, for instance).

  Less is known for weighted timed games in the presence of negative
  weights in edges and locations. In particular, no results existed
  before this work for a class that does not restrict the number of
  clocks of the timed automaton to 1.  However, several clocks are needed to keep track of how much time elapsed since multiple reference points, and negative weights are
  particularly interesting from a modelling perspective, for instance
  in case weights represent the consumption level of a resource
  (money, energy\dots) with the possibility to spend and gain some
  resource.  We thus introduce a generalisation of the strictly
  non-Zeno cost hypothesis in the presence of negative weights, that
  we call \emph{divergence}.  Under the hypothesis that clocks are bounded, we show the decidability of the class of
  divergent weighted timed games for the optimal synthesis
  problem. The technique uses
  a \emph{value iteration}
  procedure to solve weighted timed games for a bounded horizon,
  \ie~when controller has a fixed number of steps to reach his
  targets.  It follows closely the framework of \cite{AluBer04}, but
  is more symbolic and allows for negative weights. We then show that
  optimal strategies in divergent weighted timed games can be
  restricted to a bounded horizon, that matches the one obtained in
  the non-negative case from the study of~\cite{BCFL04}.

The techniques providing these decidability results cannot be extended
if the conditions are slightly relaxed. For instance, if we add the
possibility for an execution of the timed automaton following a cycle
of the region automaton to have weight \emph{exactly 0}, the decision
problem is known to be undecidable~\cite{BJM15}, even with
non-negative weights only.  For this extension, in the presence of
non-negative weights only, it has been proposed an approximation
schema to compute arbitrarily close estimates of the optimal weight
that the controller can guarantee~\cite{BJM15}.  To this end, the
authors consider regions with a refined granularity so as to control
the precision of the approximation.

Our contribution on the approximation front is two-fold. We extend the
class considered in \cite{BJM15} to the presence of negative weights,
and provide an approximation schema for the resulting class of
\emph{almost-divergent games} (again under the hypothesis that clocks are bounded). We then show that the approximation can
be obtained using a symbolic computation, that avoids an \textit{a priori}
refinement of regions. Table~\ref{tab:tbl_timed_wtg} summarises our
results on weighted timed games.

\begin{table}[tbp]
  \centering

\begin{tabular}{c|c|c|c|}
  \hhline{~|*{3}{-|}}
  \multirow{2}{*}{
  Timed
  } &
  \multicolumn{3}{c|}{\cellcolor{gray!20}weights in \N}\\
  \hhline{~|*{3}{-|}}
  & \cellcolor{gray!20}divergent
  & \cellcolor{gray!20}almost-divergent
  & \cellcolor{gray!20} all \WTG \\
  \hline
  \multicolumn{1}{|c|}{\cellcolor{gray!20}Value pb.} &
  \begin{tabular}{c}
  $2$-\EXP \\
  \cite{BCFL04}+\cite{AluBer04} \end{tabular} &
  \begin{tabular}{c}
    undecidable
    \\ \cite{BJM15} \end{tabular} &
  \begin{tabular}{c}
    undecidable
    \\ \cite{BBR05} \end{tabular} \\ \hline
  \multicolumn{1}{|c|}{\cellcolor{gray!20}Approx. pb.} &
  \cellcolor{gray!10}/ &
  \begin{tabular}{c}
  $2$-\EXP \\
  \cite{BJM15}+\cite{AluBer04} \end{tabular} & ? \\ \hline
  \multicolumn{1}{|c|}{\cellcolor{gray!20}Value $+\infty$} &
  \multicolumn{3}{c|}{\begin{tabular}{c}
  \EXP-complete \\
  \cite{BCFL04}+\cite{JurTri07} \end{tabular}} \\ \hline%
  \multicolumn{4}{c} {} \\
  \hhline{~|*{3}{-|}}
  \multirow{2}{*}{
  Timed
  } &
  \multicolumn{3}{c|}{\cellcolor{gray!20}weights in \Z}\\
  \hhline{~|*{3}{-|}}
  & \cellcolor{gray!20}divergent
  & \cellcolor{gray!20}almost-divergent
  & \cellcolor{gray!20} all \WTG \\
  \hline
  \multicolumn{1}{|c|}{\cellcolor{gray!20}Value pb.} &
  \begin{tabular}{c}
  $3$-\EXP \\
  \EXP-hard \\
  Thm.~\ref{thm:div_wtg} \end{tabular} &
  \begin{tabular}{cc}
    undecidable
    \\\cite{BJM15} \end{tabular} &
  \begin{tabular}{cc}
    undecidable
    \\\cite{BGNK+14} \end{tabular} \\ \hline
  \multicolumn{1}{|c|}{\cellcolor{gray!20}Approx. pb.} &
  \cellcolor{gray!10}/ &
  \begin{tabular}{c}
  $3$-\EXP \\
  Thm.~\ref{thm:almost-div} \end{tabular} & ? \\ \hline
  \multicolumn{1}{|c|}{\cellcolor{gray!20}Value $-\infty$} &
  \multicolumn{2}{c|}{\begin{tabular}{c}
  \EXP-complete \\
  Prop.~\ref{lm:-infty-main},\ref{prop:-infty}
   \end{tabular}} &
  \begin{tabular}{cc}
    undecidable
    \\Prop.~\ref{prop:-infty_undec} \end{tabular} \\ \hline
  \multicolumn{1}{|c|}{\cellcolor{gray!20}Value $+\infty$} &
  \multicolumn{3}{c|}{\begin{tabular}{c}
  \EXP-complete \\
  Prop.~\ref{prop:no-plus-infty-wt} \end{tabular}} \\ \hline
  \multicolumn{1}{|c|}{\cellcolor{gray!20}Membership} &
  \multicolumn{2}{c|}{\begin{tabular}{c}\PSPACE-complete\\
  Thm.~\ref{thm:div_wtg},\ref{thm:almost-div}\end{tabular}} &
  \cellcolor{gray!10}/ \\\hline
\end{tabular}

  \caption{Solving weighted timed games with arbitrary weights}
  \label{tab:tbl_timed_wtg}
\end{table}

The classes of weighted timed games that we study induce
interesting classes of \emph{finite weighted games} when there are no
clocks, that can be solved with a lower (polynomial) complexity than
arbitrary weighted games, see Table~\ref{tab:tbl_untimed_wtg}.

\begin{table}[tbp]
  \centering

\begin{tabular}{c|c|c|c|c|}
  \hhline{~|*{4}{-|}}
  \multirow{2}{*}{
  Untimed
  } &
  \cellcolor{gray!20}weights in \N
  & \multicolumn{3}{c|}{\cellcolor{gray!20}weights in \Z} \\
  \hhline{~|*{4}{-|}}
  & \cellcolor{gray!20} all games
  & \cellcolor{gray!20}divergent
  & \cellcolor{gray!20}almost-div.
  & \cellcolor{gray!20} all games
  \\
  \hline
  \multicolumn{1}{|c|}{\cellcolor{gray!20}Value pb.} &
  \begin{tabular}{c} \P-complete \\\cite{KBB+08}, Section~\ref{sec:poly-lower-bound-untimed} \end{tabular} &
  \multicolumn{2}{c|}{\begin{tabular}{c} \P-complete \\
  Thm.~\ref{thm:div_untimed},\ref{thm:almost-div_untimed} \end{tabular}} &
  \begin{tabular}{c}pseudo-poly.\\ \cite{BGHM16} \end{tabular} \\\hline
  \multicolumn{1}{|c|}{\cellcolor{gray!20}Value $-\infty$} &
  \cellcolor{gray!10}/ &
  \multicolumn{2}{c|}{\begin{tabular}{c} \P-complete \\
  Section~\ref{sec:-infty} \end{tabular}} &
  \begin{tabular}{c}pseudo-poly.\\ \cite{BGHM16} \end{tabular} \\\hline
  \multicolumn{1}{|c|}{\cellcolor{gray!20}Value $+\infty$} &
  \begin{tabular}{c} \P-complete \\\cite{KBB+08}, Section~\ref{sec:poly-lower-bound-untimed} \end{tabular} &
  \multicolumn{3}{c|}{\begin{tabular}{c} \P-complete \\Prop.~\ref{prop:no-plus-infty-wt}, \cite{BGHM16} \end{tabular}} \\\hline
  \multicolumn{1}{|c|}{\cellcolor{gray!20}Membership} &
  \cellcolor{gray!10}/ &
  \multicolumn{2}{c|}{
  \begin{tabular}{c}
    \NL-complete~(unary wt.)
    \\ \P~(binary wt.)
    \\ Thm.~\ref{thm:div_untimed},\ref{thm:almost-div_untimed} \end{tabular}
    } &
  \cellcolor{gray!10}/ \\\hline
\end{tabular}

  \caption{Solving finite weighted games with arbitrary weights}
  \label{tab:tbl_untimed_wtg}
\end{table}

Other types of payoffs than the cumulative weight we study (\ie~total
payoff) have been considered for weighted timed games.  For instance,
energy and mean-payoff timed games have been introduced
in~\cite{BreCas14}.  They are also undecidable in general.
Interestingly, a subclass called \emph{robust timed games}, not far
from our divergence hypothesis, admits decidability results for other
payoffs.  A weighted timed game is robust if, to say short, every
simple cycle (cycle without repetition of a state) has
a weight that is either non-negative or less than $-\varepsilon$, for a fixed constant $\varepsilon>0$ (the same constant for all cycles). Solving robust
timed game can be done in \EXPSPACE, and is \EXP-hard.  Moreover,
deciding if a weighted timed game is robust has complexity
$2$-\EXPSPACE\ (and $\coNEXP$-hard).  This contrasts with our \PSPACE\
results for the membership problem.\footnote {While our divergent
  games have a similar definition, the two classes are incomparable.}  It
has to be noted that extending our techniques and results to the case
of robust timed games may not be possible: indeed,
with weights in $\N$ every game is robust, making the value
problem for this class undecidable~\cite{BBR05},
with no known approximation method.

This article is structured as follows. After some preliminary
definitions of the models we study in Section~\ref{sec:prelim}, we
introduce the divergent and almost-divergent classes of weighted timed
games and state our results in Section~\ref{sec:results}. We study
structural properties of these classes in order to be able to decide
their membership in Section~\ref{sec:membership}. We then start
solving these games by studying the qualitative problem of deciding
infinite values, using the crucial and central notion of kernels, in
Section~\ref{sec:kernels-infinity}. The approximation schema for
almost-divergent timed games is presented in
Section~\ref{sec:unfolding} where we introduce and prove the
correctness of a semi-unfolding of the game. Section~\ref{sec:acyclic}
explains how to compute the value of this semi-unfolding in the
special case of a tree-shaped game (without kernels) as previously
studied in~\cite{AluBer04}. We then generalise our study in
Section~\ref{sec:computing} to take kernels into account; we also
consider the special case of divergent timed games. The first
approximation/computation we propose requires to first compute
strongly connected components and to build the semi-unfolding of the timed games,
which might be highly prohibitive as usual in the timed
setting. Therefore, we present a more symbolic algorithm in
Section~\ref{sec:symbolic-wtg}. In Section~\ref{sec:strategies}, we deduce
from our study an algorithm to synthesise finite-memory almost-optimal
strategies in divergent timed games. Finally, we treat the special
case of almost-divergent untimed games in
Section~\ref{sec:solving-wg}, where polynomial time algorithms are
provided in this case (to be compared with the pseudo-polynomial
complexity in general known so far).

This work is based on works published in \cite{BMR18,BMR17}: compared to these preliminary works, this article contains full and corrected proofs, as well as a detailed study of the acyclic case in Section~\ref{sec:acyclic}, used as a building block for our study. We have also extended the work in the untimed case (Section~\ref{sec:solving-wg}) to incorporate almost-divergent weighted games.

\section{Preliminaries}\label{sec:prelim}

\subsection{Modelling real-time constraints}

We first introduce notions that let us express timing constraints,
useful to then define weighted timed games, and introduce classical
tools for their study.

Let $\Clocks=\{\clock_1,\dots,\clock_n\}$ be a finite, non-empty set
of variables called clocks.  A \emph{valuation}
$\val\colon\Clocks\to\Rpos$ is a mapping from clocks to non-negative
real numbers, such that $\val(\clock_1),\dots,\val(\clock_n)$ are
called the coordinates of $\val$.  Equivalently, $\val$ can be seen as
a point in space $\ValSpace$.  We denote $\valnull$ the valuation such
that for all $\clock\in\Clocks$, $\val(\clock)=0$.  Given a real
number $\delay\in\R$, we define $\val+\delay$ as the valuation such that
$\forall \clock\in\Clocks, (\val+\delay)(\clock)=\val(\clock)+\delay$
if it exists.\footnote {if $\delay$ is negative, $\val+\delay$ may not
  belong to $\ValSpace$} If \delay is non-negative, we say that we
performed a time elapse of delay $\delay$.  The time-successors of
$\val$ are the valuations $\val+\delay$ with $\delay\geq 0$.
Similarly, we refer to all $\val+\delay$ in $\ValSpace$ with
$\delay\leq 0$ as time-predecessors of $\val$.  The set of points that
are either time-predecessors or time-successors of a valuation $\val$
form the unique diagonal line in $\ValSpace$ that contains $\val$.  If
$\reset$ is a subset of $\Clocks$, we define $\val[\reset:=0]$ as the
valuation such that
$\forall \clock\in\reset, (\val[\reset:=0])(\clock)=0$ and
$\forall \clock\in\Clocks\backslash\reset,
(\val[\reset:=0])(\clock)=\val(\clock)$.  This operation is called a
reset of clocks $\reset$.

We extend those notions to sets of valuations in a natural way.  The
set of time-successors of $Z\subseteq\ValSpace$, denoted
$\Posttime(Z)$, contains the valuations that are time-successors of
valuations in $Z$.  The reset of $Z\subseteq\ValSpace$ by $\reset$,
denoted $Z[\reset:=0]$, contains the valuations $\val[\reset:=0]$ such
that $\val\in Z$.

The term \emph{atomic constraint} will refer to an affine inequality in
one of the following forms:
\begin{itemize}
\item A strict (\resp~non-strict) \emph{non-diagonal} atomic
  constraint over clock $\clock\in\Clocks$ and constant $c\in\Q$ is an
  inequality of the form $\clock\bowtie c$ with
  ${\bowtie}\in\{{>},{<}\}$ (\resp~${\bowtie}\in\{{\geq},{\leq}\}$).
\item A strict (\resp~non-strict) \emph{diagonal} atomic constraint
  over clocks $\clock$ and $\clocky\in\Clocks$ and constant $c\in\Q$
  is an inequality of the form $\clock-\clocky\bowtie c$ with
  ${\bowtie}\in\{{>},{<}\}$ (\resp~${\bowtie}\in\{{\geq},{\leq}\}$).
\end{itemize}
Let $\top$ and $\bot$ denote two special atomic constraints, defined
as $\clock\geq 0$ and $\clock<0$ for an arbitrary $\clock\in\Clocks$.
A \emph{guard} \guard over \Clocks is a finite conjunction of atomic
constraints over clocks in \Clocks.  In particular, guards let us
define $\clock=c$ as shorthand for $\clock\leq c \land \clock\geq c$,
and $c_1<\clock<c_2$ as shorthand for $\clock>c_1 \land \clock<c_2$.
A guard is said strict (\resp~non-strict, diagonal, non-diagonal) if
all of its atomic constraints are strict (\resp~non-strict, diagonal,
non-diagonal).  $\Guards(\Clocks)$ denotes the set of all guards over
\Clocks, and $\Guardsnd(\Clocks)$ the subset of non-diagonal guards.
For all constants $c\in\Q$ and
${\bowtie}\in\{{\geq},{\leq},{>},{<}\}$, we say that valuation
$\val\in\ValSpace$ satisfies the atomic constraint $\clock\bowtie c$
(\resp~$\clock-\clocky\bowtie c$), and write
$\val\models \clock\bowtie c$
(\resp~$\val\models \clock-\clocky\bowtie c$), if
$\val(\clock)\bowtie c$ (\resp~$\val(\clock)-\val(\clocky)\bowtie c$).
We say that valuation $\val\in\ValSpace$ satisfies guard \guard, and
write $\val\models \guard$, if $\val$ satisfies all atomic constraints
in \guard.  For $\guard\in\Guards(\Clocks)$, let $\sem{\guard}$ denote
the set of all $\val\in\ValSpace$ such that $\val\models\guard$.
  Such sets are called \emph{zones} and form convex polyhedra of \ValSpace.
  A guard $\guard$ is said satisfiable when the zone $\sem{\guard}$ is
  non-empty, and a zone is called rectangular when the associated
  guard is non-diagonal.  The universal zone refers to
  $\sem{\top}=\ValSpace$ and the empty zone refers to
  $\sem{\bot}=\emptyset$.  Guard $\overline \guard$ is the closed
  version of a satisfiable guard $\guard$ where every strict
  constraint of comparison operator ${<}$ or ${>}$ is replaced by its
  non-strict version ${\leq}$ or ${\geq}$.  The zone
  $\sem{\overline \guard}$ is the topological closure of
  $\zone=\sem{\guard}$, and is also denoted $\overline \zone$.

  We will restrict ourselves to \emph{bounded} clocks, so that
  clocks valuations will be point in $\ValSpaceBound$ instead of $\ValSpace$, for some upper bound $\clockbound\in\Nspos$.\footnote{%
  This assumption will be discussed and formalised later on in Hypothesis~\ref{hyp:bounded}.}
  We denote by $\Guards(\Clocks,\clockbound)$ (\resp~$\Guardsnd(\Clocks,\clockbound)$) the set of guards (\resp~the set of non-diagonal guards) over $\Clocks$ bounded by $\clockbound$, in the sense that $|c|\leq \clockbound$ for all constants $c\in\Q$ appearing in the atomic constraints of the guards.

\subsection{Regions}
We will rely on the crucial notion of regions, as introduced in
the seminal work on timed automata \cite{AluDil94}.
Let $\Q_\clockgranu=\{a/\clockgranu\st a\in\Z\}$ be the set of
rational numbers of granularity $1/\clockgranu$ for
a fixed $\clockgranu\in\Nspos$.
Given a finite set of rational numbers $\mathcal S\subseteq\Q$, $\mathcal S$ is
said to be of granularity $1/\clockgranu$ if $\mathcal S\subseteq\Q_\clockgranu$.
Such an $\clockgranu$ always exists, and one can find the smallest one
by decomposing elements of $\mathcal S$ as irreducible fractions $c/c'$
with $c\in\Z$, $c'\in\Nspos$ and use the least common multiple of all $c'$ as $\clockgranu$.
A guard \guard is said to be of granularity $1/\clockgranu$
if all constants in the atomic constraints of \guard
form a set of granularity $1/\clockgranu$.  A zone is of granularity $1/\clockgranu$
if it can be described by a guard of granularity $1/\clockgranu$.
Let $\Guards_{\clockgranu}(\Clocks,\clockbound)$ denote the set of guards over \Clocks
bounded by $\clockbound$ and of granularity $1/\clockgranu$, and let
$\Guardsnd_{\clockgranu}(\Clocks,\clockbound)$ denote the non-diagonal ones.
Given a finite set of guards $G\subseteq\Guardsnd(\Clocks,\clockbound)$,
we can find $\clockgranu$ such that $G\subseteq\Guardsnd_{\clockgranu}(\Clocks,\clockbound)$,
by denoting $\mathcal S\subseteq\Q$ the set of constants used in atomic constraints of $G$
and using $\clockgranu$ the smallest integer such that $\mathcal S$ is
of granularity $1/\clockgranu$.
For all $a\in\Rpos$, $\lfloor a\rfloor\in\N$ denotes the integral part of $a$,
and $\fract(a)\in[0,1)$ its fractional part, such that $a=\lfloor a\rfloor+\fract(a)$.

\begin{defi}\label{def:regions}
With respect to the set \Clocks of clocks,
a granularity $\clockgranu\in\Nspos$ and an upper bound $\clockbound\in\Nspos$ on the valuation of clocks,
we partition the set $\ValSpaceBound$ of valuations into \emph{$1/\clockgranu$-regions}. We denote by $\NRegs$ the set of $1/\clockgranu$-regions bounded by $\clockbound$. Each such region is characterised by a pair $(\iota, \beta)$ where $\iota\colon \Clocks\to [0,\clockbound)\cap \Q_\clockgranu$
and $\beta$ is a partition of~$\Clocks$ into subsets $\beta_0\uplus \beta_1\uplus \cdots \uplus \beta_m$ (with $m\geq 0$), where $\beta_0$ can be empty but $\beta_i\neq\emptyset$ for $1\leq i\leq m$. A valuation $\val$ of $\ValSpaceBound$ belongs to the region characterised by $(\iota, \beta)$ if
\begin{itemize}
\item for all $\clock\in\Clocks$,
$\iota(\clock) \clockgranu=\lfloor \val(\clock) \clockgranu\rfloor$;
\item for all $\clock\in \beta_0$, $\fract(\val(\clock)\clockgranu)=0$;
\item for all $0\leq i\leq m$, for all $\clock,\clocky\in \beta_i$,
$\fract(\val(\clock)\clockgranu)=\fract(\val(\clocky)\clockgranu)$;
\item for all $0\leq i< j\leq m$, for all $\clock\in \beta_i$
and all $\clocky\in \beta_j$, $\fract(\val(\clock)\clockgranu)
< \fract(\val(\clocky)\clockgranu)$.
\end{itemize}
\end{defi}

\begin{figure}
\centering

\begin{tikzpicture}
  \regpt[gray] 0 0;\regpt[red] 0 1;\regpt[gray] 1 0;\regpt[gray] 1 1;
  \reghl[gray] 0 0;\reghl[gray] 0 1;\reghl[gray] 1 0;\reghl[gray] 1 1;
  \regvl[gray] 0 0;\regvl[gray] 0 1;\regvl[green] 1 0;\regvl[gray] 1 1;
  \regdl[blue] 0 0;\regdl[gray] 0 1;\regdl[gray] 1 0;\regdl[gray] 1 1;
  \reght[gray] 0 0;\reght[gray] 0 1;\reght[gray] 1 0;\reght[black] 1 1;
  \reglt[gray] 0 0;\reglt[gray] 0 1;\reglt[gray] 1 0;\reglt[gray] 1 1;

  \path[draw,->](0,0) -> (2.3,0) node[above] {$\clock_1$};
  \path[draw,->](0,0) -> (0,2.3) node[right] {$\clock_2$};
  \path[draw] (0,0) -- (2,2)
  (1,0) -- (1,2) -- (0,1) -- (2,1) -- (1,0)
  (0,2) -- (2,2) -- (2,0);
  \node () at (1,-0.25) {$1$};
  \node () at (2,-0.25) {$2$};
  \node () at (-0.2,-0.25) {$0$};
  \node () at (-0.2,1) {$1$};
  \node () at (-0.2,2) {$2$};
\end{tikzpicture}
\hspace{2cm}
\begin{tikzpicture}
  \path[draw,color=gray!60!white,very thin] (0,1.66) -- (0.33,2) -- (0.33,0) --
  (2,1.66) -- (0,1.66)
  (0,1.33) -- (0.66,2) -- (0.66,0) -- (2,1.33) -- (0,1.33)
  (0,0.66) -- (1.33,2) -- (1.33,0) -- (2,0.66) -- (0,0.66)
  (0,0.33) -- (1.66,2) -- (1.66,0) -- (2,0.33) -- (0,0.33);

  \path[draw,->,thick](0,0) -> (2.3,0) node[above] {$\clock_1$};
  \path[draw,->,thick](0,0) -> (0,2.3) node[right] {$\clock_2$};
  \path[draw] (0,0) -- (2,2)
  (1,0) -- (1,2) -- (0,1) -- (2,1) -- (1,0)
  (0,2) -- (2,2) -- (2,0);
  \node () at (1,-0.25) {$1$};
  \node () at (2,-0.25) {$2$};
  \node () at (-0.2,-0.25) {$0$};
  \node () at (-0.2,1) {$1$};
  \node () at (-0.2,2) {$2$};
\end{tikzpicture}

\caption{All $1/1$-regions in $\NRegions{1}{\{\clock_1,\clock_2\}}{2}$ on the left,
their refinement of granularity $1/3$ in $\NRegions{3}{\{\clock_1,\clock_2\}}{2}$ on the right.}
\label{fig:reg}
\end{figure}

With granularity $\clockgranu=1$, we recover the classical notion of regions
from~\cite{AluDil94} (in the case of bounded clocks), and we omit~$\clockgranu$ from previous notations about
regions, such that $1/\clockgranu$-regions are simply called regions,
and $\NRegs$ is denoted $\Regs$.
The set of valuations contained in a $1/\clockgranu$-region $\reg$ characterised
by $(\iota, \beta)$ with $\beta = \beta_0\uplus \beta_1\uplus \cdots \uplus \beta_m$
can be described by the guard $g_0\land g_1\land\cdots \land g_m$
with
\[g_0 = \bigwedge_{\clock\in \beta_0} \big(\clock=\iota(\clock)\big), \qquad
g_1 = \bigwedge_{\clock,\clocky\in \beta_1} \big(0 < \clock-\iota(\clock) = \clock-\iota(\clocky)< 1/N\big)\]
and for $i\in\{2,\ldots,m\}$,
\[g_i = \bigwedge_{\clock,\clocky\in \beta_i} \big(\clockz-\iota(\clockz) < \clock-\iota(\clock) = \clocky-\iota(\clocky) < 1/N \big)\]
where $\clockz$ is any clock of $\beta_{i-1}$.
Therefore, every $1/\clockgranu$-region is a zone of granularity $1/\clockgranu$.

$\NRegs$ indeed forms a finite partition of $\ValSpaceBound$. We bound the number of regions in the next lemma, making use of the fact that we only consider bounded regions of $\ValSpaceBound$.

\begin{lem}\label{lm:number-regions}
  The number of $1/\clockgranu$-regions is polynomial
  in $\clockbound\clockgranu$ and exponential in the number of clocks $n=|\Clocks|$: $$|\NRegs|\leq n!(2\clockbound\clockgranu)^{n}\,.$$
\end{lem}
\begin{proof}
  Since $\iota(\clock)\in\{0,\frac{1}{\clockgranu},\dots,\frac{\clockbound-1}{\clockgranu}\}$ for each clock $\clock$,
  there are $(\clockbound\clockgranu)^{n}$ possible functions $\iota$. For each of them, the partition $\beta$ can be described by a total ordering of clocks (there are $n!$ such orderings),
  and a mapping of each clock $\clock$ to either $=$ or $>$, so that
  the fractional part of $\val(\clock)\clockgranu$ is either equal to the fractional part of $\val(\clocky)\clockgranu$
  where $\clocky$ is the clock immediately preceding $\clock$ in the ordering,
  or strictly greater than it. If the first clock of the ordering is mapped to $=$, it belongs to $\beta_0$,
  otherwise $\beta_0=\emptyset$.
  Therefore there are at most $n!2^{n}$ possibilities for the partition $\beta$.
\end{proof}
If $\val$ is a valuation in $\ValSpaceBound$, $[\val]$ denotes the
unique region that contains $\val$.
Valuations in the same $1/\clockgranu$-region $\reg$
satisfy the same guards in $\Guards_{\clockgranu}(\Clocks,\clockbound)$: we denote by $\reg\models \guard$ the satisfaction of the guard $\guard$ for all valuations of the region $\reg$.
In fact, zones associated to guards
in $\Guards_{\clockgranu}(\Clocks,\clockbound)$
can be described as a finite union of regions in $\NRegs$.
If $\reg$ is a $1/\clockgranu$-region in $\NRegs$,
then the time-successor valuations in $\Posttime(\reg)\cap \ValSpaceBound$
form a finite union of regions in $\NRegs$, and the reset $\reg[\reset:=0]$
of $\reset\subseteq\Clocks$ is a region in $\NRegs$.
A $1/\clockgranu$-region $\reg'$ is said to be
a time successor of the $1/\clockgranu$-region
$\reg$ if there exists $\val\in \reg$, $\val'\in \reg'$, and $\delay>0$ such that
$\val'=\val+\delay$.

\begin{exa}
  \figurename~\ref{fig:reg} represents the 24 regions of granularity $\clockgranu=1$
  with upper bound $\clockbound=2$ over two clocks.
  The green region is characterised by $\iota=(1,0)$ and the partition $\beta$ into two sets $\beta_0=\{\clock_1\}$ and $\beta_1 = \{\clock_2\}$: it can be encoded with the guard $0=\clock_1-1<\clock_2<1$ and thus is equal to the zone
  $\sem{\clock_1=1 \land 0<\clock_2<1}$.
  The red region is characterised by $\iota=(0,1)$ and the partition $\beta$ into a single set $\beta_0=\{\clock_1,\clock_2\}$.
  The black region is characterised by $\iota=(1,1)$ and the partition $\beta$ into three sets $\beta_0=\emptyset$, $\beta_1=\{\clock_1\}$ and $\beta_2 = \{\clock_2\}$.
  The blue region is characterised by $\iota=(0,0)$ and the partition $\beta$ into two sets $\beta_0=\emptyset$ and $\beta_1 = \{\clock_1,\clock_2\}$.
  \end{exa}

\subsection{Piecewise affine functions}\label{sec:piecewise-affine-value-functions}
We will heavily rely on the class of \emph{piecewise affine
  functions}, and a way to efficiently encode them, as developed
in~\cite{AluBer04}. These functions will enable us to describe the value of weighted timed games and will therefore be restricted to a domain $\ValSpaceBound$, with values taken in $\Rbar=\R\cup\{+\infty,-\infty\}$. Let $n$ denote the number of clocks, such
that $\Clocks=\{\clock_1,\dots,\clock_n\}$. An \emph{affine function} is a mapping
$f:\ValSpaceBound\to\Rbar$ such that for all
$\val\in\ValSpaceBound$,
  $$f(\val)= a_1\cdot \val(\clock_1)+\cdots + a_n\cdot \val(\clock_n)+b$$
  with partial derivatives $a_i\in\Q$ for $1\leq i\leq n$,
  and additive constant $b\in\Q$.
  In this case, we say that $f$
  is defined by the equation $\clocky=a_1\clock_1 +\cdots +a_n\clock_n+b$
  where the variable $\clocky\not\in\Clocks$ refers
  to $f(\clock_1,\dots,\clock_n)$.
  We also consider infinite mappings $\val\mapsto+\infty$ and $\val\mapsto-\infty$
  to be affine functions, defined with null partial derivatives and an infinite constant $b$.

  Intuitively, we define a piecewise affine function as a partition
  of \ValSpaceBound into finitely many polyhedra, called cells,
  each equipped by an affine function.
  Formally, an affine inequality is an equation $I$
  of the form $$ a_1\clock_1+\cdots+a_n\clock_n+b\dbmlaop 0$$ where
  $b\in\Q$ is the {additive constant} of $I$, $\dbmla\in \{\dbmlt,\dbmleq\}$
  is its {comparison operator}, and for every $1\leq i\leq n$,
  $a_i\in\Q$ is the $i$-th {partial derivative} of $I$.
  Similarly, an affine equality is an equation $E$
  of the form $$a_1\clock_1+\cdots+a_n\clock_n+b = 0\,.$$
  We say that $\val\in\ValSpace$ satisfies $I$ (\resp~$E$),
  and write $\val\models I$  (\resp~$\val\models E$),
  if $a_1\cdot \val(\clock_1)+\cdots + a_n\cdot
  \val(\clock_n)+b\dbmlaop 0$ (\resp~$=0$) holds.
  In this case, $\sem{I}$ (\resp~$\sem{E}$) refers to the set of valuations
  that satisfy $I$ (\resp~$E$).
  Equalities (\resp~inequalities) are equivalent when they are satisfied
  by the same valuations.
  In particular, multiplying the additive constant $b$
  and all partial derivatives $a_i$ by the same factor~$N\in\Nspos$
  gives an equivalent equality (\resp~inequality),
  and we will therefore assume that they are always integers.

\begin{defi}\label{def:cells}
  A \emph{cell} is a set $c\subseteq\ValSpace$, defined by a
  conjunction of affine inequalities $I_1\land\dots\land I_m$,
  such that $\val\in c$ if and only if for all $1\leq i\leq m$, $\val\models I_i$.
  We write $c=\sem{I_1\land\dots\land I_m}$ in this case.
\end{defi}
  Cells are convex polyhedra, and the intersection of finitely many cells is a cell.
  From every affine inequality $I$, we can extract an affine equality $E(I)$, \label{defEI} of identical
  partial derivatives and additive constant.  Then, we call \emph{borders} of
  a cell $c=\sem{I_1\land\dots\land I_m}$ the affine equalities
  $E(I_1),\dots, E(I_m)$.  The closure $\overline c$ of a cell $c$
  is obtained by replacing every comparison operator $\dbmlt$ by $\dbmleq$
  in its affine inequalities.
  Note that regions and zones are particular cases of cells,
  where borders are of the form $\clock+b=0$ or
  $\clock-\clocky+b=0$. An example of cell and its borders is given in
  \figurename~\ref{fig:value-ite-borders}.

\begin{figure}
\centering

\begin{tikzpicture}[scale=1.5]

  \path[fill=black!20!white] (0,0) -- (0,1) -- (0.5,1) -- (1,0) -- (0,0);
  \path[draw,blue] (0,2) -- (1,0)
    (0,1) -- (2.2,1);

  \draw (0,2) node[below right,xshift=1mm,blue] {\small$2\clock_1+\clock_2-2=0$};
  \draw (1.5,1) node[above,blue] {\small$\clock_2-1=0$};

  \path[draw,->](0,0) -> (2.3,0) node[above] {$\clock_1$};
  \path[draw,->](0,0) -> (0,2.3) node[right] {$\clock_2$};

  \node[below] () at (1,0) {$1$};
  \node[below] () at (2,0) {$2$};
  \node[below left] () at (0,0) {$0$};
  \node[left] () at (0,1) {$1$};
  \node[left] () at (0,2) {$2$};
\end{tikzpicture}

  \caption{The cell $2\clock_1+\clock_2-2<0\land \clock_2-1<0$ in gray,
  and its borders in blue.}
  \label{fig:value-ite-borders}
\end{figure}

We use the notion of cells in order to describe specific cases of \emph{piecewise affine functions} that we will need in this article. First we use them to describe a partition of the space.
Let $E$ be an affine equality of equation $
a_1\clock_1+\cdots+a_n\clock_n+b = 0$.
We say that \ValSpace is partitioned by $E$ into three cells:
\begin{itemize}
\item $c_{<}$, defined by $a_1\clock_1+\cdots+a_n\clock_n+b < 0$;
\item $c_{>}$, defined by $a_1\clock_1+\cdots+a_n\clock_n+b > 0$,
\ie~$-a_1\clock_1-\cdots-a_n\clock_n-b < 0$;
\item $c_{=}$, defined by $a_1\clock_1+\cdots+a_n\clock_n+b=0$,
\ie~the conjunction of $a_1\clock_1+\cdots+a_n\clock_n+b \leq 0$ and
$-a_1\clock_1-\cdots-a_n\clock_n-b\leq 0$.
\end{itemize}
Then, given a set $\mathcal E=\{E_1,\dots,E_m\}$ of affine equalities,
we denote $c_{<}^j$, $c_{>}^j$ and $c_{=}^j$ the three cells obtained
from $E_j\in\mathcal E$.  For every mapping
$\phi\colon\mathcal E\to\{{<},{>},{=}\}$, we define $c_{\phi}$ as the cell
$c_{\phi(E_1)}^1\cap\dots\cap c_{\phi(E_m)}^m$.  Every valuation of
\ValSpace belongs to some~$c_{\phi}$, and if $\phi\neq\phi'$ then
$c_{\phi}\cap c_{\phi'}=\emptyset$: hence, the set of mappings
${\{{<},{>},{=}\}}^{\mathcal E}$ provides a partition of \ValSpace
into $3^m$ cells.  We say that \ValSpace is partitioned by
$\mathcal E$ into $m'\in\N$ cells if $m'$ of those $3^m$ cells are
non-empty.
We denote $\splitnum(m,n)$ the greatest $m'$ over any partition of \ValSpace by $m$ affine equalities.
Similarly, a cell $c\subseteq\ValSpace$ is partitioned by $\mathcal E$
into at most $\splitnum(m,n)$ \label{splitmn} sub-cells that have non-empty intersection with $c$.
In particular, under the bounded clocks assumption we will
partition \ValSpaceBound instead of \ValSpace.

\begin{figure}[tbp]
\centering

\begin{tikzpicture}[scale=1]
  \path[draw,->](0,0) -> (2.3,0) node[above] {$\clock_1$};
  \path[draw,->](0,0) -> (0,2.3) node[right] {$\clock_2$};
  \path[draw] (0,2) -- (2,2) -- (2,0);

  \node[below] () at (1,0) {$1$};
  \node[below] () at (2,0) {$2$};
  \node[below left] () at (0,0) {$0$};
  \node[left] () at (0,1) {$1$};
  \node[left] () at (0,2) {$2$};

  \node[fill=black!40!white,circle,inner sep=0pt, minimum size=4pt] () at (0.5,1) {};
  \path[fill=black!40!white,rounded corners=1pt,inner sep=0pt] (0.5cm+2.5pt,1cm-1.5pt) -- (0.5cm+2.5pt,1cm+1.5pt) -- (0.5cm+1.5cm,1cm+1.5pt) -- (0.5cm+1.5cm,1cm-1.5pt) -- (0.5cm+2.5pt,1cm-1.5pt) -- (0.5cm+2.5pt,1cm+1.5pt);
  \path[fill=black!40!white,rounded corners=1pt,inner sep=0pt] (0cm,1cm-1.5pt) -- (0cm,1cm+1.5pt) -- (0cm+0.5cm-2.5pt,1cm+1.5pt) -- (0cm+0.5cm-2.5pt,1cm-1.5pt) -- (0cm,1cm-1.5pt) -- (0cm,1cm+1.5pt);
  \path[fill=black!40!white,rounded corners=1pt,inner sep=0pt] (0.5cm+0pt,1cm+4pt) -- (0.5cm-3pt,1cm+2pt)
  -- (0,2cm-4pt)-- (0cm,2cm) -- (2pt,2cm) -- (0.5cm+0pt,1cm+4pt) -- (0.5cm-3pt,1cm+2pt);
  \path[fill=black!40!white,rounded corners=1pt,inner sep=0pt] (0.5cm+3pt,1cm-2pt) -- (0.5cm+0pt,1cm-4pt)
  -- (1cm-2pt,0)-- (1cm+2pt,0cm) -- (0.5cm+3pt,1cm-2pt) -- (0.5cm+0pt,1cm-4pt);
  \path[fill=black!40!white,rounded corners=1pt,inner sep=0pt] (0cm,0cm) -- (0cm,1cm-2.5pt) --
  (0.5cm-2pt,1cm-2.5pt) -- (1cm-3pt,0cm) -- (0cm,0cm) -- (0cm,1cm-2.5pt);
  \path[fill=black!40!white,rounded corners=1pt,inner sep=0pt] (0cm,1cm+2.5pt) -- (0cm,2cm-6pt) -- (0.5cm-5pt,1cm+2.5pt) -- (0cm,1cm+2.5pt) -- (0cm,2cm-6pt);
  \path[fill=black!40!white,rounded corners=1pt,inner sep=0pt] (0cm+3pt,2cm) -- (2cm,2cm) -- (2cm,1cm+2.5pt)
  -- (0.5cm+2pt,1cm+2.5pt) -- (0cm+3pt,2cm) -- (2cm,2cm);
  \path[fill=black!40!white,rounded corners=1pt,inner sep=0pt] (2cm,0cm) -- (1cm+3pt,0cm) -- (0.5cm+5pt,1cm-2.5pt)
  -- (2cm,1cm-2.5pt) -- (2cm,0cm) -- (1cm+3pt,0cm);
\end{tikzpicture}

\caption{A partition of $[0,2)^{\Clocks}$ according to the two
affine equalities of \figurename~\ref{fig:value-ite-borders}.}
\label{fig:value-ite-part}
\end{figure}
\begin{exa}
The $\splitnum(2,2)=9$ cells that partition $[0,2)^{\Clocks}$
according to $\mathcal E=\{2\clock_1+\clock_2-2=0, \clock_2-1=0\}$ are
represented in \figurename~\ref{fig:value-ite-part}.
\end{exa}

As a technical tool, we will need to better understand how the number of cells in a partition grows with respect to the number of equalities used to describe this partition. In fact, the number of cells $\splitnum(m,n)$
is bounded by $\mathcal O((2m)^n)$, which we will use abundantly to
compute the complexity of our algorithms. The technical proof is delayed in Appendix~\ref{app:proof-splitnum}.
\begin{lem}\label{lm:splitnum}
For all $m,n\geq 0$, $\splitnum(m,n)\leq 2^n(m+1)^{n}$.
\end{lem}

We summarise our representation of partitions with affine equalities in the following definition:

\begin{defi}
  A \emph{partition} $P$ is defined by
  a cell $c_{P}$ and a set $\mathcal E_{P}$ of affine equalities, such that
  $P$ encodes the set of cells that partition $c_{P}$
  according to $\mathcal E_{P}$, called \emph{base cells}.
\end{defi}
  The cell $c_{P}$ is called the domain of $P$.
  We denote $[\val]_{P}$ the base cell that contains valuation $\val\in c_P$.

\begin{defi}
  A \emph{partition function} $F$ defined over a partition $P$
  is a mapping from the base cells of $P$ to affine functions.
  It encodes a mapping from $c_P$ to \Rbar, denoted $\sem{F}$: if
  $\val\in c_P$ and $F([\val]_{P})$ is defined by
  $\clocky=a_1\clock_1+\cdots+a_n\clock_n+b$, then
  $\sem{F}(\val)$ equals
  $a_1\cdot\val(\clock_1)+\cdots+a_n\cdot \val(\clock_n)+b$.
\end{defi}
A partition function $F$ of domain $c_P$ is \emph{continuous} if
for all $\val\in c_P$, for every base cell $c_b$ such that
$\val\in\overline{c_b}$, if $F(c_b)$ is defined by
$\clocky=a_1\clock_1+\cdots+a_n\clock_n+b$ then
$\sem{F}(\val)=a_1\cdot\val(\clock_1)+\cdots+a_n\cdot\val(\clock_n)+b$.
In other words, the affine equations provided by $F$ to neighbouring
cells should match on the borders that separate~them.

  Finally, we use a pair $(P,F)$ where $P$ is a partition of domain \ValSpaceBound,
  and $F$ is a partition  function defined over $P$, to encode a \emph{piecewise affine function}
  $\sem{F}\colon\ValSpaceBound\to\Rbar$.
  The piecewise affine function
  is said \emph{continuous on regions} if for every region $r\in\Regs$,
  the restriction of $\sem{F}$ to domain $r$ is continuous.
  There could be discontinuities in $\sem{F}$, but only at borders separating
  different regions.
  In particular, if a partition function is continuous over regions,
  and $\sem{F}(\val)=+\infty$ (\resp~$-\infty$) for some $\val$,
  then for all $\val'$ in the same region as $\val$,
  $\sem{F}(\val')=\sem{F}(\val)$.

\subsection{Weighted timed games}\label{sec:wtg}

We now turn our attention to (turn-based) weighted timed two-player games with a
shortest-path objective towards a set of target locations. We will
first define weighted timed games, giving their semantics in terms of
an \emph{infinite} transition system on which the traditional notions
of strategies and values are defined.

\begin{defi}
  A \emph{weighted timed game} (\WTG) is a tuple
  $\game=\struct{\LocsMin,\LocsMax,\Clocks,\LocsT, \allowbreak\Edges,
    \weight,\weightT}$ with $\Locs=\LocsMin\uplus\LocsMax$ a finite
  set of locations split between players $\MinPl$ and $\MaxPl$ (in
  drawings, locations belonging to \MinPl are depicted by circles and
  the ones belonging to \MaxPl by squares), $\Clocks$ a finite set of
  clocks,
  $\Edges\subseteq \Locs\times\Guards(\Clocks)\times
  2^{\Clocks}\times \Locs$ a finite set of edges
  $\loc\xrightarrow{\guard,\reset}\loc'$ from location $\loc$ to
  location $\loc'$, labelled by a guard~$\guard$ and a
  set $\reset$ of clocks to reset,
  $\weight\colon \Edges\uplus\Locs \to \Z$ a weight function
  associating an integer weight with each location and edge,
  $\LocsT\subseteq \LocsMin$ a set of target locations for player
  $\MinPl$, and $\weightT\colon \LocsT\times\Rpos^\Clocks\to \Rbar$ is
  a function\footnote{We restrict the type of functions allowed in Hypothesis~\ref{hyp:final}: informally, we will only deal with piecewise affine functions, as defined previously.} mapping each target configuration to a final weight of
  $\Rbar$.
\end{defi}

The semantics of a weighted timed game $\game$ is defined in terms of
an infinite labelled transition system $\sem{\game}$ whose states are configurations
$(\loc,\val)\in \Locs\times\Rpos^\Clocks$. Configurations are split
into players according to the location $\loc$. A configuration
$(\loc,\val)$ is a target if $\loc\in\LocsT$, and its final weight is
$\weightT(\loc,\val)$. The labels of $\sem{\game}$ are given by
$\Rpos\times \Edges$ and will encode the delay that a player wants to
spend in the current location, before firing a certain edge. For every
delay $\delay\in\Rpos$, edge
$\edge=\loc\xrightarrow{\guard,\reset}\loc'\in \Edges$ and
valuation~$\val$, there is a transition
$(\loc,\val)\xrightarrow{\delay,\edge}(\loc',\val')$ if
$\val+\delay\models \guard$ and $\val'=(\val+\delay)[\reset:=0]$. The
weight of such a transition takes into account both discrete and
continuous costs, and is given by
$\delay\cdot\weight(\loc) + \weight(\edge)$.

Without loss of generality (since we can detect such deadlocks using classical attractor techniques using regions we describe later, and add a new transition towards a sink location in this case), we suppose the absence of
deadlocks in $\sem{\game}$ except on target locations:
\begin{hypothesis}
  For each
location $\loc\in \Locs\backslash\LocsT$ and valuation
$\val\in\Rpos^\Clocks$, there exist $\delay\in\Rpos$ and
$\edge\in \Edges$ such that
$(\loc,\val)\xrightarrow{\delay,\edge}(\loc',\val')$, and no edges
start from~\LocsT.
\end{hypothesis}

\emph{Plays} are maximal paths in the transition system $\sem\game$: thanks to the previous hypothesis, they are
either infinite or end in a target location.  First and last elements of a finite play $\play$
are denoted $\first(\play)$ and $\last(\play)$, respectively. For a player
$\Pl\in\{\MinPl,\MaxPl\}$, the set of non-maximal plays $\play$ (often
called \emph{finite plays}) such that $\last(\play)\in\LocsPl$ is
denoted $\FPlaysPl$.

A~\emph{strategy} $\stratpl$ for player $\Pl$ is a
mapping $\FPlaysPl \to \Rpos\times\Edges$, such that for
all~$\play\in\FPlaysPl$ ending in configuration $(\loc,\val)$, the
transition system $\sem\game$ contains a transition labelled by
$\stratpl(\play)$ from $(\loc,\val)$. A strategy is said
\emph{positional} (or \emph{memoryless}) if for all
$\play,\play'\in\FPlaysPl$ ending in the same configuration,
$\stratpl(\play)=\stratpl(\play')$. Let
$\outcome((\loc_0,\val_0),\stratmin,\stratmax)$ denote the unique
maximal play starting from configuration $(\loc_0,\val_0)$ such that
for every prefix $\play$ of
$\outcome((\loc_0,\val_0),\stratmin,\stratmax)$ in $\FPlaysPl$, the
next transition in $\outcome((\loc_0,\val_0),\stratmin,\stratmax)$ is
labelled by $\stratpl(\play)$.

The objective of \MinPl is to reach a target location, while
minimising the cumulative weight up to the target.  Hence, we
associate to every finite play
$\play=(\loc_0,\val_0)\xrightarrow{\delay_1,\edge_1}(\loc_1,\val_1)
\cdots\xrightarrow{\delay_k,\edge_k}(\loc_k,\val_k)$ its
cumulative weight
\[\weightC(\play)=\sum_{i=1}^{k} \delay_{i}\cdot\weight(\loc_{i-1}) +
  \weight(\edge_i)\] Then, the weight of a maximal play $\play$, also
denoted by $\weight(\play)$, is defined by $+\infty$ if $\play$ is
infinite and thus does not reach $\LocsT$, and
$\weightC(\play)+\weightT(\loc,\val)$ if it is finite and ends in
$(\loc,\val)$ with $\loc\in \LocsT$. Then, the respective values of
the strategies are defined by
  \begin{align*}
  \Val_\game((\loc_0,\val_0),\stratmin) &= \sup_{\stratmax}
  \weight(\outcome((\loc_0,\val_0),\stratmin,\stratmax)) \\
  \Val_\game((\loc_0,\val_0),\stratmax) &= \inf_{\stratmin}
  \weight(\outcome((\loc_0,\val_0),\stratmin,\stratmax))
  \end{align*}
  Finally, we let
  \begin{align*}
    \lVal(\loc_0,\val_0) &= \sup_{\stratmax}
                           \Val_\game((\loc_0,\val_0),\stratmax) &
    \uVal
    (\loc_0,\val_0) &= \inf_{\stratmin}
                      \Val_\game((\loc_0,\val_0),\stratmin)
  \end{align*}
  be the \emph{lower} and \emph{upper values} of configuration
  $(\loc_0,\val_0)$, respectively.  We may easily show that
  $\lVal\leq \uVal$.  We say that a strategy $\stratmin^\star$ of
  $\MinPl$ is \emph{$\varepsilon$-optimal} if, for all initial
  configurations $(\loc_0,\val_0)$
  \[\Val_\game((\loc_0,\val_0),\stratmin^\star)\leq
    \uVal(\loc_0,\val_0)+\varepsilon\,.\] It is said \emph{optimal} if
  this holds for $\varepsilon=0$. A symmetric definition holds for
  optimal strategies of $\MaxPl$.

We now impose some more hypotheses on the weighted timed games we consider in this article. First, by mimicking the transformation used for timed automata \cite[Lemma~5]{BerPet98}, we can turn all diagonal guards into non-diagonal ones by keeping a bit of information for each pair of clocks in the locations. Moreover, by multiplying each rational number in guards by an appropriate positive integer, while doing the same transformation on weights of transitions, we can turn every constant in guards into integers. We therefore assume in the rest of this article that
\begin{hypothesis}
  All guards of \WTG{s} are non-diagonal and contain only integers constants.
\end{hypothesis}

Seminal works in weighted timed games \cite{AluBer04,BCFL04} have assumed that clocks are \emph{bounded}. This is known to be without loss of generality for
(weighted) timed automata~\cite[Theorem~2]{BehFeh01}: it suffices to replace transitions with unbounded delays with self-loop transitions periodically resetting the clocks. We do not know if it is
the case for the weighted timed games defined above. Indeed, the technique of \cite{BehFeh01} cannot be directly
applied. This would give too much power to player $\MaxPl$ that would then be allowed to loop in a location where an unbounded delay could originally be taken before going to the target. In \cite{BCFL04}, the situation is simpler since the game is \emph{concurrent}, and thus \MinPl always has a chance to move outside of such a situation. Trying to detect and avoid such situations in our turn-based case seems difficult in the presence of negative weights, since  the opportunities of $\MaxPl$ crucially depend on the configurations of value $-\infty$ that $\MinPl$ could control afterwards: we will see in Proposition~\ref{prop:-infty_undec} that detecting such configurations is undecidable, which is an additional evidence to motivate the decision to focus only on bounded weighted timed games. We thus suppose from now on that
\begin{hypothesis}\label{hyp:bounded}
  The \WTG{s} are \emph{bounded}, \ie~$\sem{\game}$ is restricted to
  configurations in $\Locs\times\ValSpaceBound$ where $\clockbound\in\Nspos$ is the greatest constant appearing in guards.
\end{hypothesis}

We need to be able to represent finitely the final weight functions of weighted timed games. The simplest assumption, consisting in somehow disallowing final weights as usually done in the literature, would be to map every target location to the weight $0$. We keep more complex final weights, since we will need them in
the process of solving weighted timed games. However, it will be sufficient to assume that:
\begin{hypothesis}\label{hyp:final}
  The final weight functions $\weightT$ are described by pairs $(P,F)$ where $P$ is a partition of domain \ValSpaceBound,
  and $F$ is a partition function defined over $P$: they are thus piecewise affine functions with a finite number of
  pieces (that are cells, as defined in Definition~\ref{def:cells}), and are assumed to be continuous on regions.
\end{hypothesis}
In particular, infinite final weights are constant over regions, \ie~if some
configuration $(\locT,\val)$ has final weight $+\infty$ or $-\infty$,
then for every valuation $\val'$ in the same region as~$\val$,
$\weightT(\locT,\val)=\weightT(\locT,\val')$. The zero final
weight function satisfies this property.
Moreover, the computations we will perform in the following maintain
this property as an invariant.

\begin{figure}
  \centering
  \scalebox{1}{
     \begin{tikzpicture}[node distance=3cm,auto,->,>=latex]

      \node[player2](1){\makebox[0mm][c]{$\mathbf{-2}$}};
      \node()[below of=1,node distance=6mm]{$\loc_1$};

      \node[player1](2)[below right
        of=1]{\makebox[0mm][c]{$\mathbf{2}$}};
      \node()[below of=2,node distance=6mm]{$\loc_2$};

      \node[player1](3)[above right
        of=2, accepting]{\makebox[0mm][c]{}};
      \node()[below of=3,node distance=6mm]{$\loc_3$};
    \node()[above of=3,node distance=6mm]{$\weightT=\mathbf{0}$};

      \node[player2](4)[below right
        of=3]{\makebox[0mm][c]{$\mathbf{-1}$}};
      \node()[below of=4,node distance=6mm]{$\loc_4$};

      \node[player1](5)[above right
        of=4]{\makebox[0mm][c]{$\mathbf{-2}$}};
      \node()[below of=5,node distance=6mm]{$\loc_5$};

  \path
  (2) edge node[below left,xshift=2mm,yshift=2mm]{$
          \begin{array}{c}
  \clock\leq2 \\ \clock:=0 \\ \mathbf{0}
          \end{array}
          $} (1);

  \path
  (2) edge node[above left,xshift=4mm,yshift=-4mm]{
          $\begin{array}{c}
  1\leq \clock<3 \\ \mathbf{1}
          \end{array}$} (3);

  \path
  (2) edge node[below]{$\clock<3;\; \mathbf{0}$} (4);

  \path
  (4) edge node[above right,xshift=-4mm,yshift=-4mm]{$
          \begin{array}{c}
  2\leq \clock<3 \\ \mathbf{3}
          \end{array}$} (3);

  \path
  (4) edge node[below right,xshift=-2mm,yshift=2mm]{$
          \begin{array}{c}
  \clock<3\\ \mathbf{0}
          \end{array}$} (5);

  \path
  (1) edge node[above]{$\clock<3;\; \mathbf{0}$} (3);

  \path
  (5) edge node[above]{$\clock<3; \; \mathbf{0}$} (3);

  \path
  (1) edge [loop above] node {$
          \begin{array}{c}
  \clock<3\\ \clock:=0;\; \mathbf{3}
          \end{array}$}  (1);

  \path
  (5) edge [loop above] node {$\begin{array}{c}1<\clock<3 \\ \clock:= 0; \;\mathbf{1}\end{array}$} (5);

      \end{tikzpicture}
  }\hspace{-6mm}
\scalebox{1}{
\begin{tikzpicture}[node distance=3cm,auto,>=latex]

\draw[->] (0,0) -- (0,4);
\draw[->] (0,0) -- (4,0);

\draw[dashed] (1,0) -- (1,4);
\draw[dashed] (2,0) -- (2,4);
\draw[dashed] (3,0) -- (3,4);

\draw[dashed] (0,1) -- (4,1);
\draw[dashed] (0,2) -- (4,2);
\draw[dashed] (0,3) -- (4,3);

\draw[dashed] (.666,0) node[below,xshift=-1mm]{$2/3$}-- (.666,1.666);

\node at (4,-.5) {$\clock$};
\node at (-.5,4) {$\Val$};

\node at (0,-.5) {$0$};
\node at (1,-.5) {$1$};
\node at (2,-.5) {$2$};
\node at (3,-.5) {$3$};

\node at (-.5,0) {$0$};
\node at (-.5,1) {$1$};
\node at (-.5,2) {$2$};
\node at (-.5,3) {$3$};

\draw[red,very thick] (0,3) -- (1,1) -- (3,1);

\draw[blue,very thick] (0,1) -- (2,3) -- (3,3);

\draw[green!60!black,very thick] (0,.95) -- (0.66,1.61) -- (1,.95) -- (3,.95);

\node[rectangle,fill=white,fill opacity=.9,text opacity=1] at
(2.4,3.4) {\textcolor{blue}{$\Val(\loc_4,\cdot)$}};

\node[rectangle,fill=white,fill opacity=.9,text opacity=1]  at
(2.2,.5) {\textcolor{green!60!black}{$\Val(\loc_2,\cdot)$}};

 \end{tikzpicture}
 }
  \caption{A weighted timed game with a single clock $\clock$, and a
    depiction of its value function.  Weights are indicated in bold
    font on locations and edges. The target location is $\loc_3$,
    whose final weight function is zero. The blue curve is the values associated to location $\loc_4$, the green curve the values of location $\loc_2$, obtained as the minimum between the blue and the red curve, depicting the weight obtained from $\loc_2$ when \MinPl decides to jump directly from $\loc_2$ to $\loc_3$.}\label{fig:wtg}
\end{figure}

\begin{exa}
  An example of \WTG satisfying the previous hypotheses is depicted on \figurename~\ref{fig:wtg}.
  Observe
  that location $\loc_1$ (\resp~$\loc_5$) has value $+\infty$
  (\resp~$-\infty$).
  Indeed, from any configuration $(\loc_1,\val)$ with $\val\in[0,3)^{\Clocks}$, player \MaxPl can play the self-loop on $\loc_1$, ensuring that the target $\loc_3$ is never reached.
  Moreover, from any configuration $(\loc_5,\val)$ with $\val\in[0,3)^{\Clocks}$, player \MinPl can play the self-loop on $\loc_5$,
  leading to a reset of clock~\clock.
  From $(\loc_5,(\clock=0))$ \MinPl can pick a delay of $1.5$, and loop again while accumulating a weight of $-2(1.5)+1=-2$.
  This loop of $\loc_5$ can be iterated arbitrarily many times before playing the edge to $\loc_3$, ensuring a weight arbitrarily low, hence every configuration on $\loc_5$ (in the bounded valuation domain $[0,3)^{\Clocks}$) has value $-\infty$.
  As a consequence, the value in~$\loc_4$ is only determined by the
  transition to $\loc_3$ (since \MaxPl tries to avoid
  the target), and depicted in blue in
  \figurename~\ref{fig:wtg}. In location $\loc_2$, the
  value that \MinPl can get while jumping through the transition to
  $\loc_3$ is depicted in red. While jumping to $\loc_2$
  after a delay 0 (since $\weight(\loc_2)>0$, \MinPl wants to minimise
  the time spent in this location), \MinPl can also obtain the value
  of~$\loc_4$. Therefore, the value of location $\loc_2$ is obtained
  as the minimum of these two curves, depicted in green. Observe the
  intersection point in $\clock=2/3$ requiring to refine the usual
  regions (of granularity~1).
\end{exa}

It is known that (turn-based) weighted timed games are
determined\footnote {The result is stated in \cite{BGH+15} for
  weighted timed games (called priced timed games) with one clock, but
  the proof does not use the assumption on the number of clocks.},
\ie~$\lVal(\loc,\val)=\uVal(\loc,\val)$ for each location $\loc$ and
valuation $\val$, therefore we use the notation $\Val_\game$ to refer to
both values.

We denote by $\wmax^\Locs$ (\resp~$\wmax^\Edges$) the maximal weight
in absolute values of locations (\resp~of edges) in $\game$:
\[\wmax^\Locs = \max_{\loc\in \Locs} |\weight(\loc)|
  \qquad\text{and} \qquad\wmax^\Edges = \max_{\edge \in \Edges}
  |\weight(\edge)|\]
Moreover, %
we denote by
$\wmaxTimed$ a bound on the weight of transitions in $\sem{\game}$,
that exists since clocks are bounded by $\clockbound$:
\[\wmaxTimed = \clockbound\wmax^\Locs + \wmax^\Edges\]
The integer \wmaxTimed is at most exponential in the size of \game,
and can thus be stored in polynomial space.

We consider the \emph{value problem}: given a \WTG~$\game$, a location $\loc$  and a threshold $\alpha\in \Q$, we want
to know whether $\Val_\game(\loc,\valnull)\leq \alpha$.
We also consider the \emph{$+\infty$-value problem} (\resp~\emph{$-\infty$-value problem}) asking if $\Val_\game(\loc,\valnull)$ equals $+\infty$ (\resp~$-\infty$).
In the context of
timed games, optimal strategies may not exist, even for finite
values.\footnote {For example, a player may want to let time elapse as
  much as possible, but with delay $\delay<1$ because of a strict
  guard.} We thus generally focus on the search for
$\varepsilon$-optimal strategies, that guarantee the optimal value, up
to a small error $\varepsilon\in\Rspos$: this is the \emph{synthesis
  problem}. Moreover, when the value problem is undecidable,
we also consider the \emph{value approximation problem} that consists,
given a precision $\varepsilon\in \Qspos$, in computing an
$\varepsilon$-approximation of $\Val_\game(\loc,\valnull)$.  More
generally, we will try to compute an $\varepsilon$-approximation of
the whole value function (and not only for an initial
configuration with all clocks being 0): this means that we want to compute
(in the format of a pair $(P,F)$ with $P$ a partition of domain $\ValSpaceBound$ and $F$ a partition function, since we will show that this is sufficient) a
function $\ValIteVec\colon \Locs\times\ValSpaceBound \to \Rbar$ such that
$\|\Val_\game-\ValIteVec\|_\infty \leq \varepsilon$, where
$\|\cdot\|_\infty$ denotes the classical $\infty$-norm of mappings, so
that $\|f\|_\infty= \sup_{x} |f(x)|$.
In~particular, we ask that $\Val_\game=\ValIteVec$ on all
configurations where either function has infinite value, and as such
solving the infinite value problems will be a requirement.

For the purpose of stating complexity results we assume that the size needed to encode an input \WTG $\game=\struct{\LocsMin,\LocsMax,\Clocks,\LocsT, \allowbreak\Edges,
  \weight,\weightT}$ is linear in $|\Locs|$, $|\Edges|$, and $n=|\Clocks|$.
  We assume that integer constants are encoded in binary, so that the input depends logarithmically in $\wmax^\Locs$, $\wmax^\Edges$ and $\clockbound$. Finally, we assume that the final weight function $\weightT$ is encoded as a partition function $(P,F)$,
  so that the input is linear in the number of affine expressions used in $P$ and $F$ and logarithmic in the size of their constants.
  We encode all rational constants as irreducible fractions of binary integers.

\subsection{Related work}

In the one-player case, computing the optimal value and an
$\varepsilon$-optimal strategy for \WTG{s} (called weighted or priced
timed automata) is known to be $\PSPACE$-complete~\cite{BouBri07}.  In
the two-player case, the value problem of \WTG{s} (also called priced
timed games in the literature) is undecidable with 3 clocks
\cite{BBR05,BJM15}, or even 2 clocks in the presence of negative
weights \cite{BGNK+14} (for the existence problem, asking if a
strategy of player \MinPl can guarantee a given threshold).

To obtain decidability, one possibility has been to limit the number
of clocks to 1: then, there is an exponential-time algorithm to
compute the value as well as $\varepsilon$-optimal strategies in the
presence of non-negative weights only~\cite{BBM06,Rut11,DueIbs13},
whereas the problem is only known to be $\P$-hard.  A similar result
can be lifted to arbitrary weights, under restrictions on the resets
of the clock in cycles~\cite{BGH+15}.

The other possibility to obtain a decidability
result~\cite{BCFL04,AluBer04} with non-negative weights only is to
enforce a semantical property of divergence, originally called
strictly non-Zeno cost: it asks that every play following a cycle of
the region abstraction (see later for a formal definition) has weight at
least~$1$. In \cite{BJM15}, the authors slightly extend the strictly
non-Zeno cost property, to allow for cycles of weight exactly $0$
while still preventing those of weight arbitrarily close to
$0$. Unfortunately, this leads to an undecidable value problem, but
they propose a solution to the value approximation
problem. \textbf{This article aims at extending these two restrictions
  in the presence of both non-negative and negative weights calling it
  the divergence property, in order to obtain either decidability of the
  value problem or approximations of the value for a large class of
  multi-clocks weighted timed games in the presence of arbitrary
  weights.}

Other objectives, not directly related to optimal reachability, have
been considered in~\cite{BreCas14} for~\WTG, like mean-payoff and
parity objectives.  In this work, the authors manage to solve these
problems for the so-called class of $\delta$-robust \WTG{s} that they
introduce.\footnote {As mentioned in introduction, while our divergent
  games have a similar definition, the two classes are incomparable.}

\subsection{Region abstraction}
The partition of the configuration space into regions can be maintained throughout the
play in a \WTG. By forgetting about the precise valuation of clocks,
we obtain a \emph{region abstraction} (that is usually called the
\emph{region automaton} in the literature).

\begin{defi}
  Given a \WTG
  $\game=\struct{\LocsMin,\LocsMax,\Clocks,\LocsT, \allowbreak\Edges,
    \weight,\weightT}$ such that all clocks are bounded by
  $\clockbound$ and all guards belong to
  $\Guardsnd_{\clockgranu}(\Clocks,\clockbound)$ for some granularity
  $1/\clockgranu$, we define the \emph{region abstraction} of $\game$
  as the finite labelled transition system
  $\struct{\Locs\times\NRegs,\Trans}$ labelled over
  $\NRegs\times\Edges$, where~$\Trans$ contains all transitions
  $(\loc,\reg)\xrightarrow{\reg'',\edge}(\loc',\reg')$ such that
  $\edge=\loc\xrightarrow{\guard,\reset}\loc'$ is an edge of~$\game$,
  $\reg''$ is a time-successor of $\reg$, $\reg''\models\guard$ and
  $\reg''[\reset:=0]=\reg'$.
\end{defi}

The states of the region abstraction are called \emph{region states},
and its paths are called \emph{region paths}. As there are finitely
many regions, the region abstraction of $\game$ is a finite
transition system, where paths \rpath represent sequences of regions
alternating between letting time elapse and taking edges.
From a play
$\play=(\loc_0,\val_0)\xrightarrow{\delay_1,\edge_1}(\loc_1,\val_1)
\xrightarrow{\delay_2,\edge_2}\dots$
in $\game$, we construct a region path
$\rpath=(\loc_0,[\val_0])\xrightarrow{[\val_0+\delay_1],\edge_1}(\loc_1,[\val_1])
\xrightarrow{[\val_1+\delay_2],\edge_2}\dots$, and say that $\play$
\emph{follows}~$\rpath$.

From the region abstraction, we can construct a \emph{region game}
that can be seen as a product of the original \WTG with the region
abstraction.

\begin{defi}
  Given a \WTG
  $\game=\struct{\LocsMin,\LocsMax,\Clocks,\LocsT, \allowbreak\Edges,
    \weight,\weightT}$ such that all clocks are bounded by
  $\clockbound$ and all guards belong to
  $\Guardsnd_{\clockgranu}(\Clocks,\clockbound)$ for some granularity
  $1/\clockgranu$, we define the \emph{region game} of $\game$ as the
  \WTG
  $\Nrgame{\clockgranu}=\struct{\LocsMin\times\NRegs,
    \LocsMax\times\NRegs, \Clocks,\LocsT\times\NRegs, \Edges',
    \weight',\weightT'}$ whose locations are region states, $\weight'$
  and $\weightT'$ are obtained trivially from $\weight$ and $\weightT$
  by forgetting about regions, and $\Edges'$ is defined by
  transforming every transition
  $(\loc,\reg)\xrightarrow{\reg'',\edge}(\loc',\reg')$ in $\Trans$,
  where $\edge$ is labelled by $(\guard,\reset)$, into an edge
  $(\loc,\reg)\xrightarrow{\guard'',\reset}(\loc',\reg')$, with
  $\guard''$ a guard chosen arbitrarily so that $\sem{\guard''}=\reg''\subseteq\sem{\guard}$.
\end{defi}

Every play in $\game$ exists in $\Nrgame{\clockgranu}$ as a play
following a region path $\rpath$, and conversely every play in
$\Nrgame{\clockgranu}$ is a valid play in $\game$, by
projecting away the region information of $\Nrgame{\clockgranu}$. We
thus obtain:
\begin{lem}\label{lm:region-game}
  Games $\game$ and $\Nrgame{\clockgranu}$ contain the same plays. For
  all $\loc\in\Locs$, $1/\clockgranu$-regions $\reg$, and
  $\val\in \reg$,
  $\Val_\game(\loc,\val)=\Val_{\Nrgame
    \clockgranu}((\loc,\reg),\val)$.
\end{lem}

As we assume that guards have
integer constants, we can use the granularity $1/\clockgranu=1$, and we will write $\rgame$
instead of $\Nrgame \clockgranu$ in that case.
Finally, we denote by $|\rgame|$ the number of locations in the region game, equal to $|\Locs| |\Regions\Clocks\clockbound|$.

\subsection{Corner-point abstraction}

Despite all the interest and success of regions to study timed
systems, they are not sufficient to handle weighted timed
automata/games. Indeed, for a single-clock case, and in a location
$\loc$ of weight $1$, spending time in $\loc$ from region $(0,1)$ to
region $\{2\}$ can cost any possible weight in the interval
$(1,2)$. We therefore must also keep a more precise information of
\emph{where we are inside each region}. This is the goal of the
\emph{corner-point abstraction} introduced
in~\cite{BehFeh01,LarBeh01} to study one-player \WTG{s} with
non-negative weights, generalised in~\cite{BouBri07} to handle
negative weights, and in~\cite{BouBri08a} for the multi-cost setting.

If $\reg$ is an $1/\clockgranu$-region, let $\overline{\reg}$ denote
its topological closure, \ie~the smallest
zone that contains $\reg$ associated to a non-strict guard.
The \emph{corners} of $\reg$ are all the valuations in $\overline{\reg}$
that belong to $\Q_\clockgranu^\Clocks$.  If $\reg$ is characterised
by $(\iota, \beta)$ with $\beta$ the partition $\beta_0\uplus \beta_1\uplus\cdots \uplus \beta_m$, then $\iota$ is a corner of $\reg$. If $m=0$,
then $\reg=\{\iota\}$, otherwise $\reg$ does not include its corners
but contains valuations arbitrarily close to them. The corners of
$\reg$ are the vertices of the polytope $\overline{\reg}$, such that
$\overline{\reg}$ is their convex hull. There are at most
$n+1$ corners in each $1/N$-region.
The corners of the green region in \figurename~\ref{fig:reg} are the valuations
$(1,0)$ and $(1,1)$.

We call
corner state a triple~$(\loc,\reg,\corner)$ that contains information
about a region state $(\loc,\reg)$ of $\Nrgame \clockgranu$, and a
corner $\corner$ of the $1/\clockgranu$-region $\reg$.

Notice that reset operations preserve the corners, \ie, if $\corner$ is a corner of the region $\reg$, then $\corner[\clock:=0]$ is a corner of the region $\reg[\clock:=0]$. This allows one to enrich the region game with corner information:

\begin{defi}
  The corner-point abstraction~$\Ncgame \clockgranu$ of a \WTG~$\game$
  is the \WTG obtained as a refinement of $\Nrgame \clockgranu$ where
  guards on edges are enforced to stay on one of the corners of the
  current $1/\clockgranu$-region: the locations of
  $\Ncgame \clockgranu$ are all corner states
  of~$\Nrgame \clockgranu$, associated to each player accordingly, and
  edges are all
  $(\loc,\reg,\corner)\xrightarrow{ \guard'', \reset}
  (\loc',\reg',\corner')$ such that there exists a region $\reg''$,
  an
  edge $\rtrans=(\loc,\reg)\xrightarrow{\guard,\reset}(\loc',\reg')$ of $\Nrgame \clockgranu$ such that the model of guard
  $\guard''$ is a corner $\corner''$ of region $\reg''$ satisfying the guard
  $\overline \guard$ (recall that $\overline \guard$ is the closed
  version of $\guard$), $\corner''\in\Posttime(\corner)$,
  $\corner'=\corner''[\reset:=0]$, $\reg'=\reg''[\reset:=0]$ (the two last conditions ensure that $\corner'$ is indeed a corner of $\reg'$) and there exist two valuations
  $\val\in \reg$, $\val'\in \reg'$ such that
  $((\loc,\reg),\val)\xrightarrow{\delay',\rtrans}((\loc',\reg'),\val')$
  for some $\delay'\in \Rpos$ (the latter condition ensures that the
  edge between corners is not spurious, \ie~created by the closure of
  guards). Weights of locations and edges are trivially recovered from
  $\Ncgame \clockgranu$. We define the final weight function of
  $\Ncgame \clockgranu$ over the only valuation~$\corner$ reachable in
  location $(\loc,\reg,\corner)$ (with $\loc\in \LocsT$) by
  $\weightT((\loc,\reg,\corner),\corner)=\lim_{\val\to \corner,
    \val\in \reg}\weightT(\loc,\val)$ (the limit is well defined since
  $\weightT$ is piecewise affine with a finite number of pieces on
  region~$\reg$, by Hypothesis~\ref{hyp:final}).
\end{defi}

We denote by $|\Ncgame \clockgranu|$ the number of locations in the corner-point abstraction, bounded by $|\Locs| |\NRegs| (n+1)$.

The \WTG $\Ncgame \clockgranu$ can be seen as a finite weighted game,
\ie~a \WTG without clocks, by removing guards, resets and rates of
locations, and replacing the weights of edges by the actual weight of
jumping from one corner to another: an edge
$((\loc,\reg),\corner)\xrightarrow{\guard'',
  \reset}((\loc',\reg'),\corner')$ becomes a transition from
$((\loc,\reg),\corner)$ to $((\loc',\reg'),\corner')$ with weight
$\delay\cdot\weight(\loc) + \weight(\rtrans)$, with $\delay\in\Rpos$
the only delay such that
$\sem{\guard''}=\{\corner+\delay\}$.\footnote{In case several edges
  lead to the same transition, for instance when two transitions with
  different guards reset all clocks, we either allow for
  multi-transitions or choose the best weight according to the player
  owning the current location.}
  Note that delay $\delay$ is
  necessarily a rational of the form $\alpha/\clockgranu$ with $\alpha\in \N$,
  since it must relate corners of $1/\clockgranu$-regions.  In particular, this
  proves that the cumulative weight $\weightC(\cplay)$ of a finite play
  $\cplay$ in $\Ncgame \clockgranu$ is indeed a rational number with
  denominator~$\clockgranu$.

  We will call \emph{corner play} every play $\cplay$ in the corner-point
  abstraction $\Ncgame \clockgranu$: it can also be interpreted as an
  execution in $\game$ where all guards are closed (as explained in the
  definition above).
  It straightforwardly projects on a finite path
  $\rpath$ in the region game $\Nrgame \clockgranu$: in this case, we say again
  that $\cplay$ follows $\rpath$.  \figurename~\ref{fig:plays} depicts a play,
  its projected path in the region game and one of its associated corner
  plays.

\begin{figure}[htbp]
  \centering
\scalebox{1}{
\begin{tikzpicture}[>=latex]
  \path[draw,thick,color=black,fill=black!20!white] (0,0) -- (1,0) --
  (1,1) -- (0,0);

  \path[draw,thick,color=black] (4,0) -- (4,1);

  \path[draw,thick,color=black,fill=black!20!white] (7,0) -- (8,1) --
  (7,1) -- (7,0);

  \path[draw,thick,color=black] (11,0.3) -- (12,0.3);

  \node[draw,circle,inner
  sep=1pt,color=blue!70!black,fill=blue!40!white] (1) at (0.4,0.2) {};
  \node[draw,circle,inner
  sep=1pt,color=blue!70!black,fill=blue!40!white] (2) at (4,0.7) {};
  \node[draw,circle,inner
  sep=1pt,color=blue!70!black,fill=blue!40!white] (3) at (7.3,0.9) {};
   \node[draw,circle,inner
  sep=1pt,color=blue!70!black,fill=blue!40!white] (4) at (11.3,0.3) {};

  \path[thick,color=blue!70!black,->] (1) edge[bend left] (2);
  \path[thick,color=blue!70!black,->] (2) edge[bend left=18] (3);
  \path[thick,color=blue!70!black,->] (3) edge[bend right=40] (4);

  \node[draw,circle,inner
  sep=1pt,color=green!60!black,fill=green!20!white] (1) at (1,0) {};
  \node[draw,circle,inner
  sep=1pt,color=green!60!black,fill=green!20!white] (2) at (4,0) {};
  \node[draw,circle,inner
  sep=1pt,color=green!60!black,fill=green!20!white] (3) at (8,1) {};
  \node[draw,circle,inner
  sep=1pt,color=green!60!black,fill=green!20!white] (4) at (12,0.3) {};

  \path[thick,color=green!60!black,->] (1) edge[bend right] (2);
  \path[thick,color=green!60!black,->] (2) edge[bend right=49] (3);
  \path[thick,color=green!60!black,->] (3) edge[bend left=18] (4);

  \node[color=black] at (0,0.8) {$(\loc_0,\reg_0)$};
  \node[color=black] at (4,1.4) {$(\loc_1,\reg_1)$};
  \node[color=black] at (7.5,1.4) {$(\loc_2,\reg_2)$};
  \node[color=black] at (12.7,0.3) {$(\loc_3,\reg_3)$};
  \node[color=blue!70!black] at (2.5,1.2) {$\play$};
  \node[color=green!60!black] at (2.5,-.2) {$\cplay$};

  \path[draw,thick,color=black,->] (1.3,0.3) -- node[above]{$\guard_0,\reset_0$} (3.7,0.3);
  \path[draw,thick,color=black,->] (4.3,0.3) -- node[above]{$\guard_1,\reset_1$} (6.7,0.3);
  \path[draw,thick,color=black,->] (8.3,0.3) -- node[sloped,anchor=south]{$\guard_2,\reset_2$} (10.7,0.3);

\end{tikzpicture}
}
  \caption{A play $\play$ (in blue), its projected path $\rpath$ in the
    region game (in black), and one of its associated corner plays
    $\cplay$ (in green).}
  \label{fig:plays}
\end{figure}

\noindent Let $\cplay$ be a corner play following a region path
$\rpath$.  The weight of \cplay refers to its weight in
$\Ncgame \clockgranu$.  It is possible to find a play $\play$
following $\rpath$ close to $\cplay$, in the sense that we control the
difference between their respective cumulative weights:
\begin{lemC}[{\cite[Prop.~5]{BouBri08a}}]\label{lm:fog-exec}
  For all $\varepsilon>0$, all finite region paths $\rpath$, and all
  corner plays $\cplay$ following $\rpath$,
  there exists a play \play in \game following \rpath such that
  $|\weightC(\play)-\weightC(\cplay)|\leq \varepsilon$.
\end{lemC}

  Thus, corner plays allow one to obtain faithful information on the plays
  that follow the same path.
\begin{lem}\label{lm:cornerabstract}
  If~\rpath is a finite region path in~$\Nrgame \clockgranu$, the
  set of cumulative weights $\{\weightC(\play) \mid \play \text{ play of }\game\text{ following }
  \rpath\}$ is an interval bounded by the minimum and the maximum
  values of the
  set $\{\weightC(\cplay) \mid \cplay \text{ corner play of }
  \Ncgame \clockgranu \text{ following } \rpath \}$.
\end{lem}
\begin{proof}
  The set
  $\{\weightC(\play) \mid \play \text{ finite play following } \rpath
  \}$ is an interval as the image of a convex set by an affine
  function (see \cite[Sec.~3.2]{BouBri07} for an explanation).

  The good properties of the corner-point abstraction allow us to
  conclude, since for every play~\play following~\rpath, one can find
  a corner play following~\rpath of smaller weight and one of larger
  weight~\cite[Lemma~1]{BehFeh01},
  and for every corner play~\play following~\rpath and
  every~$\varepsilon>0$, one can find a play following~\rpath whose
  weight is at most~$\varepsilon$ away from~$\weightC(\play)$
  by Lemma~\ref{lm:fog-exec}.
\end{proof}

  An important property of the corner-point abstraction,
  derived from the assumption on the
  absence of deadlocks in the game,
  is that corner plays
  cannot get stuck as long as they follow a region path:
\begin{lemC}[{\cite[Lem.~8]{Pur00}}]\label{lm:corners-no-deadlocks}
  Let \rpath be a region path starting from $(\loc,\reg)$ and
  ending in $(\loc',\reg')$. For all corners $\corner$ of $\reg$, there exists
  a corner play following \rpath that starts in $(\loc,\reg,\corner)$.
  For all corners $\corner'$ of $\reg'$ there exists
  a corner play following \rpath that ends in $(\loc',\reg',\corner')$.
\end{lemC}

\begin{figure}
  \centering
\scalebox{1}{
\begin{tikzpicture}[>=latex]
  \path[draw,thick,color=black,fill=black!20!white] (0,0) -- (1,0) --
  (1,1) -- (0,0);

  \path[draw,thick,color=black] (4,0) -- (4,1);

  \path[draw,thick,color=black,fill=black!20!white] (7,0) -- (8,1) --
  (7,1) -- (7,0);

  \path[draw,thick,color=black,fill=black!20!white] (11,0) -- (12,1) --
  (12,0) -- (11,0);

  \node[draw,circle,inner
  sep=1pt,color=green!60!black,fill=green!20!white] (1) at (1,0) {};
  \node[draw,circle,inner
  sep=1pt,color=green!60!black,fill=green!20!white] (2) at (4,0) {};
  \node[draw,circle,inner
  sep=1pt,color=green!60!black,fill=green!20!white] (3) at (8,1) {};
  \node[draw,circle,inner
  sep=1pt,color=green!60!black,fill=green!20!white] (4) at (12,1) {};
  \node[draw,circle,inner
  sep=1pt,color=green!60!black,fill=green!20!white] (11) at (0,0) {};
  \node[draw,circle,inner
  sep=1pt,color=green!60!black,fill=green!20!white] (111) at (1,1) {};
  \node[draw,circle,inner
  sep=1pt,color=green!60!black,fill=green!20!white] (22) at (4,1) {};
  \node[draw,circle,inner
  sep=1pt,color=green!60!black,fill=green!20!white] (33) at (7,0) {};
  \node[draw,circle,inner
  sep=1pt,color=green!60!black,fill=green!20!white] (333) at (7,1) {};
  \node[draw,circle,inner
  sep=1pt,color=green!60!black,fill=green!20!white] (44) at (11,0) {};
  \node[draw,circle,inner
  sep=1pt,color=green!60!black,fill=green!20!white] (444) at (12,0) {};

  \path[thick,color=green!60!black,->] (1) edge[bend right=18] (2);
  \path[thick,color=green!60!black,->] (2) edge[bend right=59] (3);
  \path[thick,color=green!60!black,->] (3) edge[bend left=18] (444);
  \path[thick,color=green!60!black,->] (2) edge[bend right=18] (33);
  \path[thick,color=green!60!black,->] (11) edge[bend right=50] (2);
  \path[thick,color=green!60!black,->] (111) edge[bend left=18] (22);
  \path[thick,color=green!60!black,->] (22) edge[bend left=18] (333);
  \path[thick,color=green!60!black,->] (333) edge[bend left=28] (44);
  \path[thick,color=green!60!black,->] (333) edge[bend left=28] (4);
  \path[thick,color=green!60!black,->] (33) edge[bend left=6] (4);
  \path[thick,color=green!60!black,->] (33) edge[bend right=18] (44);

  \node[color=black] at (0,1) {$(\loc,\reg)$};
  \node[color=black] at (4.2,1.5) {$(\loc_1,\reg_1)$};
  \node[color=black] at (7.3,1.5) {$(\loc_2,\reg_2)$};
  \node[color=black] at (12.7,0.3) {$(\loc,\reg)$};

  \path[draw,thick,color=black,->] (1.3,0.3) -- %
  (3.7,0.3);
  \path[draw,thick,color=black,->] (4.3,0.3) -- %
  (6.7,0.3);
  \path[draw,thick,color=black,->] (8.3,0.3) -- %
  (10.7,0.3);
\end{tikzpicture}
}
\centering
\scalebox{1.2}{
\begin{tikzpicture}[>=latex]

  \path[draw,thick,color=black,fill=black!20!white] (0,0) -- (1,0) --
  (1,1) -- (0,0);
  \node[draw,circle,inner
  sep=1pt,color=green!60!black,fill=green!20!white] (1) at (1,0) {};
  \node[draw,circle,inner
  sep=1pt,color=green!60!black,fill=green!20!white] (11) at (0,0) {};
  \node[draw,circle,inner
  sep=1pt,color=green!60!black,fill=green!20!white] (111) at (1,1) {};
  \path[thick,color=blue!70!black,->] (1) edge[loop right] (1);
  \path[thick,color=blue!70!black,->] (11) edge[loop left] (11);
  \path[thick,color=blue!70!black,->] (11) edge[bend right] (1);
  \path[thick,color=blue!70!black,->] (111) edge[loop above] (111);
  \path[thick,color=blue!70!black,->] (111) edge[bend right=59] (11);
  \path[thick,color=blue!70!black,->] (11) edge[bend left=25] (111);
  \path[thick,color=blue!70!black,->] (1) edge[bend left=75] (11);
  \path[thick,color=blue!70!black,->] (1) edge[bend right] (111);
  \node[color=black] at (0,1.1) {$(\loc,\reg)$};
\end{tikzpicture}
}
  \caption{A region cycle $\rpath$ in the region game (in black), its associated
  corner plays (in green), and its folded orbit graph (in blue).
  Note that there is no edge between the top right corner and
  the bottom right corner, as no corner play goes from the former to the latter.}
  \label{fig:fog}
\end{figure}

Useful theoretical tools stem from the corner-point abstraction.
Notably, let us focus on a cycle of the region abstraction.  In order to
study some properties of the corner plays following this cycle, we
only need to consider the aggregation of all the behaviours following
it.  Inspired by the \emph{folded orbit graphs} (FOG) \label{page:FOG} introduced
in~\cite{Pur00}, we define the folded orbit graph $\FOG(\rpath)$ of a
region cycle
$\rpath=(\loc_1,\reg=\reg_1) \xrightarrow{\edge_1} (\loc_2,\reg_2)
\xrightarrow{\edge_2} \cdots \xrightarrow{\edge_k} (\loc_1,\reg)$ in
$\Nrgame \clockgranu$ as a graph whose vertices are the corners states
of region $\reg$, and that contains an edge from corner $\corner$ to
corner $\corner'$ if there exists a corner play $\cplay$ from
$(\loc_1,\reg,\corner)$ to $(\loc_1,\reg,\corner')$ following
$\rpath$.  We fix $\cplay$ arbitrarily and label the edge between
$\corner$ and $\corner'$ in $\FOG(\rpath)$ by this corner play: it is
then denoted by $\corner\xrightarrow{\cplay}\corner'$.  An example is
depicted in \figurename~\ref{fig:fog}.

The folded orbit graph inherits interesting topological properties
from the corner-point abstraction. Notably, by
Lemma~\ref{lm:corners-no-deadlocks}, for all vertices $\corner$, there
exists at least one outgoing edge
$\corner\xrightarrow{\cplay'}\corner'$, and at least one incoming edge
$\corner''\xrightarrow{\cplay''}\corner$ in~$\FOG(\rpath)$.

\subsection{Value iteration algorithm}\label{subsec:value-functions}

The value of a game has been defined as a mapping of each
configuration $(\ell,\val)$ to a value in $\Rbar$. We call \emph{value
  functions} such mappings from $\Locs\times\Rpos^\Clocks$ to
$\Rbar$. If $\ValIteVec$ represents a value function, we denote by
$\ValIteVec_\loc$ the mapping $\val\mapsto\ValIteVec(\loc,\val)$. As
observed in~\cite{BCFL04,AluBer04}, one step of the game is summarised
in the following operator $\ValIteOpe$ mapping each value function
$\ValIteVec$ to a value function $\ValIteVec'=\ValIteOpe(\ValIteVec)$
defined by $\ValIteVec'_{\loc}(\val)=\weightT(\loc,\val)$ if
$\loc\in\LocsT$, and otherwise
  \begin{equation}\ValIteVec'_{\loc}(\val)=
  \begin{cases}
    \sup_{(\loc,\val)\xrightarrow{\delay,\edge}(\loc',\val')}
     \big[\delay\cdot\weight(\loc)+\weight(\edge)+\ValIteVec_{\loc'}(\val')\big]
     &
    \text{if }\loc\in\LocsMax\\
    \inf_{(\loc,\val)\xrightarrow{\delay,\edge}(\loc',\val')}
     \big[\delay\cdot\weight(\loc)+\weight(\edge)+\ValIteVec_{\loc'}(\val') \big] &
    \text{if }\loc\in\LocsMin
  \end{cases}\label{eq:operator}\end{equation}
\noindent where $(\loc,\val)\xrightarrow{\delay,\edge}(\loc',\val')$
ranges over valid transitions in \game.

Then, starting from $\ValIteVec^0$ mapping every configuration
$(\loc,\val)$ to $+\infty$, except for the targets mapped to
$\weightT(\loc,\val)$, we let
$\ValIteVec^i= \ValIteOpe(\ValIteVec^{i-1})$ for all $i>0$. The value
function $\ValIteVec^i$ contains the value $\Val^i_\game$, which is
intuitively what \MinPl can guarantee when forced to reach the target
in at most $i$ steps. More formally, we define $\weight^i(\play)$ the
weight of a maximal play $\play$ at horizon $i$, as $\weight(\play)$
if $\play$ reaches a target in at most $i$ steps, and $+\infty$
otherwise. Then,
$\Val^i_\game(\loc,\val)=\inf_{\stratmin}\sup_{\stratmax}
\weight^i(\outcome((\loc,\val),\stratmin,\stratmax))$ refers to the
value at horizon $i$.

\begin{rem}\label{rem:attractor}
In case of non-weighted games, with or without time, the value iteration algorithm is generally called \emph{attractor}. Starting from $+\infty$ for every configuration, except for the target mapped to $0$, we compute the previous iterates (with all weights mapped to $0$). In this case, as shown in \cite{AsaMal98}, values $0$ or $+\infty$ stay constant over each regions, and there are thus a finite number of possible functions, which ensures that the computation ends in finite time.
\end{rem}

We compare value functions componentwise: if $\ValIteVec, \ValIteVec'$
are two value functions, we let $\ValIteVec\leq \ValIteVec'$ if
$\ValIteVec(\ell,\val)\leq \ValIteVec'(\ell,\val)$ for all
configurations $(\ell,\val)$. Notice that \ValIteOpe is a monotonic
operator, \ie~if $\ValIteVec\leq\ValIteVec'$, then
$\ValIteOpe(\ValIteVec)\leq\ValIteOpe(\ValIteVec')$. Moreover,
$\ValIteOpe(\ValIteVec^0)\leq \ValIteVec^0$ since $\ValIteVec^0$ maps
every non-target state to $+\infty$, and target states keep the same
value. It follows that the sequence $(\ValIteVec^i)_{i\in\N}$ is
non-increasing, as
$\ValIteVec^{i}=\ValIteOpe^i(\ValIteVec^0)\geq
\ValIteOpe^i(\ValIteOpe(\ValIteVec^0))=\ValIteVec^{i+1}$.

Let us now present known results for the special case of games with no clocks.
In this case, the definition of the
operator $\ValIteOpe$ of~\eqref{eq:operator} is simplified into
  \begin{equation}\ValIteVec'_{\loc}=
  \begin{cases}
    \min_{\ell\rightarrow\ell}
     \big[\weight(\ell,\ell')+\ValIteVec_{\ell'}\big]
     &
    \text{if }\ell\in\LocsMin\\
    \max_{\ell\rightarrow\ell'}
     \big[\weight(\ell,\ell')+\ValIteVec_{\ell'}\big]
     &
    \text{if }\ell\in\LocsMax
  \end{cases}\label{eq:operator_untimed}\end{equation}
Notice that we remove the valuation part of the notation:
configurations are thus simply locations of the game. The value
iteration algorithm proposed in~\cite{BGHM16} consists in finding the
greatest fixpoint of operator $\ValIteOpe$, \ie~the limit of the
sequence $(\ValIteVec^i)_{i\in\N}$. Indeed, this greatest fixpoint is
known to be the vector of values of the game (see, \eg,
\cite[Corollary~11]{BGHM16}). For the special case of acyclic games of
depth~$d$, the fixpoint is reached after~$d$ steps, and
$\Val= \ValIteVec^d$. In this case, the infinite values in $\game$
(\ie~configurations $\ell$ with
$\Val_\game(\ell)\in\{-\infty,+\infty\}$) are derived from reaching
targets with infinite final weights. If the game contains cycles,
infinite values can also come from arbitrarily long plays: a state can
have value $+\infty$ if $\MaxPl$ can force an infinite play, never
reaching any target, and it can have value $-\infty$ if $\MinPl$ can
enforce an arbitrarily low weight, \eg~by staying in a cycle of
negative cumulative weight. These $+\infty$ states correspond to a
safety objective for player $\MaxPl$, and can be computed in
polynomial time: it is shown in \cite{BGHM16} that for all locations
$\ell$, $\Val_\game(\ell)=+\infty$ if and only if
$\ValIteVec^{|\Locs|}_\ell=+\infty$.  In contrast, deciding if a
location has value $-\infty$ has no known polynomial solution (it is as
hard as solving mean-payoff games).  In \cite{BGHM16}, it is shown
that in the presence of negative weights the sequence
$(\ValIteVec^i)_{i\in\N}$ stabilises after a number of iterations
pseudo-polynomial on states with value in $\R\cup\{+\infty\}$, and
that states with value $-\infty$ can be detected in this computation
(they are those where the computed value goes under a given
threshold).

Such a computation of the greatest fixed point might
not be possible in the timed setting, since the value problem is
undecidable in general. However, in \cite{AluBer04}, it is shown that starting from a value function $\ValIteVec$ that is represented, for each location, by a pair $(P,F)$ where $P$ is partition and $F$ is a partition function defined over $P$, the same is true for $\ValIteOpe(\ValIteVec)$.
Indeed, cells are bounded, convex polyhedra, over which elementary operations
(emptiness, intersection and inclusion tests) can be performed with linear programming, and thus in polynomial
time. As we will formally recall in Section~\ref{sec:acyclic},
this data structure allows one to effectively compute the iterates
$(\ValIteVec^i)_{i\in \N}$. Contrary to the case without clocks,
recalled before, this sequence does not stabilise in general. The
decidability results obtained before (for tree-shaped weighted timed games \cite{AluBer04}, or strictly non-Zeno weighted timed games~\cite{BCFL04}, \eg), as well as the ones we obtain in
this article, are all based on reasons to either make the sequence
stabilise or stop its computation after a sufficient number $i$ of
turns to obtain a good approximation of the value.

\begin{rem}
  In~\cite{AluBer04}, the domain of partitions is always a single
  region, and one value function is associated to each region.  We
  define value functions over \ValSpaceBound instead, in order to
  obtain a symbolic algorithm, independent of regions.  This induces
  slight differences in the way value functions are defined, because
  the mappings of~\cite{AluBer04} are continuous everywhere while ours
  can have discontinuities at borders between regions.  They define
  their partitions with overlaps over borders, such that $\ValSpace$
  is partitioned by an affine equality into two cells, $c_{\leq}$ and
  $c_{\geq}$, instead of the three $c_{<}$, $c_{>}$ and $c_{=}$.  This
  changes the number of cells $\splitnum(m,n)$ to $\mathcal O(m^n)$ instead of $\mathcal O((2m)^n)$.
\end{rem}

\section{Divergent and almost-divergent \WTG{s}}\label{sec:results}

  In this section, we introduce several classes of weighted timed
  games for which we state the results that this article shows. All classes are defined in terms of the underlying timed automaton, without making use of the partition of locations into players.
  Let us start with the class of weighted timed games
  studied in \cite{BCFL04}, to our knowledge the greatest class of \WTG
  where the value problem is known to be decidable.
\begin{defi}
  A weighted timed game \game with non-negative weights
  satisfies the \emph{strictly non-Zeno cost} property when every finite
  play~\play in \game following a cycle in the region automaton
  $\rgame$ satisfies $\weightC(\play)\geq 1$.
\end{defi}

The intuition behind this class is that the weight of any long enough
execution in \game will ultimately grow above any fixed bound, and
diverge towards $+\infty$ for an infinite execution. Therefore, the
value of \game is equal to the value $\Val^i_\game$ at some horizon
$i$ large enough (as defined in Section~\ref{subsec:value-functions}),
making the value problem decidable.  It is shown in \cite{BCFL04} that
$i$ can be bounded exponentially in the size of
$\game$. %

  We introduce divergent weighted timed games, as a natural
  generalisation of the
  strictly non-Zeno cost property to weights in \Z.
\begin{defi}
  A \WTG \game is \emph{divergent} when every finite
  play~\play in \game following a cycle in the region automaton
  $\rgame$ satisfies $\weightC(\play)\notin (-1,1)$.%
\end{defi}

  If $\game$ has only
  non-negative weights on locations and edges, this definition
  matches with the strictly non-Zeno cost property
  of~\cite{BCFL04}, we will therefore refer to their class
  as the class of divergent \WTG with non-negative weights.

\begin{rem}
  As in~\cite{BCFL04}, we could replace $(-1,1)$ by $(-\kappa,\kappa)$ to
  define a notion of $\kappa$-divergence.  However, since weights and
  guard constraints in weighted timed games are integers, for
  $\kappa\in(0,1)$, a weighted timed game $\game$ is
  $\kappa$-divergent if and only if it is divergent.
  This will be formally implied by Proposition~\ref{prop:timed-scc-sign}
  and Lemma~\ref{lm:cornerabstract}.
\end{rem}

 Our contributions on divergent \WTG{s} summarise as follows:
\begin{thm}\label{thm:div_wtg}
  The value problem over divergent \WTG{s} is decidable in $3$-\EXP,
  and is \EXP-hard.  Moreover, deciding if a given \WTG is divergent
  is a \PSPACE-complete problem.
\end{thm}

\begin{figure}
  \centering

\begin{tikzpicture}
  \draw[-,thick] (0cm,2cm) -- (0cm,-1.5cm);
  \draw (-1.25cm,1.75cm) node {weights in \Z};
  \draw (1.25cm,1.75cm) node {weights in \N};
  \draw[-,rounded corners] (0,.5cm) -- (5cm,.5cm) -- (5cm,-.5cm) -- (-5cm,-.5cm) -- (-5cm,.5cm) -- (0,.5cm);
  \draw (-2.5cm,.25cm) node {divergent};
  \draw (2.5cm,.25cm) node {strictly non-Zeno cost};
  \draw[-,rounded corners] (0,1.5cm) -- (5.5cm,1.5cm) -- (5.5cm,-1cm) -- (-5.5cm,-1cm) -- (-5.5cm,1.5cm) -- (0,1.5cm);
  \draw (-3cm,1.25cm) node {almost-divergent};
  \draw (3cm,1.25cm) node {simple};
\end{tikzpicture}

  \caption{Classes of weighted timed games,
  and their respective restrictions to non-negative weights.}
  \label{fig:class_schema}
\end{figure}

  In \cite{BJM15}, the authors slightly extend the strictly non-Zeno cost property,
  to allow for cycles of weight exactly $0$ while still preventing those
  of weight arbitrarily close to $0$:
  \begin{defi}
    A \WTG \game with non-negative weights is called \emph{simple}
    when every finite
    play~\play in \game following a cycle in the region automaton
    $\rgame$ satisfies $\weightC(\play)\in \{0\}\cup[1,+\infty)$.
  \end{defi}

  Unfortunately, it is shown in
  \cite{BJM15} that the value problem is undecidable for simple \WTG{s}.
  They propose a solution to the value approximation problem,
  as a %
  procedure computing
  an approximation of the value of every configuration.
  The intuition is that cycles of weight exactly $0$ are only possible when
  every (non-negative) weight encountered along the cycle equals~$0$,
  allowing one to define
  a subgame where every cyclic execution has weight~$0$.
  One can then analyse this subgame separately, by applying a semi-unfolding
  procedure on \rgame. %

  We introduce a class of \WTG{s} that will extend the notion of simple \WTG{s}
  and allow negative weights, so that cyclic executions of weight exactly $0$
  are allowed, but not those close to $0$. The first attempt would
  lead to the requirement that every finite
   play following a cycle in the region automaton
   $\rgame$ has a weight in
   $(-\infty,-1]\cup\{0\}\cup[1,+\infty)$. However, we did not obtain
   positive results for the value approximation problem on this class
   of \WTG{s} since cycles of weight exactly $0$ do not have the good
   property presented above for simple \WTG{s}. Instead, we
   require a stability by decomposition for cycles of weight $0$.

   If
   $\rpath=(\loc_0,r_0)\xrightarrow{r_1,e_0} (\loc_1,r_1)
   \xrightarrow{r_2,e_1}\cdots
   (\loc_{k-1},r_{k-1})\xrightarrow{r_0,e_{k-1}} (\loc_0,r_0)$ is a
   region cycle in \rgame, it is either simple (\ie~for all $i,j$ such
   that $0\leq i<j<k$, $(\loc_i,r_i)\neq(\loc_j,r_j)$) or we can
   extract smaller cycles from it.  Indeed, if $\rpath$ is not simple,
   there exists a pair $(i,j)$ such that $0\leq i<j<k$ and
   $(\loc_i,r_i)=(\loc_j,r_j)$.  Then, for such a pair, we can write
   $\rpath=\rpath_1\rpath_2\rpath_3$ such that $|\rpath_1|=i$,
   $|\rpath_3|=k-j$.  It follows that $\rpath_2$ and $\rpath_3\rpath_1$ are both region cycles
   around $(\loc_i,r_i)$.
   This process is called a decomposition of
   $\rpath$ into smaller cycles $\rpath'=\rpath_3\rpath_1$ and
   $\rpath ''=\rpath_2$.

\begin{defi}\label{def:almost-divergent-plays}
  A \WTG~\game is \emph{almost-divergent} if every play $\play$
  following a cycle \rpath of \rgame satisfies
  $\weightC(\play)\in (-\infty,-1]\cup\{0\}\cup[1,+\infty)$, and if $\weightC(\play)=0$ then for
  every decomposition of $\rpath$ into smaller cycles $\rpath'$ and $\rpath''$, and plays $\play'$ and $\play''$ following $\rpath'$ and $\rpath''$, respectively, it holds that
  $\weightC(\play')=\weightC(\play'')=0$.\footnote{Once again, we
    could replace $-1,1$ by $-\kappa,\kappa$ with $0<\kappa<1$ to
    define an equivalent notion of $\kappa$-almost-divergence.}
\end{defi}

  Clearly, every divergent \WTG is almost-divergent.
  Moreover, as we will see in Proposition~\ref{prop:almost-divergent}, when weights are non-negative, this class
  matches the simple \WTG{s} of~\cite{BJM15} therefore inheriting
  their undecidability result.
  We will thus refer to simple \WTG{s}
  as almost-divergent \WTG{s} with non-negative weights.
  \figurename~\ref{fig:class_schema} represents the hierarchy of the classes of \WTG
  that we introduced.  %

\begin{exa}
  Consider the \WTG \game in \figurename~\ref{fig:wtg}.
  The self-loop on $\loc_1$ contains a cycle of $\rgame$ around $(\loc_1,\valnull)$ that jumps to the region $1<\clock<2$.
  For every $d\in(1,2)$, there exists a play $\play$ following this cycle that uses delay $d$, so that $\weightC(\play)=3-2d\in(-1,1)$.
  It follows that $\game$ is neither divergent nor almost-divergent.
  Changing the guard on this self-loop to $2\leq\clock<3$ makes $\game$ divergent, as every region cycle left in \game iterates the self-loops around $(\loc_1,\valnull)$ and $(\loc_5,\valnull)$ of cumulative weights in $(-3,-1]$ and $(-5,-1)$, respectively.
\end{exa}

\begin{figure}
  \centering

\begin{tikzpicture}[node distance=4cm,auto,->,>=latex]
  \node[player2](0){\makebox[0mm][c]{$\mathbf{0}$}};
  \node()[below of=0,node distance=6mm]{$\loc_0$};
  \node[player1](1)[right of=0]{\makebox[0mm][c]{$\mathbf{1}$}};
  \node()[below of=1,node distance=6mm]{$\loc_1$};
  \node[player1](2)[above left of=1]{\makebox[0mm][c]{$\mathbf{-1}$}};
  \node()[above of=2,node distance=6mm]{$\loc_2$};
  \node[player1](3)[above right of=1]{\makebox[0mm][c]{$\mathbf{1}$}};
  \node()[above of=3,node distance=6mm]{$\loc_3$};
  \node[player2](4)[below right of=3]{\makebox[0mm][c]{$\mathbf{0}$}};
  \node()[below of=4,node distance=6mm]{$\loc_4$};
  \node[player1](5)[right of=3,accepting]{\makebox[0mm][c]{$\mathbf{}$}};
  \node()[above of=5,node distance=6mm]{$\loc_t$};
  \node()[below of=5,node distance=6mm]{$\weightT(\clock_1,\clock_2)=\clock_1$};

  \path
  (0) edge node[below]{$\begin{array}{c} 0<\clock_1<1\\ \clock_1:=0 \end{array}$} node[above]{$\mathbf{0}$} (1)
  (1) edge[bend left=20] node[left,xshift=-2mm,yshift=2mm]{$\begin{array}{c} \clock_2<2\\1<\clock_1<2\\ \clock_2:=0 \end{array}$} node[above]{$\mathbf{0}$} (2)
  (2) edge[bend left=20] node[right,xshift=-3mm,yshift=2mm]{$\begin{array}{c} 1<\clock_1<2\\ \clock_1:=0 \end{array}$} node[below]{$\mathbf{1}$} (1)
  (1) edge node[right]{$\begin{array}{c} \clock_2=1\\ \clock_2:=0 \end{array}$} node[left]{$\mathbf{1}$} (3)
  (3) edge node[below,xshift=-3mm]{$\begin{array}{c} \clock_1=1 \end{array}$} node[right]{$\mathbf{0}$} (4)
  (4) edge node[below]{$\begin{array}{c} 1<\clock_1<2, \clock_2<1\\\clock_1:=0 \end{array}$} node[above]{$\mathbf{-2}$} (1)
  (3) edge node[below]{$\begin{array}{c} \clock_2=0 \end{array}$} node[above]{$\mathbf{0}$} (5);
\end{tikzpicture}

\begin{tikzpicture}[node distance=4cm,auto,->,>=latex]

  \node[draw,regular polygon,regular polygon sides=4,inner sep=-2mm,minimum size=10mm](0){\makebox[0mm][c]{
  \scalebox{0.5}{\begin{tikzpicture}
    \tikzset{font=\LARGE}
    \path[draw,->,thick](0,0) -> (2.3,0) node[above] {$\clock_1$};
    \path[draw,->,thick](0,0) -> (0,2.3) node[right] {$\clock_2$};
    \path[draw,-] (0,0) -- (2,2)
    (1,0) -- (1,2) -- (0,1) -- (2,1) -- (1,0)
    (0,2) -- (2,2) -- (2,0);
    \node () at (1,-0.3) {$1$};
    \node () at (2,-0.3) {$2$};
    \node () at (-0.3,-0.3) {$0$};
    \node () at (-0.3,1) {$1$};
    \node () at (-0.3,2) {$2$};
    \node[draw, fill=black,circle,inner sep=1mm, minimum size=1mm] () at (0,0) {};
  \end{tikzpicture}}}};
  \node()[below of=0,node distance=12.5mm]{$\loc_0, \mathbf{0}$};

  \node[player1,minimum size=10mm](1)[right of=0]{\makebox[0mm][c]{
  \scalebox{0.5}{\begin{tikzpicture}
    \tikzset{font=\LARGE}
    \path[draw,->,thick](0,0) -> (2.3,0) node[above] {$\clock_1$};
    \path[draw,->,thick](0,0) -> (0,2.3) node[right] {$\clock_2$};
    \path[draw,-] (0,0) -- (2,2)
    (1,0) -- (1,2) -- (0,1) -- (2,1) -- (1,0)
    (0,2) -- (2,2) -- (2,0);
    \node () at (1,-0.3) {$1$};
    \node () at (2,-0.3) {$2$};
    \node () at (-0.3,-0.3) {$0$};
    \node () at (-0.3,1) {$1$};
    \node () at (-0.3,2) {$2$};
    \path[draw, -, fill=black, line width=1.5mm] (0,0.1) -- (0,0.9);
  \end{tikzpicture}}}};
  \node()[below of=1,node distance=12.5mm]{$\loc_1, \mathbf{1}$};

  \node[player1,minimum size=10mm](2)[above left of=1]{\makebox[0mm][c]{
  \scalebox{0.5}{\begin{tikzpicture}
    \tikzset{font=\LARGE}
    \path[draw,->,thick](0,0) -> (2.3,0) node[above] {$\clock_1$};
    \path[draw,->,thick](0,0) -> (0,2.3) node[right] {$\clock_2$};
    \path[draw,-] (0,0) -- (2,2)
    (1,0) -- (1,2) -- (0,1) -- (2,1) -- (1,0)
    (0,2) -- (2,2) -- (2,0);
    \node () at (1,-0.3) {$1$};
    \node () at (2,-0.3) {$2$};
    \node () at (-0.3,-0.3) {$0$};
    \node () at (-0.3,1) {$1$};
    \node () at (-0.3,2) {$2$};
    \path[draw, -, fill=black, line width=1.5mm] (1.1,0) -- (1.9,0);
  \end{tikzpicture}}}};
  \node()[above of=2,node distance=12.5mm]{$\loc_2, \mathbf{-1}$};

  \node[player1,minimum size=10mm](3)[above right of=1]{\makebox[0mm][c]{
  \scalebox{0.5}{\begin{tikzpicture}
    \tikzset{font=\LARGE}
    \path[draw,->,thick](0,0) -> (2.3,0) node[above] {$\clock_1$};
    \path[draw,->,thick](0,0) -> (0,2.3) node[right] {$\clock_2$};
    \path[draw,-] (0,0) -- (2,2)
    (1,0) -- (1,2) -- (0,1) -- (2,1) -- (1,0)
    (0,2) -- (2,2) -- (2,0);
    \node () at (1,-0.3) {$1$};
    \node () at (2,-0.3) {$2$};
    \node () at (-0.3,-0.3) {$0$};
    \node () at (-0.3,1) {$1$};
    \node () at (-0.3,2) {$2$};
    \path[draw, -, fill=black, line width=1.5mm] (0.1,0) -- (0.9,0);
  \end{tikzpicture}}}};
  \node()[above of=3,node distance=12.5mm]{$\loc_3, \mathbf{1}$};

  \node[draw,regular polygon,regular polygon sides=4,inner sep=-2mm,minimum size=10mm](4)[below right of=3]{\makebox[0mm][c]{
  \scalebox{0.5}{\begin{tikzpicture}
    \tikzset{font=\LARGE}
    \path[draw,->,thick](0,0) -> (2.3,0) node[above] {$\clock_1$};
    \path[draw,->,thick](0,0) -> (0,2.3) node[right] {$\clock_2$};
    \path[draw,-] (0,0) -- (2,2)
    (1,0) -- (1,2) -- (0,1) -- (2,1) -- (1,0)
    (0,2) -- (2,2) -- (2,0);
    \node () at (1,-0.3) {$1$};
    \node () at (2,-0.3) {$2$};
    \node () at (-0.3,-0.3) {$0$};
    \node () at (-0.3,1) {$1$};
    \node () at (-0.3,2) {$2$};
    \path[draw, -, fill=black, line width=1.5mm] (1,0.1) -- (1,0.9);
  \end{tikzpicture}}}};
  \node()[below of=4,node distance=12.5mm]{$\loc_4, \mathbf{0}$};

  \node[player1,minimum size=10mm](5)[right of=3,accepting]{\makebox[0mm][c]{
  \scalebox{0.5}{\begin{tikzpicture}
    \tikzset{font=\LARGE}
    \path[draw,->,thick](0,0) -> (2.3,0) node[above] {$\clock_1$};
    \path[draw,->,thick](0,0) -> (0,2.3) node[right] {$\clock_2$};
    \path[draw,-] (0,0) -- (2,2)
    (1,0) -- (1,2) -- (0,1) -- (2,1) -- (1,0)
    (0,2) -- (2,2) -- (2,0);
    \node () at (1,-0.3) {$1$};
    \node () at (2,-0.3) {$2$};
    \node () at (-0.3,-0.3) {$0$};
    \node () at (-0.3,1) {$1$};
    \node () at (-0.3,2) {$2$};
    \path[draw, -, fill=black, line width=1.5mm] (0.1,0) -- (0.9,0);
  \end{tikzpicture}}}};
  \node()[above of=5,node distance=12.5mm]{$\loc_t$};
  \node()[below of=5,node distance=12.5mm]{$\weightT(\clock_1,\clock_2)=\clock_1$};

  \path
  (0) edge node[below]{\makebox[0mm][c]{
  \scalebox{0.3}{\begin{tikzpicture}
    \path[draw,->,thick](0,0) -> (2.3,0);
    \path[draw,->,thick](0,0) -> (0,2.3);
    \path[draw,-] (0,0) -- (2,2)
    (1,0) -- (1,2) -- (0,1) -- (2,1) -- (1,0)
    (0,2) -- (2,2) -- (2,0);
    \path[draw, -, fill=black, line width=1.5mm] (0.1,0.1) -- (0.9,0.9);
    \path[draw,->, fill=black, line width=1.5mm](2.3,-0.3) -> (0,-0.3);
  \end{tikzpicture}}}} node[above]{$\mathbf{0}$} (1)

  (1) edge[bend left=20] node[left,xshift=-5mm]{\makebox[0mm][c]{
  \scalebox{0.3}{\begin{tikzpicture}
    \path[draw,->,thick](0,0) -> (2.3,0);
    \path[draw,->,thick](0,0) -> (0,2.3);
    \path[draw,-] (0,0) -- (2,2)
    (1,0) -- (1,2) -- (0,1) -- (2,1) -- (1,0)
    (0,2) -- (2,2) -- (2,0);
    \path[draw, -, fill=black] (1.1,1.3) -- (1.7,1.9) -- (1.1,1.9) -- (1.1,1.3);
    \path[draw,->, fill=black, line width=1.5mm](-0.3,2.3) -> (-0.3,0);
  \end{tikzpicture}}}} node[above]{$\mathbf{0}$} (2)

  (2) edge[bend left=20] node[right,xshift=2mm,yshift=5mm]{\makebox[0mm][c]{
  \scalebox{0.3}{\begin{tikzpicture}
    \path[draw,->,thick](0,0) -> (2.3,0);
    \path[draw,->,thick](0,0) -> (0,2.3);
    \path[draw,-] (0,0) -- (2,2)
    (1,0) -- (1,2) -- (0,1) -- (2,1) -- (1,0)
    (0,2) -- (2,2) -- (2,0);
    \path[draw, -, fill=black] (1.3,0.15) -- (1.9,0.7) -- (1.9,0.15) -- (1.3,0.15);
    \path[draw,->, fill=black, line width=1.5mm](2.3,-0.3) -> (0,-0.3);
  \end{tikzpicture}}}} node[below]{$\mathbf{1}$} (1)

  (1) edge node[below right,xshift=1mm]{\makebox[0mm][c]{
  \scalebox{0.3}{\begin{tikzpicture}
    \path[draw,->,thick](0,0) -> (2.3,0);
    \path[draw,->,thick](0,0) -> (0,2.3);
    \path[draw,-] (0,0) -- (2,2)
    (1,0) -- (1,2) -- (0,1) -- (2,1) -- (1,0)
    (0,2) -- (2,2) -- (2,0);
    \path[draw, -, fill=black, line width=1.5mm] (0.1,1) -- (0.9,1);
    \path[draw,->, fill=black, line width=1.5mm](-0.3,2.3) -> (-0.3,0);
  \end{tikzpicture}}}} node[left]{$\mathbf{1}$} (3)

  (3) edge node[below,xshift=-2mm]{\makebox[0mm][c]{
  \scalebox{0.3}{\begin{tikzpicture}
    \path[draw,->,thick](0,0) -> (2.3,0);
    \path[draw,->,thick](0,0) -> (0,2.3);
    \path[draw,-] (0,0) -- (2,2)
    (1,0) -- (1,2) -- (0,1) -- (2,1) -- (1,0)
    (0,2) -- (2,2) -- (2,0);
    \path[draw, -, fill=black, line width=1.5mm] (1,0.1) -- (1,0.9);
  \end{tikzpicture}}}} node[right]{$\mathbf{0}$} (4)

  (4) edge node[below]{\makebox[0mm][c]{
  \scalebox{0.3}{\begin{tikzpicture}
    \path[draw,->,thick](0,0) -> (2.3,0);
    \path[draw,->,thick](0,0) -> (0,2.3);
    \path[draw,-] (0,0) -- (2,2)
    (1,0) -- (1,2) -- (0,1) -- (2,1) -- (1,0)
    (0,2) -- (2,2) -- (2,0);
    \path[draw, -, fill=black] (1.1,0.3) -- (1.7,0.85) -- (1.1,0.85) -- (1.1,0.3);
    \path[draw,->, fill=black, line width=1.5mm](2.3,-0.3) -> (0,-0.3);
  \end{tikzpicture}}}} node[above]{$\mathbf{-2}$} (1)

  (3) edge node[below]{\makebox[0mm][c]{
  \scalebox{0.3}{\begin{tikzpicture}
    \path[draw,->,thick](0,0) -> (2.3,0);
    \path[draw,->,thick](0,0) -> (0,2.3);
    \path[draw,-] (0,0) -- (2,2)
    (1,0) -- (1,2) -- (0,1) -- (2,1) -- (1,0)
    (0,2) -- (2,2) -- (2,0);
    \path[draw, -, fill=black, line width=1.5mm] (0.1,0) -- (0.9,0);
  \end{tikzpicture}}}} node[above]{$\mathbf{0}$} (5);

\end{tikzpicture}

  \caption{
  A weighted timed game \game with two clocks $\clock_1$ and $\clock_2$,
  and the portion of its region game \rgame accessible from configuration $(\loc_0,(0,0))$.
  The states of \rgame are labelled by their associated region, location and weight,
  and edges are labelled by a representation of their guards and resets.
  For example, the edge from $(\loc_0,\reg_0)$ to $(\loc_1,\reg_1)$ in $\rgame$ highlights the time successors of the region $\reg_0$ that satisfy the guard $0<\clock_1<1$, and the arrow represents the direction in which this set of points is projected by the clock reset $\clock_1:=0$, so that we end up in the region $\reg_1$.
  }
  \label{fig:example-regions}
\end{figure}
\begin{exa}\label{ex:example-regions}
  Consider the \WTG \game in \figurename~\ref{fig:example-regions}, and its region game \rgame.
We chose an example where $\rgame$ is isomorphic to \game
for readability reasons.
\rgame contains one SCC $\{\loc_1,\loc_2,\loc_3,\loc_4\}$, made of two simple cycles, $\rpath_1=\loc_1\rightarrow\loc_2\rightarrow\loc_1$ and $\rpath_2=\loc_1\rightarrow\loc_3\rightarrow\loc_4\rightarrow\loc_1$, so that:
\begin{itemize}
\item all plays following $\rpath_1$ have cumulative weight in the interval $(1,3)$,
\item and all plays following $\rpath_2$ have cumulative weight $0$.
This can be checked by Lemma~\ref{lm:cornerabstract}.
\end{itemize}

\noindent As every cycle \rpath of \rgame either iterates $\rpath_2$ only, or contains $\rpath_1$, it holds that a play $\play$ following \rpath satisfies
$\weightC(\play)\in \{0\}\cup[1,+\infty)$,
and if $\weightC(\play)=0$ then
any decomposition of $\rpath$ into smaller cycles $\rpath'$ and $\rpath''$
implies that they all follow iterates of $\rpath_2$, so that plays along them must have cumulative weight $0$.
Therefore, \game is almost-divergent.
If one removes $\loc_4$, \game becomes divergent.
\end{exa}

  Our first result on almost-divergent \WTG{s} is the following
  extension of the approximation procedure for non-negative weights:
\begin{thm}\label{thm:almost-div}
  Given an almost-divergent \WTG \game, a location $\loc$ and $\varepsilon \in \Qspos$,
  we can compute an $\varepsilon$-approximation of $\Val_\game(\loc,\valnull)$
  in time triply-exponential in the size of \game
  and polynomial in $1/\varepsilon$.
  Moreover, deciding if a \WTG is almost-divergent is
  \PSPACE-complete.
\end{thm}

To obtain these results on divergent and almost-divergent \WTG{s}, we
follow a computation schema that we now outline.  First, we will
always reason on the region game $\rgame$ of the almost-divergent \WTG
$\game$.  The goal is to compute an $\varepsilon$-approximation of
$\Val_{\rgame}((\locI,[\valnull]),\valnull)$ for some initial location
$\locI$. Techniques of \cite{AluBer04} (that we will recall in
Section~\ref{sec:acyclic}) allow one to compute the (exact) values of
a \WTG played on a finite tree, using operator $\ValIteOpe$ of
Section~\ref{subsec:value-functions}. The idea is thus to decompose as
much as possible the game $\rgame$ as a \WTG over a tree.  First, we
decompose the region game into strongly connected components (SCCs,
left of \figurename~\ref{fig:schema}): we must think about the final
weight functions as the previously computed approximations of the
values of SCCs coming after the current one in the topological order.
We will keep as an invariant that final weight functions are piecewise
affine with a finite number of pieces, and are continuous on each
region.

For an SCC of $\rgame$ and an initial state $(\locI,[\valnull])$ of
$\rgame$ provided by the SCC decomposition, we show that the game on
the SCC is equivalent to a game on a tree built from a semi-unfolding
(see middle of \figurename~\ref{fig:schema}) of $\rgame$ from
$(\locI,[\valnull])$ of finite depth, with certain nodes of the tree
being \emph{kernels} (parts of \rgame that contain all cycles of
weight 0). The semi-unfolding is stopped either when reaching a final
location, or when some location (or kernel) has been visited for a
certain fixed number of times. Notice that, for divergent \WTG{s},
there are no kernels, which simplifies the computation.

Then, we compute an approximation of $\Val_\game(\locI,\valnull)$ with
a bottom-up computation on the semi-unfolding.  This computation is
exact on nodes labelled by a single region state $\state$, but
approximate on kernel nodes $\Kernel_\state$.  For the latter, we use
the corner-point abstraction (right of \figurename~\ref{fig:schema})
over $1/\clockgranu$-regions to compute values, and prove that, with an appropriately chosen $\clockgranu$, this
provides an $\varepsilon$-approximation of values.

\begin{figure}[tbp]
  \centering

\scalebox{1}{
\begin{tikzpicture}[>=latex]
  \begin{scope}[every node/.style={draw,shape=ellipse,minimum
      height=5mm, minimum width=8mm},level/.style={sibling distance=2cm/#1},level distance=8mm,->]
    \node (c) {}
    child { node (a) {}
      child { node (d) {} }
      child { node {} } }
    child { node {}
      child { node (b) {} }
      child { node {} } };
  \end{scope}
  \path[draw,gray,dashed] (1.25,-.5) to[bend right=10] (3.7,.9);
  \path[draw,gray,dashed] (1.25,-1.1) to[bend left=35] (2.6,-2.5);
  \begin{scope}[xshift=4.5cm,yshift=8mm,every
    node/.style={draw,shape=circle,minimum size=5mm,inner sep=0.3mm},level/.style={sibling distance=2cm/#1},level distance=8mm,->]
    \node {$s_0$}
    child { node {$s$}
      child { node[regular polygon,regular polygon sides=6,inner sep=0.1mm] {\makebox[1em][c]{\smaller$\Kernel_{s'}$}}
        child { node {$s$}
          child { node[double,regular polygon,regular polygon sides=3,inner
            sep=0.1mm] (stop) {$s$} }
        }
      }
      child { node {}
        child { node[double,inner sep=0] (out) {$s_f$} }
      }
    }
    child { node {}
      child { node[regular polygon,regular polygon sides=6,inner
        sep=0.1mm] {\makebox[1em][c]{\smaller$\Kernel_{s''}$}}
        child { node[double,inner sep=0] {$s_f$} }
        child { node [double,regular polygon,regular polygon sides=3,inner
            sep=0.1mm] {} }
      }
    };
    \node[node distance=3mm,below of=stop,draw=none] () {\smaller[2]{stopped leaf}};
    \node[node distance=5mm,below of=out,draw=none] () {\smaller[2]{$\weightT(s_f)$}};
  \end{scope}
  \path[draw,gray,dashed] (5.8,-.5) to[bend right=10] (8.2,.9);
  \path[draw,gray,dashed] (5.8,-1.1) to[bend left=35] (6.5,-2.5);
  \begin{scope}[xshift=9cm,yshift=6mm,node distance=2cm,scale=0.75, every node/.style={transform shape}]
    \node[draw,circle,minimum size=7mm] (0) {
      \begin{tikzpicture}[scale=0.4, every node/.style={transform shape}]
        \path[draw] (0,0) node[fill=black,minimum size=1mm] {}-- (1,1) -- (1,0) -- cycle;
      \end{tikzpicture}
    };
    \node[draw,circle,below of=0,xshift=-1.5cm,minimum size=7mm] (1) {
       \begin{tikzpicture}[scale=0.4, every node/.style={transform shape}]
        \path[draw] (0,0) node[fill=black,minimum size=1mm] {}-- (1,1) -- (1,0) -- cycle;
      \end{tikzpicture}
    };
    \node[draw,circle,below of=0,xshift=2cm,minimum size=7mm] (2) {
      \begin{tikzpicture}[scale=0.4, every node/.style={transform shape}]
        \path[draw] (0,0)-- (1,1)  node[fill=black,minimum size=1mm] {} -- (1,0) -- cycle;
      \end{tikzpicture}
    };
    \node[draw,rectangle,below of=1,xshift=-1.3cm,minimum size=7mm] (3) {
      \begin{tikzpicture}[scale=0.4, every node/.style={transform shape}]
        \path[draw] (0,0) node[circle,fill=black,minimum size=1mm] {}-- (1,0);
      \end{tikzpicture}
    };
    \node[draw,rectangle,below of=1,xshift=7mm,minimum size=7mm] (4) {
       \begin{tikzpicture}[scale=0.4, every node/.style={transform shape}]
        \path[draw] (0,0) node[circle,fill=black,minimum size=1mm] {}-- (1,1) -- (1,0) -- cycle;
      \end{tikzpicture}
    };
    \node[draw,rectangle,below of=2,xshift=-1cm,minimum size=7mm] (5) {
      \begin{tikzpicture}[scale=0.4, every node/.style={transform shape}]
        \path[draw] (0,0) node[circle,fill=black,minimum size=1mm] {};
      \end{tikzpicture}      };
    \node[draw,circle,below of=2,xshift=1cm,minimum size=7mm] (6) {
      \begin{tikzpicture}[scale=0.4, every node/.style={transform shape}]
        \path[draw] (0,0) node[fill=black,minimum size=1mm] {}-- (1,0) -- (0,-1) -- cycle;
      \end{tikzpicture}
    };

    \path[->] (0) edge node[above left]{0} (1)
    (0) edge node[above right]{$1$} (2)
    (1) edge node[above left]{$-3$} (3)
    (1) edge[bend left=10] node[above right]{$-1$} (4)
    (2) edge node[above left]{$2$} (5)
    (2) edge node[above right]{$1$} (6)
    (3) edge node[below]{$2$} (4)
    (3) edge[bend left=40] node[left]{$3$} (0)
    (4) edge[bend left=10] node[below left]{$1$} (1)
    (4) edge node[below]{$4$} (5)
    (5) edge[bend left=5] node[left]{$-3$} (0)
    (6) edge node[below]{$1$} (5)
    (6) edge[bend left] node[below]{$-3$} (4);
  \end{scope}
\end{tikzpicture}
}

\caption{Static approximation schema: SCC decomposition of \rgame,
  semi-unfolding of an SCC, corner-point abstraction for the kernels}
  \label{fig:schema}
\end{figure}

This resolution of the value problem for divergent \WTG{s} and
approximation problem for almost-divergent \WTG{s} heavily relies on
the region abstraction, and requires one to construct \rgame entirely
and compute its SCCs, before unfolding it partially in a tree-shaped
structure.  Our second result is a more symbolic approximation schema
based on the value iteration algorithm only (in case we are able to rule out the presence of configurations with value $-\infty$, which could in particular be true if there are only non-negative weights): the computations are not
performed on the region abstraction, but instead use the cell
partitions introduced in Section~\ref{sec:piecewise-affine-value-functions} that can cover several regions.

\begin{thm}\label{thm:symbolic}
  Let $\game$ be an almost-divergent \WTG such that
  $\Val_\game(\loc,\val)>{-\infty}$ for every configuration
  $(\loc,\val)$.  Then the sequence $(\Val^k_\game)_{k\geq0}$
  converges towards $\Val_\game$ and for every
  $\varepsilon\in\Qspos$, there exists an integer $P$ such that
  $\Val^P_\game$ is an $\varepsilon$-approximation of $\Val_\game$ for
  all configurations.
\end{thm}

Note that we have to control for configurations $(\loc,\val)$ of value
$-\infty$, where the non-increasing sequence
$(\Val^k_\game(\loc,\val))_{k\in\N}$ (that starts at $+\infty$) will
diverge towards $-\infty$, but has no hope of approximating it.
However, we will show that the configurations with value $-\infty$ can
be computed in advance:
\begin{prop}\label{lm:-infty-main}
  Given an almost-divergent \WTG $\game$
  and an initial location~$\locI$, the decision problem asking
  whether $\Val_\game(\locI,\valnull)=-\infty$ is \EXP-complete.
\end{prop}

The exponential upper bound is obtained in Section~\ref{sec:kernel-infinity}.
This contrasts with the general case (not necessarily
almost-divergent), where the $-\infty$-value problem is undecidable (among other problems), as detailed in Proposition~\ref{prop:-infty_undec}.

\subsection{Hardness of value problems}

The \EXP-hardness result of Theorem~\ref{thm:div_wtg} (\resp~Proposition~\ref{lm:-infty-main})
is a reduction from the problem of solving timed games with reachability
objectives \cite{JurTri07}.

To a reachability timed game $\game$, we
simply set the weight of each edge to $1$ and the weight of each location to $0$, making
it a divergent \WTG.
We set the final weight of every target configuration at $0$ (\resp~$-\infty$).
Then, \MinPl wins the reachability timed game if
and only if the value in the \WTG is lower than threshold
$\alpha=|\rgame|$ (\resp~equals $-\infty$).  One direction of this statement is immediate by
definition of having a value smaller than $+\infty$, and the other
comes from the fact that reachability in the timed game implies
reachability in the region game in less than $\alpha$ transitions, in
turn implying that \MinPl can ensure target reachability in the \WTG
with cumulative weight below $\alpha$, \ie~$\Val_\game(\state,\val)\leq \alpha + \weightT$, with $\weightT=0$ (\resp~$-\infty$).

\begin{prop}\label{prop:-infty_undec}
  Given a \WTG $\game$
  and an initial location~$\locI$, the decision problem asking
  whether $\Val_\game(\locI,\valnull)=-\infty$ is undecidable.
\end{prop}
\begin{proof}
  The proof goes via a reduction to the existence
  problem on turn-based \WTG: given a \WTG \game (without final weight
  function), a non-negative integer threshold $\alpha$ and a starting location
  $\locI$, does there exist a strategy for \MinPl that can guarantee
  reaching the unique target location $\locT$ from $\locI$ with weight
  $<\alpha$. In the setting with non-negative weights in \game, it is proven
  in~\cite{BBM06} that the problem is undecidable for the comparison
  $\leq \alpha$. In the setting with arbitrary weights, formal proofs are given for
  all comparison signs in~\cite{BGNK+14}.

  Consider the \WTG $\game'$ built from \game by adding a transition
  from $\locT$ to $\locI$, without guards and resetting all the
  clocks, of discrete weight $-\alpha$. We add a new target location
  $\locT'$, and add transitions of weight $0$ from $\locT$ to
  $\locT'$.  Locations $\locT$ and $\locT'$ belong to \MinPl. Let us prove
  that $\Val_{\game'}(\locI,\valnull)=-\infty$ if and only if \MinPl
  has a strategy to guarantee a weight $<\alpha$ in $\game$.
  Assume first $\Val_{\game'}(\locI,\valnull)=-\infty$. If
  $\Val_\game(\locI,\valnull)=-\infty$, we are done. Otherwise, \MinPl
  must follow in $\game'$ the new transition from $\locT$ to $\locI$
  to enforce a cycle of negative value, and thus enforce a play from
  $(\locI,\valnull)$ to $\locT$ with weight less than $\alpha$.
  Therefore, there exists a strategy for \MinPl in $\game$ that can
  guarantee a weight $<\alpha$.
  Reciprocally, if there exists a strategy for \MinPl in $\game$ that can guarantee
  a weight $<\alpha$, then \MinPl can force a negative cycle play in $\game'$ and
  $\Val_{\game'}(\locI,\valnull)=-\infty$.
\end{proof}

\section{Deciding divergence and almost-divergence}\label{sec:membership}

  In this section, we will study properties that region cycles must satisfy
  in divergent or almost-divergent \WTG{s}.  This will give us a better understanding
  of the modelling power these classes confer, as well as enable us to
  provide procedures of optimal complexity to decide if a \WTG
  fulfils the divergence or almost-divergence conditions.

  \subsection{Cycles in an almost-divergent \WTG}\label{sec:0-iso}

  Let us start with properties that hold for all almost-divergent
  weighted timed games \game. A region cycle~$\rpath$ of~\rgame is
  said to be a positive cycle (\resp~a negative cycle, a 0-cycle) if
  every finite play \play following~$\rpath$ satisfies
  $\weight(\play)\geq 1$ (\resp~ $\weight(\play)\leq -1$,
  $\weight(\play)=0$).

  We start by showing that, in an almost-divergent game, all cycles
  $\rpath=\rtrans_1\cdots\rtrans_k$ of~\rgame (with
  $\rtrans_1,\ldots,\rtrans_k$ edges of~$\rgame$) are either 0-cycles,
  positive cycles or negative cycles\footnote{In contrast,
    Definition~\ref{def:almost-divergent-plays} only requires that
    each play following a region cycle has weight in
    $(-\infty,-1]\cup\{0\}\cup[1,+\infty)$, without disallowing a
    region cycle to contain plays of different types.}, and we can
  classify a cycle by looking only at one of the corner plays
  following it:
\begin{lem}\label{lm:exist0}
  Let $\game$ be an almost-divergent \WTG. A region cycle~\rpath is a
  positive cycle (\resp~a negative cycle, a 0-cycle) if and only if
  there exists a corner play~\cplay following~\rpath with
  $\weightC(\cplay)>0$ (\resp~$\weightC(\cplay)<0$,
  $\weightC(\cplay)=0$).  Moreover, every region cycle
  in \game is either positive, negative, or a 0-cycle.
\end{lem}
\begin{proof}
  If \rpath is a positive cycle
  (\resp~a negative cycle, a 0-cycle), every such corner play~\cplay will have
  weight above~$0$ (\resp~under~$0$, equal to~$0$), by Lemma~\ref{lm:cornerabstract}.
  Reciprocally, if such
  a corner play exists, all corner plays following~\rpath have
  weight above~$0$ (\resp~under~$0$, equal to~$0$): otherwise the set
  $\{\weightC(\play) \mid \play \text{ play following } \rpath \}$
  would have non-empty intersection with the set
  $(-1,0)\cup(0,1)$ by Lemma~\ref{lm:cornerabstract},
  which would contradict that the game is almost-divergent.

  Let $\rpath$ be a region cycle of $\game$. All the plays following
  $\rpath$ have a weight in $(-\infty,-1]\cup\{0\}\cup[1,+\infty)$. If
  it has a play in two different subintervals $(-\infty,-1]$, $\{0\}$,
  and $[1,+\infty)$, then Lemma~\ref{lm:cornerabstract} implies also
  that a play following $\rpath$ will have a weight in
  $(-1,0)\cup(0,1)$, which is forbidden by
  Definition~\ref{def:almost-divergent-plays}.
\end{proof}

  An important result is that the sign of cycles is stable by rotation.  This is
  not trivial because plays following a cycle can start and end in
  different valuations, therefore changing the starting region state of the
  cycle could \emph{a priori} change the plays that follow it and the sign of
  their weights.
\begin{lem}\label{lm:rotatcycle}
  Let~\rpath and $\rpath'$ be region paths of an almost-divergent~\WTG.
  If $\rpath\rpath'$
  is a positive cycle (\resp~a negative cycle, a 0-cycle), then
  $\rpath'\rpath$ is a positive cycle (\resp~a negative cycle, a 0-cycle).
\end{lem}
\begin{proof}
  Since~$\rpath_1=\rpath\rpath'$ is a cycle,
  $\first(\rpath)=\last(\rpath')$ and $\first(\rpath')=\last(\rpath)$,
  so~$\rpath_2=\rpath'\rpath$ is a cycle as well.
  First, since there are finitely many corners, by constructing a long
  enough play following an iterate of $\rpath'\rpath$, we can obtain a
  corner play that starts and ends in the same corner.  Formally, we
  define two sequences of region
  corners~$(\corner_i\in\first(\rpath))_i$
  and~$(\corner'_i\in\first(\rpath'))_i$.  We start by choosing
  any~$\corner_0\in\first(\rpath)$.  Let~$\corner'_0$ be a corner
  of~$\first(\rpath')$ such that~$\corner'_0$ is accessible
  from~$\corner_0$ by following~\rpath with a corner play $\cplay_0$.
  For every~$i>0$, let~$\corner_i$
  be a corner of~$\first(\rpath)$ such that~$\corner_i$ is accessible
  from~$\corner'_{i-1}$ by following~$\rpath'$ with a corner play $\cplay_i'$,
  and let~$\corner'_i$ be a
  corner of~$\first(\rpath')$ such that~$\corner'_i$ is accessible
  from~$\corner_i$ by following~$\rpath$ with a corner play $\cplay_i$.
  We stop the construction at the first index~$\ell$ such that there exists~$k<\ell$
  with~$\corner_\ell=\corner_k$.
  Additionally, we let~$\cplay_\ell=\cplay_k$.
  We know that this process never
  gets stuck---\ie~we can always find such corner plays iteratively---by
  Lemma~\ref{lm:corners-no-deadlocks}, and it is bounded
  since~$\first(\rpath)$ has at most~$|\Clocks|+1$ corners.

\begin{figure}[tbp]
  \centering
\begin{tikzpicture}[node distance=2cm]

  \node (0) {$\corner_0$};
  \node[right of=0] (1) {$\corner_1$};
  \node[right of=1] (2) {$\corner_k=\corner_\ell$};
  \node[right of=2] (3) {$\corner_{k+1}$};
  \node[right of=3] (4) {$\corner_{\ell-1}$};
  \node[below of=0] (00) {$\corner_0'$};
  \node[below of=1] (11) {$\corner_1'$};
  \node[below of=2] (22) {$\corner_k'=\corner_\ell '$};
  \node[below of=3] (33) {$\corner_{k+1}'$};
  \node[below of=4] (44) {$\corner_{\ell-1}'$};
  \draw[->] (0) edge node[left]{\small$\cplay_0$} (00)
  (00) edge node[left,yshift=1mm]{\small$\cplay_0'$} (1)
  (1) edge node[left,xshift=1mm,yshift=-0.5mm]{\small$\cplay_1$} (11)
  (11) edge[dashed] (2)
  (2) edge node[left]{\small$\cplay_k$} (22)
  (22) edge node[left,yshift=1mm]{\small$\cplay_k'$} (3)
  (3) edge node[left,xshift=1mm,yshift=-1mm]{\small$\cplay_{k+1}$} (33)
  (33) edge[dashed] (4)
  (4) edge node[left,,xshift=1mm,yshift=-0.5mm]{\small$\cplay_{\ell-1}$} (44)
  (44) edge[out=60,in=25,looseness=1.6] node[left,xshift=-1mm,yshift=-1mm]{\small$\cplay_{\ell-1}'$} (2);
\end{tikzpicture}

  \caption{Proof scheme of Lemma~\ref{lm:rotatcycle}.  The top labels
  are corners of $\first(\rpath)$, the bottom ones are corners of $\first(\rpath')$,
  and edges represent corner plays.}
  \label{fig:rotatcycle}
\end{figure}

  For every~$0\leq i\leq \ell$, let~$w_i$ be the weight of the corner play
  $\cplay_i$ from~$\corner_i$ to~$\corner'_i$ along~\rpath, and let~$w'_i$
  be the weight of the corner play $\cplay'_i$ from~$\corner'_i$ to~$\corner_{i+1}$
  along~$\rpath'$.  The concatenation of the two plays has weight
  $w_i+w'_i>0$ (\resp~$w_i+w'_i<0$, $w_i+w'_i=0$), since it follows the positive cycle
  (\resp~negative cycle, 0-cycle) $\rpath_1$.
  For every~$0\leq i< \ell$, the concatenation of the corner play $\cplay'_i$
  from~$\corner'_i$ to~$\corner_{i+1}$ with the corner play $\cplay_{i+1}$
  from~$\corner_{i+1}$ to~$\corner'_{i+1}$
  is a play from $\corner'_i$ to $\corner'_{i+1}$, of
  weight~$w'_i+w_{i+1}$, following~$\rpath_2$.
  Since~$\rpath_2$ is a cycle, and the game is almost-divergent,
  all possible values of $w'_i+w_{i+1}$ have the same sign by Lemma~\ref{lm:exist0}.

  Finally, we can construct a corner play from $\corner'_k$ to
  $\corner'_\ell$ by concatenating the plays
  $\cplay'_k, \cplay_{k+1}, \ldots,\cplay_{\ell-1},
  \cplay'_{\ell-1}, \cplay_{\ell}$.  We denote the weight of that play $W$, and
  $$W=\sum_{i=k}^{\ell-1} (w'_i+w_{i+1})=\sum_{i=k}^{\ell-1} (w_i+w_{i}')$$
  since $w_k=w_\ell$. %
  As $w_i+w'_i>0$ (\resp~$w_i+w'_i<0$, $w_i+w'_i=0$) holds for every $i$,
  we obtain $W>0$ (\resp~$W<0$, $W=0$).

  This implies that the
  terms $w'_i+w_{i+1}$, of constant sign, are all above~$0$ (\resp~under~$0$,
  equal to~$0$).  As a
  consequence, the concatenation of $\cplay'_k$ and $\cplay_{k+1}$ is a
  corner play following $\rpath_2$ of weight above~$0$
  (\resp~under~$0$, equal to~$0$).  By
  Lemma~\ref{lm:exist0}, we conclude that $\rpath_2$~must be a positive cycle
  (\resp~a negative cycle, a 0-cycle).
\end{proof}

  Therefore, region cycles in almost-divergent games are well-behaved: we can
  compose and rotate them while preserving their sign in the expected way.

\subsection{SCC-based characterisations}
After studying the properties of region cycles, we now focus on
strongly connected components (SCCs) of the region abstraction \rgame.
An SCC~$S$ of~\rgame is said to be positive (\resp~negative) if every
cycle in $S$ is positive (\resp~negative), \ie~if every play~\play
following a region cycle in~$S$ satisfies $\weightC(\play)\geq 1$
(\resp~$\weightC(\play)\leq -1$).
  \begin{prop}\label{prop:timed-scc-sign}
    A weighted timed game \game is divergent if and only if,
    each SCC of \rgame is either positive or negative.
  \end{prop}

  Likewise, an SCC~$S$ of~\rgame is said to be non-negative
  (\resp~non-positive) if every region cycle in~$S$ is either a positive cycle
  or a 0-cycle (\resp~either a negative cycle
  or a 0-cycle), \ie~every play~\play following a region cycle in~$S$
  satisfies either $\weightC(\play)\geq 1$ or $\weightC(\play)=0$
  (\resp~either $\weightC(\play)\leq -1$ or $\weightC(\play)=0$).
  We obtain:
\begin{prop}\label{prop:almost-divergent}
  A \WTG~\game is
  almost-divergent if and only if each SCC of~\rgame is either non-negative or
  non-positive.
\end{prop}

We now prove these two results. First, note that if \game is divergent
it has no 0-cycle, and Proposition~\ref{prop:almost-divergent} implies
that each SCC of \rgame is either positive or negative.  Conversely,
if each SCC of \rgame is either positive or negative,
Proposition~\ref{prop:almost-divergent} implies that \game is
divergent.  Therefore, Proposition~\ref{prop:timed-scc-sign} is a
corollary of Proposition~\ref{prop:almost-divergent}. The rest of this
section now proves Proposition~\ref{prop:almost-divergent}.

To prove the reciprocal implication of
Proposition~\ref{prop:almost-divergent}, we only need to show that
non-negative (\resp~non-positive) SCCs of \rgame satisfy the definition
of almost-divergent \WTG{s}.  By definition,
they only contain executions \play following region cycles $\rpath$ such that
$\weightC(\play)\in\{0\}\cup[1,+\infty)$
(\resp~$\weightC(\play)\in(-\infty,-1]\cup\{0\}$).  Then, assume
$\weightC(\play)=0$ and that $\rpath$ can be decomposed into
smaller cycles $\rpath'$ and $\rpath''$.
Definition~\ref{def:almost-divergent-plays} requires us to show that
all plays $\play'$ and $\play''$ following $\rpath'$ and $\rpath''$, respectively, are such that $\weightC(\play')=\weightC(\play'')=0$,
\ie~$\rpath'$ and $\rpath''$ are 0-cycles.
Note that by Lemma~\ref{lm:rotatcycle}, $\rpath'\rpath''$ is a 0-cycle.
As $\rpath'$ and $\rpath''$ are contained in the same SCC, they are either both non-negative cycles or both non-positive cycles.
Let $\play'\play''$ be a play following $\rpath'\rpath''$, so that
$\play'$ follows $\rpath'$ and $\play''$ follows $\rpath''$.
Then, $\weightC(\play'\play'')=\weightC(\play')+\weightC(\play'')=0$, it follows that
$\weightC(\play')=\weightC(\play'')=0$, \ie~$\rpath'$ and $\rpath''$ are 0-cycles.

For the direct implication, the situation is more complex: we need to
be careful while composing cycles with each other. To help us, we rely
on the folded orbit graphs of region cycles.
Suppose that \game is almost-divergent, and consider two cycles \rpath
and $\rpath'$ in the same SCC of $\rgame$.  We need to show that they
are both either non-positive or non-negative.
Lemma~\ref{lm:timed-touching-cycles-almost} will first take care of
the case where $\rpath$ and $\rpath'$ share a region state
$(\loc,\reg)$.

\begin{figure}[tbp]
  \centering
    \begin{tikzpicture}[node distance=4cm,auto,->,>=latex]
      \node (0) {$\phantom'\corner$};
      \node[right of=0] (1) {$\corner'$};
      \path (0) edge[bend left=20] node[above] {$P$} (1)
        (1) edge[bend left=20] node[below] {$P'$} (0)
        (0) edge[loop left, looseness=25,color=blue] node[left] {$C$} (0)
        (1) edge[loop right, looseness=25,color=red] node[right] {$C'$} (1);
    \end{tikzpicture}

  \caption{Proof scheme of
    Lemma~\ref{lm:timed-touching-cycles-almost}, with paths and
    cycles around corners of the graph
    $\FOG(\rpath,\rpath')$.}
  \label{fig:scc-sign-timed}
\end{figure}
\begin{lem}\label{lm:timed-touching-cycles-almost}
  If \game is almost-divergent and two cycles $\rpath$ and $\rpath'$ of
  \rgame share a region state~$(\loc,\reg)$, they are either both non-negative
  or both non-positive.
\end{lem}
\begin{proof}
  Suppose by contradiction that $\rpath$ is negative and $\rpath'$ is
  positive.  We assume that $(\loc,\reg)$ is the first region state of both
  $\rpath$ and $\rpath'$, possibly performing a rotation of the cycles
  if necessary (in particular this preserves their sign by Lemma~\ref{lm:rotatcycle}). We will rely on the folded orbit graphs that have been defined on page~\pageref{page:FOG}.
  We construct a graph $\FOG(\rpath,\rpath')$ as
  the union of $\FOG(\rpath)$ and $\FOG(\rpath')$ (that share the same
  set of vertices), colouring in blue the edges of $\FOG(\rpath)$ and
  in red the edges of $\FOG(\rpath')$.  A path in
  $\FOG(\rpath,\rpath')$ is said blue (\resp~red) when all of
  its edges are blue (\resp~red).

  Since $\FOG(\rpath)$ and $\FOG(\rpath')$ are finite graphs with no
  deadlocks (every corner has an outgoing edge by
  Lemma~\ref{lm:corners-no-deadlocks}), from every corner of
  $\FOG(\rpath,\rpath')$, we can reach a blue simple cycle, as well as
  a red simple cycle. Since there are only a finite number of simple
  cycles in $\FOG(\rpath,\rpath')$, there exists a blue cycle $C$ and
  a red cycle $C'$ that can reach each other in
  $\FOG(\rpath,\rpath')$. Denote by $\corner$ and $\corner'$ the first
  corners of cycles $C$ and $C'$, respectively.

  We assume first $\corner=\corner'$. Let $k$ and $k'$ be the
  respective lengths of $C$ and $C'$, so that $C$ can be decomposed as
  $\corner\xrightarrow{\cplay_1}\cdots \xrightarrow{\cplay_k}\corner$
  and $C'$ as
  $\corner\xrightarrow{\cplay_1'}\cdots
  \xrightarrow{\cplay_{k'}'}\corner$, where $\cplay_i$ are corner
  plays following $\rpath$ and $\cplay_i'$ are corner plays following
  $\rpath'$.  Let $\cplay$ be the concatenation of
  $\cplay_1,\ldots,\cplay_k$, and $\cplay'$ be the concatenation of
  $\cplay_1',\ldots,\cplay_{k'}'$.  Recall that $w=|\weightC(\cplay)|$
  and $w'=|\weightC(\cplay')|$ are integers.  Since $\rpath$ is
  negative, so is $\rpath^k$, the concatenation of $k$ copies of
  $\rpath$ (the weight of a play following it is a sum of weights all
  below $-1$).  Therefore, $\cplay$, that follows $\rpath^k$, has a
  weight $\weightC(\cplay)\leq -1$ by Lemma~\ref{lm:cornerabstract}.
  Similarly, $\weightC(\cplay')\geq 1$.
  Let $\cplay''$ be the play obtained by concatenating
  $w'$ copies of $\cplay$ and $w$ copies of $\cplay'$.  Then,
  $\weightC(\cplay'')=\weightC(\cplay)w'+ \weightC(\cplay')w=0$, and
  therefore the region cycle $\rpath''$ composed of $w'$ copies
  of $\rpath^k$ and $w$ copies of ${\rpath'}^{k'}$ is a 0-cycle.  This
  contradicts the almost-divergence of $\game$, since $\rpath''$ can
  be decomposed into smaller cycles that are not 0-cycles.

  Therefore, $\corner$ and $\corner'$ are different, but can reach
  each other in $\FOG(\rpath,\rpath')$. Let $P$ be a path from
  $\corner$ to $\corner'$, and $P'$ be a path from $\corner'$ to
  $\corner$. The situation is depicted in
  \figurename~\ref{fig:scc-sign-timed}. Consider the cycle $C''$
  obtained by concatenating $P$ and $P'$.  As a cycle of
  $\FOG(\rpath,\rpath')$, we can map it to a cycle $\rpath''$ of
  \rgame (alternating \rpath and $\rpath'$ depending on the colours of
  the traversed edges), so that $C''$ is a cycle (of length 1) of
  $\FOG(\rpath'')$.  As $\game$ is almost-divergent, $\rpath''$ is
  either positive, negative or a 0-cycle.
  Moreover, $\rpath''$ cannot be a 0-cycle as it
  can be decomposed into smaller cycles that are not 0-cycles.
  Suppose for instance that it is
  positive.  Since $(\loc,\reg)$ is the first region state of both $\rpath$ and
  $\rpath''$, we can construct the $\FOG(\rpath,\rpath'')$, in which
  $C$ is a blue cycle and $C''$ is a red cycle, both sharing the same
  first vertex. We then conclude with the previous case.  A similar
  reasoning with $\rpath'$ applies to the case that $\rpath''$ is
  negative.  Therefore, in all cases, we reach a contradiction.
\end{proof}

  To finish the proof of the direct implication of
  Proposition~\ref{prop:almost-divergent}, we suppose that the two
  cycles $\rpath$ and $\rpath'$, one positive and the other negative,
  in the same SCC of $\rgame$,
  do not share any region states.
  By strong connectivity, in $\rgame$, there exists a path
  $\rpath_1$ from the first state of~$\rpath$ to the first state of
  $\rpath'$, and a path $\rpath_2$ from the first state of $\rpath'$ to
  the first state of $\rpath$.  Consider the cycle $\rpath''$ of $\rgame$
  defined by $\rpath\rpath_1\rpath'\rpath_2$.
  By the almost-divergence of $\game$, %
  $\rpath''$ must be either positive, negative or a 0-cycle.
  Since it shares a state with both
  $\rpath$ and $\rpath'$, Lemma~\ref{lm:timed-touching-cycles-almost} allows
  us to prove a contradiction in both positive and negative cases,
  and therefore $\rpath''$ must be a 0-cycle.
  This contradicts the hypothesis as one of the decompositions of $\rpath''$
  into smaller cycles produces $\rpath$ and $\rpath_1\rpath'\rpath_2$,
  with $\rpath$ a non-0-cycle.
  This concludes the proof of Proposition~\ref{prop:almost-divergent}.

\begin{rem}
  These characterisations of divergent or almost-divergent \WTG{s} in term of SCCs
  provide an intuitive understanding of the modelling power these classes hold.
  For divergence, the model should have a global structure (the SCC decomposition)
  linking modules in an acyclic fashion. For each module, we have to choose between
  a positive dynamic, where weights always eventually increase,
  and a negative dynamic, where weights always eventually decrease.
  For almost-divergence, the modules may also have portions that are (eventually)
  neutral with regard to weight accumulation.
  In both classes, arbitrarily small weights should not be allowed
  to accumulate.
\end{rem}

\subsection{Deciding membership}

  We study the \emph{membership problem} for divergent (\resp~almost-divergent)
  \WTG{s},
  \ie~the decision problem that asks if a given \WTG is divergent
  (\resp~almost-divergent).
  As mentioned in Theorems~\ref{thm:div_wtg} and~\ref{thm:almost-div},
  we show that it is \PSPACE-complete for both of these
  classes.

  Relying on the previous characterisation of
  Propositions~\ref{prop:timed-scc-sign} and
  \ref{prop:almost-divergent}, the algorithms will consist in only
  considering region cycles of length bounded by the number of corners
  in the corner-point abstraction $\cgame$. For divergent \WTG{s},
  this will be correct by using the following result:

\begin{lem}\label{lm:timed-simple-cycles-div}
  Let \game be a \WTG.
  An SCC $S$ of \rgame is positive (\resp~negative) if and only if every
  region cycle in $S$, of length at most $|\cgame|$, is positive (\resp~negative).
\end{lem}
\begin{proof}
  The direct implication holds by definition.
  Reciprocally, let us assume that every cycle in $S$ of length at most $|\cgame|$
  is positive (\resp~negative),
  and prove that every cycle \rpath in $S$ is positive (\resp~negative),
  by induction on the length of \rpath.
  If \rpath has length above $|\cgame|$, every corner play \cplay following
  \rpath can be split as $\cplay=\cplay_1\cplay_2\cplay_3$, with
  $\cplay_2$ a corner play that starts and ends in the same corner.
  Then we can write $\rpath=\rpath_1\rpath_2\rpath_3$,
  with $\cplay_1$ (\resp~$\cplay_2$, $\cplay_3$)
  following $\rpath_1$ (\resp~$\rpath_2$, $\rpath_3$).
  Observe that $\rpath_2$ and $\rpath_1\rpath_3$ are
  region cycles of $S$, both positive (\resp~negative) by induction.
  It follows that $\weightC(\cplay_2)\geq 1$ (\resp~$\weightC(\cplay_2)\leq -1$),
  and $\weightC(\cplay_1\cplay_3)\geq 1$ (\resp~$\weightC(\cplay_1\cplay_3)\leq -1$),
  as $\cplay_1\cplay_3$ is a valid corner play following $\rpath_1\rpath_3$.
  We can therefore conclude that $\weightC(\cplay)\geq 1$
  (\resp~$\weightC(\cplay)\leq -1$). This holds for all corner plays
  following \rpath, and by Lemma~\ref{lm:cornerabstract}
  \rpath is positive (\resp~negative).
\end{proof}

Then, to decide if a game is \emph{not divergent}, using
Proposition~\ref{prop:timed-scc-sign}, it suffices to search for an
SCC of the region automaton containing a cycle such that there exists
a corner play following it of non-negative weight, and a cycle such
that there exists a corner play following it of non-positive weight,
both of length bounded by
$B=|\cgame| \leq|\Locs|\times |\Regions \Clocks
\clockbound|\times(|\Clocks|+1)$. We can test this condition in
$\NPSPACE$: we guess a starting region for each cycle, use standard
reachability analysis~\cite{AluDil94} to check that they are in the
same SCC of \rgame (in \PSPACE), and use the following result with
comparison $\geq 0$ and $\leq 0$, respectively, to check the sign of
each cycle.
\begin{lem}\label{lm:pspace-cplay}
  Consider a \WTG \game, a region state $(\loc,\reg)$ of \rgame,
  a bound $B\in\N$, and
  a comparison operator ${\bowtie}\in\{{<},{>},{\leq},{\geq},{=},{\neq}\}$.
  Deciding if there exists a corner play \cplay
  following a cycle \rpath of \rgame starting from $(\loc,\reg)$,
  such that $|\rpath|\leq B$ and $\weightC(\cplay)\bowtie 0$,
  is in \PSPACE (\ie~it can be done using space
  polynomial in $|\game|$ and $\log(B)$).
\end{lem}
\begin{proof}
  We guess a starting corner $\corner$ of $\reg$ for \cplay, and we guess
  on-the-fly the transitions of \rpath and \cplay,
  \ie~a sequences of regions with one of their corners,
  keeping in memory %
  the cumulative weight of \cplay and the length $|\rpath|$.
  At every step, we check that $|\rpath|\leq B$
  in space polynomial in $\log(B)$ and $\log(|\rpath|)\leq \log(|\rgame|)$,
  with $\log(|\rgame|)$ polynomial in $|\game|$.\footnote{
  The global clock bound $\clockbound$ is at most exponential in the size of \game,
  and $|\rgame|$ is at most exponential in $|\Clocks|$ but
  polynomial in \clockbound, therefore $|\rgame|$ is at most
  exponential in $|\game|$.}
  Similarly, we can check that \cplay is following \rpath in polynomial space.
  At some point, we guess that the cycle is complete, and we check
  that the current region state equals $(\loc,\reg)$.
  Finally, we check that $\weightC(\cplay)\bowtie 0$
  in space polynomial in $\log(\weightC(\cplay))$.
  Note that $\weightC(\cplay)$ is an integer bounded (in absolute value)
  by $B\times \wmaxTimed$,
  and can thus be stored in polynomial space.
  This shows that the problem is in \NPSPACE, and thus
  in \PSPACE\ using Savitch's theorem~\cite{Sav70}.
\end{proof}

Since the bound $B$ is at most exponential in $|\game|$, this check
can be performed in \PSPACE.  This shows that the membership problem
for divergent weighted timed games is in
$\coNPSPACE=\coPSPACE=\PSPACE$ by Savitch~\cite{Sav70}.

Let us now show the ($\mathsf{co}$) $\PSPACE$-hardness by a reduction from the reachability problem in a timed
automaton. Consider a timed automaton with a starting location and a
different target location without outgoing edges. We construct from it
a weighted timed game by distributing all locations to \MinPl, and
equipping all edges with weight $1$, and all locations with weight
$0$.  We also add a loop with weight $-1$ on the target, and an edge
from the target location to the initial location with weight $0$, both
with guard $\top$ and resetting all clocks.  Then, the \WTG is not
divergent if and only if the target can be reached from the initial
location in the timed automaton.

For almost-divergent \WTG{s}, the length of the required region cycles
is bigger, because of the possible presence of $0$-cycles.
\begin{lem}\label{lm:timed-simple-cycles-almost-div}
  Let \game be a \WTG.
  An SCC $S$ of \rgame is non-negative (\resp~non-positive) if and only if every
  region cycle in $S$, of length at most $|\cgame|^2$, is either a positive cycle
  or a 0-cycle (\resp~either a negative cycle or a 0-cycle).
\end{lem}
\begin{proof}
  The direct implication holds by definition.
  Reciprocally, suppose that every cycle in $S$ of length at most $|\cgame|^2$
  is either a positive cycle
  or a 0-cycle (\resp~either a negative cycle or a 0-cycle).
  Let us prove that every cycle \rpath in $S$ is either a positive cycle
  or a 0-cycle (\resp~either a negative cycle or a 0-cycle),
  by induction on the length of \rpath.
  Consider a region cycle \rpath with length above $|\cgame|^2$.
  Let us show that for all corner plays $\cplay$, $\cplay'$ following \rpath,
  either $\weightC(\cplay)=\weightC(\cplay')=0$, or both $\weightC(\cplay)\geq 1$
  and $\weightC(\cplay')\geq 1$ (\resp~both $\weightC(\cplay)\leq -1$
  and $\weightC(\cplay')\leq -1$) hold.
  This will allow us to conclude by Lemma~\ref{lm:cornerabstract}.

  From a pair of corner plays \cplay and $\cplay'$ following \rpath,
  we can extract a sequence of pairs of corners states
  $((\loc_i,\reg_i,\corner_i),(\loc_i,\reg_i,\corner_i'))$,
  such that $(\loc_i,\reg_i)$ is the $i$-th region state of \rpath,
  and $\corner_i$ (\resp~$\corner_i'$) is the $i$-th corner of \cplay (\resp~$\cplay'$).
  Since $|\rpath|>|\cgame|^2$, there must exist two indexes, $j$ and $k$, such that
  $j<k$, $(\loc_j,\reg_j)=(\loc_k,\reg_k)$ and
  $(\corner_j,\corner'_j)=(\corner_k,\corner'_k)$.
  In other words, we can write $\rpath=\rpath_1\rpath_2\rpath_3$,
  with $\rpath_2$ and $\rpath_1\rpath_3$ region cycles of $S$,
  and \cplay, $\cplay'$ can be split as $\cplay=\cplay_1\cplay_2\cplay_3$,
  $\cplay'=\cplay_1'\cplay_2'\cplay_3'$,
  with $\cplay_l$ (\resp~$\cplay_l'$)
  following $\rpath_l$ (\resp~$\rpath_l'$) for $l\in\{1,2,3\}$,
  such that $\cplay_2$ and $\cplay_2'$ are corner cycles,
  \ie~$\first(\cplay_2)=\last(\cplay_2)$ and $\first(\cplay_2')=\last(\cplay_2')$.
  Then, by induction either $\weightC(\cplay_2)=\weightC(\cplay_2')=0$,
  or both $\weightC(\cplay_2)\geq 1$ and $\weightC(\cplay_2')\geq 1$
  (\resp~both $\weightC(\cplay_2)\leq -1$ and $\weightC(\cplay_2')\leq -1$) hold.
  The same property holds for $\cplay_1\cplay_3$ and $\cplay_1'\cplay_3'$,
  both valid corner plays following $\rpath_1\rpath_3$.
  It follows that either $\weightC(\cplay)=\weightC(\cplay')=0$,
  or both $\weightC(\cplay)\geq 1$ and $\weightC(\cplay')\geq 1$
  (\resp~both $\weightC(\cplay)\leq -1$ and $\weightC(\cplay')\leq -1$) hold.
\end{proof}

Then, to decide if a game is \emph{not almost-divergent}, we
distinguish two cases:
\begin{itemize}
  \item There exists a region cycle,
  of length at most $B=|\cgame|^2$, and two corner plays \cplay and $\cplay'$,
  both following \rpath, such that $\weightC(\cplay)=0$ and
  $\weightC(\cplay')\neq 0$.
  \item An SCC of the region automaton contains a cycle
  such that there exists a corner play following it of negative weight,
  and a cycle such that there exists a corner play following it of
  positive weight, both of length bounded by $B=|\cgame|^2$.
\end{itemize}
We can test both conditions in $\NPSPACE$, by guessing the starting
regions of these cycles and using respectively
Lemma~\ref{lm:pspace-cplay} (for the second condition) and the
following result (for the first condition):
\begin{lem}\label{lm:pspace-cplay-double}
  Consider a weighted timed game \game, a region state $(\loc,\reg)$ of \rgame,
  a bound $B\in\N$, and
  comparison operators ${\bowtie},{\bowtie'}\in\{{<},{>},{\leq},{\geq},{=},{\neq}\}$.
  Deciding if there exists a cycle \rpath of \rgame starting from $(\loc,\reg)$,
  and two corner plays \cplay and $\cplay'$, both following \rpath,
  such that $|\rpath|\leq B$, $\weightC(\cplay)\bowtie 0$ and
  $\weightC(\cplay')\bowtie' 0$,
  is in~\PSPACE.
\end{lem}
\begin{proof}
  We follow the same non-deterministic procedure as Lemma~\ref{lm:pspace-cplay},
  but this time we guess two corner plays on-the-fly instead of one.
\end{proof}

This shows that the membership problem for divergent weighted timed
games is in $\coNPSPACE=\coPSPACE=\PSPACE$~\cite{Sav70}.

Let us now show the ($\mathsf{co}$) $\PSPACE$-hardness by a reduction from the
reachability problem in a timed automaton, similar to the one we used
for the $\PSPACE$-hardness of deciding divergence.  We consider a
timed automaton with a starting location and a different target
location without outgoing edges.  We construct from it a weighted
timed game by distributing all locations to \MinPl, and equipping all
edges with weight $0$, and all locations with weight $0$.  We also add
a loop with weight $1$ on the initial location, one with weight $-1$
on the target location, and an edge from the target location to the
initial location with weight $0$, all three with guard $\top$ and
resetting all clocks.  Then, the weighted timed game is not
almost-divergent if and only if the target can be reached from the
initial location in the timed automaton.

\section{Deciding infinite values}\label{sec:kernels-infinity}\label{sec:kernel-infinity}

There are three reasons for an infinite value to appear in a \WTG:
\begin{itemize}
\item an infinite value ($+\infty$ or $-\infty$) could appear in a
  final weight function and be propagated along a play;
\item a value $+\infty$ could be obtained if $\MinPl$ is not able to
  reach a target location;
\item a value $-\infty$ could be obtained if $\MinPl$
  is able to reach a
  negative cycle, loop arbitrarily many times inside, and then reach a target
  location.
\end{itemize}

This section explains how to detect the three situations in the region
game $\rgame$ of an almost-divergent \WTG $\game$. Before that, let us
start by formalising a way to safely remove region states of $\rgame$
for which the value of all their associated configurations is known to
be infinite. Let $\States^{+\infty}$ be a subset of
$\States=\Locs\times\Regs$, such that
$\Val_{\rgame}((\loc,r),\nu)=+\infty$ for all
$(\loc,r)\in\States^{+\infty}$ and $\nu\in r$. If a configuration of
$\game$ is in the attractor of $\MaxPl$ towards $\States^{+\infty}$ (see Remark~\ref{rem:attractor}),
then $\MaxPl$ has a strategy giving it value $+\infty$. Moreover,
the attractor being computable using regions, either all configurations in the
same region have value $+\infty$, or none have value $+\infty$: therefore we could
add this region to $\States^{+\infty}$. We will thus assume that
$\States^{+\infty}$ is \emph{closed by attractor} for $\MaxPl$,
\ie~the attractor of $\MaxPl$ towards $\States^{+\infty}$ equals
$\States^{+\infty}$.  We can define the same notion for a set
$\States^{-\infty}$ of states with value $-\infty$, that can be
assumed closed by attractor of $\MinPl$.

\begin{lem}\label{lm:safe-removal}
  In $\rgame$, let $\States^{+\infty}$ be a set of region states of
  value $+\infty$ closed by attractor of $\MaxPl$, and let
  $\States^{-\infty}$ be a set of region states of value $-\infty$
  closed by attractor of $\MinPl$. Removing the region states
  $\States^{+\infty}\cup \States^{-\infty}$ from \rgame will not
  change any other value.
\end{lem}
\begin{proof}
  We only prove the result for the removal of region states from
  $\States^{-\infty}$, since the proof for $\States^{+\infty}$ is
  entirely symmetrical.
  Let $\States'=\StatesMin'\uplus\StatesMax'$ denote
  $\States\backslash\States^{-\infty}$, and let $\rgame'$ denote the
  restriction of $\rgame$ to
  $\States'$. %
  Let us show that for all region states $(\loc, r)\in\States'$ and
  $\nu\in r$,
  $\Val_{\rgame}((\loc,r),\nu)=\Val_{\rgame'}((\loc,r),\nu))$.

  If $\stratmin$ is a strategy of \MinPl in $\rgame$, and $\stratmin'$
  is a strategy of \MinPl in $\rgame'$, such that for all plays
  $\play$ in $\FPlaysMin_{\rgame'}$
  $\stratmin(\play)=\stratmin'(\play)$, we say that $\stratmin$
  coincides with $\stratmin'$. We define the same notion for
  strategies of \MaxPl.  All strategies of \MinPl and \MaxPl in
  $\rgame'$ can be extended arbitrarily to become strategies in \rgame
  that coincide, by making the same choices on plays that stay in
  $\rgame'$, and making arbitrary choices otherwise.  It follows that
  if $\stratmin$ and $\stratmin'$ are strategies of \MinPl that
  coincide, then for all strategies $\stratmax'$ of \MaxPl in
  $\rgame'$, there exists a corresponding strategy $\stratmax$ in
  $\rgame$, such that
  $\outcome_{\rgame'}(((\loc,r),\nu),\stratmin',\stratmax')=
  \outcome_{\rgame}(((\loc,r),\nu),\stratmin,\stratmax)$, and
  therefore
  $$\Val_{\rgame'}(((\loc,r),\nu),\stratmin')\leq\Val_{\rgame}(((\loc,r),\nu),\stratmin)$$
  as the supremum over strategies of \MaxPl in \rgame is at least equal
  to the one in $\rgame'$.  Similarly, if $\stratmax$ and $\stratmax'$
  are strategies of \MaxPl that coincide, then
  $$\Val_{\rgame'}(((\loc,r),\nu),\stratmax')\geq\Val_{\rgame}(((\loc,r),\nu),\stratmax)$$

  By closure by attractor of $\MinPl$, all transitions starting in a
  region state of $\StatesMin'$ must end in a region state of
  $\States'$ (otherwise the first state would belong to the attractor
  of $\MinPl$ towards $\States^{-\infty}$).  Therefore, every strategy
  of $\MinPl$ in \rgame has a corresponding strategy in $\rgame'$ that
  makes the same choices. In particular, if we consider an $\varepsilon$-optimal
  strategy $\stratmin$ of \MinPl in $\rgame$, and its corresponding
  strategy $\stratmin'$ in $\rgame'$, it holds that
  $$\Val_{\rgame'}((\loc,r),\nu)\leq
  \Val_{\rgame'}(((\loc,r),\nu),\stratmin')$$ by definition of $\Val$
  as an infimum over strategies, and
  $$\Val_{\rgame'}(((\loc,r),\nu),\stratmin')\leq\Val_{\rgame}(((\loc,r),\nu),\stratmin)$$
  as explained before. Finally,
  $\Val_{\rgame}(((\loc,r),\nu),\stratmin)\leq
  \Val_{\rgame}((\loc,r),\nu)+\varepsilon$ by $\varepsilon$-optimality
  of $\stratmin$, and thus
  $\Val_{\rgame'}((\loc,r),\nu)\leq\Val_{\rgame}((\loc,r),\nu)+\varepsilon$. Since
  this holds for all $\varepsilon>0$, we get that
  $\Val_{\rgame'}((\loc,r),\nu)\leq\Val_{\rgame}((\loc,r),\nu)$.

  From every state $\state$ of $\StatesMax'$, a positional strategy
  that chooses a transition jumping into $\States^{-\infty}$ must have
  value $-\infty$. Any other choice would thus be equal or better, and
  by closure by attractor of $\MinPl$, there must exist a transition
  starting in $(\loc,r)$ that ends in $\States'$.  Therefore, there
  exists an $\varepsilon$-optimal strategy $\stratmax$ of $\MaxPl$ in
  \rgame with a corresponding strategy $\stratmax'$ in
  $\rgame'$. Then, it holds that
  $\Val_{\rgame'}((\loc,r),\nu)\geq
  \Val_{\rgame'}(((\loc,r),\nu),\stratmax')\geq
  \Val_{\rgame}(((\loc,r),\nu),\stratmax)\geq
  \Val_{\rgame}((\loc,r),\nu)-\varepsilon$. We conclude since this
  holds for all $\varepsilon$.
\end{proof}

\subsection{Infinite final values}

As a first step, we explain how to compute and remove all region
states with value $+\infty$, and region states with value $-\infty$
because of final weights. The only states with infinite value that
will remain are some states that derive a value of~$-\infty$ from
arbitrary accumulation of negative weights, and we will deal with them
later.

Recall that the final weight function $\weightT$ has been supposed
piecewise affine with a finite number of pieces and is continuous on
each region. In particular, final weights $+\infty$ or $-\infty$ are
given to entire regions. Then, let
$\StatesT^{-\infty}\subseteq \LocsT\times\Regs$
(\resp~$\StatesT^{+\infty}$) denote the set of target region states
that $\weightT$ maps to $-\infty$ (\resp~$+\infty$).

\begin{prop}\label{prop:no-infty-final-wt}
  If a region state of $\game$ is in the attractor of $\MinPl$ towards
  $\StatesT^{-\infty}$ (\resp~in the attractor of $\MaxPl$ towards
  $\StatesT^{+\infty}$), then it has value $-\infty$
  (\resp~$+\infty$).  Moreover, if we remove those states from \rgame,
  the value of the other configurations does not change.
\end{prop}
\begin{proof}
  If a region state $(\loc,r)$ is in the attractor of $\MinPl$ towards
  $\StatesT^{-\infty}$, then clearly $\MinPl$ has a (positional)
  strategy giving value $-\infty$ from $(\loc,r)$, and
  $\Val_{\rgame}((\loc,r),\nu)=-\infty$ for all $\nu\in r$.  An
  attractor of $\MinPl$ is always closed by attractor of \MinPl, so we
  can apply Lemma~\ref{lm:safe-removal} and conclude. Once again, a
  symmetrical proof lets us deal with the attractor of $\MaxPl$
  towards $\StatesT^{+\infty}$.
\end{proof}

Then, assuming that all final weights are finite, configurations with
value $+\infty$ are those from which \MinPl cannot reach the target
region states: thus, they can also be computed and removed using the
attractor algorithm.

\begin{prop}\label{prop:no-plus-infty-wt}
  If $\weightT$ maps all configurations to values in $\R$, then a
  configuration $((\loc,r),\nu)$ has value $+\infty$ if and only if
  $(\loc,r)$ is not in the attractor of \MinPl towards region states
  $\LocsT\times\Regs$.  Moreover, if we remove those region states
  from \rgame, the value of the other configuration does not change.
\end{prop}
\begin{proof}
  If a region state $(\loc,r)$ is not in the attractor of $\MinPl$
  towards $\LocsT\times\Regs$, then $\MaxPl$ has a (safety) strategy
  giving value $+\infty$ from all configurations $((\loc,r),\nu)$ with
  $\nu\in r$, and $\Val_{\rgame}((\loc,r),\nu)= +\infty$.  Conversely,
  if $(\loc,r)$ is in the attractor of $\MinPl$ towards
  $\LocsT\times\Regs$, then $\MinPl$ has a strategy giving finite
  value from $((\loc,r),\nu)$, and
  $\Val_{\rgame}((\loc,r),\nu) < +\infty$.  We conclude by
  Lemma~\ref{lm:safe-removal}.
\end{proof}

We can now assume that all configurations have value in
$\R\cup\{-\infty\}$ and that all target region states have final
weight in $\R$: the precomputation needed so far consists only of
attractor computations, that can be performed in time linear in
$|\rgame|$.

As previously explained, finding all states of value $-\infty$ is
harder (in particular undecidable in general), but whenever we do
manage to find them we can safely remove them by
Lemma~\ref{lm:safe-removal}. The solution is quite simple for
divergent \WTG{s}, and slightly more involved in almost-divergent
\WTG{s}, because of the appearance of $0$-cycles. In order to only
present one uniform solution, we directly deal with the more general
almost-divergent case. To do so, we introduce a new tool, the
\emph{kernels}, that will also be very useful in the rest of the
study.

\subsection{Kernel of an almost-divergent \WTG}\label{sec:kernels-in-wtg}

Kernels of an almost-divergent \WTG are a way to group all $0$-cycles
of the game. We study those kernels and give a characterisation
allowing computability. In \cite{BJM15}, kernels have also been
studied, and contain exactly all transitions and locations of weight
$0$. Contrary to their non-negative case, the situation is more
complex in our case with arbitrary weights since $0$-cycles could go
through locations or transitions that have weight different from
$0$. Moreover, it is not trivial (and may not be true in a non
almost-divergent \WTG) to know whether it is sufficient to consider
only simple $0$-cycles, \ie~cycles without repetitions.

We will now construct the kernel~$\Kernel$ as the subgraph of~\rgame
containing all 0-cycles. Formally, let \RTransK be the set of edges
of~\rgame belonging to a \emph{simple} 0-cycle, and~\RStatesK be the
set of states covered by~\RTransK. We define the kernel~$\Kernel$ of
\rgame as the subgraph of~\rgame defined by \RStatesK and \RTransK.
Edges in $\RTrans\backslash\RTransK$ with starting state in \RStatesK
are called the output edges of~\Kernel.  We define it using only
simple 0-cycles in order to ensure its computability.  However, we now
show that this is of no harm, since the kernel contains exactly all
the 0-cycles, which will be crucial in our approximation schema.

\begin{prop}\label{prop:kernelprop}
  A cycle of $\rgame$ is entirely in~\Kernel if and only if it is a 0-cycle.
\end{prop}
\begin{proof}
  We prove that every 0-cycle is in~\Kernel by induction on the length
  of the cycles.  The initialisation contains only cycles of length
  $1$, that are in~\Kernel by construction.  If we consider a
  cycle~\rpath of length above $1$, it is either simple (and thus in~\Kernel, by definition),
  or it can be rotated and decomposed into~$\rpath'\rpath''$, $\rpath'$ and
  $\rpath''$ being smaller cycles.  Let~\play be a corner play
  following~$\rpath'\rpath''$.  We denote by $\play'$ the prefix
  of~\play following~$\rpath'$ and $\play''$ the suffix
  following~$\rpath''$.  It holds that
  $\weightC(\play')=-\weightC(\play'')$, and in an almost-divergent SCC
  this implies $\weightC(\play')=\weightC(\play'')=0$.  Therefore, by
  Lemma~\ref{lm:exist0} both~$\rpath'$ and~$\rpath''$ are 0-cycles,
  and they must be in~\Kernel by induction hypothesis.

  \begin{figure}[tbp]
    {\centering
    \scalebox{1}{
      \begin{tikzpicture}[node distance=5cm,auto,->,>=latex]
        \def \n {5}
        \def \radius {1cm}
        \def \margin {9} %

        \foreach \s in {1,...,\n}
        {
          \node[draw,circle,minimum width=8pt](\s) at ({90+360/\n * (\s - 1)}:\radius) {};
          \draw ({90+360/\n * (\s - 1)+\margin}:\radius)
          arc ({90+360/\n * (\s - 1)+\margin}:{90+360/\n * (\s)-\margin}:\radius)
          node[midway]{$\rtrans_{\s}$};
        }
        \path[looseness=2]
        (1) edge[bend left=100] node[midway]{$\rpath_{\rtrans_5}$} (5)
        (5) edge[bend left=100] node[midway]{$\rpath_{\rtrans_4}$} (4)
        (4) edge[bend left=100] node[midway]{$\rpath_{\rtrans_3}$} (3)
        (3) edge[bend left=100] node[midway]{$\rpath_{\rtrans_2}$} (2)
        (2) edge[bend left=100] node[midway]{$\rpath_{\rtrans_1}$} (1);
      \end{tikzpicture}
    }
}
    \caption{Application of the claim in the proof of
      Proposition~\ref{prop:kernelprop}}
    \label{fig:flower}
  \end{figure}

  We now prove that every cycle in \Kernel is a 0-cycle.  By
  construction, every edge $\rtrans\in \RTransK$ is part of a simple
  0-cycle.  Thus, to every edge~$\rtrans\in \RTransK$, we can
  associate a path~$\rpath_\rtrans$ such that $\rtrans\rpath_\rtrans$
  is a simple 0-cycle (rotate the simple cycle if necessary). The
  situation is exemplified in \figurename~\ref{fig:flower}. We can
  prove the following claim by relying on another pumping argument on
  corners:
\begin{claim}
  If $\rtrans_1\cdots\rtrans_k$ is a path in~\Kernel, then
  $\rtrans_1\rtrans_2\cdots\rtrans_k\rpath_{\rtrans_k}
  \cdots\rpath_{\rtrans_2}\rpath_{\rtrans_1}$ is a 0-cycle of~\rgame.
\end{claim}
  \begin{claimproof}
    We prove the property by induction on~$k$.  For~$k=1$, the property
    is immediate since $\rtrans_1\rpath_{\rtrans_1}$ is a 0-cycle.
    Consider then $k$ such that the property holds for $k$, and let us prove
    that it holds for $k+1$.  We will exhibit two corner plays following
    $\rtrans_1\cdots\rtrans_{k+1}\rpath_{\rtrans_{k+1}}\cdots\rpath_{\rtrans_1}$
    of opposite weight and conclude with Lemma~\ref{lm:exist0}.

    Let~$\corner_0$ be a corner of~$\last(\rtrans_{k+1})$.  Since
    $\rtrans_{k+1}\rpath_{\rtrans_{k+1}}$ is a 0-cycle, there
    exists~$w\in\Z$, a corner play~$\play_{0}$ following~$\rtrans_{k+1}$
    ending in~$\corner_0$ with weight~$w$ and a corner play~$\play'_{0}$
    following~$\rpath_{\rtrans_{k+1}}$ beginning in~$\corner_0$ with
    weight~$-w$.  We name~$\corner'_0$ the corner of~$\last(\rtrans_k)$
    where $\play'_{0}$ ends.  We consider any corner play $\play_1$
    following $\rtrans_{k+1}$ from corner $\corner'_0$.  The corner play
    $\play'_0\play_1$ follows the path
    $\rpath_{\rtrans_{k+1}}\rtrans_{k+1}$ that is also a $0$-cycle
    by Lemma~\ref{lm:rotatcycle},
    therefore $\play_1$ has weight $w$.  We denote by $\corner_1$ the
    corner where $\play_1$ ends.  By iterating this construction, we
    obtain some corner plays $\play_0,\play_1,\play_2,\ldots$ following
    $\rtrans_{k+1}$ and $\play'_0,\play'_1,\play'_2,\ldots$ following
    $\rpath_{\rtrans_{k+1}}$ such that $\play'_{i}$ goes from corner
    $\corner_i$ to $\corner'_i$, and $\play_{i+1}$ from corner $\corner'_i$ to
    $\corner_{i+1}$, for all $i\geq 0$.  Moreover, all corner plays
    $\play_i$ have weight $w$ and all corner plays $\play'_i$ have
    weight $-w$.  Consider the first index $\ell$ such that
    $\corner_\ell=\corner_j$ for some $j<l$, which exists because the number of
    corners is finite.

    We apply the induction hypothesis to find a corner play following
    $\rtrans_1\cdots\rtrans_{k}\rpath_{\rtrans_{k}}\cdots\rpath_{\rtrans_1}$,
    going through the corner $\corner'_j$ in the middle: more formally,
    there exists $w_\alpha$, a corner play~$\play_\alpha$ following
    $\rtrans_1\cdots\rtrans_k$ ending in~$\corner'_j$ with
    weight~$w_\alpha$ and a corner play~$\play'_\alpha$
    following~$\rpath_{\rtrans_k}\cdots\rpath_{\rtrans_1}$ beginning
    in~$\corner'_j$ with weight~$-w_\alpha$.  We apply the induction hypothesis a
    second time with corner $\corner'_{\ell-1}$: there exists $w_\beta$, a
    corner play~$\play_\beta$ following $\rtrans_1\cdots\rtrans_k$
    ending in~$\corner'_{\ell-1}$ with weight~$w_\beta$ and a corner
    play~$\play'_\beta$
    following~$\rpath_{\rtrans_k}\cdots\rpath_{\rtrans_1}$ beginning
    in~$\corner'_{\ell-1}$ with weight~$-w_\beta$.

    The corner play
    $\play_\alpha\play_{j+1}\play'_{j+1}\play_{j+2}\play'_{j+2}
    \cdots\play'_{\ell-1}\play'_\beta$, of weight
    $w_\alpha+(w-w)(\ell-j)-w_\beta=w_\alpha-w_\beta$, follows the cycle
    $\rtrans_1\cdots\rtrans_k(\rtrans_{k+1}\rpath_{\rtrans_{k+1}})^{\ell-j}
    \rpath_{\rtrans_k}\cdots\rpath_{\rtrans_1}$.  The corner play
    $\play_\beta\play_\ell\play'_j\play'_\alpha$, of weight
    $w_\beta+w-w-w_\alpha=w_\beta-w_\alpha$, follows the cycle
    $\rtrans_1\cdots\rtrans_k\rtrans_{k+1}
    \rpath_{\rtrans_{k+1}}\rpath_{\rtrans_k}\cdots\rpath_{\rtrans_1}$.  Since
    the game is almost-divergent, and those two corner plays are of opposite sign and in the
    same SCC, both have weight 0.  The second corner play of weight 0
    ensures that the cycle
    $\rtrans_1\cdots\rtrans_{k+1}
    \rpath_{\rtrans_{k+1}}\cdots\rpath_{\rtrans_1}$ is a 0-cycle, by
    Lemma~\ref{lm:exist0}.
  \end{claimproof}

  Now, if $\rpath$ is a cycle of \rgame in~\Kernel, there exists a
  cycle~$\rpath'$ such that $\rpath\rpath'$ is a 0-cycle. Since
  $\rgame$ is an almost-divergent \WTG, $\rpath$ is a 0-cycle.
\end{proof}

\subsection{Values \texorpdfstring{$-\infty$}{-infinity} coming from negative cycles}\label{sec:no-infty}

Equipped with the kernels, we are now ready to remove the only
remaining configurations having a value $-\infty$ in $\rgame$.
\begin{prop}\label{prop:-infty}
  In an SCC of $\rgame$, the set of configurations with value
  $-\infty$ is a union of regions computable in time linear in the
  size of $\rgame$. Moreover, if we remove those states from \rgame,
  the value of the other configurations does not change.
\end{prop}
\begin{proof}
  If the SCC is non-negative, the cumulative weight cannot decrease
  along a cycle,
    thus, there can be no configurations with value $-\infty$.

If the SCC is non-positive,
  let  $\RTransT$ be the set of edges of \rgame whose
  end state has location in $\LocsT$.
  We prove that a configuration has value $-\infty$ if and only if it
  belongs to a region state where player~\MinPl can ensure the LTL
  formula on edges
  $\phi = ( \mathrm G\,\neg\RTransT )\wedge \neg \mathrm F \mathrm G
  \, \RTransK $: in simple words, this means that $\MinPl$ can ensure to see an edge from $\RTransT$, or exit $\RTransK$ if we enter the kernel at some point.

  If $(\loc,\reg)$ is a region state where~\MinPl can ensure $\phi$,
  \MinPl can ensure a value of $-\infty$ from all configurations
  in~$(\loc,\reg)$ by avoiding $\RStatesT$ for as long as desired,
  while not getting stuck in~\Kernel, and thus going through an
  infinite number of negative cycles by
  Proposition~\ref{prop:kernelprop}. Conversely, suppose that $(\loc,\reg)$ is a
  region state from which $\MinPl$ cannot ensure $\phi$. The winner of an $\omega$-regular timed game only depends on the (finite) region abstraction \cite{WonHof92}, and finite turn-based $\omega$-regular games are determined \cite{BucLan69}. Thus, we know that \MaxPl can ensure
  $\lnot \phi =(\mathrm F \RTransT) \vee \mathrm F \mathrm G
  \RTransK$. Then, from~$(\loc,\reg)$, \MaxPl must be able to reach
  $\RStatesT$ or stay in \Kernel forever. In both cases, \MaxPl can
  ensure a value above $-\infty$.

  The procedure to detect $-\infty$ states thus consists in checking
  LTL formula
  $( \mathrm G\,\neg\RTransT )\wedge \neg \mathrm F \mathrm G \,
  \RTransK $, and thus can be performed via three attractor
  computations in time linear in $|\rgame|$. The set of region states
  with value $-\infty$ can then be removed safely from $\rgame$, by
  Lemma~\ref{lm:safe-removal} (since it is closed by attractor of
  $\MinPl$).
\end{proof}

From this point on, \textbf{we assume that no
  configurations of $\rgame$ have value $+\infty$ or $-\infty$, and
  that the final weight function maps all configurations to $\R$.}
Since $\weightT$ is piecewise affine with finitely many pieces,
$\weightT$ is bounded.  Let $\wmax^\target$ denote the supremum of
$|\weightT|$, ranging over all target configurations.

\section{Semi-unfolding of \WTG{s}}\label{sec:unfolding}

Given an almost-divergent \WTG $\game$, we describe the construction
of its \emph{semi-unfolding} $\tgame$ (as depicted
in~\figurename~\ref{fig:schema}). This crucially relies on the
absence of states with value~$-\infty$ which has been explained in
Section~\ref{sec:kernel-infinity}.

We only build the semi-unfolding \tgame of an SCC of $\game$ starting
from some state $(\loc_0,\reg_0)\in S$ of the region game, since it is
then easy to glue all the semi-unfoldings together to get the one of
the full game. \textbf{We thus assume in this section that $\rgame$ is
  an SCC.}
Since every configuration has finite value, we will prove that values
of the game are bounded by $|\rgame|\wmaxTimed+\wmax^\target$.  As a
consequence, we can find a bound $\gamma$ linear in $|\rgame|$,
$\wmaxTimed$ and $\wmax^\target$ such that a play that visits some state
outside the kernel more than $\gamma$ times has weight strictly above
$|\rgame|\wmaxTimed+\wmax^\target$, hence is useless for the value
computation.
This leads to considering the semi-unfolding $\tgame$ of $\rgame$
(nodes in the kernel are not unfolded, see \figurename~\ref{fig:schema})
such that each node not in the kernel is encountered at most $\gamma$
times along a branch: the end of each branch is called a \emph{stopped
  leaf} of the semi-unfolding.  In particular, the depth of $\tgame$ is
bounded by $|\rgame| \gamma$, and thus is polynomial in $|\rgame|$,
$\wmaxTimed$ and $\wmax^\target$.  Leaves of the semi-unfolding are thus
of two types: target leaves that are copies of target locations of
$\rgame$ for which we set the target weight as in $\rgame$, and stop
leaves for which we set their target weight as being constant to
$+\infty$ if the SCC $\rgame$ is non-negative, and $-\infty$ if the SCC
is non-positive.

More formally, if $(\loc,\reg)$ is in~\Kernel, we
let~$\Kernel_{\loc,\reg}$ be the part of~\Kernel accessible
from~$(\loc,\reg)$ (note that $\Kernel_{\loc,\reg}$ is an SCC as
\Kernel is a disjoint set of SCCs).  We define the output edges
of~$\Kernel_{\loc,\reg}$ as being the output edges of~\Kernel
accessible from~$(\loc,\reg)$.  If~$(\loc,\reg)$ is not in~\Kernel,
the output edges of~$(\loc,\reg)$ are the edges of~\rgame starting
in~$(\loc,\reg)$.

We define a tree $T$ whose nodes are either labelled by region states
$(\loc,\reg)\in\RStates\backslash\RStatesK$ or by kernels
$\Kernel_{\loc,\reg}$, and whose edges are labelled by output edges
in~\rgame.  The root of the tree~$T$ is labelled with an initial
region state $(\loc_0,\reg_0)$, or~$\Kernel_{\loc_0,\reg_0}$ (if
$(\loc_0,\reg_0)$ belongs to the kernel), and the successors of a node
of $T$ are then recursively defined by its output edges.  When a
state~$(\loc,\reg)$ is reached by an output edge, the child is
labelled by~$\Kernel_{\loc,\reg}$ if~$(\loc,\reg)\in \Kernel$,
otherwise it is labelled by~$(\loc,\reg)$.  Edges in~$T$ are labelled
by the edges used to create them.  Along every branch, we stop the
construction when either a final state is reached (\ie~a state not
inside the current SCC) or the branch contains
$3|\rgame|\wmaxTimed+2\wmax^\target+2$ nodes labelled by the same
state ($(\loc,\reg)$ or~$\Kernel_{\loc,\reg}$).
Leaves of $T$ with a location
belonging to~\LocsT are called \emph{target leaves}, others are called
\emph{stopped leaves}.

We now transform~$T$ into a \WTG~\tgame, by replacing every node
labelled by a state~$(\loc,\reg)$ by a different copy $(\tilde{\loc},\reg)$
of $(\loc,\reg)$.  Those states are said to inherit from~$(\loc,\reg)$.  Edges
of~$T$ are replaced by the edges labelling them, and have a
similar notion of inheritance.  Every non-leaf node labelled by a
kernel~$\Kernel_{\loc,\reg}$ is replaced by a copy of the \WTG
$\Kernel_{\loc,\reg}$, output edges being plugged in the expected
way.  We deal with stopped leaves labelled by a
kernel~$\Kernel_{\loc,\reg}$ by replacing them with a single node copy
of~$(\loc,\reg)$, like we dealt with node labelled by a state~$(\loc,\reg)$.
State partition between players and weights are inherited from the
copied states of~\rgame.  The only initial state of~\tgame is the
state denoted by $(\tilde{\loc}_0,\reg_0)$ inherited from $(\loc_0,\reg_0)$
in the root of $T$ (either $(\loc_0,\reg_0)$ or $\Kernel_{\loc_0,\reg_0}$).
The final states of~\tgame are the states derived from leaves of
$T$.  If \rgame is a non-negative (\resp~non-positive) SCC, the
final weight function \weightT is inherited from~\rgame on target
leaves and set to $+\infty$ (\resp~$-\infty$) on stopped
leaves.

The semi-unfolding of the \WTG from \figurename~\ref{fig:example-regions} can be found in \figurename~\ref{fig:example-unfolding} of Appendix~\ref{app:example}.

\begin{prop}\label{prop:semi-unfolding}
  Let $\game$ be an almost-divergent \WTG, and let $(\loc_0,\reg_0)$
  be some region state of $\rgame$. The semi-unfolding $\tgame$ with
  initial state $(\tilde{\loc}_0,\reg_0)$ (a copy of a region state
  $(\loc_0,\reg_0)$) is equivalent to $\game$, \ie~for all
  $\val_0\in \reg_0$,
  $\Val_{\game}(\loc_0,\val_0) =
  \Val_{\tgame}((\tilde{\loc}_0,\reg_0),\val_0)$.
\end{prop}

The rest of this section aims at proving
Proposition~\ref{prop:semi-unfolding}, only in the case where $\rgame$
is an SCC, since the general result then follows easily. First, we
need information on the weight of finite plays in the region game.

\begin{lem}\label{lm:pathbound}
  All finite plays in~\rgame have cumulative weight (ignoring final weights)
  at least $-|\rgame|\wmaxTimed$ in the non-negative case, and at most
  $|\rgame|\wmaxTimed$ in the non-positive case.  Moreover, values of the game
  are bounded by $|\rgame|\wmaxTimed+\wmax^\target$.
\end{lem}
\begin{proof}
  Suppose first that $\rgame$ is a non-negative SCC.  Consider a
  play~\play following a path~\rpath.  This path \rpath can be decomposed into
  $\rpath=\rpath_1\rpath^c_1\cdots\rpath_k\rpath^c_k$ such that
  every~$\rpath^c_i$ is a cycle, and $\rpath_1\dots\rpath_k$ is a
  simple path in \rgame (thus $\sum_{i=1}^k |\rpath_i|\leq|\rgame|$).
  Let us define all plays~$\play_i$ and~$\play^c_i$ as the
  restrictions of~\play on~$\rpath_i$ and~$\rpath^c_i$.  Now, since
  all plays following cycles have cumulative weight at least~$0$,
  $$\weightC(\play)=\sum_{i=1}^k \weightC(\play_i)+\weightC(\play^c_i)\geq
  \sum_{i=1}^k -\wmaxTimed|\play_i| + 0\geq -|\rgame|\wmaxTimed$$
  Similarly, we can show
  that every play in a non-positive SCC has cumulative weight at most $|\rgame|\wmaxTimed$.

  For the bound on the values, consider again two cases.  If~\rgame is a
  non-negative SCC, consider a positional attractor strategy~\stratmin
  for~\MinPl toward~\RStatesT.  Since all states have values below
  $+\infty$, all plays obtained from strategies of~\MaxPl will follow
  simple paths of~\rgame, that have cumulative weight at most
  $|\rgame|\wmaxTimed$ in absolute value.  Similarly, if~\rgame is
  a non-positive SCC, following the proof of Proposition~\ref{prop:-infty},
  since all values are above $-\infty$, $\MaxPl$ can ensure
  $\lnot \phi$,
  \ie~$(\mathrm F \RTransT) \vee \mathrm F \mathrm G \RTransK$ on all
  states.  Then we can construct a strategy~\stratmax for~\MaxPl
  combining an attractor strategy toward~\RStatesT on states
  satisfying $\mathrm F \RTransT$, a safety strategy on states
  satisfying $G \RTransK$, and an attractor strategy toward the latter
  on all other states.  Then, all plays obtained from strategies
  of~\MinPl will either not be winning ($G \RTransK$) or follow simple
  paths of~\rgame.  Both cases imply that the values of the game are
  bounded by $|\rgame|\wmaxTimed+\wmax^\target$.
\end{proof}

We can use this to obtain similar results in the semi-unfolding
$\tgame$.

\begin{lem}\label{lm:rootleafplay}
  All plays in~\tgame from the initial state $(\tilde \ell_0,r_0)$ to
  a stopped leaf have cumulative weight at least
  $2|\rgame|\wmaxTimed+2\wmax^\target+1$ if the SCC \rgame is
  non-negative, and at most $-2|\rgame|\wmaxTimed-2\wmax^\target-1$ if
  it is non-positive.
\end{lem}
\begin{proof}
  Note that by construction, all finite paths in \tgame from the
  initial state to a stopped leaf can be decomposed as
  $\rpath'\rpath_1\cdots\rpath_{3|\rgame|\wmaxTimed+2\wmax^\target+1}$ with all $\rpath_i$
  being cycles around the same state.  Additionally, those cycles cannot be 0-cycles by
  Proposition~\ref{prop:kernelprop}, since they take at least one
  edge outside of~\Kernel.  Therefore the restriction of~\play
  to $\rpath_1\cdots\rpath_{3|\rgame|\wmaxTimed+2\wmax^\target+1}$ has weight at
  least~$3|\rgame|\wmaxTimed+2\wmax^\target+1$ (in the non-negative case) and at most
  $-3|\rgame|\wmaxTimed-2\wmax^\target-1$ (in the non-positive case).  The beginning of the
  play, following $\rpath'$, has cumulative weight at
  least~$-|\rgame|\wmaxTimed$ (in the non-negative case) and at most
  $|\rgame|\wmaxTimed$ (in the non-positive case), by Lemma~\ref{lm:pathbound}.
\end{proof}

Consider two plays of the same length $k$ in $\rgame$ and $\tgame$, respectively:
$$\play=((\loc_1,\reg_1),\val_1)\xrightarrow{\delay_1,\rtrans_1}\cdots
\xrightarrow{\delay_{k-1},\rtrans_{k-1}}((\loc_k,\reg_k),\val_k)$$
$$\tilde{\play}=((\tilde{\loc}_1,\reg_1),\val_1)\xrightarrow{\delay_1,
  ,\tilde{\rtrans}_1}\cdots \xrightarrow{\delay_{k-1},
  \tilde{\rtrans}_{k-1}}((\tilde{\loc}_k,\reg_k),\val_k)\,.$$
They are said to \emph{mimic} each other if every
$(\tilde{\loc}_i,\reg_i)$ is inherited from $(\loc_i,\reg_i)$ and every
edge $\tilde{\rtrans}_i$ is inherited from the edge
$\rtrans_i$ of \rgame.  Combining Lemmas~\ref{lm:pathbound} and
\ref{lm:rootleafplay}, we obtain:
\begin{lem}\label{lm:mimickedplays}
  If \rgame is a non-negative (\resp~non-positive) SCC, every
  play from the initial state and with cumulative weight less than
  $|\rgame|\wmaxTimed+2\wmax^\target+1$ (\resp~greater than $-|\rgame|\wmaxTimed-2\wmax^\target-1$) can be
  mimicked in~\tgame without reaching a stopped leaf.  Conversely,
  every play in~\tgame reaching a target leaf can be mimicked
  in~\rgame.
\end{lem}
\begin{proof}
  We prove only the non-negative case, since the other case is
  symmetrical.  Let~\play be a play of~\rgame with cumulative weight
  less than $|\rgame|\wmaxTimed+2\wmax^\target+1$.  Consider the
  branch of the unfolded game it follows.  If~\play cannot be mimicked
  in~\tgame, then a prefix of~\play reaches the stopped leaf of that
  branch when mimicked in~\tgame.  In this situation, \play starts by
  a prefix of weight at least $2|\rgame|\wmaxTimed+2\wmax^\target+1$
  by Lemma~\ref{lm:rootleafplay} and then ends with a suffix play of
  weight at least~$-|\rgame|\wmaxTimed$ by Lemma~\ref{lm:pathbound},
  and that contradicts the initial assumption.  The converse is true
  by construction.
\end{proof}

Then, intuitively, the plays of~\rgame starting in an initial
configuration that cannot be mimicked in~\tgame are not useful for
value computation. To obtain Proposition~\ref{prop:semi-unfolding}, we now
prove that for all valuations $\val_0\in \reg_0$,
$\Val_{\game}(\loc_0,\val_0) =
\Val_{\tgame}((\tilde{\loc}_0,\reg_0),\val_0)$.  By
Lemma~\ref{lm:region-game}, we already know that
$\Val_\game(\loc_0,\val_0) = \Val_{\rgame}((\loc_0,\reg_0),\val_0)$.
Recall that we only left finite values in \rgame (in the final weight
functions, in particular), and more precisely
$|\Val_{\rgame}((\loc_0,\reg_0),\val_0)|\leq
|\rgame|\wmaxTimed+\wmax^\target$ by Lemma~\ref{lm:pathbound}.  We
first show that the value is also finite in \tgame.  Indeed, if
$\Val_{\tgame}((\tilde{\loc}_0,\reg_0),\val_0)=+\infty$, since we
assumed all final weights of \rgame bounded, we are necessarily in the
non-negative case, and $\MaxPl$ is able to ensure stopped leaves
reachability.

  \begin{claim}
  If $\Val_{\tgame}((\tilde{\loc}_0,\reg_0),\val_0)=+\infty$, then
  there are no strategies in~\rgame for~\MinPl ensuring weight less than
  $|\rgame|\wmaxTimed+\wmax^\target+1$ from~$(\loc_0,\reg_0)$.
  \end{claim}
  Thus, we can obtain the contradiction $\Val_{\rgame}((\loc_0,\reg_0),\val_0)>
  |\rgame|\wmaxTimed+\wmax^\target$.
  \begin{claimproof}%
    By contradiction, consider a strategy \stratmin of \MinPl ensuring weight
    $A\leq|\rgame|\wmaxTimed+\wmax^\target+1$ in \rgame.
    Then, for all \stratmax, the cumulative weight of the play
    $\outcome_{\rgame}(((\tilde{\loc}_0,\reg_0),\val_0),\stratmin,\stratmax)$
     (reaching target configuration $(\loc,\val)$)
    is at most equal to $A-\weightT(\loc,\val)\leq |\rgame|\wmaxTimed+2\wmax^\target+1$, and by Lemma~\ref{lm:mimickedplays}
    this play does not reach a stopped leaf when mimicked in \tgame, which is absurd.
  \end{claimproof}
  If $\Val_{\tgame}((\tilde{\loc}_0,\reg_0),\val_0)=-\infty$, we are
  necessarily in the non-positive case, and, by construction, this
  implies that~\MinPl ensures stopped leaves reachability
  in~\tgame.

  \begin{claim}
  If $\Val_{\tgame}((\tilde{\loc}_0,\reg_0),\val_0)=-\infty$, then
  there are no strategies in~\rgame for~\MaxPl ensuring weight above
  $-|\rgame|\wmaxTimed-\wmax^\target-1$ from~$(\loc_0,\reg_0)$.
  \end{claim}
  Thus, we can obtain the contradiction $\Val_{\rgame}((\loc_0,\reg_0),\val_0)<
  -|\rgame|\wmaxTimed-\wmax^\target$.
  \begin{claimproof}%
    By contradiction, consider a strategy \stratmax of \MaxPl ensuring weight
    $A\geq-|\rgame|\wmaxTimed-\wmax^\target-1$ in \rgame.
    Then, for all \stratmin, the cumulative weight of the play  $\outcome_{\rgame}(((\tilde{\loc}_0,\reg_0),\val_0),\stratmin,\stratmax)$ (reaching target configuration $(\loc,\val)$)
    is at least $A-\weightT(\loc,\val)$, thus at least $-|\rgame|\wmaxTimed -2\wmax^\target -1$. By Lemma~\ref{lm:mimickedplays},
    this play does not reach a stopped leaf when mimicked in \tgame, which is absurd.
  \end{claimproof}

  If~\rgame is non-negative, for all $\varepsilon>0$ we can fix an
  $\varepsilon$-optimal strategy~\stratmin for~\MinPl in~\tgame.
  Every play derived from \stratmin in \tgame reaches a target leaf,
  and can thus be mimicked in \rgame by Lemma~\ref{lm:mimickedplays}.  Therefore,
  \stratmin can be mimicked in \rgame, where it keeps %
  the same value.
  From this we deduce $\Val_{\rgame}((\loc_0,\reg_0),\val_0) \leq \Val_{\tgame}(
  (\tilde{\loc}_0,\reg_0),\val_0)$.  If~\rgame is non-positive,
  the same reasoning applies by considering
  an $\varepsilon$-optimal strategy for $\MaxPl$ in~$\tgame$.

  Let us now show the reverse inequality.
  If~\rgame is non-negative, let us fix $0<\varepsilon<1$, an
  $\varepsilon$-optimal strategy~\stratmin for~\MinPl in~\rgame, and a
  strategy~\stratmax of~\MaxPl in~\rgame.
  Let $\play$ be their outcome $\outcome_{\rgame}(((\loc_0,\reg_0),\val_0),\stratmin,\stratmax)$,
  $\play_k$ be the finite prefix of $\play$ defining its cumulative weight
  and $(\loc_k,\val_k)$ be the configuration defining its final weight, such that
  $\weight_{\rgame}(\play)=\weightC(\play_k)+\weightT(\loc_k,\val_k)$.
  Then, $$\weight_{\rgame}(\play)\leq \Val_{\rgame}((\loc_0,\reg_0),\val_0)+\varepsilon <
  |\rgame|\wmaxTimed+\wmax^\target+1$$ therefore
  $$\weightC(\play_k)<|\rgame|\wmaxTimed+\wmax^\target+1-\weightT(\loc_k,\val_k)
  \leq |\rgame|\wmaxTimed+2\wmax^\target+1$$
  and by Lemma~\ref{lm:mimickedplays}
  all such plays $\play$ can be mimicked in~\tgame, so that
  $$\Val_{\tgame}((\tilde{\loc}_0,\reg_0),\val_0) \leq
  \Val_{\rgame}((\loc_0,\reg_0),\val_0)$$
  Once again, if~\rgame is non-positive, the same reasoning applies by considering
  an $\varepsilon$-optimal strategy for $\MaxPl$ in~$\rgame$. This
  ends the proof of Proposition~\ref{prop:semi-unfolding}.

  \begin{rem}
    Note that the semi-unfolding procedure of an SCC depends on
    $\wmax^\target$, where $\weightT$ can be the value function of an
    SCCs under the current one.  Assuming all configurations have
    finite values, we can extend the reasoning of
    Lemma~\ref{lm:pathbound} and bound all values in the full game by
    $|\rgame|\wmaxTimed+\wmax^\target$, which lets us bound uniformly
    the unfolding depth of each SCC and gives us a bound on the depth
    of the complete semi-unfolding tree:
    \begin{equation}
      |\rgame|(3|\rgame|\wmaxTimed+2\wmax^\target+2)+1\label{eq:depth-semi-unfolding}
  \end{equation}
  \end{rem}

\section{Computing values for acyclic WTGs}\label{sec:acyclic}

In this section, we focus on the class of acyclic \WTG{s} where the
value problem is decidable, as shown by~\cite{AluBer04}. In our
context, this will give us a way to compute the value of the
semi-unfolding $\tgame$ of an almost-divergent \WTG that contains no
kernels: this is equivalent for the \WTG to be divergent indeed. This
section is particularly technical, and can thus be skipped by
non-interested readers. We have chosen to give such detailed
explanations of the techniques used in~\cite{AluBer04} (with
independent proofs) for several reasons.
  On the one hand, our setting is more general, in the sense that we
  allow for negative weights and for final weights, where the authors of~\cite{AluBer04} do not
  do so explicitly.
  On the other hand, their result is stated for concurrent games, a
  generalisation of the turn-based games we consider.  This leads to
  simplifications in the proofs, and makes easier some parts of the
  complexity analysis.
  We will need, in Section~\ref{sec:computing}, to bound
  the partial derivatives of the functions we compute.  This cannot be
  deduced from their result directly.
  We present their techniques in a new, more symbolic light, by
  performing computations on the entire state-space at once instead of
  region by region.
  For reasons detailed in Section~\ref{sec:valitediscuss} and
  in~\cite{busattogaston:tel-02436831}, we are not able to replicate
  their (incomplete) complexity analysis.  We will therefore rely on a
  doubly-exponential upper bound instead of the exponential one
  claimed in~\cite{AluBer04}.
  Last but not least, this detailed analysis allows us to solve the
  synthesis of $\varepsilon$-optimal strategies for both players, as
  will be
  detailed in Section~\ref{sec:strategies}.

  The main result of this section is a symbolic algorithm for
  computing the value $\ValIteVec^i=\Val_\game^i$ defined in
  Section~\ref{subsec:value-functions}.
\begin{thm}\label{thm:valitebounded}
  Given $i\geq 0$, computing $\Val_\game^i$ can be done in time
  doubly-exponential in $i$ and
  exponential in the size of \game.
\end{thm}

The intuition behind the result is the observation that the
mappings~$\ValIteVec^0_\loc$ are piecewise affine for all $\loc$, and
a proof that $\ValIteOpe$ preserves piecewise affinity, so that all
iterates $\ValIteVec^i_\loc$ can be computed using piecewise affine
functions.  In order to bound the size of $\ValIteVec^i_\loc$ (in
particular, its number of pieces), we need the fine encoding via cells
and partition functions defined in
Section~\ref{sec:piecewise-affine-value-functions}.

\subsection{About complexity bounds}\label{sec:complexity-proof}

  In this section we will assume without loss of generality that the number of clocks $n\geq 1$. The case $n=0$ (finite weighted games) will be detailed in Section~\ref{sec:solving-wg}.
  The piecewise affine value function \ValIteVec is
  encoded as pairs $(P_\loc,F_\loc)_{\loc\in\Locs}$ such that
  $\sem{F_\loc}=\ValIteVec_\loc$ for all $\loc\in\Locs$.
	We detail now how we measure the complexity of a pair $(P,F)$. As mentioned
	before, we assume without loss of generality that (in)equalities involved in the definition of cells
	only use integers. This is not the case for value functions which are described by
	equations $\clocky= a_1\clock_1+\cdots
	  +a_n\clock_n+b$ with $a_i$
	  and~$b$ in~\Q. In order to track their size, we instead write
	  $a_\clocky \clocky=a_1\clock_1+\cdots
	  +a_n\clock_n+b$, with all $a_i$ and
	  $b$ integers of \Z, and $a_\clocky\in\Nspos$.

	\begin{defi}
		Let $(P,F)$ be a pair composed of a partition $P$ and a value function
		$F$ defined on that partition. The \emph{complexity of $(P,F)$} is
		the pair $\langle m,\beta \rangle$ where the \emph{size complexity} $m$ and
		the \emph{constant complexity} $\beta$ are defined as follows.
		First, we define the size complexity as $m = |\mathcal E_{P}| + m_0$, where
		$m_0$ denotes the number of inequalities in the encoding
		of the domain $c_P$.\footnote{If $c_P=\ValSpaceBound$, then $m_0=n$.}
		Second, the constant complexity $\beta$ is the smallest natural number that upper bounds (in absolute value)
		every finite constant (partial derivatives
		and additive constants) in the affine (in)equalities
                of $P$ and $F$. Given a piecewise affine value function $(P,F)$ of complexity $\langle m,\beta \rangle$,
	we say that $(P,F)$ has complexity at most $\langle m',\beta'\rangle$
	if $m\leq m'$ and $\beta \leq \beta'$.
	\end{defi}

	From this definition, we can bound the memory needed to store
        a partition function:
	\begin{lem}\label{lem:memory-partition}
	Let $\ValIteVec$ be represented by a pair $(P,F)$ of
        complexity at most $\langle m,\beta \rangle$.
	Then it can be represented using a space in
        $\mathcal O(2^{n}(n+2)(m+1)^{n+1}\lceil{\log\beta}\rceil)$.
	\end{lem}

	\begin{proof}
    Storing $P$ requires storing $m$ affine expressions,
    each stored in space $(n+1)\lceil{\log\beta}\rceil$.
  Moreover, each affine expression of the value function $F$, one for each base cell of $P$,
  is stored in space $(n+2)\lceil{\log\beta}\rceil$.
  $P$ has at most $2^n(m+1)^{n}$ base cells by Lemma~\ref{lm:splitnum}, and
  storing each base cell explicitly requires storing at most $m$ borders per base cell.
  In total,
  $\ValIteVec$ can be represented in space
  $2^{n+1}(n+2)(m+1)^{n+1}\lceil{\log\beta}\rceil$.
  \end{proof}

\subsection{Operations over value functions}

Our goal is now to compute the value function
$\ValIteOpe(\ValIteVec)$ encoded by
$(P_\loc',F_\loc')_{\loc\in\Locs}$, when $\ValIteVec$ is encoded by
$(P_\loc,F_\loc)_{\loc\in\Locs}$.  We decompose the computation into
smaller operations that we first introduce and compute over
partitions, while explaining how they affect the complexity parameter
$\langle m,\beta \rangle$.

  If $P_1$ and $P_2$ are partitions over the same domain~$c_P$,
  let $P_1\oplus P_2$ denote the coarsest partition
  consistent with both $P_1$ and $P_2$: each base cell $c_b$ of $P_1\oplus P_2$
  corresponds to an intersection $c_1\cap c_2$, with $c_1$ a base cell of $P_1$
  and $c_2$ a base cell of $P_2$.
  It is obtained from $\mathcal E_{P_1\oplus P_2}=\mathcal E_{P_1}\cup \mathcal E_{P_2}$.
  Note that if $P_1 \dots P_q$ are partitions of complexity
  at most $\langle m,\beta \rangle$, %
  $P_1\oplus\dots\oplus P_q$ is a partition of size complexity at most~$\langle qm,\beta \rangle$.
  Now, the minimum (\resp~maximum) of a finite set of piecewise affine value functions
  can be computed with partitions.
\begin{lemC}[{\cite[Thm.~1]{AluBer04}}]\label{lm:valitemin}
  Let $(P_i,F_i)_{1\leq i\leq q}$ be piecewise affine value functions,
  defined over the same domain $c_P$, where each $(P_i,F_i)$ has
  complexity at most $\langle m,\beta \rangle$. %
  There exists a piecewise affine value function $(P',F')$
  of domain $c_P$ and complexity at most $\langle m',\beta'\rangle$,
  where $m'=q m+q^2$,
  and $\beta'=2\beta^2$,
  such that $\sem{F'}=\min_{1\leq i\leq q} \sem{F_i}$
  (\resp~$\sem{F'}=\max_{1\leq i\leq q} \sem{F_i}$).
\end{lemC}
\begin{proof}
  Let $P'$ be $P_1\oplus\dots\oplus P_q$.
  Let $c$ denote a base cell in $P'$, corresponding to an intersection
  $c_1\cap\dots\cap c_q$ of base cells of $P_1, \dots, P_q$ respectively.
  Consider the affine value functions
  $F_1(c_1), \dots, F_q(c_q)$.
  Each of these is defined by an equation of the form
  $a_\clocky \clocky=a_1\clock_1+\cdots+a_n\clock_n+b$, that can be
  seen as
  affine equalities over variables $\Clocks\uplus\{\clocky\}$,
  denoted $E_1,\dots, E_q$,
  or equivalently as sets of points in $\R^{\Clocks\uplus\{\clocky\}}$,
  denoted $\sem{E_1},\dots,\sem{E_q}$.
  If $E$ and $E'$ are such equalities, of equations $a_\clocky
  \clocky=a_1\clock_1+\cdots+a_n\clock_n+b$
  and $a_\clocky' \clocky=a_1'\clock_1+\cdots+a_n'\clock_n+b'$, the
  intersection $\sem{E}\cap\sem{E'}$
  is either empty %
  or, by elimination of $\clocky$, it satisfies the equation
  $$(a_\clocky a_1'-a_\clocky' a_1)\clock_1+\cdots+ (a_\clocky
  a_n'-a_\clocky' a_n)\clock_n +(a_\clocky b'-a_\clocky' b)=0\,.$$
  This describes an affine equality over $\Clocks$, that we denote $E\cap_y E'$,
  with constant complexity at most $2\beta^2$.
  Now, let us partition $P_1\oplus\dots\oplus P_q$ by the set of all such intersections
  $$\{E_i\cap_y E_j \st 1\leq i,j \leq q\land \sem{E_i}\cap\sem{E_j}\neq\emptyset \}\,.$$
  On every sub-cell $c'$ of this partition, there exists
  $j\in\{1,\dots,n\}$ such that for every $\val\in c'$,
  $\sem{F_j}(\val)=\min_{1\leq i\leq q} \sem{F_i}(\val)$.  Therefore, we
  define~$F'$ on~$c'$ as equal to~$F_j(c_j)$.
  The partition $P_1\oplus\dots\oplus P_q$ has size complexity at most $qm$,
  and we added at most $q^2$ intersections $E\cap_y E'$,
  resulting in a partition of the expected complexity.
\end{proof}

  For all $\val\in\ValSpaceBound$, let $\delay_\val=
  \sup\{\delay\st \val+\delay\in\ValSpaceBound\}\in\Rpos$ denote
  the greatest valid delay from~$\val$ (in fact, $\delay_\val
    = \clockbound-\|\val\|_\infty$).
  Consider the following operations:
\begin{itemize}
  \item If $\reset\subseteq\Clocks$ is a set of clocks and $\loc\in\Locs$,
  let $\Unreset_\reset(\ValIteVec_{\loc}):\ValSpaceBound\to\Rbar$ denote
  the value function such that for all $\val$,
  $$\Unreset_\reset(\ValIteVec_{\loc})(\val)=\ValIteVec_{\loc}(\val[\reset:=0])\,.$$
  \item If $\guard$ is a guard over \Clocks and $\loc\in\Locs$,
  let $\Guard_\guard(\ValIteVec_{\loc}):\ValSpaceBound\to\Rbar$ denote
  the value function such that for all $\val$,
  $$\Guard_\guard(\ValIteVec_{\loc})(\val)=
  \begin{cases}
  \ValIteVec_{\loc}(\val) & \text{if }\val\models\guard \\
  -\infty & \text{if }\val\not\models\guard\land \loc\in\LocsMax\\
  +\infty & \text{if }\val\not\models\guard\land\loc\in\LocsMin\,.
\end{cases}$$
The values $+\infty$ and $-\infty$ in $\Guards_\guard$ %
  ensure that players cannot choose invalid delays:  By the no-deadlocks
  assumption, from every configuration there exists a transition in $\sem{\game}$,
  whose value will win against $+\infty$ or $-\infty$
  in the subsequent equality~\eqref{eq:symvalite}.
  \item If $\edge\in\Edges$ is an edge from $\loc$ to $\loc'$,
  let $\Pretime_\edge(\ValIteVec_{\loc'}):\ValSpaceBound\to\Rbar$ denote
  the value function such that for all $\val$,
  \begin{equation}
  \Pretime_\edge(\ValIteVec_{\loc'})(\val)=
  \begin{cases}
  \sup_{\delay\in[0,\delay_\val)}
  \big[\delay\cdot\weight(\loc)+\weight(\edge)+
  \ValIteVec_{\loc'}(\val+\delay)\big] & \text{if }\loc\in\LocsMax\\
  \inf_{\delay\in[0,\delay_\val)}
  \big[\delay\cdot\weight(\loc)+\weight(\edge)+
  \ValIteVec_{\loc'}(\val+\delay)\big] & \text{if }\loc\in\LocsMin\,.
  \end{cases}\label{eq:pretime}
\end{equation}

\end{itemize}

\noindent  Then, if $\ValIteVec'=\ValIteOpe(\ValIteVec)$, %
  it holds that
  \begin{equation}\label{eq:symvalite}\ValIteVec'_{\loc}=
  \begin{cases}
    \ValIteVec_\loc & \text{if }\loc\in\LocsT \\
    \max_{\edge=(\loc,\guard,\reset,\loc')}
    \Pretime_\edge(\Guard_\guard(\Unreset_\reset(\ValIteVec_{\loc'})))
     &
    \text{if }\loc\in\LocsMax\\
    \min_{\edge=(\loc,\guard,\reset,\loc')}
    \Pretime_\edge(\Guard_\guard(\Unreset_\reset(\ValIteVec_{\loc'}))) &
    \text{if }\loc\in\LocsMin\backslash\LocsT
  \end{cases}\end{equation}
  \noindent where $\edge$
  ranges over the edges in \game that start from $\loc$.
  We have detailed in Lemma~\ref{lm:valitemin} how one can perform the
  $\min$ and $\max$ operations over partitions.
  Let us now focus on the $\Guard_\guard$ and $\Unreset_\reset$ operations.

\begin{lem}\label{lm:valiteguard}
  Let $(P,F)$ be a piecewise affine value function
   of complexity at most $\langle m,\beta \rangle$.
  Let $\guard$ be a non-diagonal guard in \game.
  Then there exists a piecewise affine value function $(P',F')$
  of complexity at most $\langle m',\beta'\rangle$,
  where $m'=m+2n$ and $\beta' = \max(\beta,\clockbound)$,
  such that $\sem{F'}=\Guard_\guard(\sem{F})$.
\end{lem}
\begin{proof}
  The non-diagonal guard $\guard$ can be encoded
  as a cell $I_1\land\dots\land I_{2n}$,
  with one upper and one lower inequality for each clock.
  We define $P'$ from $P$ with
  the set of affine equalities $\mathcal E_{P}\cup\{E(I_1),\dots, E(I_{2n})\}$.
  It follows that each base cell of $P'$ is either entirely included in $\guard$
  or entirely outside of it.  We can thus define $F'$ appropriately,
  such that $\sem{F'}=\Guard_\edge(\sem{F})$.
\end{proof}

\begin{figure}[ht]
  \centering

\begin{tikzpicture}[scale=2]

  \path[fill=black!20!white] (0,0) -- (0,1) -- (0.5,1) -- (1,0) -- (0,0);
  \path[draw,blue] (0,2) -- (1,0)
    (0,1) -- (2,1);
    \path[draw,->](0,0) -> (2.3,0) node[above] {$\clock_1$};
    \path[draw,->](0,0) -> (0,2.3) node[right] {$\clock_2$};
    \path[draw] (0,2) -- (2,2) -- (2,0);

    \node[below] () at (1,0) {$1$};
    \node[below] () at (2,0) {$2$};
    \node[below left] () at (0,0) {$0$};
    \node[left] () at (0,1) {$1$};
    \node[left] () at (0,2) {$2$};
\end{tikzpicture}
\hspace{1cm}
\begin{tikzpicture}[scale=2]

  \path[fill=black!20!white] (0,0) -- (0,2) -- (1,2) -- (1,0) -- (0,0);
  \path[draw,blue] (1,0) -- (1,2);

  \path[draw,->](0,0) -> (2.3,0) node[above] {$\clock_1$};
  \path[draw,->](0,0) -> (0,2.3) node[right] {$\clock_2$};
  \path[draw] (0,2) -- (2,2) -- (2,0);

  \node[below] () at (1,0) {$1$};
  \node[below] () at (2,0) {$2$};
  \node[below left] () at (0,0) {$0$};
  \node[left] () at (0,1) {$1$};
  \node[left] () at (0,2) {$2$};
\end{tikzpicture}

  \caption{On the left, the partition
  of \figurename~\ref{fig:value-ite-part}. On the right,
  the corresponding partition obtained
  by applying Lemma~\ref{lm:valiteunreset} for reset $\reset=\{\clock_2\}$.
  The affine value function of the grey cell on the right is obtained
  from the grey cell on the left, by setting the partial derivative
  of $\clock_2$ to $0$.}
  \label{fig:value-ite-unreset}
\end{figure}

We continue our study with the reset of clocks, that we illustrate
  in \figurename~\ref{fig:value-ite-unreset}.
\begin{lem}\label{lm:valiteunreset}
  Let $(P,F)$ be a piecewise affine value function of complexity
  at most $\langle m,\beta \rangle$.
  Let $\reset$ be a set of clocks.
  Then there exists a piecewise affine value function $(P',F')$
  of complexity at most $\langle m,\beta \rangle$
  such that $\sem{F'}=\Unreset_\reset(\sem{F})$.
\end{lem}
\begin{proof}
  If $E:a_1\clock_1+\cdots + a_n\clock_n+b= 0$ is an affine equality,
  let $\Unreset_\reset(E)$ denote the affine equality
  $a_1'\clock_1+\cdots + a_n'\clock_n+b= 0$, with
  $a_i'=0$ if $\clock_i\in\reset$ and $a_i'=a_i$ otherwise, for $i\in[1,n]$.
  We extend this operator to affine inequalities~$I$ in the same way.
  For each valuation $\val\in\ValSpace$, and $E$ an affine equality (\resp~inequality)
  $\val\models \Unreset_\reset(E)$ if and only if $\val[\reset:=0]\models E$.
  Then, if $c=I_1\land \dots\land I_p$ is a cell, let $\Unreset_\reset(c)$ denote
  the cell $\Unreset_\reset(I_1)\land \dots\land \Unreset_\reset(I_p)$.
  It follows that $\val\in\Unreset_\reset(c)$ if and only if $\val[\reset:=0]\in c$.
  In particular, if $c$ does not intersect the sub-space where every clock in $\reset$
  equals $0$, then $\Unreset_\reset(c)=\emptyset$.

  Similarly, if $c$ is a base cell of $P$ and $F$ maps $c$ to the
  affine value function
  $\clocky=a_1\clock_1+\cdots + a_n\clock_n+b$, let
  $\Unreset_\reset(F(c))$ denote the affine
  function $\clocky=a_1'\clock_1+\cdots + a_n'\clock_n+b$ with
  for $i\in[1,n]$, $a_i'=0$ if $\clock_i\in\reset$ and $a_i'=a_i$
  otherwise.  Then, for every $\val\in\Unreset_\reset(c)$, it holds
  that $\ValIteVec(\val[\reset:=0])=\Unreset_\reset(F(c))(\val)$.

  We construct $P'$ from $P$ by replacing $c_{P}$ by $\Unreset_\reset(c_{P})$, and by replacing $\mathcal E_{P}=\{E_1,E_2,\dots\}$
  by $\{\Unreset_\reset(E_1),\Unreset_\reset(E_2),\dots\}$.
  If $c$ is a base cell of $P$, and $\Unreset_\reset(c)$ is non-empty,
  we let $F'(c)=\Unreset_\reset(F(c))$.
  The result is a partition $P'$ with the desired complexity,
  and a partition function $F'$
  such that $\sem{F'}=\Unreset_\reset(\sem{F})$.
\end{proof}

\subsection{Tubes and diagonals}
  All that is left is the $\Pretime_\edge$ operation.  It is more challenging,
  and requires some extra machinery, that we now detail, related to \emph{diagonal} behaviours
  that naturally arise when dealing with time-elapses.

  An affine inequality (\resp~equality) is \emph{diagonal}
  if the sum of its partial derivatives is null, \ie~$a_1+\cdots + a_n=0$.
  It follows that if $\val$ satisfies a diagonal $I$
  then $\val+\delay\models I$ for all $\delay\in\R$.
  A cell is called a \emph{tube} when all of its inequalities are diagonal.
  When the cell is a sub-cell of some domain $c_P$ in a partition $P$,
  we relax this definition slightly, to allow for non-diagonal borders
  inherited from~$c_P$.

  A \emph{tube partition} is
  a partition $P$ where the set $\mathcal E_{P}$ of affine equalities
  is split between the set $\mathcal E_{P}^{d}$
  of diagonal affine equalities and the set $\mathcal E_{P}^{nd}$
  of non-diagonal ones.
  A tube partition induces a partition of $c_{P}$
  by $\mathcal E_{P}^{d}$ into base cells that are tubes,
  called the \emph{base tubes} of $P$.

  Given two affine equalities $E:a_1\clock_1+\cdots+a_n\clock_n+b= 0$
  and $E': a_1'\clock_1+\cdots+a_n'\clock_n+b'= 0$,
  let $A=a_1+\cdots+a_n$ and $A'= a_1'+\cdots+a'_n$ denote the sums of
  their respective partial derivatives.
  We define their diagonal intersection as
  $$E\cap_d E':(A a_1'-A' a_1)\clock_1+\cdots + (A a_n'-A' a_n)\clock_n+(A b'- A'b)= 0$$
  Observe that $E\cap_d E'$ is a (possibly empty) diagonal equality
  that contains $E\cap E'$.
  Moreover, if $E$ (\resp~$E'$) is diagonal then $E\cap_d E'$
  is equivalent to $E$ (\resp~$E'$).
  Now, given a cell $c=I_1\land\dots\land I_m$ and a set $\mathcal E$
  of affine equalities,
  let $\overline{\mathcal E}$ denote\footnote{Given an affine inequality $I$, $E(I)$
  denotes the associated affine equality, see page~\pageref{defEI}.}
  $\mathcal E \cup\{E(I_1),\dots E(I_m)\}$,
  and let $\Tube(c,\mathcal E)$ denote $\{E\cap_d E'\st E,E'\in \overline{\mathcal E}\}$.
  The pair $(c,\mathcal E)$ is said \emph{atomic} if $c$
  is partitioned by $\Tube(c,\mathcal E)$ in only one cell (equal to $c$).
  Intuitively, $(c,\mathcal E)$ is atomic if
  the affine equalities in $\mathcal E$ and in the borders of $c$ do not intersect
  within the smallest tube that contains $c$.

  A tube partition is \emph{atomic} if for every base tube $c$ in the
  partition of $c_{P}$ by $\mathcal E_{P}^{d}$,
  the pair $(c,\mathcal E_{P}^{nd})$ is atomic.
  Intuitively, this means that in the non-diagonal part of~$P$,
  the equalities that split cells into sub-cells
  do not intersect
  within their base tube.  %
  Tube partitions can be made atomic, by introducing a bounded amount
  of diagonal affine equalities.
\begin{lemC}[{\cite[Lem.~3]{AluBer04}}]\label{lm:valiteatomic}
  Let $(P,F)$ be a piecewise affine value function of complexity
  at most $\langle m,\beta \rangle$.
  Then there exists a piecewise affine value function $(P',F')$
  where $P'$ is an atomic tube partition, of complexity
  at most $\langle m',\beta'\rangle$, where\footnote{$\splitnum(m,n)$ has been
  defined on page~\pageref{splitmn}.}
  $m' = m +m^2\splitnum(m,n)$
  and $\beta' = 2n\beta^2$,
  such that $\sem{F'}=\sem{F}$.
\end{lemC}
\begin{proof}
  We add a set of new diagonal
  equalities, that contain all equalities $E\cap_d E'$
  derived from the base cells included in base tubes of $P$.
  Then, there are at most $m^2$ new diagonals
  for each base cell, and the resulting tube partition must be atomic, as any $E\cap_d E'$ belongs to $\overline{\mathcal E}$, and base cells of $P$ cannot be partitioned by $\overline{\mathcal E}$.
  As $|A|,|A'|\leq n\beta$, we can bound by $2n\beta^2$ the constants in inequalities $E\cap_d E'$.
\end{proof}

\begin{figure}
\centering

\begin{tikzpicture}[scale=2]
  \path[draw,->](0,0) -> (2.3,0) node[above] {$\clock_1$};
  \path[draw,->](0,0) -> (0,2.3) node[right] {$\clock_2$};
  \path[draw] (0,2) -- (2,2) -- (2,0);

  \node[below] () at (1,0) {$1$};
  \node[below] () at (2,0) {$2$};
  \node[below left] () at (0,0) {$0$};
  \node[left] () at (0,1) {$1$};
  \node[left] () at (0,2) {$2$};

  \path[draw,blue] (0,2) -- (1,0)
    (0,1) -- (2,1);
  \path[draw,green!80!black] (0,0.5) -- (1.5,2)
    (1,0) -- (2,1);
\end{tikzpicture}

  \caption{The atomic tube partition derived from the partition
  of \figurename~\ref{fig:value-ite-part} by Lemma~\ref{lm:valiteatomic}.}
  \label{fig:value-ite-atomic}
\end{figure}

  \figurename~\ref{fig:value-ite-atomic} represents the atomic tube
  partition $P'$ associated to the partition~$P$ displayed
  in \figurename~\ref{fig:value-ite-part}.
  The tube partition $P'$ has $2$ diagonal equalities in green,
  resulting in $5$ base tubes.
  The result $P'$ is therefore an atomic tube
  partition of complexity $\langle 6,2\rangle$.\footnote{the size complexity includes the two borders of the domain}

  We can now compute $\Pretime_\edge(\ValIteVec_{\loc'})$,
  with $\ValIteVec_{\loc'}$ a value function encoded as
  a tube partition $(P,F)$, and $\edge\in\Edges$ an edge
  from $\loc$ to $\loc'$, using~\eqref{eq:pretime}.
  We will assume in the following that $\loc$ is
  a location of \MaxPl, but the case of \MinPl is symmetrical.
  Let us fix a valuation $\val\in\ValSpaceBound$.
  For every delay $\delay\in[0,\delay_\val)$, consider
  the term $\ValIteVec_{\loc'}(\val+\delay)$ of~\eqref{eq:pretime}.
  The valuations $\val+\delay$ belong to a diagonal line of \ValSpace,
  and range from $\val$ to $\val+\delay_\val$.
  The segment $[\val,\val+\delay_\val]$ intersects
  a finite set $\mathcal C_\val$ of base cells in $P$.
  For each such cell $c$, we isolate from the segment
  $[\val,\val+\delay_\val]$ two delays:
  \[\delay_1 = \inf \{d\in [0,\delay_\val]\mid \val+\delay \in c\}
    \quad\text{ and }\quad \delay_2 = \sup \{d\in [0,\delay_\val]\mid \val+\delay \in c\}\]
  As $\delay\mapsto\sem{F}(\val+\delay)$ is affine over $c$,
  so is $\delay\mapsto
  \delay\cdot\weight(\loc)+\weight(\edge)+\sem{F}(\val+\delay)$,
  thus the supremum of $\delay\cdot\weight(\loc)+\weight(\edge)+\sem{F}(\val+\delay)$
  for $\val+\delay\in c$ must either be reached at
  (or arbitrarily close to) $\delay_1$,
  or at (or arbitrarily close to) $\delay_2$.
  Note that $\val+\delay_2$ must belong to a non-diagonal border of $c$,
  while $\val+\delay_1$
  either belongs to a non-diagonal border of $c$ or equals $\val$
  (whenever $\val\in c$).
  Thus, the optimal value of $\delay$ for evaluating the supremum
  must correspond to either delay $0$ or
  to a delay leading $\val$ to a non-diagonal border (this is also proven
  in \cite{AluBer04}).

  If $B$ is a non-diagonal border of $c$, and $\val$ is a valuation of $\ValSpace$,
  there exists a unique $\delay\in\R$ such that $\val+\delay\in\sem{B}$.
  In fact, if $B$ is described by $ a_1\clock_1+\cdots+a_n\clock_n +b=0$
  and $A=a_1+\cdots +a_n$, then
  $$\delay=-\frac{a_1\cdot\val(\clock_1)+\cdots+a_n\cdot\val(\clock_n) +b}{A}$$
  We denote this delay $\delay_{\val,B}$. Observe that it must belong to $[0,\delay_\val]$
  as $\sem{B}$ is reachable from~$\val$ by time-elapse.\footnote{
    Note that it can be equal to $\delay_\val$, as $\clock-\clockbound=0$
    is a border of the cell $\ValSpaceBound$.
  }

  If $c$ is a cell of $\mathcal C_\val$, let $\mathcal B_\val(c)$
  denote the non-diagonal borders of $c$ reachable from $\val$ by time-elapse.
  The supremum $\Pretime_\edge(\sem{F})(\val)$ is then equal to
  $$\max\Big(\sem{F}(\val),\; \max_{c\in\mathcal C_\val}\max_{B\in\mathcal B_\val(c)}
  [\delay_{\val,B}\cdot\weight(\loc)+\weight(\edge)+ \sem{F(c)}(\val+\delay_{\val,B})]\Big)$$
  where $\sem{F}(\val)$ corresponds to the delay $0$,
  and $\sem{F(c)}(\val+\delay_{\val,B})$  corresponds to a jump in cell $c$ arbitrarily close
  to $B$.\footnote{In particular, if $\val+\delay_{\val,B}\not\in c$ then $F(c)$
  evaluated on $\val+\delay_{\val,B}$ may not equal $\sem{F}(\val+\delay_{\val,B})$.}

  If the tube partition $(P,F)$ is atomic, it follows that
  every other valuation in the same cell $\val'\in[\val]_P$ can reach
  the same set of borders by time elapse, \ie~$\mathcal C_\val=\mathcal C_{\val'}$
  and $\mathcal B_\val(c)=\mathcal B_{\val'}(c)$ for all $c\in\mathcal C_\val$.
  As a result, we rename those sets $\mathcal C_{c'}$ and $\mathcal B_{c'}(c)$
  if $c'=[\val]_P$ and $c\in \mathcal C_{c'}$.
  We introduce an operator $\Pretime_{\edge,c,B}$,
  indexed by an edge, a cell and a non-diagonal border of the cell,
  that maps a partition function $F$ to the value function
  $\val\mapsto
  \delay_{\val,B}\cdot\weight(\loc)+\weight(\edge)+ \sem{F(c)}(\val+\delay_{\val,B})$.
  As a consequence,
  for each base cell $c_b$ of $P$, $\Pretime_\edge(\sem{F})$ restricted
  to domain $c_b$ equals
  \begin{equation}
  \max(\sem{F},\max_{c\in\mathcal C_{c_b}}\max_{B\in\mathcal B_{c_b}(c)}
  \Pretime_{\edge,c,B}(F))
  \label{eq:pre_edge}
  \end{equation}

\begin{lem}\label{lm:valitepreborder}
  Let $(P,F)$ be a piecewise affine value function of complexity
  at most $\langle m,\beta \rangle$, where $P$ is an atomic
  tube partition.
  Let $c_b$ be a base cell of $P$, $\edge$ be an edge from $\loc$ to $\loc'$,
  $c$ be a cell in $\mathcal C_{c_b}$ and $B$ be a border in $\mathcal B_{c_b}(c)$.
  Then there exists an affine value function~$f$
   with constants bounded by
   $\beta' = 4n\beta^2\max(\wmax^\Locs,\wmax^\Edges)$,
  such that $f=\Pretime_{\edge,c,B}(F)$ on $c_b$.
\end{lem}
\begin{proof}
    Let $\val$ be a valuation in $c_b$.  As $B\in\mathcal B_{c_b}(c)$ it holds that
  $\delay_{\val,B}\in[0,\delay_\val]$, such that $$\Pretime_{\edge,c,B}(F)(\val)=
  \delay_{\val,B}\cdot\weight(\loc)+\weight(\edge)+ \sem{F(c)}(\val+\delay_{\val,B})\,.$$
  Let $a_\clocky \clocky=a_1\clock_1+\cdots+a_n\clock_n +b$ be the equation of $F(c)$,
  and let $A=a_1+\cdots+a_n$.
  Let $a_1'\clock_1+\cdots+a_n'\clock_n +b'=0$ be the equation of $B$,
  with $A'= a_1'+\cdots+a_n'\neq 0$ the sum of its partial derivatives.
  We obtain the following equalities:
  \begin{align*}
  \delay_{\val,B}&=-\frac{a_1'\cdot\val(\clock_1)+\cdots +a_n'\cdot\val(\clock_n) +b'}{A'}\\
  \sem{F(c)}(\val+\delay_{\val,B}) &=
  \frac{
  a_1\cdot\val(\clock_1)+\cdots +a_n\cdot\val(\clock_n) + A
                                      \delay_{\val,B}+b}{a_\clocky}\\
    \sem{F(c)}(\val+\delay_{\val,B})\cdot A' a_\clocky &=
                                                          (A' a_1-A a_1')\cdot \val(\clock_1)+\cdots+(A' a_n-A a_n')\cdot \val(\clock_n)  +(A' b-A b')\\
    \Pretime_{\edge,c,B}(F)(\val)\cdot A' a_\clocky&=
                                                     \sum_{i=1}^n (A' a_i-A a_i'-a_\clocky a_i'\cdot\weight(\loc))
                                                     \cdot \val(\clock_i) \\
                 &\quad
                   +(A' b-A b'-a_\clocky  b'\cdot\weight(\loc)
                   +A' a_\clocky\cdot \weight(\edge))
  \end{align*}
  Thus $\Pretime_{\edge,c,B}(F)(\val)$ is described by an equation
  $a_\clocky^f \clocky=a_1^f\clock_1+\cdots+a_n^f\clock_n +b^f$.
  The announced bound on constants holds for $n\geq 1$.
\end{proof}

  Let us now bring everything together for %
  the value iteration operator $\ValIteOpe$.
\begin{prop}\label{prop:valiteope}
  Let $\ValIteVec = (P_\loc,F_\loc)_{\loc \in \Locs}$
  be a piecewise affine value function, where
  every $(P_\loc,F_\loc)$ has complexity at most $\langle m,\beta \rangle$,
  Let $q=|\Edges|$ be the number of edges in $\game$.  Then there
  exists a piecewise affine value function
  $\ValIteVec' = (P'_\loc,F'_\loc)_{\loc \in \Locs}$ such that
  $\ValIteVec'=\ValIteOpe(\ValIteVec)$.  In addition, every
  $(P'_\loc,F'_\loc)$ has complexity at most
  $\langle m',\beta' \rangle$
  where $m' = 36q^2(4m+8n+6)^{3n(n+2)}$ and
  $\beta'$
  is polynomial in $n$, $\max(\wmax^\Locs,\wmax^\Edges)$ and $\beta$.
\end{prop}
\begin{proof}
  Fix a location $\loc\in\LocsMin$. The case $\LocsMax$ is symmetrical
  and will not be detailed.
  If $\loc\in\LocsT$, let $(P'_\loc,F'_\loc)=(P_\loc,F_\loc)$ of complexity at most $\langle m, \beta\rangle$.
  Otherwise, $$\ValIteVec'_{\loc}=
  \min_{\edge=(\loc,\guard,\reset,\loc')}
    \Pretime_\edge(\Guard_\guard(\Unreset_\reset(\ValIteVec_{\loc'})))\,.$$

  For every edge $\edge=(\loc,\guard,\reset,\loc')$,
  we construct %
  a piecewise affine value function $(P_{\edge},F_{\edge})$ encoding
  $\Pre_\edge(\Guard_\guard(\Unreset_\reset(\ValIteVec_{\loc'})))$.

    We construct from $(P_\loc,F_\loc)$, by Lemmas~\ref{lm:valiteguard}, \ref{lm:valiteunreset}
    and \ref{lm:valiteatomic},
    an atomic tube partition $(P_1,F_1)$ encoding
    $\Guard_\guard(\Unreset_\reset(\ValIteVec_{\loc'}))$.
    Following complexity results given in these lemmas, we end up
    with a complexity at most $\langle m_1,\beta_1 \rangle$, such that:
    $$
    \left \{\begin{array}{lcl}
    m_1 &= & (m+2n)(1+(m+2n)\splitnum(m+2n,n)) \\
    \beta_1 &= & 2n\max(\beta,\clockbound)^2\\
    \end{array}
    \right .
    $$
    As explained before (see equation~\ref{eq:pre_edge}),
    $\Pretime_\edge(\sem{F_1})$ is obtained
    using computations on
    base cells of $P_1$ followed by a minimum with $\sem{F_1}$ .
    Given a base cell $c_b$ of $P_1$, a cell $c$ in $\mathcal C_{c_b}$,
    and a border $B$ in $\mathcal B_{c_b}(c)$, we can apply
    Lemma~\ref{lm:valitepreborder} to deduce the existence of an affine
    value function $f_{e,c,B}$ such that
    $f_{e,c,B} = \Pretime_{\edge,c,B}(F_1)$ on $c_b$.
    We then have, for $\nu$ in base cell $c_b$:
    $$
    \begin{array}{lll}
    \Pretime_\edge(\sem{F_1})(\nu) & = & \min(\sem{F_1}(\nu),\min_{c\in\mathcal C_{c_b}}\min_{B\in\mathcal B_{c_b}(c)}
    \Pretime_{\edge,c,B}(F_1)(\nu)) \\
     & = & \min(\sem{F_1}(\nu),\min_{c\in\mathcal C_{c_b}}\min_{B\in\mathcal B_{c_b}(c)}
    f_{e,c,B}(\nu)) \\
      \end{array}
    $$

    More precisely, for each base cell $c_b$ of $P_1$, we can compute the affine value function
    $f_{e,c,B}$ over $c_b$, with Lemma~\ref{lm:valitepreborder}, with coefficients bounded
    by $\beta_2 = 4n\beta_1^2\max(\wmax^\Locs,\wmax^\Edges)$.
    There are at most $\splitnum(m,n)$ cells in $\mathcal C_{c_b}$
    and at most $2$ borders in $\mathcal B_{c_b}(c)$, thus the above formula for
    $\Pretime_\edge(\sem{F_1})$ involves a maximum amongst
    $q_1 = 2\splitnum(m_1,n) +1$ elements. Applying Lemma~\ref{lm:valitemin}, we can thus
    compute a piecewise affine value function $(P_{c_b},F_{c_b})$ representing $\Pretime_\edge(\sem{F_1})$
    over $c_b$. Using complexity results given in Lemma~\ref{lm:valitemin},
    we can deduce that constant complexity is bounded by
    $\beta_3 = 2\beta_2^2$,
    and that size complexity is bounded by
    $m_3 = q_1m_1+q_1^2$.

    Then, we can use all $(P_{c_b},F_{c_b})$ to define a unique
    $(P_{\edge},F_{\edge})$, by merging all $P_{c_b}$ using the operator $\oplus$ and by setting
    $\sem{F_{\edge}}(\val)=\sem{F_{c_b}}(\val)$ with $c_b=[\val]_{P_1}$.
    The resulting size complexity $m_e$ thus satisfies
    $m_e \le m_3 \splitnum(m_1,n)$, while the constant complexity
    is left unchanged with $\beta_e=\beta_3$.

    Finally, we can apply Lemma~\ref{lm:valitemin} %
    to compute $\min_{\edge=(\loc,\guard,\reset,\loc')} \sem{F_\edge}$.
    The result is a partition function $(P'_\loc,F'_\loc)$ encoding $\ValIteVec'_{\loc}$.
    Regarding complexity, Lemma~\ref{lm:valitemin}
    ensure that the value function $\ValIteVec'_{\loc}$
    has complexity at most $\langle m',\beta' \rangle$, with:
    \[
      m' = qm_e+q^2 \quad \text{ and } \quad \beta' =
      2\beta_e^2\]

    \noindent Putting everything together, we obtain :
    \begin{align*}
      \splitnum(m,n) & \leq (2m+2)^{n} && (\text{Lemma}~\ref{lm:splitnum}) \\
      m_1 & \leq (2m+4n+2)^{n+2} && (1\leq \splitnum(m+2n,n)) \\
      \splitnum(m_1,n) & \leq (4m+8n+6)^{n(n+2)} && ((2x^{n+2}+2)^{n}\leq (2x+2)^{n(n+2)})\\
      q_1 & \leq 3 (4m+8n+6)^{n(n+2)} \\
      m_3 & \leq 18 (4m+8n+6)^{2n(n+2)} && (q_1m_1+q_1^2\leq 2q_1^2)\\
      m_e & \leq 18 (4m+8n+6)^{3n(n+2)}\\
      m' & \leq 36q^2(4m+8n+6)^{3n(n+2)} && (qm_e+q^2\leq 2q^2m_e)
    \end{align*}
    Regarding constant complexity, we have:
    \begin{align*}
      \beta'
      & = 2(\beta_e)^2 = 2(\beta_3)^2 = 2^3\beta_2^4 \\
      &= 2^3 (4n\beta_1^2\max(\wmax^\Locs,\wmax^\Edges))^4 &&
        (\beta_2 = 4n\beta_1^2 \max(\wmax^\Locs,\wmax^\Edges))\\
      &= 2^{11}\ \max(\wmax^\Locs,\wmax^\Edges)^4\ n^4\ (2n\max(\beta,\clockbound)^2)^8
        && (\beta_1 = 2n\max(\beta,\clockbound)^2)\\
      &= 2^{19}\ \max(\wmax^\Locs,\wmax^\Edges)^4\ n^{12}\ \max(\beta,\clockbound)^{16} \tag*{\qedhere}
  \end{align*}

\end{proof}

Now that we know how to perform one step of computation of
$\ValIteOpe$, we can estimate the complexity of iterated computations
of this operator.

  \begin{cor}\label{cor:complexity}
    Let $\ValIteVec = (P_\loc,F_\loc)_{\loc \in \Locs}$
  be a piecewise affine value function, where
  every $(P_\loc,F_\loc)$ has complexity at most $\langle m,\beta \rangle$.
  Let $q=|\Edges|$ be the number of edges in $\game$.
  For every $i\ge 1$, we can compute a piecewise affine value function
  $\ValIteVec^{(i)} = (P^{(i)}_\loc,F^{(i)}_\loc)_{\loc \in \Locs}$
  such that $\ValIteVec^{(i)}=\ValIteOpe^i(\ValIteVec)$.  In addition,
  every $(P^{(i)}_\loc,F^{(i)}_\loc)$ has complexity at most
  $\langle m^{(i)},\beta^{(i)} \rangle$ where
  $m^{(i)} = \max(m,8n+6,10q)^{(6n(n+2))^i}$
  and $\beta^{(i)}$ is polynomial in $n$,
  $\max(\wmax^\Locs,\wmax^\Edges)$ and $\beta$.
  \end{cor}
  \begin{proof}
  The result can be obtained by induction on $i$.
  The initialisation is $m\leq m^{(0)}$.
  For the induction step, by Proposition~\ref{prop:valiteope}
  we must show $36q^2(4m^{(i)}+8n+6)^{3n(n+2)}\leq m^{(i+1)}$.
  Observe that $36q^2\leq (2q)^{3n(n+2)}$, and $8n+6,10q\leq m^{(i)}$, so that
  $36q^2(4m^{(i)}+8n+6)^{3n(n+2)}
  \leq (10qm^{(i)})^{3n(n+2)}
  \leq (m^{(i)})^{6n(n+2)}
  \leq m^{(i+1)}$.
  This bound is crude but sufficient for our results. We note that even finer bounds on the iteration of $m\mapsto36q^2(4m+8n+6)^{3n(n+2)}$ lead to doubly-exponential complexities.
  \end{proof}

  Thanks to
  Lemma~\ref{lem:memory-partition}, it will follow that
  $\ValIteVec^i$, obtained by applying $i$ times the operator
  $\ValIteOpe$ over $\ValIteVec^0$ of bounded complexity, can be
  stored using a space depending polynomially on $m_0$ and $\beta_0$,
  exponentially on $n$, and doubly-exponentially on $i$.
  Moreover, the operations that we have described can be performed with a time complexity at most linear in the size of the output.
  This provides a doubly-exponential upper bound on the complexity
  (with respect to~$i$) of computing $\ValIteVec^i$.  This concludes
  the proof of Theorem~\ref{thm:valitebounded}.

\subsection{Exponential vs doubly-exponential}\label{sec:valitediscuss}

  Let us now discuss the result of \cite{AluBer04}, where an exponential
  upper bound on the bounded value problem is obtained with a non-symbolic algorithm
  on a slightly different setting (with non-negative weights only,
  without final weights, in a concurrent setting). %

  Whereas the concurrent setting they use generalizes ours, the sign
  of weights has seemingly no impact in the proofs, and the symbolic
  version requires minor changes related to the continuity of value
  functions and to the way guards are handled. These changes should
  not affect the complexity significantly.  The main difference with their
  work is that their partition object is nested, it forms a tree
  structure where cells are partitionned into sub-cells as needed.
  Our techniques can be extended to nested partitions in the same way,
  and this has been detailed
  in~\cite[Chapter~10]{busattogaston:tel-02436831}.  However, even
  with nested partitions the complexity of our techniques is still
  doubly-exponential, and the reason for this exponential gap is not
  apparent.

  As a tentative answer, we make the following observation.
  If the game has no resets, \ie~$\reset=\emptyset$ on all edges,
  the complexity of our approach becomes exponential.\footnote{
  This requires nested partitions and a more involved complexity analysis
  detailed in~\cite[Chapter~10]{busattogaston:tel-02436831}.
  In this case, the $\Unreset_\reset$ steps can be skipped,
  and in Proposition~\ref{prop:valiteope} we end up with a \emph{linear}
  growth of the complexity with respect to $m$, instead of \emph{polynomial}.
  }
  In \cite{AluBer04}, the way one should deal with resets is not detailed,
  it is therefore left open whether we could obtain an exponential bound or
  whether their solution is in fact doubly-exponential.

  As an additional point of interest,
  it is shown in~\cite[Chapter~10.1.4]{busattogaston:tel-02436831} that
  there exists a tube partition $P$, of complexity
  $m$, such that if $P'$ of complexity $m'$
  is obtained after applying $\ValIteOpe$, then $m'=\Theta(m^{n-1})$.
  This is the root of the issue, as we would need $m'=\mathcal O(m)$ in order to obtain
  an exponential bound when nesting $\ValIteOpe$.

\subsection{Bounding partial derivatives}\label{app:value-ite}

  In the previous analysis, we explained that constants
  (partial derivatives and additive constants) grow polynomially at each
  elementary step, which is enough for an exponential upper bound
  (double-exponential growth of their value, stored in binary).
  This rough analysis will not be fine enough for some of our results.
  In particular, the approximation results of Section~\ref{sec:computing}
  will be sensitive to the partial derivatives in a linear (and not logarithmic) way.
  In this section, we study the growth of these partial derivatives more closely.
  This time, our focus will not be on the space required to store affine
  equations, but rather on mathematical properties of the value functions,
  namely their Lipschitz-continuity, closely related to bounds on partial derivatives.
  As a result, we revert to denoting affine equations
  as terms $\clocky=a_1\clock_1+\cdots+a_n\clock_n+b$ with rational constants instead
  of using integers with a separately stored denominator $a_\clocky$.

\begin{defi}
  A value function $\ValIteVec:\Locs\times\Rpos^\Clocks\to\Rbar$ is
  said to be \emph{\lipconst-Lipschitz-continuous on regions}, with $\lipconst\in \Rpos$ when,
  for all regions $r$ and all $\loc\in\Locs$,
  $|\ValIteVec(\loc,\val)-\ValIteVec(\loc,\val')|\leq \lipconst
  \|\val-\val'\|_\infty$ for all valuations $\val,\val'\in r$, where
  $\|\val\|_\infty=\max_{\clock\in \Clocks} |\val(\clock)|$ is the
  $\infty$-norm of vector $\val\in\R^\Clocks$.
  The function $\ValIteVec$ is said to be
  Lipschitz-continuous on regions if it is
  \lipconst-Lipschitz-continuous on regions, for some
  $\lipconst\in\Rpos$.
\end{defi}

  Since final weight functions \weightT are piecewise affine and continuous on regions,
  they are Lipschitz-continuous on regions, so is $\ValIteVec^0$.
  We will maintain as an invariant that $\ValIteVec^i$ is Lipschitz-continuous
  on regions:

\begin{prop}\label{prop:tree-lipschitz}
  If every final weight in a WTG \game is \lipconst-Lipschitz-continuous on regions
  (and piecewise affine), and $i>0$,
  then $\Val^i_\game$ is %
  $\lipconst'$-Lipschitz-continuous on regions,
  with $\lipconst'$ polynomial in $\lipconst$, $\wmax^\Locs$ and $n$,
  and exponential in $i$.
  More precisely, $\lipconst'$ is bounded by
  \[\frac{\max(n-1,2)^i-1}{\max(n-2,1)}\wmax^\Locs+(n-1)^i \max(\lipconst,\wmax^\Locs)\,.\]
\end{prop}

Note that for a piecewise affine function with finitely many pieces,
  being \lipconst-Lipschitz-continuous on regions is equivalent to being
  continuous on regions and having all partial derivatives bounded
  by \lipconst in absolute value.
  The rest of this section is dedicated to proving
  Proposition~\ref{prop:tree-lipschitz}, as a corollary of the following result.

\begin{lem}\label{lm:value-ite-partial-derivatives}
  If for all $\loc\in\Locs$, $\ValIteVec_\loc$ is piecewise affine with
  finitely many pieces %
  that have all their partial derivatives bounded
  by \lipconst in absolute value,
  then for all $\loc\in\Locs$,
  $\ValIteOpe(\ValIteVec)_\loc$ is continuous on regions and piecewise affine %
  with partial derivatives bounded
  by $\max(\lipconst,|\weight(\loc)| + (n-1) \lipconst)$ in absolute value.
\end{lem}
\begin{proof}
  We will show that for every region $\reg$, $\ValIteOpe(\ValIteVec)$ restricted
  to $\reg$ has those properties. Note that they are transmitted over finite $\min$
  and $\max$ operations.
  The continuity on regions is easy to prove because it is stable by
  infimum and supremum.
  There exists a partition function $(P_\loc,F_\loc)$ for each $\loc\in\Locs$
  that represents $\ValIteVec$.
  As explained before, a crucial property %
  is that, for a given valuation $\val$, the
  delays $\delay$ that need to be considered in the infimum and supremum
  of $\ValIteOpe(\ValIteVec)_{\loc}(\val)$ correspond to the
  intersection points of the diagonal half line containing the time
  successors of $\val$ and borders of cells (if $\val^b$ is such a
  valuation, $\delay=\|\val^b-\val\|_\infty$ is the associated delay).
  In particular, there is a finite number of such borders, and the final
  $\ValIteOpe(\ValIteVec)_\loc$ function can be written as a finite nesting of
  finite $\min$ and $\max$ operations over affine terms, %
  each corresponding to a choice of delay and an edge to take.
  Formally, there are several cases to consider to define those terms,
  depending on delay and edge choices.  For each available edge $\edge$,
  those terms can either be:
  \begin{itemize}
    \item If a delay $0$ is taken and all clocks in $\reset\subseteq\Clocks$ are reset by $\edge$,
    then  $$\weightC((\loc,\val)\xrightarrow{0}(\loc,\val)\xrightarrow{\edge}(\loc',\val[\reset:=0]))=
    \weightC(\edge) + \ValIteVec_{\loc'}(\val[\reset:=0])$$
    \item If a delay $\delay>0$ (leading to valuation $\val^b$ on border $B$) is taken
    and all clocks in $\reset\subseteq\Clocks$ are reset by $\edge$,
    then $$\weightC((\loc,\val)\xrightarrow{\delay}(\loc,\val^b)\xrightarrow{\edge}(\loc',\val^b[\reset:=0]))=
    \weightC(\loc)\cdot \delay + \weightC(\edge) + \ValIteVec_{\loc'}(\val^b[\reset:=0])$$
  \end{itemize}

\begin{figure}[ht]
\centering

\begin{tikzpicture}[scale=0.75,cap=round,>=latex]
  \draw[thick,->] (-1.1cm,-1.1cm) -- (6cm,-1.1cm) node[right,fill=white] {$\clock_1$};
  \draw[thick,->] (-1.1cm,-1.1cm) -- (-1.1cm,6cm) node[above,fill=white] {$\clock_2$};

  \draw[dotted,-] (0cm,-1cm) -- (6cm,5cm);

  \filldraw[black] (2cm,1cm) circle(1pt);
  \draw (2cm,1cm) node[below] {$\val$};

  \filldraw[black] (2cm,1.5cm) circle(1pt);
  \draw[->] (2cm,1cm) -- (2cm,1.5cm);

  \filldraw[black] (2.5cm,1cm) circle(1pt);
  \draw (2.5cm,1cm) node[right] {$\val'$};%
  \draw[->] (2cm,1cm) -- (2.5cm,1cm);

  \draw[dotted,-] (0.5cm,-1cm) -- (6.5cm,5cm);

  \draw[-] (0cm,6cm) -- (6cm,3cm);
  \draw (2cm,5cm) node[above] {$B$};

  \filldraw[black] (4.66cm,3.66cm) circle(1pt);
  \draw (4.66cm,3.66cm) node[above] {$\val^b$};

  \filldraw[black] (5cm,3.5cm) circle(1pt);
  \draw (5cm,3.5cm) node[below] {${\val'}^b$};

  \draw[dashed,-] (-1cm,1cm) -- (4cm,6cm);
  \draw[dashed,-] (2cm,-1cm) -- (6cm,3cm);
  \draw[dashed,-] (2cm,4cm) -- (4cm,1cm);
  \draw[dashed,-] (-0.25cm,1.75cm) -- (2.25cm,-0.75cm);
  \draw (2cm,3.9cm) node[below] {$c$};
\end{tikzpicture}

   \caption{A cell $c$ as described in the proof of Lemma~\ref{lm:value-ite-partial-derivatives}.
   Dashed lines are borders of $c$, dotted lines are proof constructions.
   }\label{fig:value-ite-lipschitz}
\end{figure}

  In the first case, the resulting partial derivatives are $0$ for clocks in $\reset$,
  and the same as the partial derivatives in $\ValIteVec_{\loc'}$ for all other clocks,
  which allows us to conclude that they are bounded by \lipconst.
  We now consider the second case.
  We argue that it can be decomposed as a delay followed by an edge
  of the first case, meaning that we can assume $\reset=\emptyset$ without loss of generality.

  There are again two cases: the border $B$ being inside a region or on the frontier of a region.

  If the border is not the frontier of a region, it is the
  intersection points of two affine pieces of $\ValIteVec_{\loc'}$
  whose equations (in the space $\R^{n+1}$ whose $n$ first coordinates
  are the clocks $(\clock_1,\ldots,\clock_n)$ and the last coordinate
  $\clocky$ corresponds to the value
  $\ValIteVec_{\loc'}(\clock_1,\ldots,\clock_n)$) can be written
  $\clocky=a_1\clock_1+\cdots + a_n\clock_n+b$ (before the border) and
  $\clocky=a'_1\clock_1+\cdots+a'_n\clock_n+b'$ (after the border).
  Therefore, valuations at the borders all fulfill the equation
  \begin{equation}
    (a'_1-a_1)\clock_1+\cdots+(a'_n-a_n)\clock_n + b'-b=0\label{eq:1}
  \end{equation}
  We let $A=(a'_1-a_1)+\cdots+(a'_n-a_n)$. Consider that $\loc$ is a
  location of \MinPl (the very same reasoning applies to the case of a
  location of \MaxPl).  Since $\ValIteOpe$ computes an infimum, we
  know that the function mapping the delay $\delay$ to the weight obtained
  from reaching $\val+\delay$ is decreasing before the
  border and increasing after.  These functions are locally affine
  which implies that their slopes satisfy:
  \begin{equation}
    \weight(\loc)+a_1+\cdots+a_n\leq 0\quad \text{ and } \quad
    \weight(\loc)+a'_1+\cdots+a'_n\geq 0\label{eq:2}
  \end{equation}
  We deduce from these two inequalities that $A\geq 0$.  The case
  where $A=0$ would correspond to the case where the border contains a
  diagonal line, which is forbidden, and thus $A>0$.  Consider now a
  valuation $\val$ and another valuation~$\val'$ of coordinates
  $(\val(\clock_1),\ldots,\val(\clock_{k-1}),\val(\clock_k)+\lambda,
  \val(\clock_{k+1}),\ldots,\val(\clock_n))$ with $\lambda$ small
  enough.  The delays $\delay$ and $\delay'$ needed to arrive to the
  border starting from these two valuations are such that
  $\val+\delay$ and $\val'+\delay'$ both satisfy \eqref{eq:1}.  We can
  then deduce
  $$\delay'-\delay=\lambda\frac{a_k-a'_k}{A}$$  It is now possible to compute the
  partial derivative of $\ValIteOpe(\ValIteVec)_\loc$ in the $k$-th
  coordinate using the limit when $\lambda$ tends to $0$ of the
  quotient
  \begin{equation}\label{eq:partderiv}\frac{\ValIteOpe(\ValIteVec)_{\loc,\val'}-
      \ValIteOpe(\ValIteVec)_{\loc,\val}}{\lambda} = \frac{\weight(\loc)(\delay'-\delay)+
      \ValIteVec_{\loc',\val'+\delay'}-\ValIteVec_{\loc',\val+\delay}}{\lambda}
  \end{equation}
  We may compute
  it by using the equations of the affine pieces before or after the
  border.  We thus obtain
  \begin{align*}
    \frac{\ValIteOpe(\ValIteVec)_{\loc,\val'}-\ValIteOpe(\ValIteVec)_{\loc,\val}}{\lambda}
    &=
      \frac{a_k-a'_k} A \big(\weight(\loc)+a_1+\cdots+a_n\big) + a_k\\
    \frac{\ValIteOpe(\ValIteVec)_{\loc,\val'}-\ValIteOpe(\ValIteVec)_{\loc,\val}}{\lambda}
    &=
      \frac{a_k-a'_k} A \big(\weight(\loc)+ a'_1+\cdots+a'_n\big) + a'_k
  \end{align*}
  In the case where $a_k\geq a'_k$, the first equation, with
  \eqref{eq:2}, allows us to obtain that the quotient (and thus the
  partial derivative) is at most $a_k$.  The second equation, with
  \eqref{eq:2}, allows us to obtain that the partial derivative is at
  least $a'_k$. In case $a'_k\geq a_k$, we obtain similarly that the
  partial derivative is at most $a'_k$ and at least $a_k$.  Since
  $a_k$ and $a'_k$ are bounded in absolute value by \lipconst, so is
  the partial derivative.

  We now come back to the case where the border is on the frontier of
  a region.  Then, it is a segment of a line of equation $\clock_k=c$
  for some $k$ and $c$.  Moreover, $\ValIteVec_{\loc'}$ contains at
  most three values for points of $B$: the limit coming from before
  the border, the value at the border, and the limit coming from after
  the border.  The computation of $\ValIteOpe(\ValIteVec)$ considers
  values obtained from all three and takes the minimum (or maximum).

  Now, let $\clocky=a_1\clock_1+\cdots+a_n\clock_n+b$ be the equation
  defining the affine piece of $\ValIteVec_{\loc'}$ before the border
  (\resp~at the border, after the border).  Consider a valuation
  $\val$ and another valuation $\val'$ of coordinates
  $(\val(\clock_1),\ldots,\val(\clock_{j-1}),\val(\clock_j)+\lambda,
  \val(\clock_{j+1}),\ldots,\val(\clock_n))$ with $\lambda$ small
  enough.  The delays $\delay$ and $\delay'$ needed to arrive to the
  border starting from these two valuations are such that
  $\val+\delay$ and $\val'+\delay'$ both satisfy $\clock_k=c$.  We can
  then deduce that $\delay'-\delay=0$ if $j\neq k$ and
  $\delay'-\delay=-\lambda$ if $j=k$.  It is now possible to compute
  the partial derivative of $\ValIteOpe(\ValIteVec)_\loc$ in the
  $j$-th coordinate using \eqref{eq:partderiv} again.
  We may compute it by using the equations of the affine piece before
  the border (\resp~at the border, after the border).
  Then, \begin{align*}\ValIteVec_{\loc',\val+\delay}&=a_1(\clock_1+\delay)+\cdots+a_n(\clock_n+\delay)
                                                      +b=\sum_{i=1,i\neq
                                                      k}^n
                                                      a_i(\clock_i+\delay)
                                                      + a_k c+b\\
          \ValIteVec_{\loc',\val'+\delay'}&=\sum_{i=1,i\neq k}^n
                                            a_i(\clock_i+\delay')
                                            + a_k c +b\end{align*}
                                            We thus obtain
  \[
    \frac{\ValIteOpe(\ValIteVec)_{\loc,\val'}-\ValIteOpe(\ValIteVec)_{\loc,\val}}{\lambda}
    =\begin{cases}
      a_j&\text{ if }j\neq k\\
      -\weight(\loc)-\sum_{i=1,i\neq k}^n a_i&\text{ otherwise}
      \end{cases}
  \]

  Then, the partial derivatives are bounded, in absolute value, by $|\weight(\loc)|+(n-1) \lipconst$.
\end{proof}

As a corollary, we obtain Proposition~\ref{prop:tree-lipschitz} by a simple induction on $i$.

\section{Computing values}\label{sec:computing}

In this section we conclude the proofs of Theorems~\ref{thm:div_wtg}
and~\ref{thm:almost-div},
with a triply-exponential upper bound on computing (\resp~approximating) values
in divergent (\resp~almost-divergent) weighted timed games.

These upper bounds are obtained by computations performed on the
semi-unfolding $\tgame$ described in
Section~\ref{sec:unfolding}. Whereas the computation is direct for
divergent \WTG{s}, we then extend it to almost-divergent
\WTG{s} by first explaining how to compute value approximations in
kernels.

\subsection{Divergent \WTG{s}}

We first explain how to compute the values of a divergent \WTG
$\game$, thus proving Theorem~\ref{thm:div_wtg}. By definition, in
such a divergent \WTG, there are no 0-cycles, and thus the kernel
\Kernel is empty. In this case, the semi-unfolding $\tgame$ is a tree
of depth $i$ exponential in \game, and thus of doubly-exponential
size. Proposition~\ref{prop:semi-unfolding} ensures that the values of
$\tgame$ coincide with the values of $\game$.
By Theorem~\ref{thm:valitebounded}, we can compute the (exact) values
in $\tgame$ in time doubly-exponential in $i$ and exponential in the
size of \tgame.  We thus obtain a triply-exponential algorithm
computing the value of a divergent \WTG. Equivalently, this shows that
the value problem is in $3$-\EXP\ for divergent \WTG{s}.

\subsection{Approximation of kernels}\label{sec:approx-kernels}
In order to generalise our computations from divergent to almost-divergent \WTG{s}, we start
by approximating a kernel $\Kernel$ of a \WTG \game by extending the
region-based approximation schema of \cite{BJM15}. More precisely, we
thus consider here a \WTG \game composed of a kernel (a subgraph of a
region game containing only $0$-cycles, as defined in
Section~\ref{sec:kernels-in-wtg}) and some target locations equipped
with final weights. In the setting of~\cite{BJM15}, not only cycles
but all finite plays in kernels have weight 0, allowing for a
reduction to a finite game. In our setting, we have to approximate the
timed dynamics of runs, and therefore resort to the corner-point
abstraction (as shown in the right part of~\figurename~\ref{fig:schema}).

Our goal is to compute an $\varepsilon$-approximation of the value of
the kernel (in a given initial configuration). Since final weight
functions are piecewise affine with a finite number of pieces and
continuous on regions, they are \lipconst-Lipschitz-continuous on
regions, for a given constant $\lipconst\geq 0$. The technique to
obtain the approximation is to consider regions of a refined enough
granularity: we thus pick
\begin{equation}
  \clockgranu(\varepsilon,\lipconst) = \left\lceil
    \frac{\wmax^\Locs\,|\rgame| +
      \lipconst}{\varepsilon}\right\rceil\label{eq:granularity}
\end{equation}
later denoted $\clockgranu$ when the parameters $\varepsilon$
and $\lipconst$ are clear from context.

Consider then the corner-point abstraction $\Ncgame{\clockgranu}$
described in Section~\ref{sec:wtg},
with locations of the form $(\loc,\reg,\corner)$ such that $\corner$
is a corner of the $1/\clockgranu$-region $\reg$.  Two plays $\play$
of $\game$ and $\play'$ of $\Ncgame{\clockgranu}$ are said to be
\emph{$1/\clockgranu$-close} if they follow the same path \rpath in
$\Nrgame \clockgranu$.  In particular, at each step the configurations
$(\loc,\val)$ in $\play$ and $(\loc',\reg',\corner')$ in $\play'$
(with $\corner'$ a corner of the $1/\clockgranu$-region $\reg'$)
satisfy $\loc=\loc'$ and $\val\in \reg'$, and the edges taken in both
plays have the same discrete weights. Close plays have \emph{close}
weights, in the following sense:

\begin{lem}\label{lm:close-plays}
  For all $1/\clockgranu$-close plays $\play$ of $\game$ and
  $\play'$ of $\Ncgame{\clockgranu}$,
  \[|\weight_\game(\play)-\weight_{\Ncgame \clockgranu}(\play')|\leq
    \varepsilon\]
\end{lem}
\begin{proof}
  Since $\play$ and $\play'$ encounter the same locations of $\game$,
  one reaches a target location if and only if the other does.  In the
  case where they do not reach a target location, both weights are
  infinite, and thus equal.  We now look at the case where both plays
  reach a target location, moreover in the same step.

  Consider the region path $\rpath$ of the run $\play$: it can
  be decomposed into a simple path with maximal cycles in it.  The
  number of such maximal cycles is bounded by
  $|\rgame|$ and the remaining simple
  path has length at most $|\rgame|$.  Since all cycles of a kernel are $0$-cycles, the
  parts of $\play$ that follow the maximal cycles have weight exactly
  0.

  Consider the same decomposition for the play $\play'$.  Cycles of
  $\rpath$ do not necessarily map to cycles over locations of
  $\Ncgame \clockgranu$, since the $1/\clockgranu$-corners could be distinct.  However,
  Lemma~\ref{lm:cornerabstract} shows that, for all those cycles of
  $\rpath$, there exists a sequence of finite plays of $\game$ whose
  weight tends to the weight of $\play'$.  Since all those finite plays follow a cycle
  of the region game $\rgame$ (with $\game$ being a kernel), they all
  have weight 0.  Hence, the parts of $\play'$ that follow the maximal
  cycles of $\rpath$ have also weight exactly 0.

  Therefore, the difference
  $|\weight_\game(\play)-\weight_{\Ncgame \clockgranu}(\play')|$ is
  concentrated on the remaining simple path of $\rpath$: on each edge
  of this path, the maximal weight difference is
  $1/\clockgranu\times \wmax^\Locs$ since $1/\clockgranu$ is the
  largest difference possible in time delays between plays that stay
  $1/\clockgranu$-close (since they stay in the same
  $1/\clockgranu$-regions).  Moreover, the difference between the
  final weight functions is bounded by $\lipconst/\clockgranu$, since
  the final weight function $\weightT$ is
  \lipconst-Lipschitz-continuous and the final weight function of
  $\Ncgame \clockgranu$ is obtained as limit of $\weightT$.  Summing
  the two contributions, we obtain as upper bound the constant
  $(\wmax^\Locs\,|\rgame| +
  \lipconst)/\clockgranu\leq \varepsilon$.
\end{proof}

In particular, if we start in some configurations $(\loc,\val)$ of
$\game$, and $((\loc,\reg,\corner),\corner)$ of
$\Ncgame \clockgranu$, with $\val\in \reg$, since both players
have the ability to stay $1/\clockgranu$-close all along the plays, a
bisimulation argument allows us to obtain that the values of the two
games are also close in $(\loc,\val)$ and
$((\loc,\reg,\corner),\corner)$, respectively:

\begin{lem}\label{lm:bisimulation}
  For all locations $\loc\in \Locs$, $1/\clockgranu$-regions $\reg$,
  valuations $\val\in \reg$ and corners~$\corner$ of $\reg$,
  \[|\Val_\game(\loc,\val)-\Val_{\Ncgame
      \clockgranu}((\loc,\reg,\corner),\corner)|\leq \varepsilon\]
\end{lem}
\begin{proof}
  To obtain the result, we prove that
  \[\Val_\game(\loc,\val)\leq \Val_{\Ncgame
      \clockgranu}((\loc,\reg,\corner),\corner)+\varepsilon
    \quad\text{ and }\quad \Val_{\Ncgame
      \clockgranu}((\loc,\reg,\corner),\corner)\leq
    \Val_\game(\loc,\val)+\varepsilon\] By definition and
  determinacy of turn-based \WTG{s}, this is equivalent to proving
  these two inequalities:
  \[\inf_{\stratmin}\sup_{\stratmax}
    \weight_\game(\outcome((\loc,\val),\stratmin,\stratmax))
    \leq \inf_{\stratmin'}\sup_{\stratmax'} \weight_{\Ncgame
      \clockgranu}(\outcome(((\loc,\reg,\corner),\corner),\stratmin',\stratmax'))
    +\varepsilon\]
  \[\sup_{\stratmax'}\inf_{\stratmin'} \weight_{\Ncgame
    \clockgranu}(\outcome(((\loc,\reg,\corner),\corner),\stratmin',\stratmax'))
    \leq \sup_{\stratmax}\inf_{\stratmin}
    \weight_\game(\outcome((\loc,\val),\stratmin,\stratmax))
    +\varepsilon\]
  To show the first inequality, it suffices to show that
  for all $\stratmin'$, there exists $\stratmin$ such that for
  all $\stratmax$, there is $\stratmax'$ verifying
  \begin{equation}\label{eq:beta_almost}
    |\weight_\game(\outcome((\loc,\val),\stratmin,\stratmax))
    - \weight_{\Ncgame
      \clockgranu}(\outcome(((\loc,\reg,\corner),\corner),\stratmin',\stratmax'))|\leq \varepsilon\end{equation}
  For the second, it suffices to show that
  for all $\stratmax'$, there exists $\stratmax$ such that for
  all $\stratmin$, there is $\stratmin'$ verifying the same
  equation~\eqref{eq:beta_almost}.
  We will detail the proof for the first, the second being syntactically the same,
  with both players swapped.

  Equation~\eqref{eq:beta_almost} can be obtained from Lemma~\ref{lm:close-plays},
  under the condition that
  the plays
  $\outcome((\loc,\val),\stratmin,\stratmax)$ and
  $\outcome(((\loc,\reg,\corner),\corner),\stratmin',\stratmax')$ are
  $1/\clockgranu$-close.
  Therefore, we fix a strategy $\stratmin'$ of \MinPl in the game
  $\Ncgame \clockgranu$, and we construct a strategy $\stratmin$ of \MinPl in
  $\game$, as well as two mappings
  $f\colon \FPlaysMin_\game\to \FPlaysMin_{\Ncgame \clockgranu}$ and
  $g\colon \FPlaysMax_{\Ncgame \clockgranu}\to \FPlaysMax_\game$ such~that:
  \begin{itemize}
  \item for all $\play\in \FPlaysMin_\game$, $\play$ and $f(\play)$
    are $1/\clockgranu$-close, and if $\play$ is consistent with $\stratmin$
    and starts in $(\loc,\val)$, then $f(\play)$ is consistent with
    $\stratmin'$ and starts in $((\loc,\reg,\corner),\corner)$;
  \item for all $\play'\in \FPlaysMax_{\Ncgame \clockgranu}$, $g(\play')$ and
    $\play'$ are $1/\clockgranu$-close, and if $\play'$ is consistent with
    $\stratmin'$ and starts in $((\loc,\reg,\corner),\corner)$, then $g(\play')$
    is consistent with $\stratmin$ and starts in $(\loc,\val)$.
  \end{itemize}
  We build $\stratmin$, $f$, and $g$ by induction on the length $k$ of plays, over
  prefixes of plays of length $k-1$, $k$ and $k$, respectively.
  For $k=0$ (plays of length 0 are those restricted to a single
  configuration), we let $f(\loc,\val)=((\loc,\reg,\corner),\corner)$ and $g((\loc,\reg,\corner),\corner)=(\loc,\val)$,
  leaving the other values arbitrary (since we will not use them).

  Then, we suppose $\stratmin$, $f$,
  and $g$ built until length $k-1$, $k$ and $k$, respectively
  (if $k=0$, $\stratmin$ has not been build yet), and we define
  them on plays of length $k$, $k+1$ and $k+1$, respectively.
  For every $\play\in \FPlaysMin_\game$ of length $k$, we note
  $\play'=f(\play)$.  Consider the decision
  $(\delay',\edge')=\stratmin'(\play')$ and $\play'_+$ the prefix
  $\play'$ extended with the decision $(\delay',\edge')$.  By timed
  bisimulation, there exists $(\delay,\edge)$ such that the prefix
  $\play_+$ composed of $\play$ extended with the decision
  $(\delay,\edge)$ builds $1/\clockgranu$-close plays $\play_+$ and $\play'_+$.
  We let $\stratmin(\play)=(\delay,\edge)$.
  If $\play_+\in \FPlaysMin_\game$, we also let $f(\play_+)=\play'_+$,
  and otherwise we let $g(\play'_+)=\play_+$.
  Symmetrically, consider $\play'\in\FPlaysMax_{\Ncgame \clockgranu}$ of length $k$, and
  $\play=g(\play')$.  For all possible decisions $(\delay',\edge')$, by
  timed bisimulation, there exists a decision $(\delay,\edge)$ in the
  prefix $\play$ such that the respective extended plays $\play'_+$
  and $\play_+$ are $1/\clockgranu$-close.  We then let $g(\play'_+)=\play_+$ if
  $\play_+\in\FPlaysMax_\game$ and $f(\play_+)=\play'_+$ otherwise.
  We extend the definition of $f$ and $g$ arbitrarily for other
  prefixes of plays.  The properties above are then trivially verified.

  We then fix a strategy $\stratmax$ of \MaxPl in the game $\game$,
  which determines a unique play
  $\outcome((\loc,\val),\stratmin,\stratmax)$.  We construct a
  strategy $\stratmax'$ of \MaxPl in the game $\Ncgame \clockgranu$ by
  building the unique play
  $\outcome(((\loc,\reg,\corner),\corner),\stratmin',\stratmax')$ we will be
  interested in, such that each of its prefixes is in relation, via
  $f$ or $g$, to the associated prefix of
  $\outcome((\loc,\val),\stratmin,\stratmax)$.  Thus, we only
  need to consider a prefix of play $\play'\in \FPlaysMax_{\Ncgame \clockgranu}$
  that starts in $((\loc,\reg,\corner),\corner)$ and is consistent with
  $\stratmin'$, and $\stratmax'$ built so far.  Consider the play
  $\play = g(\play')$, starting in $(\loc,\val)$ and consistent with
  $\stratmin$, and $\stratmax$ (by assumption).  For the decision
  $(\delay,\edge)=\stratmax(\play)$ (letting $\play_+$ be the extended
  prefix), the definition of $f$ and $g$ ensures that there exists a
  decision $(\delay',\edge')$ after $\play'$ that results in an extended
  play $\play'_+$ that is $1/\clockgranu$-close, via $f$ or $g$, with
  $\play_+$.  We thus can choose $\stratmax'(\play')=(\delay',\edge')$.

  We finally have built two plays
  $\outcome((\loc,\val),\stratmin,\stratmax)$ and
  $\outcome((\loc',\val'),\stratmin',\stratmax')$ that are
  $1/\clockgranu$-close, as needed, which concludes this proof.
\end{proof}

We can thus obtain an $\varepsilon$-approximation of
$\Val_\game(\loc,\val)$ by computing
$\Val_{\Ncgame \clockgranu}((\loc,\reg,\corner),\corner)$ for any
corner $\corner$ of $\reg$. Recall that $\Ncgame \clockgranu$ can be
considered as an untimed weighted game (with reachability objective).
Thus we can apply the result of~\cite{BGHM16}, where it is shown that
the optimal values of such games can be computed in pseudo-polynomial
time (\ie~polynomial with respect to the number of locations $|\Ncgame\clockgranu|$, at most
$(n+1)|\Nrgame\clockgranu|$, and the weights of transitions $\wmax$ encoded
in unary, instead of binary): more precisely, the value
$\Val_{\Ncgame \clockgranu}$ is obtained using the operator
$\ValIteOpe$ of \eqref{eq:operator_untimed}
(page~\pageref{eq:operator_untimed}), as the $i$-th iteration with
$i=((2 |\Ncgame\clockgranu|-1)\wmax+1) |\Ncgame\clockgranu|$, polynomial in $|\Locs|$, $\wmax$, $\clockbound$ and $\clockgranu$,
and exponential in $n$. We then define an $\varepsilon$-approximation
of $\Val_\game$, named $\Val_\clockgranu'$, on each
$1/\clockgranu$-region by interpolating the values of its
$1/\clockgranu$-corners in $\Ncgame \clockgranu$ with a piecewise
affine function:
If $\val$ is a valuation that belongs to the $1/\clockgranu$-region
$\reg$, then $\val$ can be expressed as a (unique) convex combination
of the $1/\clockgranu$-corners $\corner$ of $\reg$, so that
$\val=\sum_{\corner} \alpha_{\corner} \corner$ with
$\alpha_{\corner}\in [0,1]$ for all $\corner$, and
we let $\Val_\clockgranu'(\loc,\val)=\sum_{\corner} \alpha_{\corner} \Val_{\Ncgame \clockgranu}((\loc,\reg,\corner),\corner)$ for all locations $\loc$ of \game.

Moreover, we can control the growth of the Lipschitz constant of the
approximated value for further use.
\begin{lem}\label{lm:kernel-approx-regular}
  $\Val_\clockgranu'$ is an $\varepsilon$-approximation of
  $\Val_\game$,
  \ie~$\|\Val_\clockgranu'-\Val_\game\|_\infty \leq
  \varepsilon$. Moreover, $\Val_\clockgranu'$ is piecewise affine with
  a finite number of pieces and
  $2(\wmax^\Locs\,|\rgame| +
  \lipconst)$-Lipschitz-continuous over regions.
\end{lem}
\begin{proof}
  By construction, the approximated value is piecewise affine with one
  piece per $1/\clockgranu$-region.  To prove the Lipschitz constant,
  it is then sufficient to bound the difference between
  $\Val_{\Ncgame \clockgranu}((\loc,\reg,\corner),\corner)$ and
  $\Val_{\Ncgame \clockgranu}((\loc,\reg,\corner'),\corner')$, for
  $\corner$ and $\corner'$ two corners of a $1/\clockgranu$-region
  $\reg$.  We can pick any valuation \val in $\reg$ and apply
  Lemma~\ref{lm:bisimulation} twice, between \val and $\corner$, and
  between \val and
  $\corner'$.  %
  We obtain
  $|\Val_{\Ncgame
    \clockgranu}((\loc,\reg,\corner),\corner)-\Val_{\Ncgame
    \clockgranu}((\loc,\reg,\corner'),\corner')| \leq 2
  (\wmax^\Locs\,|\rgame| + \lipconst)
  /\clockgranu = 2 (\wmax^\Locs\,|\rgame|
  + \lipconst) \| \corner-\corner'\|_\infty$.
\end{proof}

The time complexity needed to compute
$\Val'_\clockgranu$ is polynomial in $|\Locs|$, $\wmax$, $\clockbound$ and $\clockgranu$,
and exponential in $n$.

\subsection{Approximation of almost-divergent
  \WTG{s}}\label{sec:approx-AD}
We now explain how to approximate the value of an almost-divergent
\WTG $\game$, thus proving Theorem~\ref{thm:almost-div}.  After having
computed the semi-unfolding $\tgame$ described in
Section~\ref{sec:unfolding}, we perform a bottom-up computation of the
approximation. The addition of kernels (with respect to the case of
divergent \WTG{s} studied before) requires us to compute an
approximation instead of the actual value. Notice that the techniques
used in Theorem~\ref{thm:valitebounded}, applied for $i=1$ step (and
not for the whole tree as for divergent \WTG{s}), allow us to compute
the value of a non-kernel node of $\tgame$, depending on the values of
its children. There is no approximation needed here, so that if we
have computed $\varepsilon$-approximations of all its children, we can
compute an $\varepsilon$-approximation of the node. More formally,
this is justified by the following lemma with $i=1$:
\begin{lem}\label{lm:final-distance}
  Let $\game$ and $\game'$ be two \WTG{s} that only differ on the
  final weight functions $\weightT$ and $\weightT'$. Then,
  $\|\Val_\game - \Val_{\game'}\|_\infty \leq
  \|\weightT-\weightT'\|_\infty$.
  Moreover, for all $i\in\N$, $\|\Val_\game^i - \Val_{\game'}^i\|_\infty \leq
  \|\weightT-\weightT'\|_\infty$.
\end{lem}
\begin{proof}
  Consider two strategies $\stratmin$ and $\stratmax$ for both
  players in $\game$ and $\game'$ (since both games are identical
  up-to the final weight functions, they share the same sets of
  strategies for both players). Then, from an initial configuration
  $(\loc,\val)$, the plays
  $\outcome_\game((\loc,\val),\stratmin,\stratmax)$ and
  $\outcome_{\game'}((\loc,\val),\stratmin,\stratmax)$ are the
  same, and their weights only differ by the final weight functions
  taken at the same configurations. Thus,
  \[|\weight(\outcome_\game((\loc,\val),\stratmin,\stratmax)) -
    \weight(\outcome_{\game'}((\loc,\val),\stratmin,\stratmax))|
    \leq \|\weightT-\weightT'\|_\infty\]
  In particular,
  \[\weight(\outcome_\game((\loc,\val),\stratmin,\stratmax))
    \leq
    \weight(\outcome_{\game'}((\loc,\val),\stratmin,\stratmax))
    + \|\weightT-\weightT'\|_\infty\]
  so that, taking infimum over all $\stratmin$ and supremum over
  all $\stratmax$ gives
  \[\Val_\game(\loc,\val) \leq \Val_{\game'}(\loc,\val) +
    \|\weightT-\weightT'\|_\infty\]
  A symmetrical argument allows us to conclude that $\|\Val_\game - \Val_{\game'}\|_\infty \leq
  \|\weightT-\weightT'\|_\infty$.
  Using the bounded horizon versions of $\Val$ and $\weight$ in the previous proof, we also obtain $\|\Val_\game^i - \Val_{\game'}^i\|_\infty \leq
  \|\weightT-\weightT'\|_\infty$ for any $i$.
\end{proof}

We now explain in detail the full process of approximation of the
value $\Val_\game(\loc_0,\val_0)$ of an almost-divergent \WTG $\game$:
it is a bottom-up computation on the tree $T$ rooted in
$(\loc_0,\reg_0)$ (with $\reg_0$ the region of $\val_0$) that we used
to describe the semi-unfolding $\tgame$.  By
Proposition~\ref{prop:semi-unfolding}, the value we want to
approximate is equal to $\Val_{\tgame}((\tilde{\loc}_0,\reg_0),\val_0)$.
For a node $s$ in $T$, let $\ValIteVec_s$ denote the exact value
function of the corresponding node in $\tgame$. In particular, the
value function at the root of $T$ is equal to
$\val\mapsto \Val_{\tgame}((\tilde{\loc}_0,\reg_0),\val)$. Our algorithm
iteratively computes an approximated value function $\ValIteVec'_s$
for all nodes $s$ of $T$.

To obtain an adequate $\varepsilon$-approximation of
$\Val_\game(\loc_0,\val_0)$, we will thus need to guarantee a
precision in kernels that depend on the number of kernels we visit in
the semi-unfolding. Let $\alpha$ be the maximal number of kernels
along a branch of the tree $T$. For a given node $s$ in $T$, we also
let $\alpha(s)$ be the maximal number of kernels along the branches of
the subtree rooted in $s$. Finally, let us denote by $h(s)$ the
maximal length of a branch in the subtree rooted in $s$.

We will maintain, along the bottom-up computation, that $\ValIteVec'_s$
is a $\lipconst_s$-Lipschitz-continuous mapping on regions such that
\begin{equation}
  \lipconst_s \leq \max(2,n)^{h(s)}(2 \wmax^\Locs\,|\rgame| + \lipconst) \quad\text{ and }\quad
  \|\ValIteVec_{s}-\ValIteVec'_{s}\|_\infty\leq
  \alpha(s)\varepsilon/\alpha\label{eq:invariant}
\end{equation}
where $\lipconst$ is the maximal
Lipschitz constant of the final weight functions of $\game$. In
particular, at the root of $T$ where $\alpha(s)=\alpha$, we indeed
recover an $\varepsilon$-approximation of the value of the game in
configuration $(\loc_0,\val_0)$.

For each leaf node $s$ of $T$, $\ValIteVec_s$ and $\ValIteVec'_s$ are
equal to the final weight function $\weightT$ of $\tgame$, and we thus
get the invariant \eqref{eq:invariant} for free (knowing that
$\alpha(s)=h(s)=0$).

For an internal node $s$ of $T$ (that either gives rise to a single state
$(\loc,\reg)$ of $\tgame$, or to a kernel $\Kernel_{\loc,\reg}$), let us
suppose that the value of each child $s'$ of $s$ in $T$ have been
computed and verify the invariant \eqref{eq:invariant}.

If $s$ is a node of the form $(\loc,\reg)$ (that is not part of a
kernel), we define two \WTG{s} $\widetilde\game$ and
$\widetilde\game'$
that contain the
state $(\loc,\reg)$ as well as its children in $T$. The children $s'$
are made target states with respective final weight functions given by
$\ValIteVec_{s'}$ and $\ValIteVec'_{s'}$. Moreover, we know that
\[\alpha(s)=\max_{s'}\alpha(s') \text{ and } h(s)=\max_{s'}h(s')+1\]
By definition, $\ValIteVec_s$ is equal to
$\val\mapsto\Val_{\widetilde\game}(s,\val)$ and thus to
$\val\mapsto\Val_{\widetilde\game}^1(s,\val)$ since $\widetilde\game$ is acyclic of depth $1$. Our approximation algorithm
consists in letting
$\ValIteVec'_{s}=\val\mapsto\Val_{\widetilde\game'}^1(s,\val)$.
By Lemma~\ref{lm:final-distance} (with $i=1$ and $\varepsilon$
being $\max_{s'}\alpha(s')\varepsilon/\alpha$),
$\|\ValIteVec_{s}-\ValIteVec'_{s}\|_\infty\leq
\alpha(s)\varepsilon/\alpha$. By
Lemma~\ref{lm:value-ite-partial-derivatives}, $\ValIteVec'_s$ is
$\lipconst_s$-Lipschitz-continuous on regions with
$\lipconst_s\leq \max(\max_{s'}\lipconst_{s'},|\weight(\loc)| + (n-1)
\max_{s'}\lipconst_{s'})$: with the help of the invariant
\eqref{eq:invariant} for all children $s'$, and since
$|\weight(\loc)|\leq \wmax^\Locs$ and $h(s')\leq h(s)-1$, we obtain
$\lipconst_s\leq \max(2,n)^{h(s)}(2 \wmax^\Locs\,|\rgame| + \lipconst)$.

If $s$ is a node of the form $\Kernel_{\loc,\reg}$, we define two
\WTG{s} $\widetilde\game$ and $\widetilde\game'$,
that contain the locations of the kernel $\Kernel_{\loc,\reg}$ of
$\tgame$, as well as the children of $s$ in $T$ (reached by output
edges of $\Kernel_{\loc,\reg}$). The children~$s'$ are made target
states of respective final weight functions $\ValIteVec_{s'}$ and
$\ValIteVec'_{s'}$.  Moreover, we know that
\[\alpha(s)=\max_{s'}\alpha(s') + 1 \text{ and } h(s)=\max_{s'}h(s')+1\]
Thus, games $\widetilde\game$ and $\widetilde\game'$ are identical,
except for their final weight functions that are
$\varepsilon$-close.
Thus, we know by Lemma~\ref{lm:final-distance} that
$\|\Val_{\widetilde\game} - \Val_{\widetilde\game'}\|_\infty \leq
\max_{s'} \alpha(s')\varepsilon/\alpha$. Moreover, by definition,
$\ValIteVec_s$ is equal to
$\val\mapsto\Val_{\widetilde\game}(s,\val)$. Our approximation
algorithm consists in letting $\ValIteVec'_{s}$ be equal to an
$\varepsilon/\alpha$-approximation of $\Val_{\widetilde\game'}$,
obtained by Lemma~\ref{lm:kernel-approx-regular} with a granularity
$\clockgranu(\varepsilon/\alpha,\max_{s'}\lipconst_{s'})$: we thus
have
$\|\Val_{\widetilde\game'}(s,\cdot) - \ValIteVec'_{s}\|_\infty\leq
\varepsilon/\alpha$, and $\ValIteVec'_{s}$ is
$\lipconst_s$-Lipschitz-continuous on regions with
$\lipconst_s\leq 2(\wmax^\Locs\,|\rgame| +
\max_{s'}\lipconst_{s'})$. By triangular inequality, we deduce that
$\|\ValIteVec_{s}-\ValIteVec'_{s}\|_\infty\leq (\max_{s'}
\alpha(s'))\varepsilon/\alpha+\varepsilon/\alpha =
\alpha_s\varepsilon/\alpha$. Moreover, with the help of
invariant~\eqref{eq:invariant} for $\lipconst_{s'}$, we once again
obtain that
$\lipconst_s\leq \max(2,n)^{h(s)}(2 \wmax^\Locs\,|\rgame| + \lipconst)$.

We thus obtain an algorithm that faithfully computes an
$\varepsilon$-approximation of the value of the game. Let us now
discuss the complexity of the algorithm. Overall, the biggest
Lipschitz constant for $\ValIteVec'$ is
$\max(2,n)^{h}(2 \wmax^\Locs\,|\rgame| +
\lipconst)$ with $h$ the height of the semi-unfolding that is
$|\rgame|(3|\rgame|\wmaxTimed+2\wmax^\target+2)+1$ as noticed
in~\eqref{eq:depth-semi-unfolding}
(page~\pageref{eq:depth-semi-unfolding}). This Lipschitz constant is
thus at most doubly-exponential with respect to the size of
$\game$. Therefore, the biggest granularity $\clockgranu$ used in
kernel approximations (described in~\eqref{eq:granularity}) can be
globally bounded as doubly-exponential in the size of $\game$ and
polynomial in $1/\varepsilon$. This entails that each kernel
approximation can be performed in time doubly-exponential in the size
of $\game$, and polynomial in $1/\varepsilon$. As the height of the
semi-unfolding is at most exponential in the size of $\game$, the
number of kernel approximations needed is at most doubly-exponential.
The rest of the algorithm consist in applying
Theorem~\ref{thm:valitebounded}, outside of the kernels. To obtain the
overall complexity, we prove for kernel nodes
an equivalent to Proposition~\ref{prop:valiteope}, explaining how one step of
computation on the semi-unfolding increases the complexity of the encoding of piecewise
affine value functions.

\begin{lem}\label{lem:kernel-approximation-complexity}
  For a kernel \WTG $\tilde\game'$, if all piecewise affine value functions
  $\ValIteVec'_{s'}$ of target children $s'$ of $s$ in $T$ have
  complexity at most $\langle m,\beta\rangle$, then $\ValIteVec'_{s}$, obtained by Lemma~\ref{lm:kernel-approx-regular} with granularity $\clockgranu$,
  has complexity at most $\langle m',\beta'\rangle$ with
  $m' = n(n+2)(MN+1)$ and
  $\beta'= (\wmax+n\beta)(n+1)!(MN)^{n+2}$, where $n$ denotes the number of clocks.
\end{lem}
\begin{proof}
  The bound on $m'$ comes from the fact that value function
  $\ValIteVec'_s$ is obtained by interpolating the value computed for
  all $1/N$-corners. The partition thus needs to separate each $1/N$-region from others (in the worst case). By the characterisation following Definition~\ref{def:regions}, we thus need all (hyperplane) equations of the form $N\clock_i=k$ with
  $k\in\{0,1,\ldots,MN\}$ (there are $n(MN+1)$ such equations), and
  $N(\clock_i-\clock_j) = k$ with $1\leq i<j\leq n$ and
  $k\in \{-MN,-MN+1,\ldots,MN\}$ (there are $n(n+1)/2 \times (2MN+1)$ such equations). In total, we thus need $n(MN+1) + n(n+1)(2MN+1)/2 \leq n(n+2)(MN+1)$ equations in the worst case.

  Once multiplied by $N$, these equations use constants bounded by $MN$ in absolute value. To bound $\beta'$, we
  also need to bound the constants appearing in the partition
  function $F$ describing $\ValIteVec'_s$ on each base cell (that is a
  $1/N$-region). We deduce from \cite[Proposition~19]{BGHM16} that the
  value of each of the corners in the $1/N$-region game
  $\Ncgame \clockgranu$ is the value of an acyclic path in the untimed
  game $\Ncgame \clockgranu$ (since \MaxPl always has a memoryless
  optimal strategy), and is thus the sum of a final weight and of the
  weight of at most $(n+1)|\Locs||\NRegs|$ transitions (this is the
  number of states of the untimed game).
  \begin{itemize}
  \item The final weight is obtained by taking the value of the
    partition function $F_{s'}$ of one of the children of $s$ in
    $T$ at a corner of a $1/N$-region. If $F_{s'}$ is described by the
    equation $a_{\clocky}\clocky = a_1\clock_1+\cdots + a_n\clock_n+b$
    (with integers $|a_{\clocky}|,|a_1|,\ldots, |a_n|,|b|\leq \beta$)
    then the value at corner
    $(\clock_1\mapsto k_1/N, \ldots, \clock_n\mapsto k_n/N)$, with natural
    numbers $0\leq k_1,\ldots,k_n\leq NM$, is of the form
    $(a_1 k_1 + \cdots + a_n k_n)/(N a_\clocky) + b/ a_\clocky$. It can
    thus be written $A/B$ with $|A|$ and $|B|$ integers at most
    $nMN\beta$.
  \item Each transition of $\Ncgame \clockgranu$ has a weight of the
    form $d \, \weight(\loc)+\weight(e)$ with $d$ a delay that separates two
    corners: thus $d\in \{0,1/N,2/N,\ldots,MN/N\}$, and the weight can
    be written as $A/B$ with $|A|$ and $|B|$ integers at most
    $MN \wmax^\Locs + N\wmax^\Edges$.
  \end{itemize}
  In total, corners have values that can be written as $A/B$ with
  $|A|$ and $|B|$ integers at most $MN (\wmax + n\beta)$.

  The partition function $F$ describing $\ValIteVec'_s$ on a
  $1/\clockgranu$-region $\reg$ is obtained by interpolating the
  values of each of its corners. Fix $\val\in \reg$: it can be
  expressed as a (unique) convex combination of the
  $1/\clockgranu$-corners $\corner_1,\ldots,\corner_{n+1}$ of $\reg$,
  so that $\val=\sum_{i=1}^{n+1} \alpha_{i} \corner_i$ with
  $\alpha_{i}\in [0,1]$. Then, we let
  $\sem{F}(\val)=\sum_{i=1}^{n+1} \alpha_{i} \Val_{\Ncgame
    \clockgranu}((s,\reg,\corner),\corner)$.  To further bound
  $\beta'$, we need to express the coefficients $\alpha_i$ in terms of
  $\val(\clock_1),\ldots,\val(\clock_n)$. Indeed, we have the
  matricial relation
  \[\left(
      \begin{array}{c}
        \val(\clock_1) \\ \vdots\\\val(\clock_n) \\ 1
      \end{array}
    \right) = \left(
      \begin{array}{cccc}
        \corner_1(\clock_1) & \corner_2(\clock_1) & \cdots &
                                                             \corner_{n+1}(\clock_1)
        \\
        \vdots & \vdots &  & \vdots \\
        \corner_1(\clock_n) & \corner_2(\clock_n) & \cdots &
                                                             \corner_{n+1}(\clock_n)
        \\
        1 & 1 &\cdots & 1
      \end{array}\right) \times \left(
      \begin{array}{c}
        \alpha_1\\ \vdots\\\alpha_n \\\alpha_{n+1}
      \end{array}\right)\]
  the last line coming from the fact that the sum of coefficients
  $\alpha_i$ must be equal to 1. The square matrix $\mathcal M$ of
  this equality can be shown invertible, since no three corners of a
  $1/N$-region are aligned. Therefore, each coefficient $\alpha_i$ can
  be written as an affine expression over
  $\val(\clock_1),\ldots,\val(\clock_n)$ whose coefficients are given
  by the inverse of matrix $\mathcal M$. By Cramer's rule, this
  inverse can be written as the cofactor matrix $\mathcal C$ divided by the
  determinant of $\mathcal M$. Coefficients of $\mathcal C$ are
  determinants of $n\times n$ submatrices of $\mathcal M$. Since
  coefficients of $\mathcal M$ are of the form $k/N$ with
  $k\in\{0,\ldots,MN\}$, coefficients of $\mathcal C$ are of the form
  $k/N^{n}$ with $k\in\{0,1,\ldots,n!(MN)^{n}\}$. The determinant of
  $\mathcal M$ is similarly of the form $k'/N^{n+1}$ with
  $k'\in\{1,\ldots,(n+1)!(MN)^{n+1}\}$. We can thus write each
  $\alpha_i$ as an affine coefficient of
  $\val(\clock_1),\ldots,\val(\clock_n)$ whose coefficients are of the
  form $Nk/k'$ with $k\in\{0,1,\ldots,n!(MN)^{n}\}$ and
  $k'\in\{1,\ldots,(n+1)!(MN)^{n+1}\}$.  Multiplying these
  coefficients with the values of corners, we obtain that each partial
  derivative and additive constant of the equalities defining $F$ on
  each $1/N$-region is bounded by $MN(\wmax+n\beta)(n+1)!(MN)^{n+1}$.
\end{proof}

Pairing this lemma with Proposition~\ref{prop:valiteope} that explains
how the complexity grows along a non-kernel node, we obtain (as in
Corollary~\ref{cor:complexity}) the maximal complexity of the
partition functions obtained from the leaves of the
semi-unfolding in $i$ steps. Notice that the $m$-complexity is reset
each time we go through a kernel, while the $\beta$-complexity
continues to grow polynomially in $\beta$ (but exponentially in $n$
which does not change the final computations). Thus, we conclude once
again that the partition functions after $i$ steps have a
complexity doubly-exponential in $i$, and thus (since the height $h$
of the semi-unfolding is exponential in the size of the game), we
obtain a triply-exponential algorithm computing an
$\varepsilon$-approximation of the value of an almost-divergent \WTG,
which concludes the proof of Theorem~\ref{thm:almost-div}.
This complexity is polynomial in $1/\varepsilon$ as $\clockgranu$ is linear in $1/\varepsilon$.
An example of execution of the approximation schema can be found in
Appendix~\ref{app:example}.

\section{Symbolic algorithms}\label{sec:symbolic-wtg}

The previous approximation result suffers from several drawbacks.  It
relies on the SCC decomposition of the region automaton.  Each of these
SCCs have to be analysed in a sequential way, and their analysis
requires an a priori refinement of the granularity of regions. We show that this can be overcome, in case we suppose that no configuration has value $-\infty$ (which could be guaranteed by the designer of the game for some particular reasons; this is for instance always the case if only non-negative weights are used). We
prove in this section that the symbolic approach based on the value
iteration paradigm, \ie~the computation of iterates of the
operator~$\ValIteOpe$, recalled in page~\pageref{eq:operator}, is an
approximation schema, as stated in Theorem~\ref{thm:symbolic}.

\subsection{Symbolic approximation algorithm}\label{sec:symbolic}

Notice that configurations with value $+\infty$ are stable through
value iteration, and do not affect its other computations.  Since
Theorem~\ref{thm:symbolic} assumes the absence of configurations of
value $-\infty$, we will therefore consider in the following that all
configurations have finite value in \game (we discuss further this hypothesis in Section~\ref{sec:discussion.}).

Consider first a game $\game$ that is a kernel, with final weight functions that are \lipconst-Lipschitz-continuous on regions.  By Lemma~\ref{lm:bisimulation},
there exists an integer~$\clockgranu$ such
that solving the untimed weighted game $\Ncgame \clockgranu$ computes an
$\varepsilon/2$-approximation of the value of $1/\clockgranu$ corners.
We denote $\clockgranu(\varepsilon,\lipconst)$ this granularity,
and recall that
$$\clockgranu(\varepsilon,\lipconst) = \left\lceil\cfrac{2\wmax^\Locs\,|\rgame| +
2\lipconst}{\varepsilon}\right\rceil\,.$$

Using the results of~\cite{BGHM16} for untimed weighted games, we know
that those values are obtained after a finite number of steps of (the
untimed version of) the value iteration operator, depending on the
number
$|\Ncgame{\clockgranu(\varepsilon,\lipconst)}|\leq |\Locs| |\NRegions
{\clockgranu(\varepsilon,\lipconst)}\Clocks\clockbound| (n+1)$
of locations of the corner-point abstraction.  More precisely, if one considers a
number of iterations
$$P_\Kernel(\varepsilon,\lipconst) = |\Ncgame {\clockgranu(\varepsilon,\lipconst)}| (2 (|\Ncgame{\clockgranu(\varepsilon,\lipconst)}| -1) \wmaxTimed+1)\,,$$ then
$\Val^{P_\Kernel(\varepsilon,\lipconst)}_{\Ncgame{\clockgranu(\varepsilon,\lipconst)}}((\loc,\reg,\corner),\corner)=\Val_{\Ncgame{\clockgranu(\varepsilon,\lipconst)}}((\loc,\reg,\corner),\corner)$.
From this observation, we deduce the following property of $P_\Kernel(\varepsilon,\lipconst)$:
\begin{lem}\label{lm:symbolic-kernel}
  If $\game$ is a kernel with no configurations of infinite value and with final weight functions that are \lipconst-Lipschitz-continuous on regions, and $\varepsilon>0$,
  then it holds for all configurations $(\loc,\val)$ of $\game$ that
  $|\Val_\game(\loc,\val)-\Val^{P_\Kernel(\varepsilon,\lipconst)}_\game(\loc,\val)|\leq \varepsilon$.
\end{lem}
\begin{proof}
  We already know that
  $\Val^{P_\Kernel(\varepsilon,\lipconst)}_{\Ncgame{\clockgranu(\varepsilon,\lipconst)}}((\loc,\reg,\corner),\corner)=\Val_{\Ncgame{\clockgranu(\varepsilon,\lipconst)}}((\loc,\reg,\corner),\corner)$
  for all configurations $((\loc,\reg,\corner),\corner)$ of $\Ncgame{\clockgranu(\varepsilon,\lipconst)}$.  Moreover,
  Lemma~\ref{lm:bisimulation} ensures that
  $|\Val_\game(\loc,\val)-\Val_{\Ncgame{\clockgranu(\varepsilon,\lipconst)}}((\loc,\reg,\corner),\corner)|\leq
  \varepsilon/2$ whenever $\val$ is in the $\frac{1}{\clockgranu(\varepsilon,\lipconst)}$-region
  $\reg$.  Therefore, we only need to prove that
  $|\Val^{P_\Kernel(\varepsilon,\lipconst)}_\game(\loc,\val)-\Val^{P_\Kernel(\varepsilon,\lipconst)}_{\Ncgame{\clockgranu(\varepsilon,\lipconst)}}((\loc,\reg,\corner),\corner)|\leq
  \varepsilon/2$ to conclude.  This is done as for
  Lemma~\ref{lm:bisimulation}, since Lemma~\ref{lm:close-plays} (that
  we need to prove Lemma~\ref{lm:bisimulation}) does not depend on the
  length of the plays $\play$ and $\play'$, and both runs reach the
  target state in the same step, \ie~both before or after the horizon
  of $P_\Kernel(\varepsilon,\lipconst)$ steps.
\end{proof}

Once we know that value iteration converges on kernels, we can use the semi-unfolding
of Section~\ref{sec:unfolding} to prove that it also converges on non-negative SCCs
when all values are finite.

We define the following notations:
\begin{itemize}
  \item let $h$ be the height of the semi-unfolding of $\game$,
  \item let $\alpha$ be the maximum number of kernels in a branch of the semi-unfolding,
  \item let $(\lipconst_i)_{i\in\N}$ be a sequence of Lipschitz constants, defined by $\lipconst_0'=\lipconst$ and by
    \[\lipconst_{i+1} = \cfrac{\max(n-1,2)^{P_\Kernel(\frac{\varepsilon}{\alpha},\lipconst_i)}-1}{\max(n-2,1)}\wmax^\Locs+(n-1)^{P_\Kernel(\frac{\varepsilon}{\alpha},\lipconst_i)} \lipconst_i\,,\]
  \item a number of iterations $P_+(\varepsilon,\lipconst) = h +\alpha (P_\Kernel(\frac{\varepsilon}{\alpha},\lipconst_{h})-1)$.
\end{itemize}

 In the worst case, $P_+(\varepsilon,\lipconst)$
 can be non-elementary:
$P_\Kernel(\frac{\varepsilon}{\alpha},\lipconst)$ grows polynomially in $\lipconst$, so that $\lipconst_h$  is a tower of $\alpha$ exponentials.

\begin{lem}\label{lm:symbolic-scc-plus}
  If $\game$ is a non-negative SCC with no configurations of infinite
  value and with final weight functions that are \lipconst-Lipschitz-continuous on regions, and $\varepsilon>0$, then
  $|\Val_\game(\loc,\val)-\Val^{P_+(\varepsilon,\lipconst)}_\game(\loc,\val)|\leq \varepsilon$
  for all configurations $(\loc,\val)$ of $\game$.
\end{lem}
\begin{proof}
The idea is to unfold every kernel of the semi-unfolding game \tgame
according to its bound in Lemma~\ref{lm:symbolic-kernel}.
More precisely, consider a non-negative SCC \game, a precision $\varepsilon$, and an initial configuration $(\loc_0,\val_0)$.
Let \tgame be its finite semi-unfolding (obtained from the labelled tree $T$,
as in Section~\ref{sec:unfolding}), such that
$\Val_\game(\loc_0,\val_0)=\Val_{\tgame}((\tilde{\loc}_0,\reg_0),\val_0)$.
Let $\alpha$ be the maximum number of kernels along one of the branches of $T$.

We describe a bottom-up transformation of $T$ into $T'$, that keeps track
of the Lipschitz constants and the precision of the approximation,
and define a new weighted timed game $\mathcal T'(\game)$ from $T'$ by applying the method used to create \tgame in Section~\ref{sec:unfolding}.

Leaf nodes are left unchanged, with final weight functions that are $\lipconst$-Lipschitz-continuous on regions.
Note that the associated state in $\tgame$ has the same value.
For an internal node $s$ of $T$ (that gives rise to a single state
$(\loc,\reg)$ of $\tgame$, or a kernel $\Kernel_{\loc,\reg}$), let us
suppose that the value of all children $s'$ of $s$ in $T'$ are $\lipconst_i$-Lipschitz-continuous on regions, for some $\lipconst_i$, and that they are $\varepsilon_i$-close to their
values in $\tgame$.

If $s$ is a node of the form $(\loc,\reg)$ (that is not part of a
kernel), we keep it unchanged in $T'$, and therefore by Lemma~\ref{lm:final-distance} the value of $s$ in $\mathcal T'(\game)$
is $\varepsilon_i$-close to its value in $\tgame$, and is $\max(\lipconst_i,|\weight(\loc)|+(n-1)\lipconst_i)$-Lipschitz-continuous on regions by Lemma~\ref{lm:value-ite-partial-derivatives}.

If $s$ is a node of the form $\Kernel_{\loc,\reg}$,
we replace it by an unfolding of $\Kernel_{\loc,\reg}$ of depth $P_\Kernel(\frac{\varepsilon}{\alpha},\lipconst_i)$,
where $P$ is the bound
of Lemma~\ref{lm:symbolic-kernel} for the kernel $\Kernel_{\loc,\reg}$.
Then, by Lemmas~\ref{lm:symbolic-kernel} and~\ref{lm:final-distance}, the value of $s$ in $\mathcal T'(\game)$
is $\varepsilon_i+\frac{\varepsilon}{\alpha}$-close to its value in $\tgame$, and is $\lipconst_{i+1}$-Lipschitz-continuous on regions by Proposition~\ref{prop:tree-lipschitz},
with
\[\lipconst_{i+1}=\cfrac{\max(n-1,2)^{P_\Kernel(\frac{\varepsilon}{\alpha},\lipconst_i)}-1}{\max(n-2,1)}\wmax^\Locs+(n-1)^{P_\Kernel(\frac{\varepsilon}{\alpha},\lipconst_i)} \lipconst_i \,.\]
In particular, note that $P_\Kernel(\frac{\varepsilon}{\alpha},\lipconst_i)$ is polynomial in $\lipconst_i$, therefore $\lipconst_{i+1}$ is exponential in $\lipconst_i$.
Since $P_\Kernel(\frac{\varepsilon}{\alpha},\lipconst_i)\geq 1$, $\lipconst_{i+1}$ is greater than the $\max(\lipconst_i,|\weight(\loc)|+(n-1)\lipconst_i)$
expression obtained for non-kernel nodes.
Thus, we can bound the Lipschitz constants globally by $\lipconst_h$, so that each kernel is unfolded for at most $P_\Kernel(\frac{\varepsilon}{\alpha},\lipconst_h)$ steps.

At the root of $T'$, we recover a node whose value function in $\mathcal T'(\game)$ is $\varepsilon$-close to its value in $\tgame$, and that is $\lipconst_h$-Lipschitz-continuous on regions.

Observe that $\mathcal T'(\game)$ is not a semi-unfolding, it is instead a (complete) unfolding of
\rgame, of a certain bounded depth $P_+(\varepsilon,\lipconst)$ (at most $(h-\alpha)$ plus $\alpha$ times $P_\Kernel(\frac{\varepsilon}{\alpha},\lipconst_h)$).

Finally, let us show that the value computed by $\Val^{P_+(\varepsilon,\lipconst)}_{\game}$ on the root state
$(\tilde{\loc}_0,\reg_0)$
is bounded between $\Val_\game$ and $\Val_{\mathcal T'(\game)}$, which allows us to
conclude.
Consider $\mathcal T''(\game)$, the (complete) unfolding of \rgame with
unfolding depth $P_+(\varepsilon,\lipconst)$, where kernels are also unfolded.
By construction, $\Val_{\mathcal T''(\game)}((\tilde{\loc}_0,\reg_0),\val_0)
=\Val_{\mathcal T''(\game)}^{P_+(\varepsilon,\lipconst)}((\tilde{\loc}_0,\reg_0),\val_0)$.
Then, we can prove that $\Val_{\mathcal T''(\game)}^{P_+(\varepsilon,\lipconst)}((\tilde{\loc}_0,\reg_0),\val_0)
=\Val_\game^{P_+(\varepsilon,\lipconst)}(\loc_0,\val_0)$ (same strategies at bounded horizon ${P_+(\varepsilon,\lipconst)}$),
which implies $\Val_\game(\loc_0,\val_0)=\Val_{\rgame}((\loc_0,\reg_0),\val_0)\leq
\Val_{\mathcal T''(\game)}((\tilde{\loc}_0,\reg_0),\val_0)$ (the function $\Val^k$ decreases monotonously as $k$ increases).
By another monotonicity argument,
we can also prove $\Val_{\mathcal T''(\game)}((\tilde{\loc}_0,\reg_0),\val_0)\leq
\Val_{\mathcal T'(\game)}((\tilde{\loc}_0,\reg_0),\val_0)$:
the tree $\mathcal T''(\game)$ contains $\mathcal T'(\game)$ as a prefix, and therefore one can construct $\mathcal T'(\game)$ from $\mathcal T''(\game)$ by replacing sub-trees of $\mathcal T''(\game)$ by stopped leaves of final weight $+\infty$ (remember that we are in a non-negative SCCs).
Bringing everything together we obtain
$|\Val_\game^{P_+(\varepsilon,\lipconst)}(\loc_0,\val_0)-\Val_\game(\loc_0,\val_0)|\leq\varepsilon$.
\end{proof}

Proving the same property on non-positive SCCs requires more work,
because the semi-unfolding gives final weight $-\infty$ to stop
leaves, %
which does not integrate well with value iteration (initialisation at
$+\infty$ on non-target states).  However, by unfolding those SCCs
slightly more (at most $|\rgame|$ more steps), we can obtain the
desired property with a similar bound~$P_-(\varepsilon,\lipconst)=h+|\rgame| +\alpha (P_\Kernel(\frac{\varepsilon}{\alpha},\lipconst_{h+|\rgame|})-1)$.
\begin{lem}\label{lm:symbolic-scc-minus}
  If $\game$ is a non-positive SCC with no configurations of infinite value and with final weight functions that are \lipconst-Lipschitz-continuous on regions, and $\varepsilon>0$, then
  $|\Val_\game(\loc,\val)-\Val^{P_-(\varepsilon,\lipconst)}_\game(\loc,\val)|\leq
  \varepsilon$ for all configurations $(\loc,\val)$ of~$\game$.
\end{lem}
\begin{proof}
Consider a non-positive SCC \game, a precision $\varepsilon$, and an initial configuration $(\loc_0,\val_0)$.
Let \tgame be its finite semi-unfolding (obtained from the labelled tree $T$,
as in Section~\ref{sec:unfolding}), such that $\Val_\game(\loc_0,\val_0)=\Val_{\tgame}((\tilde{\loc}_0,\reg_0),\val_0)$.

We now change $T$, by adding a subtree under each stopped leaf: the complete unfolding of \rgame,
starting from the stopped leaf, of depth $|\rgame|$.  Let us name $T^+$ this unfolding tree.
We then construct $\mathcal T^+(\game)$ as before, based on $T^+$.
Since we are in a non-positive SCC, $\mathcal T^+(\game)$ must have
final weight $-\infty$ on its stopped leaves.
It is easy to see that $\Val_\game(\loc_0,\val_0)=\Val_{\mathcal T^+(\game)}((\tilde{\loc}_0,\reg_0),\val_0)$
still holds (the proof of Lemma~\ref{lm:symbolic-scc-plus} was based on branches being long enough, and we
increased the lengths).
We now perform a small but crucial change: the final weight of stopped leaves
in $\mathcal T^+(\game)$ is set to $+\infty$ instead of $-\infty$.
Trivially $\Val_{\tgame}((\tilde{\loc}_0,\reg_0),\val_0)\leq
\Val_{\mathcal T^+(\game)}((\tilde{\loc}_0,\reg_0),\val_0)$ (we increased the final weight function).
Let us prove that $$\Val_{\mathcal T^+(\game)}((\tilde{\loc}_0,\reg_0),\val_0)\leq
\Val_{\tgame}((\tilde{\loc}_0,\reg_0),\val_0)\,.$$

For a fixed $\eta>0$, consider $\stratmin$ an $\eta$-optimal strategy for player \MinPl
in \tgame.
Let us define $\stratmin^+$, a strategy for \MinPl in $\mathcal T^+(\game)$, by
making the same choice as \stratmin on the common prefix tree, and once a node
that is a stopped leaf in \tgame is reached, we switch to a positional attractor
strategy of \MinPl towards target states.
Consider any strategy $\stratmax^+$ of \MaxPl in $\mathcal T^+(\game)$,
and let \stratmax be its projection in \tgame.
Let $\play^+$ denote the (maximal) play $$\outcome_{\mathcal T^+(\game)}(((\loc_0,\reg_0),\val_0),\stratmin^+,\stratmax^+))\,,$$
and $\play$ be $\outcome_{\tgame}(((\loc_0,\reg_0),\val_0),\stratmin,\stratmax))$.
By construction, $\play^+$ does not reach a stopped leaf in $\mathcal T^+(\game)$.
If the play $\play^+$
stays in the common prefix tree of $T$ and $T^+$, then $\play=\play^+$,
and $$\weight_{\mathcal T^+(\game)}(\play^+)\leq \Val_{\tgame}((\tilde{\loc}_0,\reg_0),\val_0)+\eta\,.$$
If it does not, then $\play^+$ has a prefix that reaches a stopped leaf
in \tgame: this must be $\rho$.
This implies (see Lemma~\ref{lm:mimickedplays}) that $$\weight_{\mathcal T^+(\game)}(\play^+)<-|\rgame|\wmaxTimed-\wmax^\target\leq\Val_{\tgame}((\tilde{\loc}_0,\reg_0),\val_0)\,.$$
Since this holds for all $\eta>0$, we proved $\Val_{\mathcal T^+(\game)}((\tilde{\loc}_0,\reg_0),\val_0)\leq
\Val_{\tgame}((\tilde{\loc}_0,\reg_0),\val_0)$, which finally implies that
the two values are equal.

Then, we can follow the proof of Lemma~\ref{lm:symbolic-scc-plus}
(with $T^+$ and $\mathcal T^+(\game)$ instead of $T$ and $\mathcal T(\game)$) in order to conclude.
\end{proof}

Now, if we are given an almost-divergent \WTG \game and a precision $\varepsilon$,
we can add the bounds for value iteration obtained from each SCC by Lemmas~\ref{lm:symbolic-scc-plus} and~\ref{lm:symbolic-scc-minus},
and obtain a final bound $P$ such that for all $k\geq P$,
$\Val^k_\game$ is an $\varepsilon$-approximation of $\Val_\game$.
This concludes the proof of Theorem~\ref{thm:symbolic}.

\begin{rem}
We note that one can get a (non-elementary) upper bound for $P$ without constructing the semi-unfolding nor the region game, by using the inequalities $|\rgame|\leq|\Locs|n!(2\clockbound)^n$,
  $|\Ncgame{\clockgranu(\varepsilon,\lipconst)}|\leq|\Locs|n!(2\clockgranu(\varepsilon,\lipconst)\clockbound)^n(n+1)$ and
 $\alpha\leq h\leq |\rgame|(3|\rgame|\wmaxTimed+2\wmax^\target+2)+1$ to
 fix global values for $\alpha=h$, followed by observing
 $P_+(\varepsilon,\lipconst)\leq P_-(\varepsilon,\lipconst)$
 so that $P\leq |\rgame| P_-(\varepsilon,\lipconst)$.
\end{rem}

\subsection{Discussion}\label{sec:discussion.}
This symbolic procedure avoids the three %
drawbacks (SCC decomposition,
sequential analysis of the SCCs, and refinement of the granularity of
regions) of the previous approximation schema
if we assume that no configuration has value $-\infty$.
However, computing the number of steps $P$ of Theorem~\ref{thm:symbolic} is non-elementary,
with an upper bound complexity that is
of the order of a tower of $\alpha$ exponentials, with
$\alpha$ exponential in the size of the input \WTG.
Instead, we could directly launch the value iteration algorithm on the game $\game$, and
stop the computation whenever we are satisfied enough by the
approximation computed: however, there are no formal guarantees whatsoever on
the quality of the approximation before the number of steps $P$ given
above.

If $\game$ is not guaranteed to be free of configurations of value
$-\infty$, then we must first perform the SCC decomposition of \rgame,
and, as \game is almost-divergent, identify and remove regions whose
value is $-\infty$, by Proposition~\ref{prop:-infty}.  Then, we can apply the
value iteration algorithm.

We also note that if $\game$ is a divergent \WTG,
the unfolding of kernels is not needed, so that
Lemmas~\ref{lm:symbolic-scc-plus} and~\ref{lm:symbolic-scc-minus}
construct bounds $P_+(\varepsilon,\lipconst)$
and $P_-(\varepsilon,\lipconst)$ allowing one to compute the exact values of $\game$.
Moreover, these bounds become exponential in the size of $\game$
instead of being non-elementary, so that the overall complexity
of the symbolic algorithm is $3$-\EXP, matching the
result of Section~\ref{sec:computing}.

As a final remark, notice that our correctness proof strongly relies
on Section~\ref{sec:approx-kernels}, and thus would not hold with the
approximation schema of \cite{BJM15} (which
does not preserve the continuity on regions of the computed value
functions, in turn needed to define final weights on $\frac{1}{\clockgranu(\varepsilon,\lipconst)}$-corners).

\section{Strategy synthesis}\label{sec:strategies}

We are also interested in the synthesis problem, that asks for an
$\varepsilon$-optimal strategy of \MinPl.
In this section, we will prove the following result:
\begin{thm}\label{thm:strategies}
  Let $\varepsilon>0$ and let $\game$ be a divergent \WTG.
  We can compute in triple exponential time
  an $\varepsilon$-optimal strategy for \MinPl.
\end{thm}

Recall that in the value iteration algorithm of Section~\ref{subsec:value-functions}, one step of the game is summarised
by the operator $\ValIteOpe$ mapping each value function
$\ValIteVec$ to a value function $\ValIteVec'=\ValIteOpe(\ValIteVec)$
defined by
  \begin{equation}\ValIteVec'_{\loc}(\val)=
  \begin{cases}
    \weightT(\loc,\val)
    &
   \text{if }\loc\in\LocsT\,,\\
    \sup_{(\loc,\val)\xrightarrow{\delay,\edge}(\loc',\val')}
     \big[\delay\cdot\weight(\loc)+\weight(\edge)+\ValIteVec_{\loc'}(\val')\big]
     &
    \text{if }\loc\in\LocsMax\,,\\
    \inf_{(\loc,\val)\xrightarrow{\delay,\edge}(\loc',\val')}
     \big[\delay\cdot\weight(\loc)+\weight(\edge)+\ValIteVec_{\loc'}(\val') \big] &
    \text{if }\loc\in\LocsMin\,.
  \end{cases}\label{eq:operator-recall}\end{equation}
Intuitively, $\varepsilon$-optimal
strategies can be extracted from the value iteration operator, by mapping each
play that ends in $(\loc,\val)$ with $\loc\in\LocsMin$ to a choice
$(\delay,\edge)$ such that the transition
$(\loc,\val)\xrightarrow{\delay,\edge}(\loc',\val')$ is
$\varepsilon$-optimal. However, the choice depends on the step $i$ in
the value iteration computation.  Formally, if $A$ is a set, $f$ is a
mapping from $A$ to $\Rbar$ and $\varepsilon>0$, let
$\arginf_{a\in A}^{~\varepsilon} f(a)$ denote the set of elements
$a^*\in A$ such that $f(a^*) \leq \inf_{a\in A} f(a)+\varepsilon$.
Then, let us name $\stratmin^{i,\varepsilon}$ a strategy defined from
the application of \eqref{eq:operator-recall} in
$\ValIteVec^i= \ValIteOpe(\ValIteVec^{i-1})$, so that for all finite
plays $\rho$ ending in $(\loc,\val)$ and $\loc\in\LocsMin$,
\begin{equation}
  \stratmin^{i,\varepsilon}(\rho)\in
  \arginf_{(\loc,\val)\xrightarrow{\delay,\edge}(\loc',\val')}^{~\varepsilon}
  \big(\delay\cdot\weight(\loc)+\weight(\edge)+
  \ValIteVec_{\loc'}^{i-1}(\val')\big)\label{eq:strategy}
\end{equation}

  Now, let $\stratmin^{\star,i,\varepsilon}$ denote a strategy that maps
  every finite play $\rho$ ending in $\LocsMin\setminus\LocsT$ to $(\delay,\edge)=\stratmin^{j,\varepsilon}(\rho)$,
  with $j=\max(0,i-|\rho|)$.

\begin{prop}\label{prop:optstratvalite}
  The strategy $\stratmin^{\star,i,\varepsilon}$ is $\varepsilon i$-optimal
  for $\MinPl$ during the first $i$ steps:
  $$|\Val^i((\loc,\val),\stratmin^{\star,i,\varepsilon/i})-\Val^i(\loc,\val)|\leq\varepsilon$$
\end{prop}
  \begin{proof}
    Let us show by induction on $i$ that for all configurations $(\loc,\val)$
    in \game,
    $\Val^i_\game((\loc,\val),\stratmin^{\star,i,\varepsilon})
    \leq \ValIteVec^i_\loc(\val)+\varepsilon i$.
    If $i=0$, both sides are equal to either $+\infty$ or to $\weightT(\loc,\val)$
    if $\loc\in\LocsT$.
    Let us assume that $\Val^{i-1}_\game((\loc,\val),\stratmin^{\star,i-1,\varepsilon})
    \leq \ValIteVec^{i-1}_\loc(\val)+\varepsilon (i-1)$.

    We note that for all strategies $\stratmin$,
    if $\loc\in\LocsMax$ then
    $$\Val^{i}_\game((\loc,\val),\stratmin)=\sup_{(\loc,\val)\xrightarrow{\delay,\edge}(\loc',\val')}
     \big[\delay\cdot\weight(\loc)+\weight(\edge)+\Val^{i-1}_\game((\loc',\val'),\stratmin') \big]\,,$$
     where $\stratmin'$ appends $(\loc,\val)\xrightarrow{\delay,\edge}(\loc',\val')$ in front of paths
     and then calls $\stratmin$.
     In particular, if $\stratmin=\stratmin^{\star,i,\varepsilon}$
     then $\stratmin'$ matches the definition of
     $\stratmin^{\star,i-1,\varepsilon}$.
     Then, $\Val^{i}_\game((\loc,\val),\stratmin^{\star,i,\varepsilon})
     \leq \ValIteVec^{i}_\loc(\val)+\varepsilon (i-1)$.

     Similarly, if $\loc\in\LocsMin\backslash\LocsT$ then
     $$\Val^{i}_\game((\loc,\val),\stratmin)=
     \delay\cdot\weight(\loc)+\weight(\edge)+
     \Val^{i-1}_\game((\loc',\val'),\stratmin')$$ where
     $(\loc,\val)\xrightarrow{\delay,\edge}(\loc',\val')$ is the
     decision made by $\stratmin$ on $(\loc,\val)$ and $\stratmin'$
     appends $(\loc,\val)\xrightarrow{\delay,\edge}(\loc',\val')$ in
     front of paths and then calls $\stratmin$.  Then,
     $$\Val^{i}_\game((\loc,\val),\stratmin^{\star,i,\varepsilon}) -
     \varepsilon \leq
     \delay\cdot\weight(\loc)+\weight(\edge)+\Val^{i-1}_\game((\loc',\val'),\stratmin^{\star,i-1,\varepsilon})
     \leq \ValIteVec^{i}_\loc(\val)+\varepsilon (i-1)$$

      Finally, if $\loc\in\LocsT$ then $\Val^i_\game((\loc,\val),\stratmin^{\star,i,\varepsilon})=\ValIteVec^i_\loc(\val)=\weightT(\loc,\val)$.
    \end{proof}

    We now explain how to extract some strategies
    $\stratmin^{i,\varepsilon}$ from the partition functions, in
    order to solve the synthesis problems on acyclic and divergent
    games (games where a known $i\in\N$ implies $\Val=\Val^i$).

We will use affine inequalities over $n+1$ variables to encode constraints
on the choice of delays. Formally, an affine delay inequality is an equation $I$
of the form $a_\clockd \clockd \dbmlaop a_1\clock_1+\cdots+a_n\clock_n+b$,
with all $a_i$, $b$ and $a_\clockd$ integers of $\Z$, $a_\clockd\neq 0$
and $\dbmla\in \{\dbmlt,\dbmleq\}$.
We say that $I$ is a lower bound if $a_\clockd<0$, and
that it is an upper bound if $a_\clockd>0$.

\begin{defi}
  A \emph{partition strategy function} $S$ defined over a partition $P$
  is a mapping from the base cells of $P$ to a tuple $(\edge,I_{\mathbf l},I_{\mathbf u},p)$,
  where \edge is an edge of \game, $I_{\mathbf l}$ and $I_{\mathbf u}$ are
  two affine delay inequalities that are respectively lower and upper bounds,
  and finally $p\in\{{\mathbf l},{\mathbf u}\}$ selects one of the two.
\end{defi}

Given a precision $\varepsilon>0$, a partition strategy function
encodes a mapping from $c_P$ to pairs $(\edge,d^\varepsilon)$ denoted
$\sem{S}^{\varepsilon}$, with $\edge\in\Edges$ and
$d^\varepsilon\in\Rpos$.  Let $\val\in c_P$ and $S([\val]_{P})$ be
equal to $(\edge,I_{\mathbf l},I_{\mathbf u},p)$, with $I_{\mathbf l}$
defined by
$a_\clockd \clockd \dbmlaop a_1\clock_1+\cdots+a_n\clock_n+b$ and
$I_{\mathbf u}$ defined by
$a_\clockd' \clockd \dbmlaop' a'_1\clock_1+\cdots+a'_n\clock_n+b'$.
Let $D\subseteq\Rpos$ denote the interval of delays $\clockd$ that
satisfy both
$a_\clockd \clockd \dbmlaop
a_1\val(\clock_1)+\cdots+a_n\val(\clock_n)+b$ and
$a_\clockd' \clockd \dbmlaop'
a'_1\val(\clock_1)+\cdots+a'_n\val(\clock_n)+b'$.  We denote by
$d_{\mathbf l}$ and $d_{\mathbf u}$ the lower and upper bounds of $D$.
Let $D^\varepsilon$ denote
$D\cap (d_{\mathbf l}-\varepsilon,d_{\mathbf l}+\varepsilon)$ if
$p={\mathbf l}$, and
$D\cap (d_{\mathbf u}-\varepsilon,d_{\mathbf u}+\varepsilon)$ if
$p={\mathbf u}$.  We denote by $d_{\mathbf l}^\varepsilon$ and
$d_{\mathbf u}^\varepsilon$ the lower and upper bounds of
$D^\varepsilon$.  Then,
$\sem{S}^\varepsilon(\val)=(\edge,d^\varepsilon)$ with
$d^\varepsilon=\frac{d_{\mathbf l}^\varepsilon+d_{\mathbf
    u}^\varepsilon}{2}\in D^\varepsilon$.

We argue that partition strategy functions can be used to compute a
positional strategy $\stratmin^{i,\varepsilon}$ given an encoding
of $\Val^{i-1}_\game$ as partition functions, so
that~\eqref{eq:strategy} holds.

We assume that for all $\loc\in\Locs$, $\ValIteVec_\loc^{i}$ is piecewise affine with
finitely many pieces %
and is \lipconst-Lipschitz-continuous
over regions.
By Proposition~\ref{prop:tree-lipschitz},
we know that \lipconst is at most exponential
in the size of $\game$ and $i$.

\begin{prop}\label{prop:strategyfun}
  For all $\varepsilon>0$ and $i\in\Nspos$, we can compute in time
  doubly-exponential in $i$ and
  exponential in the size of \game a partition $P_\loc^i$ and
  a partition strategy function $S_\loc^i$ for each location $\loc\in\LocsMin\backslash\LocsT$,
  so that
  $\sem{S_\loc^i}^{\frac{\varepsilon}{\lipconst}}(\val)
  =\stratmin^{i,\varepsilon}(\loc,\val)$.
\end{prop}
\begin{proof}
In Section~\ref{sec:acyclic}, we detailed how to compute
piecewise affine value functions $(P_\loc^i,F_\loc^i)$ encoding
$\ValIteVec_\loc^{i}$ for all locations $\loc$ and a fixed horizon $i$.
We will now explain how $S_\loc^i$ can be obtained as a side-product
of the step $\ValIteVec^i=\ValIteOpe(\ValIteVec^{i-1})$.

Recall that if $\loc\in\LocsMin\backslash\LocsT$,
  it holds that
  $$\ValIteVec^i_{\loc}=
    \min_{\edge=(\loc,\guard,\reset,\loc')}
    \Pretime_\edge(\Guard_\guard(\Unreset_\reset(\ValIteVec^{i-1}_{\loc'})))$$
  \noindent where $\edge$
  ranges over the edges in \game that starts from $\loc$.
  We can identify, for every base cell $c$ of the resulting partition $P^i_\loc$,
  an edge $\edge$ so that $\ValIteVec^i_{\loc}$ restricted to domain $c$
  is equal to $\Pretime_\edge(\Guard_\guard(\Unreset_\reset(\ValIteVec^{i-1}_{\loc'})))$.
  This is the edge that $S_\loc^i$ will play on $c$.
  For the choice of delay, let us explain how to obtain the delay inequalities
  $I_{\mathbf l}$ and $I_{\mathbf u}$, and the selection choice $p$,
  from the computation of $\Pretime_\edge$.

  Let $(P,F)$ be an atomic tube partition encoding
  $\Guard_\guard(\Unreset_\reset(\ValIteVec^{i-1}_{\loc'}))$.  If
  $c_b$ is a base cell of $P$, recall that $\mathcal C_{c_b}$ is the
  set of cells reachable from $c_b$ by time elapse, and that if
  $c\in C_{c_b}$, $\mathcal B_{c_b}(c)$ is the set of borders of $c$
  reachable from $c_b$ by time-elapse. In particular,
  $\mathcal B_{c_b}(c)$ contains either one non-diagonal border (a
  single choice of delay reaches $c$ from any $\val\in c_b$) or two
  non-diagonal borders (in which case there is a range of delays
  reaching $c$).  Then, $\Pretime_\edge(\sem{F})$ restricted to domain
  $c_b$ equals
  $\min(\sem{F},\min_{c\in\mathcal C_{c_b}}\min_{B\in\mathcal
    B_{c_b}(c)} [\Pretime_{\edge,c,B}(F)])$.  These minimum operators
  are applied over piecewise affine value functions, and we can
  identify, for every base cell $c'$ of the output, the arguments that
  optimise them.  In~particular, the base cells $c'$ where
  $\Pretime_\edge(\sem{F})=\sem{F}$ correspond to a choice of delay
  zero, so that $S_\loc^i$ maps $c'$ to
  $(\edge,-\clockd \leq 0,\clockd \leq 0,{\mathbf l})$.
  Alternatively, the base cells $c'$ where
  $\Pretime_\edge(\sem{F})=\Pretime_{\edge,c,B}(F)$ for some
  $c\in\mathcal C_{c_b}$ and some $B\in\mathcal B_{c_b}(c)$
  correspond, for all $\val\in c'$, to a choice of delay arbitrarily
  close to $\delay_{\val,B}$, the delay that lets valuation $\val$
  reach $\sem{B}$ by time-elapse.  Recall that if $B$ is described by
  $a_1\clock_1+\cdots+a_n\clock_n +b=0$ and $A=a_1+\cdots+a_n$, then
  $$-A\delay_{\val,B}=a_1\cdot\val(\clock_1)+\cdots+a_n\cdot\val(\clock_n) +b$$
  There are now two cases.  If $\mathcal B_{c_b}(c)=\{B\}$, then
  $S_\loc^i$ maps $c'$ to
  $(\edge,I_{\mathbf l},I_{\mathbf u},{\mathbf l})$, with
  $I_{\mathbf l}$ defined by
  $-A \clockd \leq a_1\clock_1+\cdots+a_n\clock_n+b$ and
  $I_{\mathbf u}$ defined by
  $A \clockd \leq -a_1\clock_1-\cdots-a_n\clock_n-b$.  Otherwise
  $\mathcal B_{c_b}(c)=\{B_{\mathbf l},B_{\mathbf u}\}$, with
  $B_{\mathbf l}$ the lower border and $B_{\mathbf u}$ the upper
  one.\footnote{We can assume without loss of generality that
    $\delay_{\val,B_{\mathbf l}}\leq \delay_{\val,B_{\mathbf u}}$ for
    all $\val\in c_b$ as the tube partition $(P,F)$ is atomic.}  As
  $B_{\mathbf l}$ and $B_{\mathbf u}$ are borders of $c$, they
  correspond to affine inequalities
  $a_1\clock_1+\cdots+a_n\clock_n +b \dbmla 0$ and
  $a'_1\clock_1+\cdots+a'_n\clock_n +b' \dbmla' 0$, respectively.  Let
  $A=a_1+\cdots+a_n$ and $A'=a_1'+\cdots+a'_n$.  We let
  $I_{\mathbf l}$ be defined by
  $-A \clockd \dbmlaop a_1\clock_1+\cdots +a_n\clock_n+b$ and
  $I_{\mathbf u}$ be defined by
  $A' \clockd \dbmlaop' -a'_1\clock_1-\cdots -a'_n\clock_n-b'$.  Then we let
  $S_\loc^i$ map $c'$ to $(\edge,I_{\mathbf l},I_{\mathbf u},p)$,
  where $p={\mathbf l}$ if $B=B_{\mathbf l}$ and $p={\mathbf u}$ if
  $B=B_{\mathbf u}$.

  It follows that $\sem{S_\loc^i}^\varepsilon(\val)$ corresponds to a
  choice of edge \edge and delay \delay from~$(\loc,\val)$ leading to
  a valuation $(\loc',\val')$, so that
  $\|(\val+\delay) - (\val+\delay_{\val,B})\|_{\infty}<\varepsilon$,
  with $\delay_{\val,B}$ the optimal delay for minimizing the
  $\Pretime_\edge$ step.  As the output piecewise affine functions are
  \lipconst-Lipschitz-continuous over regions,
  $|\ValIteVec^i_\loc(\val+\delay)-\ValIteVec^i_\loc(\val+\delay_{\val,B})|<
  \lipconst \varepsilon$.  Then by definition of $\Pretime_\edge$ we
  obtain
  \begin{align*}
  \sem{S_\loc^i}^\varepsilon(\val)\in\arginf_{(\loc,\val)\xrightarrow{\delay,\edge}(\loc',\val')}^{~\lipconst
    \varepsilon}
  \big(\delay\cdot\weight(\loc)+\weight(\edge)+\ValIteVec_{\loc'}^{i-1}(\val')\big) \tag*{\qedhere}
   \end{align*}
 \end{proof}

By Proposition~\ref{prop:optstratvalite}, we can therefore solve the
synthesis problem in triple exponential time for all weighted time games \game so that $\Val_\game=\Val_\game^i$ with $i$ at most exponential in the size of $\game$.
This holds for acyclic games, where $i$ is bounded by $|\rgame|$, and also for divergent games, where $i$
can be bounded by the results of Section~\ref{sec:symbolic}
(Lemmas~\ref{lm:symbolic-scc-plus} and~\ref{lm:symbolic-scc-minus} in
the special case with no kernel), assuming that an exponential time pre-computation using Proposition~\ref{prop:-infty} is used to remove the valuations of value $-\infty$.
This concludes the proof of Theorem~\ref{thm:strategies}.

\section{Weighted timed games with no clocks}\label{sec:solving-wg}

In this section, we study weighted timed games where the set of clocks is empty.
Their semantics is a finite transition system, so we will simply call them
(finite) weighted games.
For notational convenience, we omit guards, resets, valuations and delays
from all notations,
and merge the notions of locations and configurations into the notion of states.
We also assume that all delays are null and that weights are integers.

\begin{defi}
  A \emph{weighted game}\footnote
  {Weighted games are called \emph{min-cost
  reachability games} in \cite{BGHM16}.}
  is a tuple $\game=\struct{\StatesMin,\StatesMax,\StatesT,\Trans,\weight,\weightT}$
  with $\States=\StatesMin\uplus\StatesMax$ a finite
  set of states split between players $\MinPl$ and $\MaxPl$,
  $\Trans\subseteq \States\times \States$ a finite set of transitions
  $\state\rightarrow\state'$ from state $\state$ to
  state $\state'$,
  $\weight\colon \Trans \to \Z$ a weight function
  associating an integer weight with each transition\footnote{If $\trans=\state\rightarrow\state'$ is a transition in \game,
    we denote $\weight(\state,\state')$ the weight $\weight(\trans)$.},
  $\StatesT\subseteq \StatesMin$ a set of target states for player
  $\MinPl$, and $\weightT\colon \StatesT\to \Zbar$ is
  a function mapping each target state to a final weight of
  $\Zbar= \Z\cup\{-\infty,+\infty\}$.
\end{defi}

  A cycle is a finite play $\play$, of length at least $1$,
  such that $\first(\play)=\last(\play)$.

\subsection{Solving weighted games}\label{sec:value_ite_untimed}

  Let $\game$ be a finite weighted game.
  The value iteration algorithm described in Section~\ref{subsec:value-functions} can be applied on $\game$ as a greatest fixpoint computation over a vector of $|\States|$ numbers,
  as detailed in~\cite{BGHM16}.
  This algorithm computes the value of any finite weighted game in pseudo-polynomial time.
  In particular, along the fixpoint computation states with a value of $-\infty$ will be associated with a sequence of values that converges towards $-\infty$. When the value of such a state gets low enough, it is recognised as a $-\infty$ state and is set to its fixpoint value directly.

\subsubsection{Optimal strategies}

Let us fix an initial state $\state$.  By definition of lower and
upper values, there exists for each player $\Pl\in\{\MinPl,\MaxPl\}$ a
sequence of strategies $(\stratpl^i)_{i\in\N}$ such that
$\lim_{i\to\infty}\Val_\game(\state,\stratpl^i) = \Val_\game(\state)$
and such that the sequence $(\Val_\game(\state,\stratmin^i))_{i\in\N}$
is non-increasing over \Zbar, while
$(\Val_\game(\state,\stratmax^i))_{i\in\N}$ is non-decreasing.  If the
sequence $(\Val_\game(\state,\stratpl^i))_{i\in\N}$ stabilizes for all
$i\geq k$, then $\stratpl^k$ is an optimal strategy of player \Pl for
$\state$, \ie~$\Val_\game(\state,\stratpl^k)=\Val_\game(\state)$.
Therefore, if $\Val_\game(\state)>-\infty$ then \MinPl must have an
optimal strategy for \state (an infinite decreasing sequence over \Z
stabilizes), and if $\Val_\game(\state)<+\infty$ then \MaxPl has an
optimal strategy for $\state$.  Moreover, if an optimal strategy for
\Pl exists for all states \state, then they can be combined into an
(overall) optimal strategy for \Pl.

  In fact, there always exists a \emph{positional} strategy $\stratmax^\star$ for \MaxPl
  that is optimal, %
  even if some states have value $+\infty$.
  The strategy $\stratmax^\star$ can be obtained in the value iteration algorithm,
  by memorizing
  for every state $\state\in\StatesMax$ the transition that maximizes
  $\max_{\state\rightarrow\state'}
  \big[\weight(\state,\state')+\ValIteVec_{\state'}\big]$ in the last application of $\ValIteOpe$.
  However, this does not hold for $\MinPl$, as there might be no
  sequence of positional strategies for player $\MinPl$ whose value at \state
  converges towards $\Val_\game(\state)$.
  In \cite{BGHM16}, it is shown that value iteration can also compute an
  optimal strategy for $\MinPl$
  (or a sequence of strategies in the $-\infty$ case),
  by switching between two positional strategies $\stratmin^\star$
  and $\stratmin^\dagger$:
  $\stratmin^\star$ accumulates negative weight by following negative cycles,
  and $\stratmin^\dagger$ ensures reaching a target.
  The optimal strategy of \MinPl follows the decisions of
  $\stratmin^\star$, until switching to the decisions of $\stratmin^\dagger$ when
  the length of the play is greater than a finite bound $k$.
  These strategies thus require finite memory in the form of a counter.

  An interesting case happens if \game has no cycles of negative
  cumulative weight, \eg~if weights are non-negative.

  \begin{lem}\label{lm:no-neg-cycles-untimed}
    If \game has no cycles of negative cumulative weight, then
    both players have optimal strategies that are positional.
    Moreover, $\Val_\game=\Val^{|\States|}_\game$ and
    the optimal strategies can be computed in polynomial time.
  \end{lem}
  \begin{proof}
    Any strategy of \MinPl is optimal on states of value $+\infty$,
    so we will ignore those.
    As there are no negative cycles, value $-\infty$ can only be obtained
    through reaching a target with final weight $-\infty$.
    As a consequence, \MinPl has an optimal strategy $\stratmin$
    that switches between $\stratmin^\star$ and $\stratmin^\dagger$ when
    the length of the play is greater than some $k$, as detailed
    in~\cite{BGHM16}.
    The strategy $\stratmin^\star$ is only compatible with cycles
    of negative cumulative weight, and $k\geq|\States|+1$,
    therefore %
    $\stratmin^\dagger$ is never used.
    Thus, \MinPl has an optimal strategy $\stratmin^\star$ that is positional.
    Then, consider $\Val^{|\States|}_\game$, the value with bounded horizon
    $|\States|$.  For every state \state, we have
    $\Val^{|\States|}_\game(\state,\stratmin^\star)=\Val_\game(\state)$,
    so that $\Val^{|\States|}_\game\leq \Val_\game$.
    As $\Val_\game\leq \Val^{k}_\game$ holds for any $k\ge 0$,
    $\Val_\game=\Val^{|\States|}_\game$.
    Finally, since value iteration converges in $|\States|$ steps, the computation
    of optimal strategies described in~\cite{BGHM16} runs in polynomial time.
  \end{proof}

\subsection{Divergent and almost-divergent weighted games}\label{sec:almost-div-wg}

  Our contribution on finite weighted games is to solve in polynomial time the value
  problem, for a subclass of finite weighted games.
  This class corresponds to the games that are \emph{almost-divergent} when seen as timed games with zero clocks.
  To the best of our knowledge, this is the first
  attempt to solve a non-trivial class of weighted games with arbitrary
  weights in polynomial time.

  Let us first define the class of divergent weighted games in the untimed setting:

\begin{defi}
  A weighted game \game is divergent if every cycle \play of \game
  satisfies $\weightC(\play)\neq 0$.
\end{defi}

  In particular, if $\game$ is seen as a weighted timed game with no clocks, then $\rgame$ is isomorphic to $\game$, so that divergent weighted games are divergent weighted timed games.
  We will obtain the following results:
\begin{thm}\label{thm:div_untimed}
  The value problem over finite divergent weighted games is
  \P-complete. Moreover, deciding if a given finite weighted game is
  divergent is an \NL-complete problem when weights are encoded in unary,
  and is in \P\ when they are encoded in binary.
\end{thm}

  With divergent weighted games, we described a class where the value problem is
  polynomial instead of pseudo-polynomial.  This gain in complexity came at a cost,
  the absence of cycles of weight $0$.
  We argue that some of those cycles can be allowed,
  and apply the almost-divergent notion to the untimed setting.
  \begin{defi}
    A weighted game \game is \emph{almost-divergent} if every 0-cycle \play of \game
    satisfies the following property:
    for every decomposition of $\play$ into smaller cycles $\play'$ and $\play''$,
    $\play'$ and $\play''$ are 0-cycles.
  \end{defi}

  We will obtain the following results, extending Theorem~\ref{thm:div_untimed}:
  \begin{thm}\label{thm:almost-div_untimed}
    The value problem over finite almost-divergent weighted games is
    \P-complete.  Moreover, deciding if a given finite weighted game is
    almost-divergent is an \NL-complete problem when weights are encoded in unary,
    and is in \P\ when they are encoded in binary.
  \end{thm}

  We prove Theorems~\ref{thm:div_untimed} and~\ref{thm:almost-div_untimed}
  at the same time, by proving a \P upper bound for the value problem on almost-divergent game,
  a \P lower bound on divergent games, and finally by studying the decision
  problems associated with membership to the divergent and almost-divergent classes.

\subsubsection{Polynomial upper bound}\label{sec:-infty}

  Our algorithm solving almost-divergent games in polynomial time follows the semi-unfolding scheme, that we detailed
  in the timed setting as a proof of Theorem~\ref{thm:almost-div}.
  The SCC characterisation of almost-divergence and the notion of kernel apply, and the only changes concern the complexity analysis.
  In particular, with $0$ clocks, we have $|\rgame|=|\game|$, therefore computing the kernel and the states of value $+\infty$ or $-\infty$ can be done in polynomial time.
  Moreover, the semi-unfolding of Section~\ref{sec:unfolding} has polynomial depth.

    In order to compute the value of a state $\state_0$ of \game, one could
    construct the semi-unfolding of root $\state_0$, and compute its value.
    Indeed, every cycle in \tgame belongs to \Kernel, so they must all be 0-cycles,
    therefore by Lemma~\ref{lm:no-neg-cycles-untimed} we can compute
    $\Val_{\tgame}(\tilde{\state}_0)$ in time polynomial in the size of \tgame.
    However, this would be an exponential time algorithm, since the number of nodes in $T$
    can be exponential in $|\States|$.
    We argue that this exponential blow-up can be avoided:
    when two nodes of $T$ are at the same depth
    and are labelled by the same state they can be merged, producing a graph $T$
    that is acyclic instead of tree-shaped, with at most quadratically many states.
    This does not change the value of the resulting weighted game $\tgame$ at its root,
    because the two merged nodes had the same sub-tree, and therefore were states
    with the same value in \tgame.
    This optimization on the construction of \tgame is performed on-the-fly,
    while the semi-unfolding is constructed, such that constructing \tgame
    (and solving it by Lemma~\ref{lm:no-neg-cycles-untimed})
    can be done in time polynomial in the size of \game.

\subsubsection{Polynomial lower bound}\label{sec:poly-lower-bound-untimed}

  Let us show that the value problem is \P-hard on divergent weighted games.
  This comes from a reduction (in logarithmic
  space) of the problem of solving finite games with reachability
  objectives \cite{Imm81}.  To a reachability game, we simply set the weight
  of every transition to $1$ and the final weight of every target to $0$,
  making it a divergent weighted game.  Then,
  \MinPl wins the reachability game if and only if the value in the
  weighted game is lower than $|\States|$.
  The same reduction can be used to show the \P-hardness of the $+\infty$ and $-\infty$-value problems on divergent weighted games.

\subsubsection{Deciding divergence and almost-divergence}

  Let us study the \emph{membership problem} for divergent and almost-divergent weighted games,
  \ie~the decision problem that asks if a given weighted game is divergent or almost-divergent,
  and prove the results of Theorems~\ref{thm:div_untimed}
  and~\ref{thm:almost-div_untimed}.

  We start with almost-divergent games, and will rely on the characterization of almost-divergent games in term of SCCs
  given in Proposition~\ref{prop:almost-divergent}.
  First, we note that simple cycles are enough to ensure that
  an SCC is non-negative (\resp~non-positive), providing us with an efficient way to check
  this property:
\begin{lem}\label{lm:simple-cycles-almost-div}
  An SCC~$S$ is non-negative (\resp~non-positive) if and only if every simple cycle
  in $S$ is non-negative (\resp~non-positive).
  Moreover, deciding if an SCC is non-negative (\resp~non-positive) is
  in \NL\ when weights are encoded in unary,
  and is in \P\ when they are encoded in binary.
\end{lem}
\begin{proof}
  The direct implication holds by definition.
  Reciprocally, let us assume that every simple cycle in $S$
  is non-negative (\resp~non-positive),
  and prove that every cycle \play in $S$ is non-negative (\resp~non-positive).
  The cycle \play can be decomposed into simple cycles,
  all belonging to $S$.  Therefore they are all
  non-negative (\resp~non-positive).  As the cumulative weight of $\play$ is
  the sum of the cumulative weights of these simple cycles,
  \play must be non-negative (\resp~non-positive).

  As a corollary, an SCC~$S$ is non-negative (\resp~non-positive) if and only if
  every cycle in $S$, of length at most $|\States|$,
  is non-negative (\resp~non-positive).

  To decide if a strongly connected \game is non-negative (\resp~non-positive),
  we outline two procedures: one is deterministic and will provide the polynomial
  upper bound on time-complexity, the other will guess a logarithmic number of bits and
  provide \NL\ membership.

  The deterministic algorithm proceeds as follows:
  with Floyd-Warshall's algorithm, one can compute the shortest paths
  (\resp~greatest paths) adjacency matrix $M$ in cubic time, such that
  $M(\state,\state')$ contains the minimal (\resp~maximal) value in the finite set
  $\{\weight(\play) \mid \play\text{ simple path from }\state\text{ to }\state'\}$.
  Then, $S$ is non-negative (\resp~non-positive) if and only if for every
  state \state it holds that $M(\state,\state)\geq 0$
  (\resp~$M(\state,\state)\leq 0$).
  If there exists a state \state such that
  $M(\state,\state)<0$ (\resp~$M(\state,\state)>0$), then $S$ is
  not non-negative (\resp~non-positive) as there is a negative (\resp~positive) cycle.
  Conversely, if $M(\state,\state)\geq0$  (\resp~$M(\state,\state)\leq0$), we know that all simple paths from $\state$
  to $\state$ have non-negative (\resp~non-positive) weight.

  Let us now assume that weights are encoded in unary,
  and present a non-deterministic procedure.
  Then, note that a (binary) register containing integer values
  in $[-B,B]$, with $B$ polynomial in $\wmax$ and $|\States|$,
  requires a number of bits at most logarithmic in the size of \game.
  An SCC is \emph{not} non-negative (\resp~not non-positive), if and only if
  it contains a cycle of positive (\resp~negative) cumulative weight,
  of length bounded by $|\States|$.
  We can guess such a cycle $\play$ on-the-fly, keeping in memory its
  cumulative weight (smaller than $B=\wmax\times |\States|$ in absolute value),
  its initial state, and its current length, all in logarithmic space.
  If the length of the cycle exceeds $|\States|$, the guess is invalid.
  Similarly, we can verify that the last state equals the first, and that
  the computed cumulative weight is indeed positive (\resp~negative).
  Therefore, deciding if $S$ is non-negative (\resp~not non-positive) is
  in~$\coNL=\NL$~\cite{Imm88,Sze88}.
  Note that when weights are encoded in binary this procedure only
  gives \coNP\ membership.
\end{proof}

  Let us now explain why the membership problem is an $\NL$-complete problem
  when weights are encoded in unary.
  First, to prove the
  membership in \NL, notice that a weighted game is \emph{not almost-divergent}
  if and only if there is a positive cycle and a negative cycle, both of
  length at most $|\States|$, and belonging to the same SCC.
  This can be tested in $\NL$, using a non-deterministic procedure
  similar to the one from Lemma~\ref{lm:simple-cycles-almost-div}.
  We first guess a starting state for both
  cycles.  Verifying that those are in the same SCC can be done
  in \NL\ by using standard reachability analysis.
  Then, we once again guess the two cycles on-the-fly, keeping in memory their
  cumulative weights in logarithmic space.  Therefore, testing
  divergence is in $\coNL=\NL$ \cite{Imm88,Sze88}.

  The $\NL$-hardness (indeed $\coNL$-hardness, which is equivalent
  \cite{Imm88,Sze88}) is shown by a reduction from the reachability
  problem in a finite graph.  More precisely, we consider a finite
  automaton with a starting state and a different target state without
  outgoing transitions.  We construct from it a weighted game by
  distributing all states to \MinPl, and equipping all transitions with
  weight $0$.  We also add a loop with weight $1$ on the initial state,
  one with weight $-1$ on the target state, and a transition from
  the target state to the initial state with weight $0$.
  Then, the game is not almost-divergent if and only if the target
  can be reached from the initial state in the automaton.

  When weights are encoded in binary, the previous decision procedure
  gives \NP\ membership.  However, we can achieve a \P\ upperbound
  by computing the strongly connected components and then
  using Lemma~\ref{lm:simple-cycles-almost-div} to check that each SCC
  is either non-negative or non-positive.

  This concludes the proofs of Theorem~\ref{thm:almost-div_untimed}.
  For the membership results in Theorem~\ref{thm:div_untimed}, we note that a weighted game \game is
  divergent if and only if it is almost-divergent and its kernel is empty,
  \ie~\game does not contain any simple cycle of weight zero.
  Variants of the previous techniques can be used to either compute the kernel
  in polynomial time, guess a 0-cycle of length at most $|\States|$ in \NL\ when
  weights are encoded in unary, or prove \NL-hardness.

\section{Conclusion}

We have obtained new results for several controller synthesis problems
on timed automata, that we now summarise.
Our study of weighted timed games
belongs to a series of works
that explore the frontier of decidability.
We introduced the first decidable class of weighted
timed games with arbitrary weights and
no restrictions on the number of clocks.
We have given an approximation procedure for a larger
class of weighted timed games, where the exact problem
becomes undecidable.
In addition, we have proved the correctness of a symbolic approximation schema,
that does not start by splitting exponentially every region, but only
does so when necessary.
We argue that this paves
the way towards an implementation of value approximation
for weighted timed games.  Such tool would likely struggle
with instances of moderate size, but could help with the design
and testing of alternative approaches that trade theoretical guarantees
with performance.

Another perspective is to extend this work to the concurrent setting,
where both players play simultaneously and the shortest delay is
selected.
It should be noted that several known results on weighted timed games with
non-negative weights~\cite{BCFL04,AluBer04,BJM15} are stated in such
a concurrent setting.
We did not consider this setting in this work because
concurrent \WTG{s} are not determined, and several of our proofs rely
on this property for symmetrical arguments (mainly to lift results of
non-negative strongly connected components to non-positive ones).

A long-standing open problem is the approximation of weighted timed
games, \ie~whether one can compute an arbitrarily close approximation
of the value of a given game. We successfully solved this problem on
the class of almost-divergent games, but we were not able to extend
further our techniques to more general games. As a first step, we
could try to consider the slightly larger class of 0-isolated games,
where we ask for every cycle of the region game to have a weight
either $\geq 1$, or $\leq -1$, or exactly $0$.  We do not have
approximation results on this 0-isolated class, and as such it forms a
natural intermediate step between the best known decidable class and
the general case. However we must prepare ourselves to possibly
negative answers: the value of a weighted timed game could be
\emph{non approximable}, though we are not aware of any such game; especially we do not even know what kind of real numbers (irrational, or even transcendental) could be obtained as the value in the null valuation. Therefore, pursuing better lower bounds in
various settings could help in the future, in order to close the
remaining complexity gaps.

The divergence and almost-divergence classes that we have studied in this article are independent of the partition of locations into players.
It would be interesting, as future work, to investigate restrictions taking into account the interaction of the two players.

\bibliographystyle{alphaurl}
\bibliography{biblio}

\clearpage
\appendix

\section{Proof of Lemma~\ref{lm:splitnum}}\label{app:proof-splitnum}

  The number $\splitnum(m,n)$ is known in discrete geometry
  as the number of faces obtained when partitioning $\R^n$ by $m$ hyperplanes. We refer to \cite[Chap.~6.1]{Mat02} for more details on this problem, and make use of their vocabulary in this proof, so that cells are called faces and affine equalities are called hyperplanes. The faces are relatively open convex sets, with dimensions ranging from $0$ to $n$. We call $k$-faces the faces of dimension $k$.
  For example, in \figurename~\ref{fig:value-ite-part} the $0$-face is the point at the intersection of the two hyperplanes, the four $1$-faces are half-lines, and the four $2$-faces partition the remaining points of $\R^2$.

  According to \cite[Prop.~6.1.1]{Mat02}, the number of $n$-faces in a space of dimension $n$ is $\sum_{\ell=0}^{\min(n,m)} \binom{m}{\ell}$. These faces are known as faces of full dimension.
  If $n$ or $m$ equals $0$, $\splitnum(m,n)=1$.
  If $n,m\geq 1$, the $(n-1)$-faces
  are obtained in each of our hyperplanes (spaces of dimension $n-1$), by considering their partition by the $m-1$ other hyperplanes. In these partitions, the $(n-1)$-faces are of full dimension,
  so that each hyperplane contains $\sum_{\ell=0}^{\min(n-1,m-1)} \binom{m-1}{\ell}$ such faces.
  The total number of $(n-1)$-faces is therefore $m\sum_{\ell=0}^{\min(n-1,m-1)} \binom{m-1}{\ell}$.
  Let us now count the number of $k$-faces for $k\in[0,n-2]$.

  As explained in \cite[Chap.~6.1]{Mat02}, the number of faces is maximised when the hyperplanes are in general position, \ie~for any $k\in[2,\min(n+1,m)]$, the intersection of any set of $k$ hyperplanes is $(n-k)$-dimensional.\footnote{for $n-k<0$ this means that the intersection is empty.}
  Then, if $n\geq 2$ the faces of dimension $k\in[\max(0,n-m),n-2]$ lie at the intersection
  of $n-k$ hyperplanes, and for each such intersection $L$ the $k$-faces are obtained by partitioning $L$ by the $m-n+k$ other hyperplanes.
  They are faces of full dimension in $L$, so that $L$ contains $\sum_{\ell=0}^{\min(k,m-n+k)} \binom{m-n+k}{\ell}$ such faces.
  There are $\binom{m}{n-k}$ intersections $L$, so that the total number of $k$-faces is $\sum_{\ell=0}^{\min(k,m-n+k)} \binom{m}{n-k} \binom{m-n+k}{\ell}$.
  Note that this formula matches the number of $(n-1)$-faces and the number of $n$-faces previously given,
  and that if $m\leq n$ there are no faces of dimension $k<n-m$.
  Therefore, \[\splitnum(m,n)= \sum_{k=\max(0,n-m)}^{n}\sum_{\ell=0}^{\min(k,m-n+k)} \binom{m}{n-k} \binom{m-n+k}{\ell}\,.\]

  Let us show $\splitnum(m,n)\leq 2^n(m+1)^{n}$.
  First, note that for all $0\leq \ell \leq k \leq n$ with $n-k\leq m$ and $\ell\leq m-n+k$,  %
  it holds that
  $\binom{m}{n-k} \binom{m-n+k}{\ell} = \binom{m}{n-k+\ell} \binom{n-k+\ell}{n-k} \leq \binom{m}{n-k+\ell} \binom{n}{k}$.
  Then, $\splitnum(m,n)\leq \sum_{k=\max(0,n-m)}^{n} \left[\binom{n}{k}\sum_{\ell=0}^{\min(k,m-n+k)} \binom{m}{n-k+\ell}\right]$.
  Moreover, $\sum_{\ell=0}^{\min(k,m-n+k)} \binom{m}{n-k+\ell}\leq \sum_{\ell=0}^{\min(m,n)} \binom{m}{\ell}$.

  Then, note that $\sum_{\ell=0}^{n}\binom{m}{\ell}$ is the number of selections of size at most $n$ out of a set of $m$ elements, without repetitions. It is smaller than the number of selections of $n$ elements, with repetitions,
  in a set of size $m+1$ (the extra element can be used as padding). Therefore, $\sum_{\ell=0}^{\min(n,m)}\binom{m}{\ell}\leq (m+1)^{n}$.
  On the other hand, $\sum_{k=\max(0,n-m)}^{n} \binom{n}{k} \leq \sum_{k=0}^{n} \binom{n}{k} = 2^n$.
  It follows that $\splitnum(m,n)\leq 2^n(m+1)^n$.

\section{Example of an execution of the approximation schema}\label{app:example}

Consider the \WTG \game in \figurename~\ref{fig:example-regions} and some
$\varepsilon>0$.
We want to compute
an $\varepsilon$-approximation of its value in location $\loc_0$
for the valuation $(\clock_1{=}0,\clock_2{=}0)$, denoted $\Val_\game(\loc_0,(0,0))$.
In this example, we will use $\varepsilon = 15$ because the computations would
not be readable with a smaller precision.
Since in this example each location \loc of \game leads to a unique state $(\loc,\reg)$ of \rgame,
we will refer to states of $\rgame$ by their associated location label.
As explained in Example~\ref{ex:example-regions}, \rgame contains one SCC made of two simple cycles,
$\rpath_1=\loc_1\rightarrow\loc_2\rightarrow\loc_1$ is a positive cycle and $\rpath_2=\loc_1\rightarrow\loc_3\rightarrow\loc_4\rightarrow\loc_1$
is a 0-cycle.

Therefore, \rgame only contains non-negative SCCs and is almost-divergent.
Since all states are in the attractor of \MinPl towards $\LocsT$,
all cycles are non-negative and the final weight function is bounded (on all reachable regions),
there are no configurations in \rgame with value $+\infty$ or $-\infty$.

We let the kernel $\Kernel$ be the sub-game of \rgame defined by $\rpath_2$,
and we construct a semi-unfolding \tgame of \rgame of equivalent value.
We should unfold the game until every
stopped branch contains a state seen at least $3|\rgame|\wmaxTimed+2\wmax^\target+2=3\times 3\times 4+2\times 1+2=40$ times.
We will unfold with bound $4$ instead of $40$ for readability (it is enough on this example).
Thus the infinite branch $(\loc_1\loc_2)^\omega$ is stopped when $\loc_1$ is reached for the fourth time,
as depicted in \figurename~\ref{fig:example-unfolding}.

\begin{figure}[ht]
  \centering

\begin{tikzpicture}[node distance=2cm,auto,->,>=latex]
\begin{scope}
  \node[player1] (1) {$\loc_1'$};
  \node[player1] (3) [below right of=1] {$\loc_3'$};
  \node[player2] (4) [below left of=1] {$\loc_4'$};
  \node[] (5) [below right of=3] {$\loc_t'$};
  \node[] (2out) [below left of=4] {$\loc_2'$};
  \node[gray] (0) [above right of=1,xshift=-3mm] {$\loc_0$};
  \node[] (2in) [above left of=1] {$\loc_2$};
  \node[draw,dashed,regular polygon,regular polygon sides=6,inner sep=12mm]  [below of=1,yshift=10mm] {$\Kernel_{\loc_1}'$};

  \path
  (1) edge (3)
  (3) edge (4)
  (4) edge (1)
  (0) edge[dashed,gray,bend right=30] (1)
  (2in) edge[bend left=30] (1)
  (1) edge[bend right=30] (2out)
  (3) edge (5)
  ;
\end{scope}

\path[draw,gray,dashed] (1.7,1) to[bend right=10] (5.4,-0.8);
\path[draw,gray,dashed] (1.7,-3) to[bend left=35] (5.4,-1.2);

\begin{scope}[xshift=7cm,yshift=20mm,every
  node/.style={draw,shape=circle,minimum size=5mm,inner sep=0.3mm},
  level/.style={sibling distance=5cm/#1},level distance=10mm,->]
  \node[rectangle] (0) {$\loc_0$}
  child { node[regular polygon,regular polygon sides=6,inner sep=0.1mm] {$\Kernel_{\loc_1}$}
    child { node[] {$\loc_2$}
      child { node[regular polygon,regular polygon sides=6,inner sep=0.1mm] (kernel) {$\Kernel_{\loc_1}'$}
        child { node[] {$\loc_2'$}
          child { node[regular polygon,regular polygon sides=6,inner sep=0.1mm] {$\Kernel_{\loc_1}''$}
            child { node[] {$\loc_2''$}
              child { node[accepting] (stop) {$\loc_1'''$} }
            }
            child { node[accepting] (out3) {$\loc_t''$} }
          }
        }
        child { node[accepting] (out2) {$\loc_t'$} }
      }
    }
    child { node[accepting] (out1) {$\loc_t$} }
  };
  \node[node distance=5mm,below of=stop,draw=none] () {\smaller[2]{$\weightT(\clock_1,\clock_2)=+\infty$}};
  \node[node distance=5mm,below of=out1,draw=none] () {\smaller[2]{$\weightT(\clock_1,\clock_2)=\clock_1$}};
  \node[node distance=5mm,below of=out2,draw=none] () {\smaller[2]{$\weightT(\clock_1,\clock_2)=\clock_1$}};
  \node[node distance=5mm,below of=out3,xshift=3mm,draw=none] () {\smaller[2]{$\weightT(\clock_1,\clock_2)=\clock_1$}};
\end{scope}
\end{tikzpicture}

  \caption{The kernel $\Kernel$ (with input state $\loc_1$),
  and a semi-unfolding \tgame such that $\Val_\game(\loc_0,(0,0))=\Val_{\tgame}(\loc_0,(0,0))$.
  We denote by $\loc_i$, $\loc_i'$ and $\loc_i''$ the locations in $\Kernel$, $\Kernel'$ and $\Kernel''$.}
  \label{fig:example-unfolding}
\end{figure}
Let us now compute an approximation of $\Val_{\tgame}$.
Let us first remove the states of value $+\infty$: $\loc_1'''$ and $\loc_2''$.
Then, we start at the bottom and compute an $(\varepsilon/3)$-approximation
of the value of $\loc_1''$ in the game defined by
$\Kernel_{\loc_1}''$ and its output edge to $\loc_t''$.
Following Section~\ref{sec:approx-kernels}, we should use $\clockgranu\geq 3(4+1)/\varepsilon$
and compute values in the $1/\clockgranu$-corners game ${\mathcal C}_{\clockgranu}(\Kernel_{\loc_1}'')$ in order to obtain
an $(\varepsilon/3)$-approximation of the value function.
For $\varepsilon=15$ we will use $\clockgranu=1$ (in this case the computation happens
to be exact and would also hold with a small $\varepsilon$).
We construct this corner game, and obtain the finite (untimed) weighted game in \figurename~\ref{fig:example-corner-game}.

\begin{figure}[ht]
  \centering

\begin{tikzpicture}[node distance=3cm,auto,->,>=latex]

  \node[player1,minimum size=10mm](1top){\makebox[0mm][c]{
  \scalebox{0.5}{\begin{tikzpicture}
    \tikzset{font=\LARGE}
    \path[draw,->,thick](0,0) -> (2.3,0) node[above] {$\clock_1$};
    \path[draw,->,thick](0,0) -> (0,2.3) node[right] {$\clock_2$};
    \path[draw,-] (0,0) -- (2,2)
    (1,0) -- (1,2) -- (0,1) -- (2,1) -- (1,0)
    (0,2) -- (2,2) -- (2,0);
    \node () at (1,-0.2) {};%
    \node () at (2,-0.2) {};%
    \node () at (-0.2,-0.2) {};%
    \node () at (-0.2,1) {};%
    \node () at (-0.2,2) {};%
    \path[draw, -, fill=black, line width=1.5mm] (0,0.1) -- (0,0.9);
    \node[draw, fill=black,circle,inner sep=1mm, minimum size=1mm] () at (0,1) {};
  \end{tikzpicture}}}};
  \node()[above of=1top,node distance=12mm]{$c_1'$};

  \node[player1,minimum size=10mm](1bot)[below of=1top,yshift=5mm]{\makebox[0mm][c]{
  \scalebox{0.5}{\begin{tikzpicture}
    \tikzset{font=\LARGE}
    \path[draw,->,thick](0,0) -> (2.3,0) node[above] {$\clock_1$};
    \path[draw,->,thick](0,0) -> (0,2.3) node[right] {$\clock_2$};
    \path[draw,-] (0,0) -- (2,2)
    (1,0) -- (1,2) -- (0,1) -- (2,1) -- (1,0)
    (0,2) -- (2,2) -- (2,0);
    \node () at (1,-0.2) {};%
    \node () at (2,-0.2) {};%
    \node () at (-0.2,-0.2) {};%
    \node () at (-0.2,1) {};%
    \node () at (-0.2,2) {};%
    \path[draw, -, fill=black, line width=1.5mm] (0,0.1) -- (0,0.9);
    \node[draw, fill=black,circle,inner sep=1mm, minimum size=1mm] () at (0,0) {};
  \end{tikzpicture}}}};
  \node()[above of=1bot,node distance=12mm,yshift=-1mm]{$c_1$};

  \node[player1,minimum size=10mm](3top)[right of=1bot]{\makebox[0mm][c]{
  \scalebox{0.5}{\begin{tikzpicture}
    \tikzset{font=\LARGE}
    \path[draw,->,thick](0,0) -> (2.3,0) node[above] {$\clock_1$};
    \path[draw,->,thick](0,0) -> (0,2.3) node[right] {$\clock_2$};
    \path[draw,-] (0,0) -- (2,2)
    (1,0) -- (1,2) -- (0,1) -- (2,1) -- (1,0)
    (0,2) -- (2,2) -- (2,0);
    \node () at (1,-0.2) {};%
    \node () at (2,-0.2) {};%
    \node () at (-0.2,-0.2) {};%
    \node () at (-0.2,1) {};%
    \node () at (-0.2,2) {};%
    \path[draw, -, fill=black, line width=1.5mm] (0.1,0) -- (0.9,0);
    \node[draw, fill=black,circle,inner sep=1mm, minimum size=1mm] () at (1,0) {};
  \end{tikzpicture}}}};
  \node()[below of=3top,node distance=12mm]{$c_3'$};

  \node[player1,minimum size=10mm](3bot)[right of=3top]{\makebox[0mm][c]{
  \scalebox{0.5}{\begin{tikzpicture}
    \tikzset{font=\LARGE}
    \path[draw,->,thick](0,0) -> (2.3,0) node[above] {$\clock_1$};
    \path[draw,->,thick](0,0) -> (0,2.3) node[right] {$\clock_2$};
    \path[draw,-] (0,0) -- (2,2)
    (1,0) -- (1,2) -- (0,1) -- (2,1) -- (1,0)
    (0,2) -- (2,2) -- (2,0);
    \node () at (1,-0.2) {};%
    \node () at (2,-0.2) {};%
    \node () at (-0.2,-0.2) {};%
    \node () at (-0.2,1) {};%
    \node () at (-0.2,2) {};%
    \path[draw, -, fill=black, line width=1.5mm] (0.1,0) -- (0.9,0);
    \node[draw, fill=black,circle,inner sep=1mm, minimum size=1mm] () at (0,0) {};
  \end{tikzpicture}}}};
  \node()[below of=3bot,node distance=12mm]{$c_3$};

  \node[draw,regular polygon,regular polygon sides=4,inner sep=-2mm,minimum size=10mm](4bot)[left of=1bot]{\makebox[0mm][c]{
  \scalebox{0.5}{\begin{tikzpicture}
    \tikzset{font=\LARGE}
    \path[draw,->,thick](0,0) -> (2.3,0) node[above] {$\clock_1$};
    \path[draw,->,thick](0,0) -> (0,2.3) node[right] {$\clock_2$};
    \path[draw,-] (0,0) -- (2,2)
    (1,0) -- (1,2) -- (0,1) -- (2,1) -- (1,0)
    (0,2) -- (2,2) -- (2,0);
    \node () at (1,-0.2) {};%
    \node () at (2,-0.2) {};%
    \node () at (-0.2,-0.2) {};%
    \node () at (-0.2,1) {};%
    \node () at (-0.2,2) {};%
    \path[draw, -, fill=black, line width=1.5mm] (1,0.1) -- (1,0.9);
    \node[draw, fill=black,circle,inner sep=1mm, minimum size=1mm] () at (1,0) {};
  \end{tikzpicture}}}};
  \node()[below of=4bot,node distance=12mm]{$c_4$};

  \node[draw,regular polygon,regular polygon sides=4,inner sep=-2mm,minimum size=10mm](4top)[left of=4bot]{\makebox[0mm][c]{
  \scalebox{0.5}{\begin{tikzpicture}
    \tikzset{font=\LARGE}
    \path[draw,->,thick](0,0) -> (2.3,0) node[above] {$\clock_1$};
    \path[draw,->,thick](0,0) -> (0,2.3) node[right] {$\clock_2$};
    \path[draw,-] (0,0) -- (2,2)
    (1,0) -- (1,2) -- (0,1) -- (2,1) -- (1,0)
    (0,2) -- (2,2) -- (2,0);
    \node () at (1,-0.2) {};%
    \node () at (2,-0.2) {};%
    \node () at (-0.2,-0.2) {};%
    \node () at (-0.2,1) {};%
    \node () at (-0.2,2) {};%
    \path[draw, -, fill=black, line width=1.5mm] (1,0.1) -- (1,0.9);
    \node[draw, fill=black,circle,inner sep=1mm, minimum size=1mm] () at (1,1) {};
  \end{tikzpicture}}}};
  \node()[below of=4top,node distance=12mm]{$c_4'$};

  \node[player1,minimum size=10mm](5top)[right of=1top,xshift=5mm,accepting]{\makebox[0mm][c]{
  \scalebox{0.5}{\begin{tikzpicture}
    \tikzset{font=\LARGE}
    \path[draw,->,thick](0,0) -> (2.3,0) node[above] {$\clock_1$};
    \path[draw,->,thick](0,0) -> (0,2.3) node[right] {$\clock_2$};
    \path[draw,-] (0,0) -- (2,2)
    (1,0) -- (1,2) -- (0,1) -- (2,1) -- (1,0)
    (0,2) -- (2,2) -- (2,0);
    \node () at (1,-0.2) {};%
    \node () at (2,-0.2) {};%
    \node () at (-0.2,-0.2) {};%
    \node () at (-0.2,1) {};%
    \node () at (-0.2,2) {};%
    \path[draw, -, fill=black, line width=1.5mm] (0.1,0) -- (0.9,0);
    \node[draw, fill=black,circle,inner sep=1mm, minimum size=1mm] () at (1,0) {};
  \end{tikzpicture}}}};
  \node()[above of=5top,node distance=12mm]{$c_t', \weightT=\mathbf{1}$};

  \node[player1,minimum size=10mm](5bot)[right of=5top,xshift=-5mm,accepting]{\makebox[0mm][c]{
  \scalebox{0.5}{\begin{tikzpicture}
    \tikzset{font=\LARGE}
    \path[draw,->,thick](0,0) -> (2.3,0) node[above] {$\clock_1$};
    \path[draw,->,thick](0,0) -> (0,2.3) node[right] {$\clock_2$};
    \path[draw,-] (0,0) -- (2,2)
    (1,0) -- (1,2) -- (0,1) -- (2,1) -- (1,0)
    (0,2) -- (2,2) -- (2,0);
    \node () at (1,-0.2) {};%
    \node () at (2,-0.2) {};%
    \node () at (-0.2,-0.2) {};%
    \node () at (-0.2,1) {};%
    \node () at (-0.2,2) {};%
    \path[draw, -, fill=black, line width=1.5mm] (0.1,0) -- (0.9,0);
    \node[draw, fill=black,circle,inner sep=1mm, minimum size=1mm] () at (0,0) {};
  \end{tikzpicture}}}};
  \node()[above of=5bot,node distance=12mm]{$c_t, \weightT=\mathbf{0}$};

  \path
  (1bot) edge node[below]{$\mathbf{2}$} (3top)
  (3top) edge[bend left=35] node[below,yshift=1mm]{$\mathbf{0}$} (4bot)
  (4bot) edge node[below]{$\mathbf{-2}$} (1bot)
  (1top) edge node[above,xshift=-5mm,yshift=-2mm]{$\mathbf{1}$} (3bot)
  (3bot) edge[bend left=30] node[below]{$\mathbf{1}$} (4top)
  (4top) edge[bend left=10] node[above]{$\mathbf{-2}$} (1top)
  (4bot) edge node[below]{$\mathbf{-2}$} (1top)
  (3top) edge[bend right=40] node[right]{$\mathbf{0}$} (5top)
  (3bot) edge[bend right=20] node[right]{$\mathbf{0}$} (5bot)
  ;

\end{tikzpicture}

  \caption{
  The finite weighted game obtained from ${\mathcal C}_{1}(\Kernel_{\loc_1}'')$,
  where $c_i$ and $c_i'$ are the corners of $\loc_i''$ in \tgame.}
  \label{fig:example-corner-game}
\end{figure}

We can compute the values in this game to obtain $\Val(c_1')=1$ and $\Val(c_1)=3$.
We then define a value for every configuration in state $\loc_1''$ by linear interpolation,
obtaining: $$(\clock_1,\clock_2)\mapsto 3-2\clock_2\,.$$
This happens to be exactly $(\clock_1,\clock_2)\mapsto\Val_{\tgame}(\loc_1'',(\clock_1,\clock_2))$
in this case, but would only be an $\varepsilon/3$-approximation of it in general.
Now, we can compute an $\varepsilon/3$-approximation of $\Val_{\tgame}(\loc_2')$
with one step of value iteration, obtaining $$(\clock_1,\clock_2)\mapsto\inf_{0<\delay<2-\clock_1} (-1)\times\delay+1+3-2(0+\delay)=3\clock_1-2\,.$$

The next step is computing an $\varepsilon/3$-approximation of the value of $\loc_1'$ in the game defined by
$\Kernel_{\loc_1}'$ and its output edges to $\loc_t'$ and $\loc_2'$, of
respective final weight functions $(\clock_1,\clock_2)\mapsto \clock_1$ and $(\clock_1,\clock_2)\mapsto 3\clock_1-2$.
This will give us a $2\varepsilon/3$-approximation of $\Val_{\tgame}(\loc_1')$.

Following Section~\ref{sec:approx-kernels} once again, we should use $\clockgranu\geq 3(5+3)/\varepsilon$
and compute values in the $1/\clockgranu$-corners game ${\mathcal C}_{\clockgranu}(\Kernel_{\loc_1}')$.
For $\varepsilon=15$ this gives $\clockgranu=2$ (which will once again keep the computation exact).
We can construct a finite (untimed) weighted game as in \figurename~\ref{fig:example-corner-game},
and obtain a value for each $1/2$-corner of state $\loc_1'$:
\begin{itemize}
\item On the $1/2$-region $(\clock_1=0,0<\clock_2<1/2)$, corner $(0,0)$ has value $2$ and corner $(0,1/2)$
has value $2$.
\item On the $1/2$-region $(\clock_1=0,\clock_2=1/2)$, corner $(0,1/2)$ has value $2$.
\item On the $1/2$-region $(\clock_1=0,1/2<\clock_2<1)$, corner $(0,1/2)$ has value $2$ and corner $(0,1)$
has value $1$.
\end{itemize}
From these results, we define a piecewise affine function by interpolating
the values of corners on each $1/2$-region, and obtain
$$(\clock_1,\clock_2)\mapsto\begin{cases} 2 & \text{if } \clock_2\leq 1/2 \\ 3-2\clock_2 & \text{otherwise}\end{cases}$$
as depicted in \figurename~\ref{fig:example-value}.

\begin{figure}[ht]
  \centering

\begin{tikzpicture}[node distance=3cm,auto,>=latex]

  \draw[->] (0,0) -- (0,2.3);
  \draw[->] (0,0) -- (4.3,0);

  \draw[dashed] (4,0) -- (4,2.3);

  \draw[dashed] (0,1) -- (4.3,1);
  \draw[dashed] (0,2) -- (4.3,2);

  \draw[dashed] (2,0) node[below,xshift=-1mm]{$1/2$}-- (2,2.3);

  \node at (4.3,-.3) {$\clock_2$};
  \node at (-.5,2.5) {$\Val_{\loc_1'}(0,\clock_2)$};

  \node at (0,-.5) {$0$};
  \node at (4,-.5) {$1$};

  \node at (-.5,0) {$0$};
  \node at (-.5,1) {$1$};
  \node at (-.5,2) {$2$};

  \draw[blue,thick] (0,2) -- (2,2) -- (4,1);
  \node () at (0,2) {};
  \node[draw, fill=black,circle,inner sep=0.5mm] () at (0,2) {};
  \node[draw, fill=black,circle,inner sep=0.5mm] () at (2,2) {};
  \node[draw, fill=black,circle,inner sep=0.5mm] () at (4,1) {};

\end{tikzpicture}

  \caption{
  The value function $(\clock_1,\clock_2)\mapsto\Val_{\tgame}(\loc_1',(\clock_1,\clock_2))$, projected on $\clock_1=0$.
  Black dots represent the values obtained for $1/2$-corners using the corner-point
  abstraction.}
  \label{fig:example-value}
\end{figure}

This gives us a $2\varepsilon/3$-approximation of $(\clock_1,\clock_2)\mapsto\Val_{\tgame}(\loc_1',(\clock_1,\clock_2))$
(in fact exactly $\Val_{\tgame}(\loc_1')$).
Now, we can compute a $2\varepsilon/3$-approximation of $\Val_{\tgame}(\loc_2)$
on region $(1<\clock_1<2, \clock_2=0)$ with one step of value iteration,
obtaining :
\[(\clock_1,\clock_2)\to  \inf_{0<\delay<2-\clock_1}
\begin{cases} 3-\delay & \text{if } \delay\leq 1/2 \\
   4-3d & \text{otherwise}
\end{cases} \allowbreak=
\begin{cases} 3\clock_1-2 & \text{if } \clock_1\leq 3/2 \\ \clock_1+1 & \text{otherwise}
\end{cases}\]

Then, we need to compute an $\varepsilon/3$-approximation of the value of $\loc_1$ in the game defined by
$\Kernel_{\loc_1}$ and its output edges to $\loc_t$ and $\loc_2$, of
respective final weight functions $(\clock_1,\clock_2)\mapsto \clock_1$ and
$(\clock_1,\clock_2)\mapsto 3\clock_1-2 \text{ if } \clock_1\leq 3/2; \clock_1+1 \text{ otherwise}$.
This will give us an $\varepsilon$-approximation of $\Val_{\tgame}(\loc_1)$.

Following Section~\ref{sec:approx-kernels} one last time, we should use $\clockgranu\geq 3(5+3)/\varepsilon$
and compute values in the $1/\clockgranu$-corner game ${\mathcal C}_{\clockgranu}(\Kernel_{\loc_1})$.
This time, let us use $\clockgranu=3$ to showcase an example where the computed value is not exact.
We can construct a finite (untimed) weighted game as in \figurename~\ref{fig:example-corner-game},
and obtain a value for each $1/3$-corner of state $\loc_1'$.
From these results, we define a piecewise affine function by interpolation,
as depicted in \figurename~\ref{fig:example-value-approx}.

\begin{figure}[H]
  \centering

\begin{tikzpicture}[node distance=3cm,auto,>=latex]

 \draw[->] (0,0) -- (0,2.3);
 \draw[->] (0,0) -- (4.3,0);

 \draw[dashed] (4,0) -- (4,2.3);

 \draw[dashed] (0,1) -- (4.3,1);
 \draw[dashed] (0,2) -- (4.3,2);

 \draw[dashed] (1.333,0) node[below,xshift=-1mm]{$1/3$}-- (1.333,2.3);
  \draw[dashed] (2.666,0) node[below,xshift=-1mm]{$2/3$}-- (2.666,2.3);

 \node at (4.3,-.3) {$\clock_2$};
 \node at (-.5,2.5) {$\Val_{\loc_1}(0,\clock_2)$};

 \node at (0,-.5) {$0$};
 \node at (4,-.5) {$1$};

 \node at (-.5,0) {$0$};
 \node at (-.5,1) {$1$};
 \node at (-.5,2) {$2$};

 \draw[red,thick] (0,2) -- (2,2) -- (4,1);
 \draw[blue,thick] (0,2) -- (1.333,2) -- (2.666,1.666) -- (4,1);
 \node () at (0,2) {};
 \node[draw, fill=black,circle,inner sep=0.5mm] () at (0,2) {};
 \node[draw, fill=black,circle,inner sep=0.5mm] () at (1.333,2) {};
 \node[draw, fill=black,circle,inner sep=0.5mm] () at (2.666,1.666) {};
 \node[draw, fill=black,circle,inner sep=0.5mm] () at (4,1) {};

\end{tikzpicture}

  \caption{
  The value function $(\clock_1,\clock_2)\mapsto\Val_{\tgame}(\loc_1,(\clock_1,\clock_2))$,
  projected on $\clock_1=0$,
  is depicted in red.
  Black dots represent the values obtained for $1/3$-corners using the corner-points
  abstraction, and the derived approximation of the value function is depicted in blue}
  \label{fig:example-value-approx}
\end{figure}

Finally, from this $\varepsilon$-approximation of $\Val_{\tgame}(\loc_1)$,
we can compute an $\varepsilon$-approximation of $\Val_{\tgame}(\loc_0)$
using one step of value iteration, and conclude.  On our example this
ensures $$\Val_{\tgame}(\loc_0,(0,0))=\sup_{0<\delay<1} \Val_{\tgame}(\loc_1,(0,\delay))
\in[2-\varepsilon,2+\varepsilon]\,.$$

\end{document}